\documentclass[a4paper,11pt]{article}

\pdfoutput=1 
\usepackage{jheppub} 
\usepackage{makeidx}
\usepackage[totoc]{idxlayout}
\usepackage[T1]{fontenc} 
\definecolor{darkgreen}{rgb}{0.0, 0.5, 0.0}
\usepackage{amsmath}
\usepackage{hhline}
\usepackage{multirow}
\usepackage{rotating}
\usepackage{mathdots}
\usepackage{tikz}
\usetikzlibrary{matrix}
\pgfdeclarelayer{background}
\pgfsetlayers{background,main}
\usepackage{slashed}

\usepackage{cancel}

\usepackage{slashed}
\usepackage{url}

\usepackage{amssymb}  
\usepackage{float} 
\usepackage{xcolor} 
\usepackage{graphicx} 
\usepackage{mathdots} 
\usepackage[normalem]{ulem}
\usepackage{bbm}
\usepackage{dsfont}
\usepackage[mathscr]{euscript}
\usepackage{ntheorem}
\usepackage{tensor} 
\usepackage{pdfpages} 

\newcommand{\ptcend}[1]{}

\newcommand{\Qnobeta}{\blue{\mathring Q}}

\newcommand{\zTSlap}{\blue{\mathring \Delta}}

\newcommand{\nohatL}{\blue{\operatorname{L}}}

\newcommand{\Eqref}[1]{\blue{Equation~\eqref{#1}}}

\newcommand{\HSmthree}[1]{\blue{H^{k_{#1}-3}(\secN)}}

\newcommand{\secNp}{{\mcN}}
\newcommand{\CrHk}{\blue{ H^{k_{\gamma};+}_{\secNp}}}
\newcommand{\CrHkm}{\blue{ H^{k_{\gamma}-1;+}_{\secNp}}}
\newcommand{\CrHkmi}{\blue{ H^{k_{\gamma}-2i-1;+}_{\secNp}}}

\newcommand{\CrHkmmi}{\blue{ H^{k_{\gamma}-2i ;+}_{\secNp}}}

\newcommand{\CrHkmoj}{\blue{ H^{k_{\gamma}- j  ;+}_{\secNp}}}
\newcommand{\CrHkmmell}{\blue{ H^{k_{\gamma}-2\ell ;+}_{\secNp}}}
\newcommand{\CrHkmm}{\blue{ H^{k_{\gamma}-2;+}_{\secNp}}}

\newcommand{\HkdM}{\blue H^{k}({\dmanif})}

\newcommand{\Hkng}{\blue H^{k}({\secN})}

\newcommand{\Hkm}{\blue{ H^{k_{\gamma}-1}({\secN})}}
\newcommand{\Hkmi}{\blue{ H^{k_{\gamma}-2i-1}({\secN})}}
\newcommand{\Hkmmi}{\blue{ H^{k_{\gamma}-2i}({\secN})}}
\newcommand{\Hkmmj}{\blue{ H^{k_{\gamma}-2j}({\secN})}}
\newcommand{\Hkmm}{\blue{ H^{k_{\gamma}-2}({\secN})}}
\newcommand{\Hkp}{\blue{ H^{k_{\gamma}+1}({\secN})}}
\newcommand{\Hkpp}{\blue{ H^{k_{\gamma}+2}({\secN})}}
\newcommand{\Hkdt}{\blue{H^{k_{\gamma}}_{\mathrm{Bo}}({\secN})}}
\newcommand{\Hkzeta}{\blue{H^{k_{\gamma}}_{\zeta}({\secN})}}

\newcommand{\MHktwoi}{\blue H^{k_{\gamma}+2i}}
\newcommand{\MHktwoin}{\blue H^{k+n-3}}

\newcommand{\MHk}{\blue H^{k} }
\newcommand{\MHka}{\blue H^{k+a} }
\newcommand{\MHkab}{\blue H^{k+a+b} }
\newcommand{\MHell}{\blue H^{\ell} }
\newcommand{\MHlk}{\blue H^{\ell+k} }

\newcommand{\MHkm}{\blue{ H^{k-1} }}

\newcommand{\MHkp}{\blue{ H^{k+1} }}
\newcommand{\MHkpp}{\blue{ H^{k+2} }}

\newcommand{\MHkpppp}{\blue{ H^{k+4} }}

\newcommand{\MMHktwoi}{\blue{H^{k+2i}(\dmanif)}}
\newcommand{\MMHktwon}{\blue H^{k+2n}({\dmanif}) }
\newcommand{\MMHktwonm}{\blue H^{k+2(n-2)}({\dmanif}) }

\newcommand{\MMHktwoip}{\blue H^{k+2i+1}({\dmanif}) }
\newcommand{\MMHkn}{\blue H^{k+n}({\dmanif}) }
\newcommand{\MMHknm}{\blue H^{k+(n-2)}({\dmanif}) }

\newcommand{\MMHk}{\blue H^{k}({\dmanif})}
\newcommand{\MMHka}{\blue H^{k+a}({\dmanif})}
\newcommand{\MMHkab}{\blue H^{k+a+b}({\dmanif}) }

\newcommand{\MMHlk}{\blue H^{\ell+k}({\dmanif})}

\newcommand{\MMHkp}{\blue{ H^{k+1}({\dmanif})}}
\newcommand{\MMHkpp}{\blue{ H^{k+2}({\dmanif})}}

\newcommand{\kgamma}{\blue{k_\gamma}}

\newcommand{\ofP}{\blue{(\TSzlap,P)}}
\newcommand{\ofPnoP}{\red{\TSzlap,P}}

\newcommand{\ofDC}{\blue{(\TSzlap,\zdivtwo\circ\, C)}}
\newcommand{\ofDCnoDC}{\red{\TSzlap,\zdivtwo\circ\, C}}

{\catcode `\@=11 \global\let\AddToReset=\@addtoreset}
\AddToReset{equation}{section}

{\catcode `\@=11 \global\let\AddToReset=\@addtoreset}
\AddToReset{figure}{section}
\renewcommand{\thefigure}{\thesection.\arabic{figure}}

{\catcode `\@=11 \global\let\AddToReset=\@addtoreset}
\AddToReset{table}{section}
\renewcommand{\thetable}{\thesection.\arabic{table}}

\newcommand{\mrL}{\blue{\!\operatorname{\mathring{\,L}}}}

\newcommand{\zmetric}{\blue{\zzhTBW}}

\newcommand{\bluek}{\blue{k}}
\newcommand{\hak}{\blue{\kgamma}}

\newcommand{\Ck}{\blue{C^k_u\, C^\infty_{(r,x^A)}}}

\newcommand{\zmu}{\blue{\mathring{\mu}}}
\newcommand{\zlambda}{\blue{\mathring{\lambda}}}

\newcommand{\pref}[1]{\ref{#1}, p.~\pageref{#1}}

{\catcode `\@=11 \global\let\AddToReset=\@addtoreset}
\AddToReset{figure}{section}
\renewcommand{\thefigure}{\thesection.\arabic{figure}}

{\catcode `\@=11 \global\let\AddToReset=\@addtoreset}
\AddToReset{table}{section}
\renewcommand{\thetable}{\thesection.\arabic{table}}
\newcommand{\peqref}[1]{\eqref{#1}, p.~\pageref{#1}}

\newcommand{\myGauss}{{\blue{\twoscsign}}}

\newcommand{\Done}{{\purple{\operatorname{L_1}}}}
\newcommand{\Dtwo}{{\purple{\operatorname{L_2}}}}

\newcommand{\myhatopL}{\blue{\operatorname{\widehat L}}}

 \newcommand{\ddim}{{\blue{d}}}
 \newcommand{\dmanif}{\blue{{}^\ddim M}}
 \newcommand{\dmanifold}{\dmanif}
 \newcommand{\dmetric}{{\blue{\mathring{\gamma}}}}
 \newcommand{\dnabla}{\blue{\mathring{D}}}

 \newcommand{\CKV}{{\blue{\mathrm{CKV}}}}
\newcommand{\CKVp}{{\blue{\mathrm{CKV}^\perp}}}
 \newcommand{\KV}{{\blue{\mathrm{KV}}}}

\newcommand{\TSzlap}{\blue{\mathring \Delta}}
\newcommand{\zDelta}{\TSzlap}

\newcommand{\TTt}{{\blue{\mathrm{TT}}}}
\newcommand{\TTtp}{{\blue{\mathrm{TT}^\perp}}}

\newcommand{\zdivtwo}{\blue{\operatorname{{}\mathring{\text div}_{(2)}}}}
\newcommand{\zDivtwo}{\zdivtwo}
\newcommand{\zdivone}{\blue{\operatorname{{}\mathring{\text div}_{(1)}}}}
\newcommand{\zDivone}{\zdivone}

\newcommand{\zdivonedagger}{\blue{\operatorname{(\mathring{\text div}_{(1)})^\dagger}}}
\newcommand{\zdivtwodagger}{\blue{\operatorname{(\mathring{\text div}_{(2)})^\dagger}}}

\newcommand{\Lopdagger}{\red{\operatorname{\blue{L}_n^\dagger}}}
\newcommand{\opL}{\operatorname{L}}

\newcommand{\hLopdagger}{\hLop_n^\dagger}

\newcommand{\Lop}{\red{\operatorname{\blue{L}_n}}}
\newcommand{\Lndagger}{\red{\operatorname{\blue{L}_n^\dagger}}}

\newcommand{\hLop}{\operatorname{\!\!\blue{\check{\,\,L}}}}

\newcommand{\TS}{\mathop{\mathrm{TS}}}

\newcommand{\twoscsign}{{\blue{\varepsilon}}}

\newcommand{\interph}{\blue{v}}

\newcommand{\overadd}[2]
{\overset{ (#1) }{#2}}

\newcommand{\secN}{\blue{\mathbf{S}}}
\newcommand{\secNone}{\blue{\mathbf{S}_1}}

\newcommand{\wh}{\blue{w}}

\newcommand{\red}[1]{{\color{red} #1}}

\newcommand{\bluec}{\color{blue}}
\newcommand{\barh}{\blue{ h }}
\newcommand{\hBo}{\blue{ h^{\mathrm{Bo}}}}

\newcommand{\R}{\mathbb{R}}

\newcommand{\purple}[1]{{\purplec#1}}

\newcommand{\Ipsi}{{\lambda}}

\newcommand{\zhTB}{\blue{\check h}}

\newcommand{\zhTBW}{\blue{\gamma}}
\newcommand{\zzhTBW}{\blue{\mathring{\zhTBW}}}

\newcommand{\mcN}{{\mycal N}}

\newcommand{\sectionofScri}%
{{ \,\,\,\,\mathring{\!\!\!\!\mcN}}}

\newcommand{\mcL}{{\mycal L}}

\newcommand{\ringh}{{  \zzhTBW }}%

\newcommand{\redc}{\color{red}}

\newcommand{\zR}{\blue{\mathring{R}}}

\newcommand{\eq}[1]{(\ref{#1})}

\newcommand{\eeal}[1]{\label{#1}\end{eqnarray}}

\newtheorem{theorem}{\sc  Theorem\rm}[section]

\newtheorem{Theorem}[theorem]{\sc  Theorem\rm}

\newtheorem{corollary}[theorem]{\sc  Corollary\rm}

\newtheorem{lemma}[theorem]{\sc Lemma\rm}
\newtheorem{Lemma}[theorem]{\sc Lemma\rm}

\newtheorem{proposition}[theorem]{\sc Proposition\rm}
\newtheorem{Proposition}[theorem]{\sc Proposition\rm}

\theorembodyfont{}

\newtheorem{remark}[theorem]{\sc Remark\rm}
\newtheorem{Remark}[theorem]{\sc Remark\rm}

\newtheorem{remarks}[theorem]{\sc Remarks\rm}

\DeclareFontFamily{OT1}{rsfs}{}
\DeclareFontShape{OT1}{rsfs}{m}{n}{ <-7> rsfs5 <7-10> rsfs7 <10-> rsfs10}{}
\DeclareMathAlphabet{\mycal}{OT1}{rsfs}{m}{n}

\newcommand{\mnote}[1]
{\protect{\stepcounter{mnotecount}}$^{\mbox{\footnotesize $
\bullet$\themnotecount}}$ \marginpar{
\raggedright\tiny\em $\!\!\!\!\!\!\,\bullet$\themnotecount: #1} }

\newcommand{\zgamma}{\blue{\mathring \gamma}}

\newcommand{\Z}{\mathbbm{Z}}

\newcommand{\blue}[1]{{\color{blue} #1}}

\definecolor{applegreen}{rgb}{0.55, 0.71, 0.0}
\definecolor{armygreen}{rgb}{0.29, 0.33, 0.13}

\definecolor{caribbeangreen}{rgb}{0.0, 0.8, 0.6}

\newcounter{mnotecount}[section]

\renewcommand{\themnotecount}{\thesection.\arabic{mnotecount}}

\newcommand{\mnotex}[1]
{\protect{\stepcounter{mnotecount}}$^{\mbox{\footnotesize
$
\bullet$\themnotecount}}$ \marginpar{
\raggedright\tiny\em
$\!\!\!\!\!\!\,\bullet$\themnotecount: #1} }

\newcommand{\bel}[1]{\begin{equation}\label{#1}}
\newcommand{\bea}{\begin{eqnarray}}
\newcommand{\bean}{\begin{eqnarray}\nonumber}
\newcommand{\beal}[1]{\begin{eqnarray}\label{#1}}
\newcommand{\eea}{\end{eqnarray}}

\newcommand{\nn}{\nonumber}

\newcommand{\qed}{\hfill $\Box$}
\newcommand{\qedskip}{\hfill $\Box$\medskip}
\newcommand{\proof}{\noindent {\sc Proof:\ }}
\newcommand{\be}{\begin{equation}}
\newcommand{\eeq}{\end{equation}}
\newcommand{\ee}{\end{equation}}
\newcommand{\beqa}{\begin{eqnarray}}
\newcommand{\eeqa}{\end{eqnarray}}
\newcommand{\beqan}{\begin{eqnarray*}}
\newcommand{\eeqan}{\end{eqnarray*}}
\newcommand{\ba}{\begin{array}}
\newcommand{\ea}{\end{array}}
\newcommand{\dt}[1]{\blue{\delta\Psi_{\mathrm{Bo}}[#1,\bluek]}}

\newcommand{\const}{\mbox{\rm const}} 

\newcommand{\seccheck}[1]{{\color{darkgreen}\mnotex{reread by ptc up to here on #1}}}

\newcommand{\ptclater}[1]{{\color{darkgreen}\mnotex{ptc , do later: #1}}}
\newcommand{\ptcheck}[1]{{\color{darkgreen}\mnotex{ptc : checked on #1}}}
\newcommand{\wancheck}[1]{{\color{darkgreen}\mnotex{wan : checked on #1}}}
\def\beq{\begin{eqnarray}}
\def\eeq{\end{eqnarray}}

\def\a{\alpha}
\def\b{\beta}

\def\be{\begin{equation}}
\def\ee{\end{equation}}
\def\bea{\begin{eqnarray}}
\def\eea{\end{eqnarray}}

\newcommand{\spaceD}{{ D}}
\newcommand{\zspaceD}{ {\mathring D}}

\newcommand{\mcE}{\mycal E}

\newcommand{\nobarg}{\blue{g}}
\newcommand{\nobarzg}{\blue{\mathring g}}
\newcommand{\zguu}{\blue{\mathring g_{uu}}}

\newcommand{\tmcN}{{\,\,\,\,\,\widetilde{\!\!\!\!\!\mcN}}}

\newcommand{\Lie}{{\mathcal{L}}}

\newcommand{\TSxip}{\blue{\zeta}}
\newcommand{\TSxi}{\blue{\xi}}
\newcommand{\TSr}{\blue{r}}
\newcommand{\TSu}{\blue{u}}

\newcommand{\TSoLie}{\blue{\mathring \Lie}}

\newcommand{\Gmap}{\blue{z^*}}
\newcommand{\gsim}{\blue{\sim_{\mathrm{gauge}}}}

\newcommand{\kerpsi}{\blue{\mu}}

\newcommand{\T}{\mathbb{T}}
\newcommand{\N}{\mathbb{N}}
\newcommand{\sm}{d\mu_{\zzhTBW}}
\newcommand{\ip}[2]{\langle #1, #2\rangle }

\newcommand{\wc}[1]{{\mnote{{\bf wan :}
#1 }}}

\newcommand{\tdu}{\red{u}}
\newcommand{\tdr}{r}
\newcommand{\tdA}{A}
\newcommand{\tdB}{B}

\newcommand{\Hf}{\blue{H}}

\newcommand{\qh}{\blue{\hat q}}

\newcommand{\hkappa}{\hat{\kappa}}

\newcommand{\ochi}{\red{\chi}}

\newcommand{\kphi}[1]{\blue{\overset{[#1]}{\hat\varphi}}{}}

\newcommand{\kphit}[2]{\red{\overset{[#2]}{{\hat\varphi}_{#1}}}{}}

\newcommand{\vphi}[1]{\red{\overset{[#1]}{\varphi}}}
\newcommand{\kxi}[1]{\red{\overset{(#1)}{\xi}}{}}

\newcommand{\ck}[2]{\mathcal{\blue{K}}\left(#1,#2\right)}
\newcommand{\zck}[2]{\mathring{\mathcal{\blue{K}}}\left(#1,#2\right)}
\newcommand{\cka}[1]{\mathcal{\blue{K}}_{[\alpha]}(#1)}
\newcommand{\ckm}[1]{\mathcal{\blue{K}}_{[m]}(#1)}
\newcommand{\hck}[2]{\mathcal{\blue{\tilde K}}\left(#1,#2\right)}

\renewcommand{\kxi}[1]{\red{\overset{(#1)}{\red{\xi}}}{}}

\newcommand{\kzeta}[1]{\red{\overset{(#1)}{\red{\zeta}}}}

\newcommand{\TTtpvec}[1]{\overset{[#1]}{w}}

\newcommand{\kQ}[2]{\red{\overset{[#1]}{\red{Q}}_{#2}}}

\newcommand{\im}{\text{im}}
\newcommand{\tric}{\blue{\mathring{\mathscr{R}}}}

\newcommand\hlight[1]{\tikz[overlay, remember picture,baseline=-\the\dimexpr\fontdimen22\textfont2\relax]\node[rectangle,fill=gray!50,rounded corners,fill opacity = 0.2,text opacity =1] {$#1$};} 

\newcommand{\FGp}[1]{{
\mnote{{
{\bf finnCheck:}
#1} }}}

\title{\boldmath Characteristic Gluing with $\Lambda$: II. Linearised equations in higher dimensions
\protect\footnote{Preprint: UWThPh-2024-2}	
}

 \author{Wan Cong,}
 \author{Piotr T.\ Chru\'sciel,}
 \author{and Finnian Gray}
 \affiliation{University of Vienna, Faculty of Physics
  \\Boltzmanngasse 5, A 1090 Vienna, Austria}
 \emailAdd{wan.cong@univie.ac.at}
 \emailAdd{piotr.chrusciel@univie.ac.at}
 \emailAdd{finnian.gray@univie.ac.at}

\abstract{We prove a gluing theorem for any finite number of derivatives for linearised vacuum gravitational fields in Bondi gauge on  a class of characteristic hypersurfaces in  static vacuum $(n+1)$-dimensional backgrounds with cosmological constant $ \Lambda  \in \R$, $n\ge 4$.   This provides the key step for a full nonlinear characteristic gluing for the vacuum Einstein equations near the family of metrics considered. Our work extends, in the linearised case, the pioneering analysis of Aretakis, Czimek and Rodnianski, carried-out for two derivatives on light cones in four-dimensional Minkowski spacetime, as well as our previous work on four-dimensional spacetimes.
}

\makeindex

\renewcommand{\redc}{}
\renewcommand{\bluec}{}
\renewcommand{\purple}[1]{#1}
\renewcommand{\blue}[1]{#1}
\renewcommand{\red}[1]{#1}

\renewcommand{\FGp}[1]{}
\renewcommand{\wancheck}[1]{}
\renewcommand{\wc}[1]{}
\renewcommand{\ptcheck}[1]{}
\renewcommand{\ptclater}[1]{}
\renewcommand{\checkmark}{}
\renewcommand{\seccheck}[1]{}
\renewcommand{\hBo}{h}

\begin{document}
\maketitle
\flushbottom
\index{psi@$\overadd{i}{\psi}\ofP$!$\overadd{i}{\tilde\psi}\ofDC$|see{U@$\tilde{U}\ofDC$}}%
\index{chi@$\overadd{i}{\chi}\ofDC$!$\overadd{i}{\tilde\chi}\ofDC$|see{U@$\tilde{U}\ofDC$}}%
\index{K@$\ck{k}{\ofPnoP}$!$\hck{k}{\TSzlap,\zdivtwo\circ\, C}$|see{U@$\tilde{U}\ofDC$}}%

\index{psi@$\overadd{i}{\psi}\ofP$!$\overadd{i}{\tilde\psi}\ofDC$|see{$\tilde{U}\ofDC$}}%
\index{chi@$\overadd{i}{\chi}\ofDC$!$\overadd{i}{\tilde\chi}\ofDC$|see{$\tilde{U}\ofDC$}}%
\index{K@$\ck{k}{\ofPnoP}$!$\hck{k}{\TSzlap,\zdivtwo\circ\, C}$|see{$\tilde{U}\ofDC$}}%

\ptclater{tilde K does not show in the index;tilde S and N I are together?; see command in the index does not work...}

\ptclater{do unobstructed gluing to perturbed data in  Minkowski on a toroidal flat null hypersurface?}

\ptclater{the higher dimensional nonlinear version and the seesaw gluing could provide optimal decay on hyperboloids? in asymptotically Euclidean already taken care of by Mao and Tao~\cite{MaoTao}}

\ptclater{perhaps relevant to Le Pengyu on Null Penrose, to check; also prove linearised Null Penrose in this context?}
\section{Introduction}

Mathematical general relativity is  concerned with  mathematical properties of physically relevant solutions of Einstein equations. The ultimate goal is to understand  all  key features of physically relevant spacetimes. This requires the ability to construct  generic solutions,  as well as special solutions exhibiting significant features.
Because of the nature of the Einstein equations, solutions  can be constructed by solving various Cauchy problems. The celebrated work of Yvonne Four\`es-Bruhat~\cite{ChBActa} provided the first key tool for this, by showing well posedness of the general relativistic spacelike Cauchy problem. An alternative is provided by the characteristic Cauchy problem~\cite{RendallCIVP} (compare \cite{ChPaetz,CCG}), 
which provides a construction of spacetimes by evolving   characteristic data. The importance and usefulness of this Cauchy problem has been growing, starting with Christodoulou's proof of cosmic censorship for spherically symmetric gravitating scalar fields~\cite{ChristodoulouAnnals99}, followed   by  Christodoulou's monumental treatise on trapped-surface formation~\cite{ChristodoulouBHFormation}.

An important new insight into the characteristic Cauchy problem has been provided recently by Aretakis, Czimek and Rodnianski~\cite{ACR2,ACR3}, who showed how to glue together characteristic initial data sets, together with two transverse derivatives,  near a four-dimensional Minkowskian light cone. This led to significant improvements~\cite{ACR3,CzimekRodnianski} in previous gluing constructions for spacelike Cauchy data, and to the construction of dynamical black hole spacetimes~\cite{KehleUnger2,KehleUnger1,KehleUnger3} with interesting properties.

 A serious shortcoming of these constructions is the gluing of only two transverse derivatives on the characteristic hypersurface.  This leads to poor differentiability of the resulting spacetimes, see e.g.\ \cite{ChristodoulouMzH,MzHAIHP} or \cite{LukRodnianski} in a small data setting. 
In a recent work~\cite{ChCong1} we showed how to do the linearised gluing with any number of transverse derivatives in the setting of \cite{ACR2,ACR3}. We further included a cosmological constant in the analysis, and  allowed for non-spherical section of the characteristic hypersurface.

The aim of this work is to extend these results  to all higher dimensions.

The analysis here provides the key step to a full nonlinear gluing, which is carried-out in the accompanying paper~\cite{ChCongGray2}.

The work here differs from that in~\cite{ChCong1} in several respects. First, the constraint equation for $\partial_u h_{AB}$
(Equation~\eqref{3IX22.1HD} below),
contains a term  which vanishes when $n=3$, with the new term forcing a different approach. Next, the kernels of some operators, such as the divergence operator acting on trace-free two-covariant symmetric tensors,  are now infinite dimensional, which requires considerably more care. Last but not least,
 there is a qualitative difference
in the analysis when $n>3$,
 in that two separate cases arise:
\begin{enumerate}
  \item
  either $n$ is even, or $n$ is odd and the number of $u$-derivatives to be glued is strictly less than $\frac{n-3}{2}$; or
  \item
 $n$ is odd, and the number of $u$-derivatives to be glued exceeds the threshold $\frac{n-3}{2}$.
\end{enumerate}
In the former case,
 the constraint equations involving $\partial_u^ih_{uA}$ and $\partial_u^\ell h_{AB}$ are decoupled, in the sense that they involve different free ``gluing fields'', $\vphi{j}$ (roughly speaking these correspond to $\int_{r_1}^{r_2} s^{-j} h_{AB} ds$, see \eqref{27VII22.1a+}) --- the indices $j$ for $\partial_u^ih_{uA}$ are always integers,
while those for $\partial_u^ih_{AB}$ are always half-integers. In case 2., for general values of the background mass parameter $m$, the equations involving $\partial_u^ih_{uA}$ and $\partial_u^\ell h_{AB}$ are coupled, calling for a more involved strategy as compared to case 1.

To make things precise, we consider the linearisation of the vacuum Einstein equations at a metric
\be
{\nobarg}  = - \big(\twoscsign
 -\frac{2\Lambda r^2 }{n(n-1)}  -\frac{2m}{r^{n-2}}
  \big) du^2-2du \, dr
 + r^2 \ringh_{AB}dx^A dx^B
   \,,
    \quad
    n>3
    \,,
 \quad
   m \in \R
  \,,
   \label{23VII22.3intro}
\ee
where
$\ringh_{AB}dx^A dx^B$ is a $u$- and $r$-independent \emph{Einstein metric} with scalar curvature equal to $(n-1)(n-2)\twoscsign$, with $\twoscsign\in \{0,\pm1\}$.
As in~\cite{ChCong1}, the question addressed here is the following: given two smooth linearised solutions of the vacuum Einstein equations defined near the null hypersurfaces $\{u=0\,,\ r<r_1\}$ and  $\{u=0\,,\ r > r_2\}$, where $r_2>r_1$, can we find characteristic initial data on the missing region $\{u=0\,,\ r_1\le r\le r_2\}$ which,
when evolved to a solution of the linearised Einstein equations,  provide a linearised metric perturbation which coincides on $\{u=0\}$, together with $u$-derivatives up to order $k$, with the original data?
We refer to this construction as the $\Ck$-gluing.
 
 Our results can be summarised by the following theorem:
 
\begin{Theorem}
 \label{T25XI23.1}
A  $\Ck$-linearised vacuum data set on $\mcN_{(r_0,r_1]}$ can be smoothly glued to another  such set on $\mcN_{[r_2,r_3)}$ up to gauge if and only if the obstructions listed in  
Tables~\ref{T11III23.2}-~\ref{T11III23.1} are satisfied.
\end{Theorem}
\begin{remark}
 \label{R19VI24.1}
In \cite[Table~4.1]{ChCong1}, we have similarly listed the obstructions in four spacetime dimensions and  $m=0$. However,
 some of
the obstructions there were not \textit{gauge-invariant},
 as needed for the non-linear analysis in \cite{ChCongGray2}. Table~\ref{T11III23.1} provides a complete list of gauge-invariant obstructing radial charges which also applies for $n=3$; 
 see Appendix~\ref{A19VI24.1}.

    \qed
\end{remark}
\begin{table}[t]
\small
\hspace{-2cm}
 \begin{tabular}{|p{.1cm}||c|c|c|c||}
  \hhline{~|====|}
  \multicolumn{1}{c||}{}
    &
    & Gluing field & Gauge field & Obstruction
\\
  \hhline{~|----|}
  \multicolumn{1}{c||}{}
    &$h_{AB} $
        &  $\interph_{AB} $
                & -
                     & -
\\
\hhline{~|----|}
  \multicolumn{1}{c||}{}
    &$\partial^i_u\tilde\hBo_{ur}\,,\  i\ge 0$
        & -
                &   $\partial_u^{i+1}\kxi{1}^{u }$
                and $\partial_u^{i+1}\kxi{2}^{u }$
                     & -
  \\
  \hhline{~|----|}
  \multicolumn{1}{c||}{}
    &$\tilde\hBo_{uu}$
        &   $\kphi{5-n}_{AB}^{[(\ker \mathring L)^\perp]}$
                &   only on $S^{n-1}$: $(\zspaceD^A\kxi{2}_A)^{[=1]}$
                     &     $\kQ{2}{}(\lambda^{[1]})$
  \\
\hhline{~|----|}
  \multicolumn{1}{c||}{}
     &$\partial_r \tilde\hBo_{uA}$
        &     $\kphi{4-n}_{AB}^{[\TTt^\perp]}$
            &
        only on $S^{n-1}$:  $(\kxi{2}^u)^{[=1]}$
                &     $\kQ{1}{}(\pi^A)$
                \\
                \multicolumn{1}{c||}{}
                &
                 &
                    &
                        &  $\pi^A$ -- KV of $\secN$
\\
  \hhline{~|----|}
  \multicolumn{1}{c||}{}
    &$\tilde\hBo_{uA}^{[\CKVp]}$
        & $\kphi{4}{}^{[\TTtp]}_{AB}$
            & -
                & -
\\
\hhline{~|----|}
  \multicolumn{1}{c||}{}
    &$\{\partial^p_u\tilde\hBo_{uA}^{[\CKV]}\}$, $0\leq p \leq k$
        &   -
            & $\{\partial^{p+1}_u\kxi{2}{}_A^{[\CKV]}\}$
                & -
\\
\hline
\multirow{4}{0pt}{
\begin{sideways}\textbf{Convenient}
\end{sideways}
}
    & $\{\partial_u^p\tilde\hBo_{AB}$, $1\leq p \leq k\}$
        &
            &
                &
\\
    & $\alpha = 0$
        & $\{\kphi{j}{}_{AB}\}_{j\in[\frac{9-n}{2}\,,\frac{7-n+2k(n-1)}{2}]}$
            &  -
                & -
\\
    & $\alpha \neq 0$
        & the above fields or
            &
                &
\\
    &
        & \hspace{-1.5cm} $\{\kphi{j}{}_{AB}\}_{j\in[\frac{7-n-2k}{2}\,,\frac{7-n-2}{2}]}$
            &  -
                & -
\\
    &
        & \hspace{1cm} ${}_{\cup[\frac{9-n+2k}{2},\frac{7-n+2k(n-1)}{2}]}$
            &
                &
\\
  \hhline{~|----|}
   & $\{\partial_u^p\tilde\hBo_{uA}^{[\CKVp]}$,  $1\leq p \leq k\}$
        &  $\{\kphi{j}_{AB}^{[\TTtp]}\}_{j\in[5\,,k(n-1)+4]}$
            &
            -
                & -
\\
  &&&&
\\
\hline
\multirow{4}{0pt}{
\begin{sideways}\textbf{Inconvenient\!\!\!}
\end{sideways}
}
 &&&&\\
    & $\{\partial_u^p\tilde\hBo{}^{[\TTtp]}_{AB},
    \partial_u^p\tilde\hBo{}^{[\CKVp]}_{uA} \,,
    1\leq p \leq k \}$
        & $\{\kphi{j}{}_{AB}^{[\TTtp]}\}_{j\in[\frac{9-n}{2}\,,1]\cup[5,4+k(n-1)]}$
            &  $\{\partial_u^j\kxi{2}{}^{[\CKVp]}_A\}_{j\in [0,k_n]}$%
                &    %
\\
    &
        & and $\kphi{2}{}_{AB}^{[V]}$
            &  and $(\zspaceD_A\kxi{2}^u)^{[\CKVp]}$
                & -
\\
  \hhline{~|----|}
  & $\partial_u^p\tilde\hBo{}^{[\TTt]}_{AB},
    1\leq p \leq k $
        & $\kphit{AB}{\frac{7-n+2p(n-1)}{2}}^{[\ker\overadd{p}{\psi}\cap\TTt]}$
            &  -
                &   - %
\\
  &
        & and $\kphit{AB}{\frac{7-n+2p}{2}}^{[(\ker\overadd{p}{\psi})^\perp\cap\TTt]}$
            &
                &

\\
  \hline
  \multicolumn{1}{c||}{}
   & $\partial_u^p\tilde\hBo_{uu}$,  $1\leq p \leq k$
        & -
            & -
                & -
\\
  \hhline{~|----|}
  \multicolumn{1}{c||}{}
   & $\partial_u^p\partial_r\tilde\hBo_{uA}$,  $1\leq p \leq k$
        & -
            & -
                & -
  \\
  \hhline{~|----|}
  \hhline{~|----|}
\end{tabular}
  \caption{Fields used for gluing when $m\ne 0$ for $(n,k)$ convenient or, $m\neq 0$, $\alpha=0$ for $(n,k)$ inconvenient. The gluing of the fields in the last two lines follows from Bianchi identities.
  A pair $(n,k)$ is deemed convenient if $n$ is even or if $n$ is odd and $k\le (n-5)/2$.
  The fields $\kxi{a}^u$ and $\kxi{a}^A$ are ``gauge fields'' at $\secN_a$, cf.~\eqref{3XII19.t1HD}-\eqref{3XII19.t2HD}.
  The operator $\mrL$ is defined in \eqref{28XI23.p1}.
 The fields $\tilde\hBo_{\mu\nu}$ are the gauge-transformed
  $ \hBo_{\mu\nu}$ fields,
 using the gauge fields $\protect\kxi{1}$
 for $r\le r_1$ and $\protect\kxi{2}{}$ for $r\ge r_2$.
 The fields $\interph_{AB}$ and $\protect\kphi{p}_{AB}$
  are defined  in \eqref{16III22.2old} and \eqref{27VII22.1a}-\eqref{27VII22.1a+};
    $\TTt$ denotes the space of
  transverse-traceless symmetric two-tensors; $V$ is the space of tensors which are Lie derivatives
  of the metric with respect
  to divergence-free vectors;
  $\KV$ denotes the space of Killing vectors and $\CKV$ that of conformal Killing vectors;
  projections such as
  $(\cdot)^{[\TTt^\perp]}$ or $(\cdot)^{[\CKV]}$
  are defined in Section~\ref{s28XI23.1};
  the radial charges  $\protect\kQ{a}{}$, $a=1,2$, are defined in \eqref{24VII22.4} and \eqref{20VII22.1}; the integer $k_n$ is defined in \eqref{2X23.1}.
  }
  \label{T11III23.2}
\end{table} %
\begin{table}[t]
\small
\hspace{-1.6cm}
 \begin{tabular}{||c|c|c|c||}
   \hline
    & Gluing field & Gauge field & Obstruction
\\
   \hline
     $h_{AB}$
        &  $\interph_{AB} $
                & -
                     & -
\\
 \hline
    $\partial^i_u\tilde\hBo_{ur}\,,\  i\ge 0$
        & -
                &   $\partial_u^{i+1}\kxi{1}^{u }$
                and $\partial_u^{i+1}\kxi{2}^{u }$
                     & -
  \\
   \hline
    $\tilde\hBo_{uu}^{[(\im \mrL)^\perp]}$
        &  -
                &   only on $S^{n-1}$: $(\zspaceD^A\kxi{2}_A)^{[=1]}$
                     &     $\kQ{2}{}(\lambda^{[1]})$
  \\
 \hline
     $\partial_r \tilde\hBo_{uA}^{[\CKV]}$ 
        &     -
            &
        only on $S^{n-1}$:  $(\kxi{2}^u)^{[=1]}$
                &     $\kQ{1}{}(\pi^A)$
                \\
                 & 
                    & 
                        &  $\pi^A$ -- KV of $\secN$
\\
 \hline
    $\{\partial^p_u\tilde\hBo_{uA}^{[\CKV]}\}$, $0\leq p \leq k$
        &   -
            & $\{\partial^{p+1}_u\kxi{2}{}_A^{[\CKV]}\}$
                & -
\\
\hline
     $\{\tilde h_{uu}\,,\partial_r\tilde h_{uA}\,, \partial_u^p\tilde\hBo{}^{[\TTtp]}_{AB},
   $
        & $\{\kphi{j}{}_{AB}^{[\TTtp]}\,, j\in[\max\{\frac{7-n-2k}{2},4-n\}\,,1]$
            &  $\{\partial_u^j\kxi{2}{}^{[\CKVp]}_A\}_{j\in [0,k_n]}$%
                &    %
\\
    \hspace{1cm} $\partial_u^p\tilde\hBo{}^{[\CKVp]}_{uA} \,,
    1\leq p \leq k \}$
        & $\cup[4,4+k(n-1)]\}$ and $\kphi{2}{}_{AB}^{[V]}$
            &  and $(\zspaceD_A\kxi{2}^u)^{[\CKVp]}$
                & -
\\
  \hline
   $
  \{\partial_u^p\tilde\hBo{}^{[\TTt]}_{AB},
    1\leq p \leq k \}$
        & $\{\kphi{j}_{AB}^{[\TTt]}\}_{j \in \left[\max\{\frac{7-n-2k}{2},4-n\},
\frac{7-n+2k(n-1)}{2}\right]}$
            &  -
                &   - %
\\
  \hline
    $\partial_u^p\tilde\hBo_{uu}$,  $1\leq p \leq k$
        & -
            & -
                & -
\\
   \hline
    $\partial_u^p\partial_r\tilde\hBo_{uA}$,  $1\leq p \leq k$
        & -
            & -
                & -
   \\
   \hline
   \hline
\end{tabular}
  \caption{Fields for gluing in the $(n,k)$ inconvenient, $m\alpha\ne 0$ case. The notation, and the last two lines, are as in Tables~\ref{T11III23.2}-\ref{T11III23.1}.}
  \label{T11XII23.1}
\end{table}%
%
\begin{table}[t]
\footnotesize
 \begin{tabular}{|p{.1cm}||c|c|c|c||}
  \hhline{~|====|}
  \multicolumn{1}{c||}{}
    &
    & Gluing field & Gauge field & Obstruction
\\
  \hhline{~|----|}
  \multicolumn{1}{c||}{}
    &$h_{AB} $
        & $\interph_{AB} $
                & -
                     & -
  \\
  \hhline{~|----|}
  \multicolumn{1}{c||}{}
    &$\partial^i_u\tilde\hBo_{ur}\,,\  i\ge 0$
        & -
                &  $\partial_u^{i+1}\kxi{1}^{u }$
                and $\partial_u^{i+1}\kxi{2}^{u }$
                     & -
  \\
  \hhline{~|----|}
  \multicolumn{1}{c||}{}
    &$\tilde\hBo_{uu}$
        &  $\kphi{5-n}{}{}^{[(\ker \mathring L)^\perp]}_{AB}$
                &  -
                     &     $\kQ{2}{}$
  \\
  \hhline{~|----|}
  \multicolumn{1}{c||}{}
    &$\partial_r \tilde\hBo_{uA}$
        &   $\kphi{4-n}_{AB}^{[\TTtp]}$
            & -
                &  $\kQ{1}{}(\pi^A)$
 \\
  \multicolumn{1}{c||}{}
    &
         &
             &
                & $\pi^A\in\CKV$
\\
  \hhline{~|----|}
  \multicolumn{1}{c||}{}
    &$\tilde\hBo_{uA}$
        &  $\kphi{4}_{AB}^{[\TTtp]}$
             & $\partial_u\kxi{2}{}_A^{[\CKV]}$
                 & -
\\
\hline
\multirow{5}{0pt}{
\begin{sideways}\textbf{$(n,k)$ convenient \!\!\!\!\!\!\!\!}
\end{sideways}
}
    & $\partial_u^p\tilde\hBo_{AB}$, $1\leq p \leq k$
        &
            &   %
                &    %
\\
    &$\alpha = 0$
        & $\kphit{AB}{(7-n+2p)/2}^{[(\ker\overadd{p}{\psi} )^\perp]}$
            &  -
                & $ \overadd{p}{q}{}_{AB}^{[\ker(\overadd{p}{\psi}  )]}$
\\
    & $\alpha \neq 0$
        & {\small the above field and}
            &
                &
\\
    &
        & $\kphit{AB}{(7-n-2p)/2}^{[(\ker\overadd{p}{\psi})]}$
            &  -
                & -
\\
  \hhline{~|----|}
   & $\partial_u^p\tilde\hBo_{uA}$,  $1\leq p \leq k$
        &  $\kphi{p+4}_{AB}^{[(\ker( \overadd{p}{\chi}\circ\, C  ))^\perp]}$
            & 
            $\partial_u^{p+1} \kxi{2}_A^{[\ker(\overadd{p}{\chi}   \circ\, C)]}$
                & -
\\
\hline
\multirow{8}{0pt}{
\begin{sideways}\textbf{$(n,k)$ inconvenient}
\end{sideways}
}
    & $\partial_u^p\tilde\hBo_{AB}$, $1\leq p \leq \frac{n-5}{2}$
        &
            &   %
                &    %
\\
    &$\alpha = 0$
        & $\kphit{AB}{(7-n+2p)/2}^{[(\ker\overadd{p}{\psi} )^\perp]}$
            &  -
                & $ \overadd{p}{q}{}_{AB}^{[\ker(\overadd{p}{\psi}  )]}$
\\
    & $\alpha \neq 0$
        & {\small the above field and}
            &
                &
\\
    &
        & $\kphit{AB}{(7-n-2p)/2}^{[(\ker\overadd{p}{\psi})]}$
            &  -
                & -
\\
  \hhline{~|----|}
   & $\partial_u^p\tilde\hBo_{AB}$, $\frac{n-3}{2}\leq p \leq k$
        &
            & %
                & %
\\
    & $p=\frac{n-3}{2}$,
    $\alpha = 0$
        &  $\kphi{2}_{AB}^{[(\ker \overadd{\frac{n-3}{2}}{\psi}  )^\perp]}$
            &
            $(\kxi{2}^u)^{[(\ker \hLop_n)^{\perp}]}$
                &  $\overadd{\frac{n-3}{2}}{q}{}_{AB}^{[\ker \overadd{\frac{n-3}{2}}{\psi}   \cap (\ker \hLop_n^\dagger)]}$
\\
    & \phantom{$p=\frac{n-3}{2}$,}
    $\alpha \neq 0$
        & {\small the above field and}
            &
                &
\\
    &
        & $\kphi{5-n}_{AB}^{[\ker\overadd{\frac{n-3}2}{\psi}\cap \ker\mrL]}$
            & $(\kxi{2}^u)^{[(\ker \mrL\circ\hLop_n)^{\perp}]}$
                &  $\kQ{3}{}{}^{[(\im (\mrL\circ \hLop_n))^\perp]} $
\\
   & $p=\frac{n-1}{2}$,
   $\alpha = 0$
         & $\kphi{3}_{AB}^{[(\ker \overadd{\frac{n-1}{2}}{\psi}  )^\perp]}$
            & $(\kxi{2}^A)^{[(\ker \Lop)^{\perp}]}$
                &  
                $\overadd{\frac{n-1}{2}}{q}{}_{AB}^{[(\im \Lop)^\perp) \cap \ker \overadd{\frac{n-1}{2}}{\psi}  ]}$
\\
   & \phantom{$p=\frac{n-1}{2}$,}
   $\alpha \neq 0$
         & $\kphi{4-n}_{AB}^{[\TTt]}$
            & $(\kxi{2}^A)^{[(\ker \zdivtwo \circ \Lop)^{\perp}]}$
                &   $\kQ{4}{}{}^{[(\im (\zdivtwo \circ \Lop))^\perp]}$
\\
   & $p=\frac{n+1}{2} + j$, $j\geq 0$
         & $\kphi{4+j}_{AB}^{[(\ker \overadd{\frac{n+1}{2}+j}{\psi}  )^\perp]}$
            & $(\partial_u^{j+1}\kxi{2}^A)^{[(\ker \Lop)^{\perp}]}$
                &  
                $\overadd{p}{q}{}_{AB}^{[ (\im \Lop)^\perp) \cap  \ker \overadd{p}{\psi}  ]}$
\\
  \hhline{~|----|}
   & $\partial_u^p\tilde\hBo_{uA}$,  
        & 
            & 
                & 
\\
    &$1\leq p \leq k - \frac{n+1}{2}$
        & $\kphi{p+4}_{AB}^{[(\ker( \zspaceD^A\overadd{p}{\chi}  ))^\perp]}$
            &$(\partial_u^{p+1}\kxi{2}^A)^{[\CKV]}$
                &$\myGauss>0$ only: $\kQ{5,p}{}{}_B$
\\
    &$k - \frac{n+1}{2}< p \leq k $
        & $\kphi{p+4}_{AB}^{[(\ker( \zspaceD^A\overadd{p}{\chi}  ))^\perp]}$
            &$(\partial_u^{p+1}\kxi{2}^A)^{[\ker( \overadd{p}\chi\circ\, C)]}$
                & -
            
\\
  \hline
  \multicolumn{1}{c||}{}
   & $\partial_u^p\tilde\hBo_{uu}$,  $1\leq p \leq k$
        & -
            & -
                & -
\\
  \hhline{~|----|}
  \multicolumn{1}{c||}{}
   & $\partial_u^p\partial_r\tilde\hBo_{uA}$,  $1\leq p \leq k$
        & -
            & -
                & -
  \\
  \hhline{~|----|}
  \hhline{~|----|}
\end{tabular}
\caption{Fields used for gluing when $m=0$,
in space-dimension $n\ge 3$ (see Remark~\ref{R19VI24.1}).
The notation, and the last two lines, are as in Table~\ref{T11III23.2}. 
    Next, the the fields $\overadd{j}{\Hf}_{uA}$ and $\overadd{j}{q}_{AB}$ are defined in \eqref{6III23.w4} and \eqref{6III23.w8} respectively, in the gauge $\delta \beta = 0$. Still with $\delta \beta = 0$, the charges $\protect\kQ{3}{}$ are defined in \eqref{17X23.1} (for $n>3$), $\protect\kQ{4}{}$ in \eqref{17X23.4} and $\protect\kQ{5,i}{}{}$ in \eqref{7VI24.21} (for $n>3$).
   $C$ is the conformal Killing operator; 
   the remaining operators are defined as follows: $\mrL$ in \eqref{28XI23.p1}; 
   $\protect\overset{(p)}{\ochi}=\protect\overset{(p)}{\ochi}\ofP$ in \eqref{6III23.w2}; 
   $\protect\overset{(p)}{\psi}=\protect\overset{(p)}{\psi}\ofP$ in \eqref{6III23.w9} and \eqref{17IV23.4}; $P$ in \eqref{16V22.1b}; 
   $\hLop_n$ in \eqref{24IV23.1xs}; $\Lop$ in \eqref{24IV23.3} for $n>3$. 
   When $n=3$, $\kQ{3}{}$ is defined in \eqref{19VI24.7}, $\protect\kQ{5,i}{}{}$ below \eqref{19VI24.21} and $\Lop$ in \eqref{19VI24.1}.
    }
        \label{T11III23.1}
\end{table} %

This  theorem is the main ingredient of the nonlinear gluing~\cite{ChCongGray2}, where a suitable implicit function theorem is used. In fact, particular care has been taken
here  to do the linearised gluing  in a way which can be promoted to a nonlinear one.

It was found by Aretakis et al., in the case $k=2$, $n=3$, $\Lambda=0$ and $\myGauss=1$, that there exists a ten-parameter family of obstructions to do such a gluing, when requiring continuity of  two $u$-derivatives of the metric components along the null-hypersurface.
Our analysis shows that the result is affected by the dimension, the cosmological constant, the topology of sections of the level sets of $u$ (which we assume to be compact), the mass, and the number of transverse derivatives which are required to be continuous.
%

We note that some introductory material below is essentially identical to that  in~\cite{ChCong1}, except for minor modifications related to the change of dimensions.

Unless explicitly indicated otherwise, we assume throughout that the Lorentzian metrics relevant for the problem at hand are $(n+1)$-dimensional with $n>3$.

\bigskip

\noindent{\sl Acknowledgements.} The first two authors are grateful to   the Beijing Institute for Mathematical Sciences and Applications for   hospitality and support during part of the work on this paper. 

\section{Notation}
 \label{s28XI23.1}

An index of notation  has been added at the end of this work for the convenience of the reader.

Unless explicitly indicted otherwise we use the notations of~\cite{ChCong1}, see Section~2 there, {and of~\cite{ACR2}. In particular, on a $d$-dimensional sphere $S^d$, the notation $t^{[= \ell]}$ will denote the $L^2$-orthogonal projection of a tensor $t$ on the space of $\ell$-spherical harmonics, with
 \begin{equation}\label{29X22}
   t^{[\le \ell]} = \sum_{i=0}^\ell t^{[=i]}
    \,,
    \qquad
   t^{[> \ell]} = t - t^{[\le \ell]}
   \,,
 \end{equation}
and  with obvious similar definition of $ t^{[< \ell]}$, etc.
 As in \cite{ChCong1}
 }
 we denote by $\zdivone$ the divergence operator  on vector fields,%
\index{div@$\zdivone$}%
 \begin{equation}\label{30X22.CKV5a}
    \zdivone \xi := \dnabla_A \xi^{A}
   \,,
 \end{equation}
and by  $\zdivtwo$ that on two-symmetric trace-free tensors:
\index{div@$\zdivtwo$}
 \begin{equation}\label{30X22.CKV5}
   (\zdivtwo h)_A:= \dnabla^B h_{AB}
   \,.
 \end{equation}
 However, the operator $\mrL$ is now defined to be
\index{L@$\mrL$}
\begin{equation}
\label{28XI23.p1}
 \mrL:= \zdivone\circ\zdivtwo\,.
\end{equation}
\index{Delta@$\TSzlap$}%
The Laplace operator of the metric $\ringh$ is denoted by $\TSzlap$.

We let $H^k(\secN)$ denote the Hilbert space of tensor fields whose derivatives of order less than or equal to $k$ are in $L^2(\secN)$.
Given any tensor subspace $X\subseteq H^k(\secN)$ and tensor field $T\in H^k(\secN)$, we will write $T^{[X]}$ for the $L^2$-orthogonal projection of $T$ on $X$. In particular, given any linear differential operator $\hat D$ acting on $H^k(\secN)$, we will write $T^{[\ker \hat D]}$ for the $L^2$-orthogonal projection of $T$ on the kernel of $\hat D$, and $T^{[(\ker \hat D)^\perp]}$ for the projection on  $(\ker \hat D)^\perp$; $T^{[\im \hat D]}$ and $T^{[(\im \hat D)^\perp]}$ are defined similarly.

The formal $L^2$-adjoint of an operator $\hat D$ will be denoted by $\hat D^{\dagger}$ .

 We  define the subspaces $S,V \subset \Hkng$ of vector fields on $\secN$ as%
\index{S@$S$}\index{V@$V$}%
 $$
 S = \{ \xi_A: \xi_A = \zspaceD_A \phi\,,
 \phi \in H^{k+1}(\secN) \}
 \,,\quad
 V = \{\xi_A \in \Hkng: \zspaceD^A\xi_A = 0  \} \,,
 $$
 where the differentiability index $k$ should be clear from the context.
 When $\secN$ is compact and boundaryless, the spaces $S$ and $V$ are $L^2$-orthogonal.
Any vector field $\xi\in \Hkng$, $k\ge 1$, can thus be decomposed into its ``scalar'' and ``vector'' parts  (cf. Appendix~\pref{app16X23.1}), which we shall denote by%
\index{S@$S$}\index{V@$V$}%
 \begin{align}
    \xi = \xi^{[S]} + \xi^{[V]} \,.
 \end{align}

 We denote the decomposition of a symmetric traceless two-covariant tensor field $h$  into its ``scalar'', ``vector'', and ``tensor'' part as
\begin{equation}
\label{14XII23.1}
h = h^{\red[S]} + h^{\red[V]} + h^{[\TTt]}
\,,
\end{equation}
where $h^{[\TTt]}$ is a transverse-traceless ($\TTt$) tensor, $h^{\red[V]}$ is the Killing operator acting on a divergence-free vector, and $h^{\red[S]}$ is the trace-free part of the Hessian of a function (cf. Appendix~\ref{ss8IX23.1}).

 Note that we use the notation $S$ both for the space of ``scalar vectors'' $ \xi^{[S]} $ and ``scalar tensors'' $h^{[S]} $, similarly for $V$, hoping that the distinction will be clear from the context.

As already pointed out, different strategies are needed when the number $k$ of transverse derivatives that one wants to glue exceeds the threshold $k=(n-3)/2$ for odd values of $n$. To address this we will use the following terminology: consider a pair $(n,k)\in \N^2$, where $n>3$ stands for the space-dimension and $k$ for the number of transverse derivatives that are glued. We say that $(n,k)$ is
\emph{convenient} if $n$ is even and  $k$ is arbitrary, or if $n$ is odd  and $k < \frac{n-3}{2}$. Otherwise the pair $(n,k)$ will be said to be \emph{inconvenient}.

\ptclater{flat torus commented out}

 %
\section{Characteristic constraint equations in Bondi coordinates}
 \label{s3X22.1}

Let $(\mathcal{M},g)$ be an $(n+1)$-dimensional spacetime. Locally near a null hypersurface with non-vanishing  divergence scalar one can assign Bondi-type coordinates $(u,r,x^A)$ in which the metric takes the form (see~\cite{MaedlerWinicour} and references therein)
\begin{align}
    g_{\alpha \beta}dx^{\alpha}dx^{\beta}
   &=  -\frac{V}{r}e^{2\beta} du^2-2 e^{2\beta}dudr
  \nonumber
\\
 &\qquad
   +r^2\zhTBW_{AB}\Big(dx^A-U^Adu\Big)\Big(dx^B-U^Bdu\Big)
    \label{23VII22.1}
    \,,
\end{align}
with $\zhTBW_{AB}$ satisfying
\begin{equation}
    \det [\gamma_{AB} ] = \det[ \zzhTBW_{AB}]
     \,,
\end{equation}
where $\det [ \zzhTBW_{AB}]$  is  $r$- and $u$-independent. This implies in particular
\begin{equation}
    \gamma^{AB}\partial_r\gamma_{AB} = 0\,,\qquad \gamma^{AB}\partial_u\gamma_{AB} = 0\,.
    \label{23VII22.2}
\end{equation}
The inverse metric reads
\begin{equation}
    g^{\sharp} =  e^{-2\beta} \frac{V}{r}\, \partial_r^2 - 2 e^{-2\beta} \, \partial_u\partial_r - 2e^{-2\beta} U^A \,  \partial_r \partial_A +\frac{1}{r^2} \gamma^{AB} \, \partial_A\partial_B\,.
\end{equation}
In these coordinates, each level set of $u$   is a null hypersurfaces with  normal field $\partial_r$, while $r$ is a parameter which varies along the null generators. The $x^C$'s are local coordinates on the codimension two surfaces of constant $(u,r)$, which foliate each constant $u$ null hypersurface.

The restriction of the Einstein equations to a level set of $u$  gives rise to a set of transport equations for the metric functions $(V,\beta,U^A,\gamma_{AB})$. In this paper we consider the linearised equations around a
Birmingham-Kottler background, which includes the Minkowski, anti-de Sitter, and  de Sitter background.
 In Bondi coordinates the  background metrics can be written as
\be
{\nobarzg}\equiv {\nobarzg}_{\a \b} dx^\a dx^\b = \zguu  du^2-2du \, dr
 + r^2 \ringh_{AB}dx^A dx^B
   \,,
   \label{23VII22.3}
\ee
with
\begin{equation}
\label{17III23.12}
\zguu :=
-\left(\twoscsign
 -\alpha^2 r^2- {\frac{2m}{r^{n-2}}}\right)
\,,
 \quad
 \twoscsign \in \{0,\pm 1\}
 \,,
 \quad
  \alpha \in \{0,\sqrt{\frac{2|\Lambda|}{n(n-1)}}, \sqrt{\frac{2|\Lambda|}{n(n-1)}} i\}
  \,,
\end{equation}
and where
$\ringh_{AB}dx^A dx^B$ is a $u$- and $r$-independent
Einstein metric of scalar curvature equal to $(n-1)(n-2)\twoscsign$, with the associated Ricci tensor, which we denote by $\zR_{AB}$, thus equal to%
\index{R@$\zR _{AB}$}%
\begin{equation}
\label{17III23.1}
  \zR_{AB} = (n-2) \myGauss\,\ringh_{AB}
   \,.
\end{equation}
We emphasise that  $\alpha \in \R\cup i \R$, with a purely imaginary value of $\alpha$ allowed to accommodate for a cosmological constant $\Lambda <0$.  Finally, the parameter $m$ is related to the total mass of the spacetime.  The inverse background metrics read
$$
\nobarzg^{\alpha\beta}\partial_\alpha\partial_\beta = -2 \partial_u\partial_r -  \zguu  (\partial_r)^2
 + r^{-2}\zzhTBW^{AB}\partial_A\partial_B
 \,.
$$

Consider now a perturbation of the metric of the form%
 \index{h@$h_{\mu\nu}$}%
\begin{equation}
 \label{3VIII22.5}
 \nobarzg_{\mu\nu} \rightarrow \nobarzg_{\mu\nu} + \epsilon \hBo_{\mu\nu}\,.
\end{equation}
The conditions on the linearised fields such that the perturbed metric is still in the Bondi form to $O(\epsilon)$ are,
\begin{equation}
\label{23VII22.4}
  h_{rA}=h_{rr}=\zzhTBW^{AB} h_{AB}=0
  \,.
\end{equation}
We shall sometimes denote the metric perturbations by $\{\delta V,\delta \beta, \delta U_A\}$. These correspond respectively to
\begin{equation}\label{3X22.1p}
 \{\delta V + 2 V \delta \beta,\delta \beta, \delta U_A\}
 \equiv
 \{- r h_{uu}, - h_{ur} / 2, - h_{uA}/ r^2\}
 \,.
\end{equation}
 We will also use the notation%
 \index{h@$h_{\mu\nu}$!$\zhTB_{\mu\nu}$}%
\begin{equation}\label{10XII23.9}
 \zhTB_{\mu\nu}:=
  \frac{\hBo_{\mu\nu}}{r^2}
  \,.
\end{equation}

\subsection{The linearised $\Ck$-gluing problem}

\index{N@$\mcN_{I}$}%
Let $\mcN_{I}$ be a null hypersurface $\{ u=u_0, r\in I\}$, where $I$ is an interval in $\R$, and let $\secN$ be a cross-section of $\mcN_{I}$, i.e. a two-dimensional submanifold of $\mcN_{I}$ meeting each null generator of $\mcN_{I}$ precisely once.
Let $2\le \bluek  \in \N$ be the number of derivatives of the metric that we want to control at $\secN$. Using the  Bondi parameterisation of the metric, we define
linearised Bondi \emph{cross-section data} of order $k$ as the collection of fields
\begin{equation}\label{23III22.992}
  x = (\partial_u^{j}h_{AB}|_{\secN}
  ,\,  \partial_r^jh_{AB}|_{\secN}
  ,\,  \partial_u^{j}\delta\beta|_{\secN},
  \, \partial_u^{j}\delta U^A|_{\secN}
  ,\, \partial_r \delta U^A|_{\secN}
  ,\,   \delta V|_{\secN}
  )_{0\le j\leq k}
  \,.
\end{equation}
and we denote the set of all possible data as $\dt{\secN} $.
The data $\dt\secN$ correspond  to (an optimised version of) the linearised data $\Psi_{\mathrm{Bo}}[\secN,k]$ of ~\cite[Section~5]{ChCong0}.

Unless explicitly indicated otherwise we assume that all the  fields in \eqref{23III22.992}
are smooth, though a finite sufficiently large degree of differentiability would suffice for our purposes, as can be verified by chasing the number of derivatives in the relevant equations.

In the linearised $\Ck$-gluing problem we start with two cross-sections $\secN_1$ and $\secN_2$, each equipped with some linearised Bondi cross-section data of order $k$, which we denote as $x_1\in \dt{\secN_1}$ and
$x_2\in\dt{\secN_2}$. The goal is to construct linearised fields
\begin{align}
y:=(\partial^{\ell}_u \delta V,\partial^{\ell}_u \delta \beta, \partial^{\ell}_u \delta U^{A},\partial^{\ell}_u h_{AB})
\end{align}
for $0\leq \ell \leq k$ which interpolate between  $x_1$ and $x_2$ along a null hypersurface $\mcN_{[r_1,r_2]}$ such that (i) $x_1$ agrees with the restriction to $r_1$ of $y$ ;  (ii) $x_2$ agrees with the restriction  to $r=r_2$ of $y$; and (iii) $y$ satisfy the linearised null constraint equations.

Since linearised Bondi data are defined up to linearised gauge transformations, we shall use these transformations to remove part of the obstructions to the gluing.

\subsection{Gauge Freedom}
\label{ss11III23.1}

\index{gauge transformations|(}
Recall that linearised gravitational fields are defined up to a gauge transformation
\begin{equation}\label{4XII19.12}
  h\mapsto h+ \Lie_\TSxip g
\end{equation}
determined by a vector field $\TSxip$. Once the metric perturbation have been put into Bondi gauge, there remains the freedom to make gauge transformations which preserve this gauge:
\begin{eqnarray}
\Lie_{\TSxip} g_{\TSr \TSr}  &=&0
 \,,
 \label{4II20.3}
  \\
   \label{4II20.4}
  \Lie_{\TSxip} g_{\TSr A}  &=&0   \,,
\\
 g^{AB} \Lie_{\TSxip} g_{AB}
  &  =  &
  0
   \label{4XII19.14}
 \, .
\end{eqnarray}
For the metric  \eqref{23VII22.3} this is solved by (cf., e.g.,~\cite{ChHMS})
\begin{align}
    \zeta^u(u,r,x^A) & = \xi^u(u, x^A)\,,
    \label{3XII19.t1HD}
\\
    \zeta^B(u,r,x^A) & = \xi^B(u,x^A) - \frac{1}{r}\zspaceD^B \xi^u(u,x^A)\,,
     \label{1VIII22.1HD}
\\
    \zeta^r(u,r,x^A) &= -\frac{r}{n-1} \zspaceD_B \xi^B(u,x^A) + \frac{1}{n-1}\TSzlap \xi^u(u,x^A)\,,
    \label{3XII19.t2HD}
\end{align}
for some fields $\TSxi^{u} (\TSu,x^A)$, $\TSxi^{B} (\TSu,x^A)$, and where $\zspaceD_A$ and $\TSzlap$ are respectively the covariant derivative and the Laplacian operator associated with the $(n-1)$-dimensional metric $\ringh_{AB}$ appearing in \eqref{23VII22.3}.

We will use the symbol%
\index{L@$\TSoLie_\TSxip$}
$$
 \TSoLie_\TSxip
$$
to denote   Lie-derivation in the $x^A$-variables
with respect to the vector field $\TSxip^A\partial_A$.

The transformation~\eqref{4XII19.12} can be viewed as a result of linearised coordinate transformation to new coordinates $\tilde{x}^{\mu}$ such that
\begin{equation}
    x^{\mu} = \tilde{x}^{\mu} + \epsilon \zeta^{\mu}(\tilde{x}^{\mu})\,,
    \label{3VIII22.3}
\end{equation}
where $\epsilon$ is as in \eqref{3VIII22.5}.
Recalling that $\zguu$ reads
$$
\zguu =
  -\twoscsign
	+\alpha^2 r^2+ {\frac{2m}{r^{n-2}}}
  \,,
$$
under \eqref{3VIII22.3} the linearised metric transforms as%
\index{gauge transformation law!$h_{\mu\nu}$}%
 \FGp{done}
\begin{eqnarray}
  \hBo_{uA}
   &\to &
  \tilde{\hBo}_{uA} = \barh_{uA} + \mcL_\zeta g_{uA}
  \nonumber
  \\
  & &
  =
  \barh_{uA} +
   \partial_{\tdA}(\zguu \zeta ^u - \zeta ^r)
  + r^2 \ringh_{AB} \partial_{\tdu}\zeta ^B
  \nonumber
\\
& &
 =
  \barh_{uA}  -\frac{1}{n-1}\partial_A \,
    \big[\, (\TSzlap \xi^u + (n-1) \twoscsign \xi^u) - r( \zspaceD_B \xi^B - (n-1)\partial_u \xi^u)
    \big]
 \nonumber
 \\
 & & \quad + r^2 (\ringh_{AB} \partial_u \xi^B + (\alpha^2 +{2m}{r^{-n}} ) \partial_A \xi^u)\,,
  \label{24IX20.1HD}
  \\
  \hBo_{ur}
   &\to &
   \tilde{\hBo}_{ur} =
  \barh_{ur} + \mcL_\zeta g_{ur} =
  \barh_{ur}
  - \partial_{\tdu} \zeta^u
  + \zguu \partial_{\tdr} \zeta^u
  - \partial_{\tdr} \zeta^r
  \nonumber
\\
& &
=
  \barh_{ur}
  - \partial_{\tdu} \xi^u
  + \frac{1}{n-1} \zspaceD^{\tdA}
   \xi_A
    \,,
    \label{24IX20.23aHD}
\\
      \hBo_{uu}
    & \to &
 \tilde \hBo_{uu} =  \hBo_{uu}  + \mcL_\zeta g_{uu} =
  \barh_{uu}
  +   \TSxip^{\TSr} \partial_{\tdr} \zguu
	+2 \partial_{\tdu}(\zguu \zeta^u - \zeta^r)
\nonumber
\\
& &
=
\barh_{uu}
 -2\big(\twoscsign+\frac{1}{n-1}\TSzlap\big)
  \partial_u\xi^u
   +
    \frac{2r}{n-1}
    \Big(\zspaceD_B \partial_u\xi^B +\big(
         \alpha^2-\frac{(n-2)m}{r^n}
          \big)
           \TSzlap \xi^u
           \Big) \nonumber
\\
 & &+ 2 r^2 \left[\left(\alpha^2+\frac{2m}{r^n}\right)\partial_u\xi^u - \frac{1}{n-1}\left(\alpha^2-\frac{(n-2)m}{r^{n}}\right)\zspaceD_B \xi^B\right]\, ,
 \label{24IX20.23xHD}
\\
  \hBo_{AB}
    & \to  &
    \tilde \hBo_{AB} =
  \barh_{AB} + \mcL_\zeta g_{AB} =
  \barh_{AB}
  +2 r \TSxip^{\TSr} \ringh_{AB}
	+
	r^2 \TSoLie_{\TSxip} \ringh_{AB}
\nonumber
\\
 &&  =
  \barh_{AB}
  + r^2\TS[
	\TSoLie_{\TSxip} \ringh_{AB} ]
 \,,
  \label{24IX20.23HD}
\end{eqnarray}
with
\index{TS@$\TS$}%
$$
\TS[X_{AB}] := \frac 12
    \big(
     X_{AB}+X_{BA} - \frac{2}{n-1}\ringh ^{CD} X_{CD} \ringh_{AB}
    \big)
$$
denoting
the traceless symmetric part of a tensor on a section $\secN$.

Given a section%
\index{S@$\secN_{u_0,r_0}$}%
$$
 \secN_{u_0,r_0}:=\{u=u_0,r=r_0\}
$$
of a null hypersurface $\mcN_I$ together with linearised Bondi sphere data $x\in \dt{\secN_{u_0,r_0}}$, Equations~\eqref{24IX20.1HD}-\eqref{24IX20.23HD} and their
$u$- and $\tdr$-derivatives provide a set of order-$k$ cross-section data $\tilde x \in \dt{\tilde\secN_{u_0,r_0}}$ on%
\index{S@$\tilde{\secN}_{u_0,r_0}$}%
$$
  \tilde{\secN}_{u_0,r_0}:= \{\tilde u=u_0,\tilde r=r_0\}=\{u=u_0+ \epsilon
  \zeta^u(u_0,r_0,x^A),
   r=r_0+\epsilon \zeta^r(u_0,r_0,x^A)\}
 \,,
$$
which is a small deformation of  the original $\secN_{u_0,r_0}$, in terms of the gauge fields
$$
z:=\{\partial^{\ell}_{\tdu}\xi^B|_{\tdu = u_0},\partial^{\ell}_{\tdu}\xi^u|_{u = u_0}\}_{0\leq \ell \leq k+1}$$
as well as the original data   $x$.
We shall write the gauge transformation as a map $\Gmap$ such that
\index{G@$\Gmap$}%
\begin{align}
\label{12V24.1}
\tilde x = \Gmap( x)
\,,
\end{align}
with $\Gmap( x)$ given by the Equations \eqref{24IX20.1HD}-\eqref{24IX20.23HD} and their
$u$- and $\tdr$-derivatives

Equation \eqref{24IX20.23aHD} shows that we can always choose $z$ so that
\begin{equation}\label{8III22.1}
  \tilde \hBo_{ur} = 0
  \,.
\end{equation}
After having done this, we are left with a residual set of gauge transformations,   defined by a $\tdu$-parameterised family of vector fields $\TSxi^A(\tdu,\cdot)$ on $\secN$, together with
$\zeta^r(\tdu,\cdot)$ and
\begin{equation}\label{5XII19.1aHD}
	 \partial_{\tdu}  \xi^{u}(\tdu,x^A) =
\frac{ \zspaceD_{\tdB} \TSxi^{B}(\tdu, x^{A})}{n-1}
 \,.
\end{equation}
Taking into account~\eqref{5XII19.1aHD}, the transformed fields now read
\ptcheck{18III23, crosschecked with NullGluing around eq. 3.16, and note that there are slightly more general equations allowing for a change in beta}
\begin{eqnarray}
  \tilde \hBo_{uA}
   & = &
  \hBo_{uA} -\frac{1}{n-1}\zspaceD_{\tdA} \TSzlap \xi^u + \zguu \zspaceD_{\tdA} \xi^u + r^2 \partial_{\tdu} \xi_A
  \nonumber
\\
& = &
  \hBo_{uA} -\frac{1}{n-1}\partial_A \, [\, (\TSzlap \xi^u + (n-1)\twoscsign \xi^u)] \nonumber\\
& &
	+ r^2 (\ringh_{AB} \partial_u \xi^B+ (\alpha^2+ {{2m}{ r^{-n} }}) \partial_A \xi^u)
 \,,
  \label{17III22.1HD}
\\
      \tilde \hBo_{uu}
    & = &
\barh_{uu}-\frac{2}{n-1} \Big(\twoscsign+\frac{1}{n-1}\TSzlap  {-\frac{nm}{r^{n-2}} }\Big)
  \zspaceD_B\xi^B\nonumber\\
& &	+ \frac{2r}{n-1}\Big(\big(\alpha^2- {\frac{(n-2)m}{r^n} } \big) \TSzlap \xi^u+\zspaceD_B \partial_u\xi^B \Big)
\, ,
 \label{13III22.2HD}
\\
  \tilde \hBo_{AB}
    & = &
 \barh_{AB}
  + 2 r^2
 \TS[\zspaceD_{\tdA}\xi _B]
  - 2 r \TS[\zspaceD_{A}\zspaceD_B \xi^u]
 \,,
  \label{17III22HD}
\end{eqnarray}
where $\xi_A:=\ringh_{AB}\xi^B$.

In the linearised gluing problem, we shall allow for such gauge perturbations to the data.
That is, we consider gluing along a null hypersurface of the perturbed data, which we will denote as $\tilde x_1 \in \dt{\tilde{\secN}_1}$ and $\tilde x_2 \in\dt{\tilde{\secN}_2}$, with the freedom of choosing gauge fields to achieve the gluing.

To simplify notation we will write%
\index{L@$\Done$}%
\index{L@$\Dtwo$}%
\index{C@$C$}%
\begin{align}
    \Done(\xi^u)_A &:=-\frac{1}{n-1}\zspaceD_{\tdA} \, [\, \TSzlap \xi^u + (n-1)\twoscsign \xi^u]
   =
   \checkmark
   -\frac{1}{n-2}\zspaceD^B\TS[\zspaceD_A\zspaceD_B \xi^u]
    \,,
    \label{30III23.2}
\\
    C(\zeta)_{AB}&:=\TS[\zspaceD_{\tdA}\zeta_B]
    \,,
     \label{18X22.41}
\\
    \Dtwo(\xi)&:=-\frac{2}{n-1}\left(\twoscsign+\frac{1}{n-1}\TSzlap\right)
   \zspaceD_{\tdB}\xi^B\,.
    \label{30III23.1}
\end{align}
For further convenience  we note the transformation laws, in this notation,
 \FGp{Checked 29II23. }
	\begin{eqnarray}
	\tilde \hBo_{uA}
	& =
	\checkmark
	&
	\hBo_{uA} +\Done(\xi^u)_A
	+ r^2 \left( \partial_{\tdu} \xi_A + \left(\alpha^2  {+{2m}{r^{-n}} }\right) \zspaceD_{\tdA} \xi^u\right)
	\,,
	\label{26IX22.1HD}
	\\
	\partial_u^i \tilde \hBo_{uA}
	& =
	\checkmark
 &
	\partial_u^i  \hBo_{uA} + \frac{1}{n-1} \Done(\zspaceD_{B}  \partial_u^{i-1}\xi^B)_A
	+ r^2 (  \partial_u^{i+1} \xi_A + \frac{\alpha^2  {+2mr^{-n}} }{n-1} \zspaceD_{\tdA} \zspaceD_{B}  \partial_u^{i-1}\xi^B)
	\,,
	\    i\ge 1
	\,,
	\phantom{xxxxxx}
	\label{26IX22.1iHD}
	\\
	\tilde \hBo_{uu}
	& = &
	\barh_{uu} + \frac{2r}{n-1}\big((\alpha^2 {-(n-2)mr^{-n}}) \TSzlap \xi^u +   \zspaceD_{\tdB} \partial_{\tdu}\TSxi^{B}   {+nmr^{-n+1}\zspaceD_B\xi^B}
	\big) + \Dtwo(\xi)  
	\, ,
	\label{26IX22.2HD}
	\\
	\tilde \hBo_{AB}
	& =
	\checkmark
 &
	\barh_{AB}
	+ 2 r^2 C(\zeta)_{AB}
	\nonumber
	\\
	& =
	\checkmark
 &
	\barh_{AB}
	+ 2 r^2 C(\xi )_{AB} - 2 r \TS[\zspaceD_A\zspaceD_B \xi^u]
	\,,
	\\
	\partial_u^i
	\tilde \hBo_{AB}
	& =
	\checkmark
 &
	\partial_u^i
	\barh_{AB}
	+ 2 r^2 C(
	\partial_u^i \xi )_{AB} -  \frac{2r}{n-1} \TS[\zspaceD_A\zspaceD_B
	\zspaceD_C \partial_u^{i-1}\xi^C]
	\,,
	\ i \ge 1
	\,,
	\\
	\zspaceD^A \tilde \hBo_{uA}
	& = &
	\zspaceD^A \hBo_{uA}-\frac{1}{n-1}\TSzlap \, [\, (\TSzlap \xi^u
	+ (n-1) \twoscsign \xi^u)
	]
	+ r^2 (\zspaceD_A  \partial_{\tdu} \xi^A +\left(\alpha^2  {+{2m}{r^{-n}} }\right)\TSzlap \xi^u)
	\,,
	\label{17III22.1apHD}
	\\
	\zspaceD^B \tilde \hBo_{AB}
	& =
\checkmark
 &
	\zspaceD^B \barh_{AB} +   r^2
	\big((\TSzlap + (n-2)\twoscsign ) {\tensor{\delta}{_A^B}}  + \frac{n-3}{n-1}\zspaceD_A\zspaceD^B\big)\xi _B
	+  2r  {(n-2)} \Done( \xi^u)_A
	\,,
	\label{17III22.1aprHD}
	\\
	\zspaceD^A \zspaceD^B\tilde \hBo_{AB}
	& =
\checkmark
   &
	\zspaceD^A \zspaceD^B\barh_{AB} +
	r^2  \bigg[\frac{2(n-2)}{n-1}\TSzlap + 2(n-2)\twoscsign \bigg]
	\zspaceD_{\tdA}\xi^A
	-    \frac{2r  {(n-2)}}{n-1}  \TSzlap \big(\TSzlap + (n-1)\twoscsign \big) \red{\xi^u}
	\nonumber
	\\
	& =
\checkmark
 &
	\zspaceD^A \zspaceD^B\barh_{AB} -
	(n-1)(n-2)r^2  \Dtwo(\xi)
	-    \frac{2r {(n-2)}}{n-1}  \TSzlap \big(\TSzlap + (n-1)\twoscsign \big) \red{\xi^u}\,.
	\label{17III22pHD}
	\end{eqnarray}
\index{gauge transformations|)}
\seccheck{12X23}

\subsection{Null Constraint Equations}
\label{sec:28VII22.1}

We now turn our attention to the Einstein equations,
\begin{equation}
    G_{\mu\nu} := R_{\mu\nu}-\frac{1}{2}g_{\mu\nu}R = 8\pi T_{\mu\nu} - \Lambda g_{\mu\nu}
    \,,
    \label{2X22.2}
\end{equation}
and their linearisation, in Bondi coordinates,
at a Birmingham-Kottler metric \eqref{23VII22.3}-\eqref{17III23.12}. We shall see that the restrictions of the Einstein equations on a null hypersurface $\mcN_{[r_1,r_2]}$ can be written as \textit{transport equations} for various linearised metric functions in $r$, that is, ordinary differential equations in the variable $r$, with source terms depending on $h_{AB}$.
 In particular, given $x_{r_1}\in\dt{\secN_{r_1}}$ and any $h_{AB}(r)|_{\mcN_{[r_1,r_2]}}$, $C^k$ in the $r$ variable, compatible with $x_{r_1}$,
the linearised Einstein equations can be solved on $\mcN_{[r_1,r_2]}$ by integrating the transport equations, such as e.g.\ \eqref{25VII22.3} below. Thus the characteristic gluing problem amounts to solving for a field $h_{AB}|_{\mcN}$ that can interpolate continuously between the given $x_{r_1}$ and $x_{r_2} \in \dt{\secN_{r_2}}$ up to gauge. More details will be provided in the following sections.

\subsubsection{$h_{ur}$}
 \label{CHGss29VII20.1}

The $G_{rr}$ component of the Einstein tensor (compare~\cite{MaedlerWinicour}), reads:
\ptcheck{17III23, with FG}%
\index{G@$G_{\mu\nu}$!$G_{rr}$}%
\index{The Einstein tensor!$G_{rr}$}%
\begin{equation}
         \label{CBCHG:beta_eq}
          \frac{r}{2(n-1)} G_{rr} = \partial_{r} \beta - \frac{r}{8(n-1)}\zhTBW^{AC}\zhTBW^{BD} (\partial_{r} \zhTBW_{AB})(\partial_r \zhTBW_{CD})
          \,.
\end{equation}
Since the second term on the right-hand side of \eqref{CBCHG:beta_eq} is quadratic in $\partial_r\gamma_{AB}$,  after linearising in vacuum we find that the transport equation of $\delta\beta$ is given by
\ptcend{unchanged}
\begin{equation}\label{CHG28XI19.4aa}
 \partial_r \delta \beta =0  \quad
   \Longleftrightarrow
   \quad
   \delta \beta = \delta \beta (u,x^A)
  \,.
\end{equation}
Using a terminology somewhat similar to that of~\cite{ACR2}, we thus obtain a pointwise radial conservation law for $\delta \beta$, and a seeming obstruction to gluing: two linearised Bondi cross-section data can be glued together if and only if their Bondi functions $\delta \beta$ coincide.

However, it follows from \eqref{8III22.1} that we can always choose a gauge so that $\delta \beta \equiv 0$. Thus, \eqref{CHG28XI19.4aa} does not lead to an obstruction to gluing-up-to-gauge. The vanishing of $\delta \beta$ will be assumed when gluing.

But, in the current section we will  \emph{not} assume $\delta \beta =0$
unless explicitly indicated otherwise.

\subsubsection{$h_{uA}$}
 \label{CHGss29VII20.2}

From the $G_{rA}$-component of the Einstein equations one has
 \ptcheck{17III23, with FG}
 \ptcend{does not change anything?}
 \index{G@$G_{\mu\nu}$!$G_{rA}$}%
\index{The Einstein tensor!$G_{rA}$}%
\begin{eqnarray}
    \label{CBCHG:UA_eq}
          &&  \partial_r \left[r^{n+1} e^{-2\beta}\zhTBW_{AB}(\partial_r U^B)\right]
             =   2r^{2(n-1)}\partial_r \Big(\frac{1}{r^{n-1}}\spaceD_A\beta  \Big)
                 \nonumber \\ &&\qquad
                 -r^{n-1}\zhTBW^{EF} \spaceD_E (\partial_r \zhTBW_{AF})
                  +16\pi r^{n-1}    T_{rA}
                  \,.
           \end{eqnarray}
The linearisation of $G_{rA}$ at a Birmingham-Kottler metric gives
\begin{eqnarray}
          &&
           2r^{n-1}  \delta G_{rA} =
           \partial_r \left[r^{n+1}  \ringh_{AB}(\partial_r \delta U^B)\right]
         -  2r^{2(n-1)}\partial_r \Big(\frac{1}{r^{n-1}}\zspaceD_A\delta \beta  \Big)
           \nonumber
\\
      &&
      \qquad \qquad \qquad \qquad
      +
                    r^{n-1}
                   \partial_r \left(r^{-2}
                    \zspaceD^B\red{h}_{AB}\right)
                 \,.
                            \label{CHG28XI19.5}
           \end{eqnarray}
The linearised vacuum Einstein equation thus gives
\begin{align}
    \partial_r\left(r^{n+1}\partial_r  \zhTB_{uA}
    - r^{n-3}  \zspaceD^B  h_{AB} \right)
     &=
      -  2r^{2(n-1)}\partial_r \Big(\frac{1}{r^{n-1}}\zspaceD_A\delta \beta  \Big)
                 -
                 (n-1)  r^{n-4} \zspaceD^B h_{AB}
                  \,.
                   \label{24VII22.1a}
\end{align}
Using \eqref{CHG28XI19.4aa}, we can rewrite this as
\begin{align}
    \partial_r\left(r^{n+1}\partial_r  \zhTB_{uA}
     -
   2 {r^{n-1}}\zspaceD_A\delta\beta
    - r^{n-3}  \zspaceD^B  h_{AB}
        \right)
     =
                 -
                 (n-1)  r^{n-4} \zspaceD^B h_{AB}
                  \,.
                   \label{24VII22.1c}
\end{align}
We write this as a transport equation for a field $\red{\overadd{*}{\Hf}_{uA}}$ (the rationale behind this notation will become clear in Section \ref{sec6III23.1}):
\index{H@$H_{uA}$!$\overadd{*}{\Hf}_{uA}$}%
\begin{align}
    \red{\overadd{*}{\Hf}_{uA}} :=
    r^{n+1}\partial_r  \zhTB_{uA}  - r^{n-3}  \zspaceD^B  h_{AB}
    -
   2 {r^{n-1}}\zspaceD_A\delta\beta
                  \,,
\quad
    \partial_r \red{\overadd{*}{\Hf}_{uA}}
     = -(n-1)  r^{n-4} \zspaceD^B h_{AB}
     \,.
    \label{24VII22.1b}
\end{align}
Integrating, one obtains a representation formula for $\red{\overadd{*}{\Hf}_{uA}}$ on $\mcN_{[r_1,r_2]}$ which reads
\ptcend{varying end point would introduce a term linear in $\xi^r(r_2)$ multiplied by the integrand at $r_2$; this will introduce innocuous boundary terms multiplying a gauge term, so quadratic, so will not change anything?}
\begin{align}
    \red{\overadd{*}{\Hf}_{uA}}(r,\cdot)
      =
     \red{\overadd{*}{\Hf}_{uA}}(r_1,\cdot)
                     - (n-1)\int_{r_1}^r \hat{\kappa}_{4-n}(s)\zspaceD^B h_{AB}(s,\cdot)\,ds \,,
    \label{25VII22.3}
\end{align}
with%
\index{kappa@$\hat{\kappa}_i$}%
\begin{align}
    \hat{\kappa}_{-(n-4)}(s):=s^{n-4}\,,
    \label{25VII22.2}
\end{align}
and where the change of sign in the exponent, compared to the sign of the index of $\hat{\kappa}$, is motivated by consistency of notation
with~\cite{ChCong1}.

\index{C@$\CKV$}
Now, the kernel of $\zdivtwodagger$ is spanned by the space, which we denote by $\CKV$, of conformal Killing vectors of $\secN$,  i.e.\ solutions of the system
\begin{equation}\label{6III22.3}
 \TS[\zspaceD_A \pi_B] =0
  \,,
\end{equation}
with $\pi_A= \pi_A(u,x^B)$.
On $S^{n-1}$ the dimension of $\CKV$ is  $ \frac{ n(n+1) }{2} $.

A classical theorem in conformal geometry shows that, for all remaining compact Riemannian Einstein manifolds, conformal Killing vectors are Killing vectors.
(This follows e.g.\ from ~\cite[Theorem~24]{Kuehnel}, since the flow of a proper conformal Killing vector, i.e., a conformal Killing vector which is not a Kiling vector, provides the conformal factor in this theorem, which cannot be constant  on compact manifolds.)
On a $(n-1)$-dimensional torus $\T^{n-1}$, solutions of \eqref{6III22.3} belong to the $(n-1)$-dimensional space of covariantly constant vectors. Finally, the space of solutions of \eqref{6III22.3} on a $(n-1)$-dimensional negatively curved compact manifold is trivial;
compare Appendix~\ref{App30X22}.

The projection of the transport Equation~\eqref{24VII22.1b} onto $\pi_A \in \CKV$ gives
\ptcheck{17III23}
\begin{eqnarray}
          &&
           \partial_r \int_{\secN}\pi^A \red{\overadd{*}{\Hf}_{uA}} \,\sm =
                   -(n-1)\int_{\secN}\pi^A  r^{n-4}
                    \zspaceD^B \hBo_{AB} \,\sm = 0
                 \,,
                 \label{24VII22.3}
           \end{eqnarray}
and thus%
\index{Q@$\kQ{1}{}$}%
\begin{eqnarray}
          &&
           \kQ{1}{}(\pi^A)  := \int_{\secN}\pi^A \red{\overadd{*}{\Hf}_{uA}}
           = \int_{\secN}\pi^A \left[r^{n+1}  \partial_r(r^{-2} \red{h}_{uA})  -
   2 {r^{n-1}}\zspaceD_A\delta\beta \right] \,\sm
           \label{24VII22.4}
           \end{eqnarray}
forms a family of \textit{conserved radial   charges}, with
$$\partial_r \kQ{1}{}(\pi^A) = 0 \,,\quad \pi^A\in\CKV\,,$$
along any $u = $ constant null hypersurfaces with the gauge choice $\delta\beta = 0$. Whenever useful we shall denote the dependence of
$\kQ{1}{}$ on $x\in \dt{\secN}$ as $ \kQ{1}{}[x]$.
Two data sets $x_{r_1}\in \dt{\secN_1}$ and $x_{r_2}\in\dt{\secN_2}$ can only be glued-up-to-gauge if one can find gauge transformations $z_{r_1}$ and $z_{r_2}$  such that
\begin{equation}
\label{31VII22.12}
 \kQ{1}{}[z_{r_1}^*(x _{r_1})](\pi^A) = \kQ{1}{}[z_{r_2}^*(x _{r_2})](\pi^A)\,.
 \end{equation}

Under  gauge transformations 
$ \red{\overadd{*}{\Hf}_{uA}}$ transforms as%
\index{gauge transformation law!H@$H_{uA}$!$\overadd{*}{\Hf}_{uA}$}%
\index{H@$H_{uA}$!$\overadd{*}{\Hf}_{uA}$!gauge transformation law}%
\begin{align}
    \red{\overadd{*}{\Hf}_{uA}} \rightarrow
    \red{\overadd{*}{\Hf}_{uA}} + 2 \zspaceD^B \big( r^{n-2} \TS[\zspaceD_A\zspaceD_B \xi_u] - r^{n-1} \TS[\zspaceD_A\xi_B]\big) - 2 (r^{n-2}\Done(\xi_u) + n m \zspaceD_A \xi_u) \,;
    \label{30IX23.1}
\end{align}
recall that $\Done$ has been defined in \eqref{30III23.2}.
Upon projection onto $\CKV$, only the last term on the right-hand side of \eqref{30IX23.1} survives and thus $\kQ{1}{}$ transforms as%
\index{gauge transformation law!Q@$\kQ{1}{}$}%
\index{Q@$\kQ{1}{}$!gauge transformation law}%
\ptcheck{22III}
\begin{align}
        \int_{\secN} \pi^A \overadd{*}{\Hf}_{uA} \,\sm
        \to &
        \int_{\secN} \big(  \pi^A \overadd{*}{\Hf}_{uA}
         + {2mn \zspaceD_A \pi^A \xi^u}\big)\,\sm
          \,.
         \label{14VIII22.5}
    \end{align}
(Either by an explicit calculation, or by general considerations, the last two formulae remain the same in the gauge $\delta\beta=0$ when using gauge transformations preserving this gauge.)

If $m=0$ we see that  $\kQ{1}{}$ is gauge invariant, hence
\begin{equation}
 \kQ{1}{}[\tilde x_{r_1}] = \kQ{1}{}[\tilde x_{r_2}]
 \quad
  \iff
  \quad
  \kQ{1}{}[ x_{r_1}] = \kQ{1}{}[ x_{r_2}]
 \,,
 \label{31VII22.1}
\end{equation}
and the equality on the left in \eqref{31VII22.1} has to hold in this case to achieve gluing-up-to-gauge. We will call gauge-invariant radial charges ``\textit{obstructions to gluing-up-to-gauge}''.

However, if $m\neq 0$, the gauge field $\xi^u$ can be used to overcome the obstruction associated to \eqref{31VII22.12} when $\zspaceD_A \pi^A$ is non-vanishing (see Section \ref{s12I22.1} below). This is possible only when $\secN$ is the round sphere and $\pi^A$ is a proper conformal Killing vector of $S^{n-1}$. In all other topologies, the radial charges  $\kQ{1}{}(\pi^A)$ are gauge-invariant independently of whether or not the mass parameter $m$ vanishes.


Next, we can rewrite \eqref{24VII22.1c} with $\delta\beta=0$  as a transport equation involving the field $h_{uA}$:%
  \ptcheck{25III}
\index{H@$H_{uA}$!$\overadd{i}{\Hf}_{uA}$!$\overadd{0}{\Hf}_{uA}$}%
\begin{align}
    &
    \partial_r
     (
     r^{n+1} \partial_r \zhTB_{uA}
      )
      = (n+1) r^n \partial_r \zhTB_{uA} + r^{n+1} \partial^2_r \zhTB_{uA}
    =
    r^{n-1} \partial_r(r^{-2} \zspaceD^B h_{AB})
    \nonumber
\\
    & \implies
     \partial_r\bigg(
    \underbrace{ n \zhTB_{uA} + r \partial_r \zhTB_{uA}
     - \frac{1}{r^3} \zspaceD^B h_{AB} }_{=: \overadd{0}{\Hf}_{uA}}
     \bigg)
     =
    \frac{1}{r^4}  \zspaceD^B h_{AB} \,.
    \label{21II23.1}
\end{align}
Integrating this equation gives
\begin{align}
    \overadd{0}{\Hf}_{uA}(r,\cdot)
     =
     \overadd{0}{\Hf}_{uA}(r_1,\cdot)
     +
    \int_{r_1}^r \frac{1}{s^4}  \zspaceD^B h_{AB} \ ds \,.
    \label{21II23.2}
\end{align}
Under gauge transformations
preserving $\delta \beta=0$
we have%
  \FGp{ 13VI23}
\index{H@$H_{uA}$}%
\index{gauge transformation law!H@$H_{uA}$}%
\index{H@$H_{uA}$!gauge transformation law}%
\begin{align}
     \overadd{0}{\Hf}_{uA}
    \rightarrow &
    \overadd{0}{\Hf}_{uA}
    -\frac2r\zspaceD^B C(\xi)_{AB}
    + n \partial_u \xi_A
    {-\frac{n-2}{r^2}} 
    \Done (\xi^u)_A + \alpha^2 n \zspaceD_A \xi^u
    \,.
    \label{28II23.1}
\end{align}

\begin{Remark}
\label{R20V24.1} For the purpose of Section~\ref{s12I22.1} we will need to analyse the regularity of the fields resulting from our construction. In the linearised case there are many ways to obtain a consistent setup, using H\"older spaces, or Sobolev spaces, possibly with a non-integer regularity index. Strong constraints arise  from the requirement of compatibility with the nonlinear analysis, carried out in the accompanying paper~\cite{ChCongGray2}. For simplicity and definiteness, here we will only consider $L^2$-based Sobolev spaces, which are natural for the evolution problem. The diligent reader will note that all the results below apply in functional spaces where the Agmon-Douglis-Nirenberg type estimates are available, in particular in H\"older spaces, or in $L^p$-based Sobolev spaces with $p\in (1,\infty)$.
\end{Remark}

\paragraph{Regularity.}
\index{regularity|(}
 Let $H^\ell(\secN)$ denote the Hilbert space of tensor fields whose derivatives of order less than or equal to $\ell$ are in $L^2(\secN)$.

Let $\kgamma\in \N$. A tensor  field  $A$ will be said to be in $\CrHk$ if
  it holds that%
\index{H@$\CrHk$}%
\begin{equation}\label{28VIII23.1a}
  A|_{r=r_1}\in H^{\kgamma}(\secN)
  \ \mbox{and} \  \partial_r A \in H^{\kgamma}( [r_1,r_2]\times \secN)
  \,.
\end{equation}
One easily checks that
\begin{equation}\label{28VIII23.1a-}
  \forall \ j\in \Z\cap [1,\kgamma]
  \qquad
  \partial^j _r \CrHk \subset \CrHkmoj
  \,,
\end{equation}
and
\begin{equation}\label{28VIII23.1a+}
  \forall \ j\in \Z\cap [0,\kgamma]
  \,, \ r  \in [r_1,r_2] \qquad \partial_r^j A(r,\cdot) \in H^{\kgamma-j}(\secN)
  \,.
\end{equation}
Indeed, \eqref{28VIII23.1a+} with  $j=0$ is obtained by integration in $r$ (cf.\ \cite[Proposition~3.2]{ChCongGray2}).
 When $j=1$ we have
$$
 \partial_r  A(r_1,\cdot) \in H^{\kgamma - 1/2}(\secN)
  \subset H^{\kgamma - 1}(\secN)
  \,,
$$
where ``$\in$'' is obtained from a standard trace theorem (cf., e.g., \cite{Adams:Sobolev}).
  This proves \eqref{28VIII23.1a-} with $j=1$, and \eqref{28VIII23.1a+} again by integration. For $j>1$ the result is obtained by induction.

We have the trivial observation: for all $s\in \R$
\begin{equation}\label{28VIII23.1}
  A \in \CrHk
  \qquad
  \Longleftrightarrow
  \qquad
  r^s A \in \CrHk
  \,.
\end{equation}
From \eqref{25VII22.3} and \eqref{21II23.2}, we  have:
\begin{eqnarray}
 \lefteqn{
  h_{AB} \in \CrHk
  \,, \
  h_{uA}|_{r=r_1}
  \in \Hkm
  \,, \
  \partial_rh_{uA}|_{r=r_1}\in \Hkm
   }
   &&
 \nonumber
\\
  &&  \Longrightarrow
   \Big(
   h_{uA} \in \CrHkm
  \
  \Longleftrightarrow
  \
    r^n \zhTB_{uA} \in \CrHkm
  \
  \Longleftrightarrow
  \
    \overadd{0}{\Hf}_{uA} \in \CrHkm
  \nonumber
\\
 &&
 \phantom{ \Longrightarrow
   \Big(
   }
  \Longleftrightarrow
  \
    \red{\overadd{*}{\Hf}_{uA}} \in \CrHkm
    \Big)
    \,.
 \label{28VIII23.3}
\end{eqnarray}
We also note that the differentiability class $\overadd{0}{\Hf}_{uA} \in \CrHkm$ is preserved under gauge transformations \eqref{28II23.1}  if
\begin{equation}\label{28VIII23.4}
  \xi^u\in  \Hkpp
  \,,
  \
  \xi^A\in  \Hkp
  \,,
  \
 \partial_u \xi^A\in  \Hkm
  \,.
\end{equation}
\index{regularity|)}

\subsubsection{$\hBo_{uu}$}
 \label{ss29VII20.4}


To obtain the transport equation for the function $V$ occurring in the Bondi form of the metric, it turns out to be convenient to consider the expression for $ 2 G_{ur} + 2 U^A G_{rA} - V/r\, G_{rr} $:
 \ptcheck{17III23, with FG against his mathematica file}
 \ptcend{watch out for the R term }
\index{G@$G_{\mu\nu}$!$ 2 G_{ur} + 2 U^A G_{rA} - V/r\, G_{rr}$}%
\index{The Einstein tensor!$ 2 G_{ur} + 2 U^A G_{rA} - V/r\, G_{rr}$}%
 \begin{eqnarray}
            &&
               r^2 e^{-2\beta} (2 G_{ur} + 2 U^A G_{rA} - V/r\, G_{rr} )
                =
                R[\zhTBW]
                -2\zhTBW^{AB}  \Big[\spaceD_A \spaceD_B \beta
                + (\spaceD_A\beta) (\spaceD_B \beta)\Big]\nonumber\\
                &&\qquad
               +\frac{e^{-2\beta}}{r^{2(n-2)} }\spaceD_A \Big[ \partial_r (r^{2(n-1)}U^A)\Big]
                -\frac{1}{2}r^4 e^{-4\beta}\zhTBW_{AB}(\partial_r U^A)(\partial_r U^B)
                -\frac{(n-1)}{r^{n-3}} e^{-2\beta} \partial_r( r^{n-3} V)\,.
                \nonumber
                \\
                  \label{3X22.1HD}
           \end{eqnarray}
(It follows directly from the definition of $G_{\mu\nu}$ and the Bondi parametrisation of the metric that $ r^2 e^{-2\beta}(2 G_{ur} + 2 U^A G_{rA} - V/r\, G_{rr} )$ can equivalently be written as $r^2 g^{AB}R_{AB}$.)
\ptcheck{17III23}
In vacuum one thus obtains
\ptcend{a constant times some powers of r2 times zeta r at the end point added to the formula for delta V?}
\begin{eqnarray}
            &&
            2\Lambda r^2
                =
                R[\zhTBW]
                -2\zhTBW^{AB}  \Big[\spaceD_A \spaceD_B \beta
                + (\spaceD_A\beta) (\spaceD_B \beta)\Big]
               +\frac{e^{-2\beta}}{r^{2(n-2)} }\spaceD_A \Big[ \partial_r (r^{2(n-1)}U^A)\Big]
               \nonumber
\\
                &&\qquad
                -\frac{1}{2}r^4 e^{-4\beta}\zhTBW_{AB}(\partial_r U^A)(\partial_r U^B)
                -\frac{(n-1)}{r^{n-3}} e^{-2\beta} \partial_r( r^{n-3} V)\,.
                   \label{eq:V_eqnHD}
           \end{eqnarray}
%

Recall that  $\zR _{AB} = (n-2)\twoscsign \ringh _{AB}$
 denotes the Ricci tensor of the metric $\ringh_{AB}$. As $h_{AB}$ is $\ringh$-traceless we have
  \ptcheck{17III23}
\begin{eqnarray}
\nonumber 
 r^2\delta
               ( R[\zhTBW])|_{\zhTBW=\ringh}
  & = &
   -\zspaceD^A\zspaceD_A( \ringh^{BC }h_{BC})+\zspaceD^A\zspaceD^B h_{AB}-\zR ^{AB}h_{AB}
\\
  & =  &
  \zspaceD^A\zspaceD^B h_{AB}
  \,.
\end{eqnarray}

Linearising~\eqref{eq:V_eqnHD} around the Birmingham-Kottler background and rearranging terms thus gives
\ptcend{varying end point would introduce a term linear in $\xi^r(r_2)$ after integration?}
 \begin{align}
                 \partial_r( r^{n-3} \delta V
                 &
                  -  \frac{ r^{n-1  }}{n-1}
                \zspaceD_A \delta U^A)
                =
                \Big( 2 r^{n-3} (n-2) \myGauss -\frac{4 \Lambda r^{n-1}}{(n-1)} \Big)\delta\beta
                           \nonumber
 \\
                &   + \frac{r^{n-3}}{(n-1)}
                 \Big\{
                  \zspaceD^A\zspaceD^B\zhTB_{AB}
                -2\zhTBW^{AB}  \zspaceD_A \zspaceD_B \delta\beta\Big\}
                +
                 r^{ n-2}\zspaceD_A\delta U^A
                  \,.
                   \label{27III2022.3b}
           \end{align}
We note that since $\delta(G_{ur} + U^A G_{rA}) = \delta G_{ur}$,
Equation~\eqref{27III2022.3b} is equivalent to the equation $2r^{n-1}/(n-1) (\delta G_{ur}-\Lambda \hBo_{ur})= 0  $.

Let us show that \eqref{27III2022.3b} gives another
family of conserved radial   charges:
    \begin{align}
        \kQ{2}{}(\lambda)&:= \int_{\secN}\lambda\Big[r^{n-3}\delta V
         - \frac{r^{n-2}}{n-1}\partial_r\big(r^2 \zspaceD^A \delta U_A\big) - \tfrac{2 r^{n-2}}{n-1} \TSzlap \delta\beta
         - 2 r^{n-3} V \delta\beta
 \Big]\sm \,,
      \label{20VII22.1}
    \end{align}
\index{Q@$\kQ{2}{}$}%
where the functions $\lambda(x^A) \in \ker(\mrL^{\dagger})$ are solutions of the equation
\begin{equation}
  \TS[\zspaceD_A \zspaceD_B \lambda] = 0\,.
    \label{24VII22.6}
\end{equation}
(Recall that  the operator $\mrL$ is defined as:%
\index{L@$\mrL$}%
\begin{equation}
    \mrL := \zdivone\circ\zdivtwo .
    )
    \label{31V23.1}
\end{equation}
The only solutions of \eqref{24VII22.6} on a compact Einstein manifold are constants, except  on a round  $S^{n-1}$, in which case  such $\lambda$'s are linear combinations of $\ell=0$ or $\ell=1$ spherical harmonics~\cite[p.~127]{Kuehnel}; cf.\ also~\cite[Section~3]{BoucettaArxiv}.

When $\lambda$ is a constant the conservation equation $\partial_r \kQ{2}{}=0$ is essentially straightforward, as several terms  both in $  \kQ{2}{} $ and in $\partial_r \kQ{2}{} $ integrate out to zero, and using $\partial_r \delta \beta =0$  (cf.\ \eqref{CHG28XI19.4aa}). For general $\lambda$'s satisfying \eqref{24VII22.6} the $r$-independent of $\kQ{2}{}$ follows
from the following calculation,
 where the linearised Einstein equation
 $\partial_r \delta \beta = 0$ has again been used:
    \begin{align}
        \partial_r  \bigg[r^{n-3}\delta V
         &
         - \frac{r^{n-2}}{n-1}\partial_r\bigg(r^2 \zspaceD^A \delta U_A\bigg)\bigg]
        =
        \partial_r (r^{n-3}\delta V)
        - \partial_r \bigg(\frac{r^{n-2}}{n-1}\partial_r\bigg(r^2 \zspaceD^A \delta U_A\bigg)\bigg)
        \nonumber
\\
        &=
         \underbrace{\partial_r (r^{n-3}\delta V) -2 r^{n-2}\zspaceD_A \delta U^A -\frac{r^{n-1}}{n-1}\zspaceD_A\partial_r \delta U^A}_{
         \tfrac{r^{n-3}}{(n-1)} \Big(\zspaceD^A\zspaceD^B \zhTB_{AB} - 2 \TSzlap \delta \beta \Big)
         + \Big( 2 r^{n-3} (n-2) \myGauss -\frac{4 \Lambda r^{n-1}}{(n-1)} \Big)\delta\beta
         \text{ by \eqref{27III2022.3b}}}
         \nonumber
\\
         &\qquad
         \underbrace{-\frac{1}{n-1}( r^n \partial_r^2 \zspaceD_A \delta U^A
         + (n+1) r^{n-1}\partial_r \zspaceD_A \delta U^A)}_{\tfrac{1}{(n-1)r} \partial_r \big(2r^{n-1} \TSzlap \delta \beta\big)+ \tfrac{r^{n-2} }{n-1} \partial_r\big(\zspaceD^A\zspaceD^B\zhTB_{AB}\big)
          \text{ by \eqref{24VII22.1c} }}
         \nonumber
\\
         & = \zspaceD^A\zspaceD^B\bigg[\frac{r^{n-3}}{n-1} \zhTB_{AB}
         + \tfrac{r^{n-2}}{n-1}  \partial_r\zhTB_{AB}\bigg]
         \nonumber
\\
    & \quad
    + \partial_r\bigg[\bigg(\tfrac{2 r^{n-2}}{n-1} \left((n-1) \epsilon + \TSzlap \right)-\frac{4 \Lambda  r^n}{(n-1) n}\bigg)\delta\beta \bigg]
    \nn
\\
    &=
    \frac{1}{n-1}\zspaceD^A\zspaceD^B
     \big[\partial_r (r^{n-2}\zhTB_{AB}) - (n-3)r^{n-3}\zhTB_{AB}
      \big]
               \nonumber
\\
    & \quad
    + \partial_r\bigg[\bigg(\tfrac{2 r^{n-2}}{n-1} \left((n-1) \epsilon + \TSzlap \right)-\frac{4 \Lambda  r^n}{(n-1) n}\bigg)\delta\beta \bigg]
      \,.
    \label{24VII22.5HD}
    \end{align}
Collecting some $r$-derivatives on one side, we find
    \begin{align}
        \partial_r  \Big[r^{n-3}\delta V
         - \frac{r^{n-2}}{n-1}\partial_r\Big(
         &
         r^2 \zspaceD^A \delta U_A\bigg) - \bigg(
         \tfrac{2 r^{n-2}}{n-1} \big((n-1) \epsilon + \TSzlap \big)-\frac{4 \Lambda  r^n}{(n-1) n}\Big)\delta\beta
         \Big]
        \nn
\\
    &=
    \frac{1}{n-1}\zspaceD^A\zspaceD^B
     \big[\partial_r (r^{n-2}\zhTB_{AB}) - (n-3)r^{n-3}\zhTB_{AB}
      \big]
      \,.
    \label{19V24.2}
    \end{align}
Adding $4 \partial_r(m\ \delta \beta) = 0$ to the left-hand side allows one to rewrite this in a more condensed and, as we will see shortly,   gauge-invariant form:
\ptcend{varying end point would introduce a term linear in $\xi^r(r_2)$ after integration?}
    \begin{align}
\partial_r  \Big[r^{n-3}\delta V
         - \frac{r^{n-2}}{n-1}\partial_r
          \big(
          &
          r^2 \zspaceD^A \delta U_A\big)
          - \tfrac{2 r^{n-2}}{n-1} \TSzlap \delta\beta
         - 2 r^{n-3} V \delta\beta
         \Big]
        \nn
\\
    &=
    \frac{1}{n-1}\zspaceD^A\zspaceD^B
     \big[\partial_r (r^{n-2}\zhTB_{AB}) - (n-3)r^{n-3}\zhTB_{AB}
      \big]
      \,.
    \label{19V24.3}
    \end{align}
Hence
\ptcend{if this is direct integration, maybe not affected?}
\begin{equation}
    \partial_r \kQ{2}{} = \frac{ 1}{n-1} \int_{\secN}\lambda
     \mrL \big[\partial_r (r^{n-1}\zhTB_{AB}) - (n-3)r^{n-3}\zhTB_{AB}\big]\,\sm = 0\,,
\end{equation}
and the condition, for $x_{r_1}\in \dt{\secN_1}$ and for $x_{r_2}\in \dt{\secN_2}$
\begin{equation}
    \kQ{2}{}[x_{r_1}] = \kQ{2}{}[x_{r_2}]
\end{equation}
provides another family of obstructions to linearised characteristic gluing.

In the gauge $\delta\beta = 0$,
and under gauge transformations preserving this, it holds that
\index{Q@$\kQ{2}{}$!gauge transformation law}%
\index{gauge transformation law!Q@$\kQ{2}{}$}%
    \begin{align}
 \kQ{2}{}[x](\lambda)  &\to
    \int_{\secN}\lambda\Bigg[
   r^{n-3}\delta V
   + 2r^{n-2}\left(\twoscsign+\frac{1}{n-1}\TSzlap\right)\left(\frac{1}{n-1} \zspaceD_B\xi^B
      \right)
      \nonumber
      \\
   &- \frac{2r^{n-1}}{n-1}(\zspaceD_B \partial_u\xi^B +\left(\alpha^2  {-2m(n-2)r^{-n}}\right) \TSzlap \xi^u)   {+nmr^{n-3}\zspaceD_B\xi^B}) 
      \nonumber
   \\
   &\qquad
    +\frac{r^{n-2}}{n-1}\bigg(\zspaceD^A \partial_r\hBo_{uA}
    + 2 r (\zspaceD_B \partial_u \xi^B + \left(\alpha^2  {-2m(n-2)r^{-n}}\right) \TSzlap \xi^u)
    \bigg)
    \Bigg]\,\sm
    \nonumber
\\
    &=
\kQ{2}{}[x](\lambda)
    +\int_{\secN}\lambda\bigg[   2r^{n-2}\left(\twoscsign + \frac{1}{n-1}\TSzlap  {-\frac{n m}{r^{n-2}}}\right)\left(\frac{1}{n-1} \zspaceD_B\xi^B
      \right)
    \bigg]
    \,\sm
     \,.
    \label{20VII22.2}
    \end{align}
    It can be verified that in the general Bondi gauge with $\delta\beta \neq 0$, the gauge transformation of $\kQ{2}{}$ continues to be given by \eqref{20VII22.2}.

    Taking $\zspaceD^A$ of~\eqref{24VII22.6} gives,
    \begin{equation}
        \bigg(  {1}-\frac{1}{n-1}\bigg)\zspaceD_B\TSzlap\lambda = -\zR _{AB}\zspaceD^A\lambda \equiv  -(n-2)\myGauss \zspaceD_B\lambda
        \quad
         \Longleftrightarrow
         \quad  \frac{1}{n-1} \zspaceD_B\TSzlap\lambda =   - \myGauss \zspaceD_B\lambda
          \,.
    \end{equation}
    %
   This shows that
    when $m=0$, or when $(\secN,\ringh$) is \emph{not} the round sphere, the second integral in~\eqref{20VII22.2} vanishes and hence $\kQ{2}{}{}$ is gauge invariant. \underline{When  $m\neq 0$,} on a round sphere the radial charges $\kQ{2}{}{(\lambda^{[=1]})}{}$ can be gauged-away using $(\zspaceD_B\xi^B)^{[=1]}$.
    \FGp{6IV23}

Equation \eqref{24VII22.5HD} in the gauge $\delta\beta = 0$ can be rewritten as
\ptcheck{23IV23}
    \begin{align}
        \partial_r  \bigg[r^{n-3}\delta V - \frac{r^{n-2}}{n-1}\partial_r\bigg(r^2 \zspaceD^A \delta U_A\bigg)-\frac{r^{n-2}}{n-1}\zspaceD^A\zspaceD^B \zhTB_{AB}\bigg]
    &=
    -\frac{(n-3)r^{n-5}}{n-1}\mrL(h_{AB})\,,
    \label{1XI22.w1}
    \end{align}
which gives a representation formula for $\delta V$ after integration.
\ptcend{should not be affected?}

From \eqref{27III2022.3b} we have, setting $\delta\beta=0$,
 \begin{align}
                -\frac{r^n}{n-1}\partial_r\zspaceD_A \delta U^A
                &=
                -r \partial_r( r^{n-3} \delta V)  + 2r^{n-1}
                \zspaceD_A \delta U^A
                + \frac{r^{n-2}}{n-1}
                  \zspaceD^A\zspaceD^B\zhTB_{AB}
                  \,.
                   \label{20II23.1}
           \end{align}
We use this to rewrite the term in the square brackets of \eqref{1XI22.w1}, in the gauge $\delta \beta =0$, as
 \index{chi@$\chi$}%
\begin{align}
       - \chi &:=
       r^{n-3}\delta V - \frac{r^{n-2}}{n-1}\partial_r\bigg(r^2 \zspaceD^A \delta U_A\bigg)-\frac{r^{n-2}}{n-1}\zspaceD^A\zspaceD^B \zhTB_{AB}
      \nonumber
\\
    &=
   \frac{2(n-2)r^{n-1}}{n-1}\zspaceD^A \delta U_A
    - r^2\partial_r (r^{n-4}\delta V)
=
   -\frac{2(n-2)r^{n-3}}{n-1}\zspaceD^A h_{uA}
    + r^2\partial_r (r^{n-3} h_{uu})
    \,.
    \label{20II23.2}
\end{align}

Under gauge transformations preserving $\delta\beta = 0$, the function $\chi$ transforms as%
 \FGp{19V23}
\index{chi@$\chi$}%
\index{gauge transformation law!chi@$\chi$}%
\begin{align}
\label{1III23.1}
    \chi &\rightarrow
    \chi + \frac{2r^{n-2}(n-3)}{n-1} \bigg(\frac{1}{n-1}\TSzlap + \twoscsign\bigg) \zspaceD_B \xi^B
    {+}  \frac{2 n m}{n-1} \zspaceD_B \xi^B
    \nonumber
    \\
    & \quad
    - \frac{2r^{n-3}(n-2)}{(n-1)^2} \TSzlap ( \TSzlap + (n-1) \twoscsign) \xi^u\,.
\end{align}
\paragraph{Regularity.}
It holds that%
\index{regularity|(}%
\begin{eqnarray}
 \delta V|_{r=r_1} \in \Hkmm\,,
  &&
  h_{AB} \in \CrHk
  \,,
  \
   h_{uA} \in \CrHkm
   \nonumber
\\
 &&   
  \Longrightarrow
   \quad
   \Big(\delta V \in \CrHkmm
  \
  \Longleftrightarrow
  \
    \chi \in \CrHkmm
    \Big)
    \,.
 \label{28VIII23.2}
\end{eqnarray}
The differentiability class $\chi \in \CrHkmm$ is preserved under gauge transformations \eqref{1III23.1}  if
\begin{equation}\label{28VIII23.5}
  \xi^u\in  \Hkpp
  \,,
  \
  \xi^A\in  \Hkp
  \,,
\end{equation}
consistently with \eqref{28VIII23.4}.%
\index{regularity|)}%
\ptcheck{20V for the new differentiability}

\seccheck{10XII23}

\subsubsection{$\partial_u \hBo_{AB}$}
 \label{ss29VII20.3}

We continue with the equation involving $\partial_u \hBo_{AB}$:
 \ptcheck{17III23, with FG}
\begin{eqnarray}\label{eq:ev_eqnHD}
     \lefteqn{
     \TS\Big[e^{2\beta} R[\gamma]_{AB} +
      r^{(5-n)/2}\partial_r [r^{\frac{n-1}2}  (\partial_u \zhTBW_{AB})]
     	 - \frac{1}{2} r^{3-n} \partial_r[ r^{n-2}V  (\partial_r \zhTBW_{AB})]
          }
          &&
           \nonumber
\\
         && -2e^{\beta} \spaceD_A \spaceD_B e^\beta +\frac{1}{r^{n-3}} \zhTBW_{CA} \spaceD_B[ \partial_r (r^{n-1}U^C) ]
          - \frac{1}{2} r^4 e^{-2\beta}\zhTBW_{AC}\zhTBW_{BD} (\partial_r U^C) (\partial_r U^D)
          \nonumber \\
       &&
       +
               \frac{r^2}{2}  (\partial_r \zhTBW_{AB}) (\spaceD_C U^C )
              +r^2 U^C \spaceD_C (\partial_r \zhTBW_{AB})
                \nonumber \\
       &&
      + \frac{1}{2} rV \gamma^{CD} \partial_r\gamma_{AC}\partial_r\gamma_{BD}-\frac{1}{2}r^2\gamma^{CD}(\partial_r\gamma_{BD}\partial_u\gamma_{AC}+\partial_u\gamma_{BD}\partial_r\gamma_{AC})
       \nonumber \\
       &&
        -
	r^2 (\partial_r \zhTBW_{AC}) \zhTBW_{BE} (\spaceD^C U^E -\spaceD^E U^C)
      +\Lambda e^{2\beta} g_{AB}  -8\pi e^{2\beta}T_{AB}
       \Big]
        =0.
\end{eqnarray}
 It is convenient to rewrite this equation as
 \FGp{28III23; Crossed-checked by FG and Wan}
 \ptcend{moving  the first term on the rhs to the left would give an equation which can be integrated with a moving end  }
\begin{eqnarray}\label{eq:30III22.1bHD}
     \lefteqn{
    \partial_r
    \Big[
     r^{\frac{n-1}2}   \partial_u \zhTBW_{AB}
     	 - \frac{1}{2} r^{\frac{n-3}{2} } V   \partial_r \zhTBW_{AB}
     	 -  \frac{n-1}{4} r^{\frac{n-5}2  } V    \zhTBW_{AB}
     \Big]
          }
          &&
           \nonumber
\\
         &=&
     	 -   \frac{n-1}{4} \partial_r(r^{\frac{n-5}2  } V )   \zhTBW_{AB}
     	- \frac{1}{2} r^{\blue{\frac{n-3}{2}}}V \gamma^{CD} \partial_r\gamma_{AC}\partial_r\gamma_{BD}
     	 \nonumber
\\
		&&+
		\frac{1}{2}r^{\blue{\frac{n-1}2}}\gamma^{CD}(\partial_r\gamma_{BD}\partial_u\gamma_{AC}+\partial_u\gamma_{BD}\partial_r\gamma_{AC})
		\nonumber
\\
        &&
       -r^{\frac{n-5}2}
       \TS\Big[
      e^{2\beta}   R[\gamma]_{AB} -2e^{\beta} \spaceD_A \spaceD_B e^\beta
      \nonumber
\\
 	  &&
      + r^{3-n} \zhTBW_{CA} \spaceD_B[ \partial_r (r^{n-1}U^C) ]
          - \frac{1}{2} r^4 e^{-2\beta}\zhTBW_{AC}\zhTBW_{BD} (\partial_r U^C) (\partial_r U^D)
          \nonumber \\
       &&
       +
               \frac{r^2}{2}  (\partial_r \zhTBW_{AB}) (\spaceD_C U^C )
              +r^2 U^C \spaceD_C (\partial_r \zhTBW_{AB})
                \nonumber \\
       &&
        -
	r^2 (\partial_r \zhTBW_{AC}) \zhTBW_{BE} (\spaceD^C U^E -\spaceD^E U^C)
       -8\pi e^{2\beta}T_{AB}
       \Big]
       \,.
\end{eqnarray}

Let us define the operators $P$ and $\tric$ acting on two-covariant tensors as
\index{P@$P$}
\index{R cal@$\tric$}
\begin{equation}\label{16V22.1b}
  P(h)_{AB} := \TS[\zspaceD_A \zspaceD^C h_{BC}] \,,\quad
  \tric (h)_{AB} := \TS[\zR_A{}^C{}_B{}^D h_{CD}]
  \,,
\end{equation}
\index{R@$\zR_{ABCD}$}%
where $\zR_{ABCD}$ is the curvature tensor of $\ringh$.
We note that for any tensor $X_{AB}$ we have
\begin{equation}\label{16V22.1bxc}
 \tric (X)_{AB} = \tric\big(\TS(X)\big){}_{AB}
 \,.
\end{equation}
If $h$ is symmetric and trace-free  it holds that
\begin{equation}\label{16V22.1bx}
  \tric (h)_{AB}
  = \zR_A{}^C{}_B{}^D h_{CD}
  \,,
\end{equation}
and if moreover $\ringh$ is a space form we have
\begin{equation}
 \label{12V23.1}
  \tric (h)_{AB} = -\myGauss h_{AB}
  \,.
\end{equation}

Denoting by   $\zR$ the Ricci scalar of $\ringh$,
we find
 \ptcheck{22III}
\begin{align}
    \delta \TS[R[\gamma]_{AB}] &=
    \TS[\delta R[\gamma]_{AB}]
    - \frac{1}{n-1} \zR h_{AB}
    \nonumber
\\
    &=
    \TS[ -\frac{1}{2}\TSzlap h_{AB} + \zspaceD^C\zspaceD_A h_{BC}]
    - (n-2) \myGauss h_{AB}
    \nonumber
\\
    & =
    -\frac{1}{2}\TSzlap h_{AB} + P (h)_{AB}
      - \tric(h)_{AB}\,.
      \label{15VI24.1}
\end{align}
Thus
the linearisation of \eqref{eq:30III22.1bHD} around a  Birmingham-Kottler background, in vacuum, reads,
\begin{eqnarray}
\label{eq:31III22.3p0}
\lefteqn{0 = r^{\frac{n-5}2} \TS[\delta G_{AB}] }
&&
\nonumber
\\ &&
=
     \partial_r \Big[r^{\frac{n-1}2}  \blue{\partial_u \zhTB_{AB}}
     - \frac{ r^{\frac{n-3}{2}}}{2}  V  \partial_r \zhTB_{AB}
     -  \frac{n-1 }{4}  r^{\frac{n-5}2}   V   \zhTB_{AB}
     -
            {r^{\frac{n-1}2}} \TS\big[\zspaceD_A   \zhTB_{uB}]
     \Big]
      \nonumber
\\&&
       +\frac{n-1}{4}  \partial_r  (r^{\frac{n-5}2}  V)  \zhTB_{AB}
         -
          r^{\frac{n-5}2}
          \Big(
          2\zspaceD_A \zspaceD_B \delta\beta
           +
           \frac{n-1}{2 }r \TS\big[\zspaceD_A   \zhTB_{uB})
       \big]
       \Big)
       \nonumber
\\&&
    + r^{(n-9)/2}( -\frac{1}{2}\TSzlap h_{AB} + \red{P(h)}_{AB}
      - \red{\tric(h)}_{AB})
        \,.
        \phantom{xxx}
\end{eqnarray}

To obtain a transport equation  involving $\partial_u h_{AB}$, with source terms depending only on the field $h_{AB}$, we take $\frac{2}{(n+1) r^{(n+1)/2}}\times C$ of
 \eqref{24VII22.1c} with $\delta\beta=0$, giving
\begin{align}
    \frac{2}{(n+1) r^{(n+1)/2}}\partial_r\left[r^{n+1}\partial_r  (r^{-2}\TS[\zspaceD_B\red{h}_{uA}])  \right]
    - \frac{2}{n+1}r^{\frac{n-3}{2}}\partial_r(r^{-2}P(\hBo)_{AB})
     &=0
                  \,.
                 \label{3IX22.w2HD}
\end{align}
Subtracting \eqref{3IX22.w2HD} from \eqref{eq:31III22.3p0} with $\delta \beta=0$  leads to the desired equation: 
 \ptcheck{29III23}
\begin{align}
     & \partial_r \Big[r^{\frac{n-1}2}  \blue{\partial_u \zhTB_{AB}}
     - \frac{ r^{\frac{n-3}{2}}}{2}  V  \partial_r \zhTB_{AB}
     -  \frac{n-1 }{4}  r^{\frac{n-5}2}   V   \zhTB_{AB}
     -\frac{2}{(n+1)r^{(n+1)/2}}\partial_r(r^{n-1}TS\big[\zspaceD_A   h_{uB}])
     \Big]
      \nonumber
\\
    &\qquad
    =
     -\frac{n-1}{4}  \partial_r  (r^{\frac{n-5}2}  V)  \zhTB_{AB}
       - \frac{2}{n+1}r^{\frac{n-3}{2}}\partial_r(r^{-2}P(h)_{AB})
       \nonumber
\\
       & \qquad\qquad
       - r^{(n-9)/2}( -\frac{1}{2}\TSzlap h_{AB} + P(h)_{AB}
      - \red{\tric(h)}_{AB})
        \,.
        \label{3IX22.1aHD}
\end{align}
Setting
\index{q@$q_{AB}$}%
\begin{eqnarray}
 q_{AB}
  & :=
  &  r^{\frac{n-1}2}  \blue{\partial_u \zhTB_{AB}}
     - \frac{ r^{\frac{n-3}{2}}}{2}  V  \partial_r \zhTB_{AB}
     -  \frac{n-1 }{4}  r^{\frac{n-5}2}   V   \zhTB_{AB}
     \nonumber
\\
   & &
     -\frac{2}{(n+1)r^{(n+1)/2}}\partial_r(r^{n-1}TS\big[\zspaceD_A   h_{uB}])
     +\frac{2r^{\frac{n-3}{2}}}{n+1}P(\zhTB)_{AB}
     \,,
     \label{28VIII23.6}
\end{eqnarray}
Equation~\eqref{3IX22.1aHD}
can be rewritten as
 \ptcheck{29III23 }
\ptcend{varying end point would introduce a term linear in $\xi^r(r_2)$ multiplied by the integrand at $r_2$; this will introduce innocuous boundary terms multiplying a gauge term, so quadratic, so will not change anything?}
\begin{align}
    \partial_r  q_{AB}
    =
    &
     -\frac{n-1}{4}  \partial_r  (r^{\frac{n-5}2}  V)  \zhTB_{AB}
       + \frac{n-3}{n+1}r^{(n-9)/2}P(h)_{AB}
        \nonumber
\\
       &
       - r^{(n-9)/2}
        \big(
         -\frac{1}{2}\TSzlap h_{AB} + P(h)_{AB}
      - \red{\tric(h)}_{AB}
      \big)
       \nonumber
\\
   = &
   \
     \frac{1}{8}  \left(n^2-1\right)\alpha ^2
   r^{\frac{n-5}2} h_{AB}
    - \frac{(n -1)^2 }{4} m r^{-(n+5)/2}h_{AB}
    \nonumber
\\
       &
       - r^{(n-9)/2}\Big[\frac{4}{n+1} P(h)_{AB}-  \red{\tric(h)}_{AB}- \big(\frac{1}{2}\TSzlap
       -
       \frac{(n-3)(n-1)\myGauss}{8}\big)h_{AB}
       \Big]
        \,.
        \label{3IX22.1HD}
\end{align}
Under a gauge transformation 
preserving the gauge condition $\delta \beta=0$, 
the field $q_{AB}$ transforms as%
 \FGp{31V23
 \\--\\
 wan : mass checked 23VI}
\index{q@$q_{AB}$}%
\index{gauge transformation law!q@$q_{AB}$}%
\index{q@$q_{AB}$!gauge transformation law}%
 \begin{align}
     q_{AB} \mapsto\  & q_{AB}\nonumber
    \\
    &
    - r^{\frac{n-3}{2}}
    \left(
    \frac{2}{n-1} \TS[\zspaceD_A\zspaceD_B\zspaceD_C\xi^C] - \left( \frac{4}{n+1} P -\frac{n-1}{2} (\twoscsign - r^2 (\alpha^2+\ {2mr^{-n}}) \right)  C(\xi)_{AB}
    \right)
     \nonumber
\\
    &
     + \frac{r^{\frac{n-5}2}}{2} \left( (n-3)\twoscsign - (n+1) r^2 \alpha^2 \ {-\frac{(n-1)^2}{{n+1}} 2mr^{-n+2}} - \frac{4(n-3)}{(n-2)(n+1)}P\right) \TS[\zspaceD_A \zspaceD_B \xi^u]
    \,.
    \label{7III23.w1}
 \end{align}
\paragraph{Regularity.}%
\index{regularity|(}%
\Eqref{3IX22.1HD} shows  that
\ptcheck{20V24; for higher diff}
\begin{eqnarray}
\lefteqn{
  q_{AB}|_{r=r_1} \in \Hkmm
  \,,
  \
  h_{AB} \in \CrHk
  \,,
  \
   h_{uA} \in \CrHkm
  }
  &&
  \nonumber
\\
 &&
  \phantom{xxxxxxxxxxxx}
  \Longrightarrow
  \quad
   \Big( q_{AB} \in \CrHkmm
  \
  \Longleftrightarrow
  \
    \partial_u h_{AB} \in \CrHkmm
    \Big)
    \,.
 \label{28VIII23.9}
\end{eqnarray}
The differentiability class $q_{AB} \in \CrHkmm$ is preserved under gauge transformations \eqref{7III23.w1}  if
\begin{equation}\label{28VIII23.10}
  \xi^u\in  \Hkpp
  \,,
  \
  \xi^A\in  \Hkp
  \,,
\end{equation}
again consistently with \eqref{28VIII23.4}.%
\index{regularity|)}%

\subsection{The remaining Einstein equations}
\label{sec:25VII22.1}

In order to appreciate what follows the reader is invited to consult the argument involving the Bianchi identities presented at  the beginning of Section~3.5 of~\cite{ChCong1}. The reasoning presented there applies as is  in all dimensions, with irrelevant dimension-dependent changes in the equations, and therefore will not be repeated here.

\subsubsection{$\partial_u\partial_r\hBo_{uA}$}
\label{ss3VIII22.9}

The set of equations $ \mcE_{u A}=0$ can be found in~\cite{ChCongGray2} and
 is too long to be usefully displayed here.
Its linearisation $ \delta \mcE_{u A}\equiv - \delta \mcE^r{}_A +(\epsilon-\alpha^2 r^2)\delta \mcE_{rA}$  in vacuum reads
\FGp{20IV23}
\begin{eqnarray}
2 \delta \mcE_{u A}
& = &
\frac{1}{ r^2} \Bigg[
\bigg(  \frac{V}{r}(n-1)(n-2)
+ 2r(n-2)\partial_r \left(\frac{V}{r}\right)
+ r^2\partial^2_r \left(\frac{V}{r}\right)
  - R[\gamma] -\zspaceD^{B}\zspaceD_{B}
 \bigg){h}_{uA}
\nonumber
\\
&&\quad\quad
+\zspaceD^{B}\zspaceD_{A}{h}_{uB}
+ \partial_{u} \zspaceD^{B}{h}_{A B}
+ r^{5-n}\partial_{r}(r^{n-4}\zspaceD_{A}\delta V)
\nonumber
\\
&&\quad\quad
  -  r^{2}
   \Big(
\frac{V}{r^{n-2}}\partial_{r}(r^{n-3}\partial_r h_{uA})
  -
    r^2 \partial_{r}\partial_{u}\left(\frac{h_{uA}}{r^{2}}\right)
 \Big)
- 2r \left(
            r \partial_u
            \zspaceD_A\delta \beta
        \right)
      \Bigg]
   +    2 \red{\Lambda}  h_{uA}
   \nonumber
\\
&=&
\frac{1}{r^2} \bigg[
\zspaceD^{B}\zspaceD_{A}{h}_{uB} - \zspaceD^{B}\zspaceD_{B}{h}_{uA} - 2 (n-2) (r^2\alpha^2 + \frac{2m}{r^{n-2}})h_{uA}
+
\partial_{u} \zspaceD^{B}{h}_{A B}
\nonumber
\\
&&\quad\quad
  + r^{5-n}\partial_{r}(r^{n-4}\zspaceD_{A}\delta V)
  -  r^{2}
   \Big(
\frac{V}{r^{n-2}}\partial_{r}(r^{n-3}\partial_r h_{uA})
  -
    r^2 \partial_{r}\partial_{u}\left(\frac{h_{uA}}{r^{2}}\right)
 \Big)
\nonumber
\\
&& \quad\quad
- \ 2r^2 
             \partial_u 
             \zspaceD_A\delta \beta
       \bigg]
\,.
\label{9XI20.t1}
\end{eqnarray}
%

 Assuming $\delta G_{rA}=0 = \delta \beta$, using the transport equation \eqref{24VII22.1c}
 to eliminate $\partial_r^2\zhTB_{uA}$ and
the identity \eqref{20II23.2} to eliminate $\partial_r\hBo_{uu}$, we can rewrite \eqref{9XI20.t1} as
\wancheck{23VI: mass terms checked}
\ptcend{varying end point would introduce a term linear in $\xi^r(r_2)$ multiplied by the integrand at $r_2$; this will introduce innocuous boundary terms multiplying a gauge term, so quadratic, so will not change anything?}
\begin{align}
- r^{n+1}\partial_{r}\partial_{u}\left(\frac{h_{uA}}{r^{2}}\right)
& =
r^{n-3}\zspaceD^{B}\zspaceD_{A}{h}_{uB} - r^{n-3}\zspaceD^{B}\zspaceD_{B}{h}_{uA} - 2 (n-2) r^{n-1}(\alpha^2 + {2m r^{-n}}) h_{uA}
\nonumber
\\
&\quad
+  r^{n-3}\partial_{u} \zspaceD^{B}{h}_{A B}
- r^{2}\partial_{r}(r^{n-3}\zspaceD_{A}{h}_{uu})
\nonumber\\
&\quad
-r^{(n-2)}(\twoscsign - r^2(\alpha^2+ {2mr^{-n}}) )\underbrace{((n-3)\partial_r h_{uA} + r\partial^2_r h_{uA})}_{= r \zspaceD^B\partial_r\zhTB_{AB} + 2(n-2)/r h_{uA}}
\nn
\\
& =
r^{n-3}\zspaceD^{B}\zspaceD_{A}{h}_{uB} - r^{n-3}\zspaceD^{B}\zspaceD_{B}{h}_{uA}
+
r^{n-3}\partial_{u} \zspaceD^{B}{h}_{A B}
\nonumber
\\
&\quad
- r^{2}\partial_{r}(r^{n-3}\zspaceD_{A}{h}_{uu})
-r^{n-1} (\twoscsign - r^2(\alpha^2+ {2mr^{-n}})) \zspaceD^B\partial_r\zhTB_{AB}
\nonumber
\\
&\quad
- 2(n-2) r^{n-3}\twoscsign  h_{uA}
\nonumber
\\
&=
r^{n-3}\bigg( {-{2}}\zspaceD^B\TS[\zspaceD_B h_{uA}]  + \partial_{u} \zspaceD^{B}{h}_{A B}\bigg)
+ \zspaceD_A \chi
\nonumber
\\
&\quad
-r^{n-1} (\twoscsign - r^2(\alpha^2+  {2mr^{-n}}) ) \zspaceD^B\partial_r\zhTB_{AB}
\, .
\label{20II23.3}
\end{align}

\paragraph{Regularity.}%
\index{regularity|(}%
We can use \eqref{20II23.3} to determine algebraically $\partial_r\partial_u h_{uA}|_{\secNone}$ in terms of the remaining fields. We obtain 
\begin{eqnarray}
  h_{AB} \in \CrHk
  \,,
  \
   h_{uA}|_{r=r_1} \in \Hkm
 \quad
  \Longrightarrow
  \quad
  \partial_r\partial_u h_{uA}|_{r=r_1} \in \HSmthree\gamma 
    \,.
 \label{28VIII23.9ar}
\end{eqnarray}
%

\subsubsection{$\partial_u \hBo_{uu}$}
\label{ss3VIII22.8}

The equation $\mcE_{uu}=0$, to be found in~\cite{ChCongGray2}, is likewise too long to be usefully displayed here.
Its linearised version is shorter and reads
\FGp{24IV23}
\ptcend{we are not integrating this equation, so moving end will not change anything?}
\begin{eqnarray}
   0  & =
    & 2 \delta \mcE_{uu}
 \nonumber
\\
 &     = &
  \frac{1}{ r^2}\bigg[
       2\partial_{u} \zspaceD^{A} {h}_{u A}
        + \partial_r \left(\frac{V}{r}\right)  \zspaceD^{A} {h}_{u A}
        - \frac{2 V}{r^{n-1}}  \partial_{r} \bigg(r^{n-2} \zspaceD^{A} {h}_{u A}\bigg)
    +\frac{V}{r^3} \bigg(\zspaceD^{A} \zspaceD^{B} -\zR^{AB} \bigg) h_{A B}
 \nonumber
\\
 &&
 +\left(r (n-1) (\partial_{u}
 -\partial_r\left(\frac{V}{r}\right)) +  \zR + \zspaceD^{A} \zspaceD_{A}\right)  \frac{\delta V}{r}
 - \frac{(n-1)V}{r^{2n-4}} \partial_{r}(r^{2n-\blue{5}} \delta V)
 \nonumber
\\
& &
  -\frac{2 V}{r} \big(\zspaceD^A\zspaceD_A
  -(n-1)(\twoscsign (n-2) +V\partial_r -r\partial_u)\delta\beta
   \bigg]
    + 2 \Lambda h_{uu}
\, . 
\label{13VIII20.t3}
\end{eqnarray}
 This must be satisfied by all $x_{r_1}\in \dt{\secN_1}$ and $x_{r_2}\in\dt{\secN_2}$ when the linearised vacuum Einstein equations hold,
 and allows us to determine in particular $ \partial_u V$ at $r=r_1$  in terms of the remaining fields there.
%
%

 \section{Further $u$-derivatives}
\label{sec6III23.1}

Assuming now that the Einstein equations
\begin{equation}\label{24V23.1}
  \mbox{$\mcE_{rr} = 0$, $\mcE_{rA} = 0$, and $\TS[\mcE_{AB}] = 0$}
\end{equation}
are satisfied, we can obtain transport equations involving $\partial_u^i h_{uA}$ and $\partial_u^{i+1} h_{AB}$ for $i\geq 1$ by combining suitably the transport equation \eqref{21II23.1} involving $h_{uA}$ with \eqref{3IX22.1HD}  involving $\partial_u h_{AB}$ and with higher $u$-derivatives of these equations. We present the transport equations in this section; a detailed derivation can be found in Appendix \ref{App6III23.1}. In this section, we assume the gauge $\delta\beta= 0$.

\subsection{Transport equation involving $\partial_u^i h_{uA}$}
Recall that in the gauge $\delta \beta = 0$ (cf.\ \eqref{21II23.1})
\begin{align}
    \overadd{0}{\Hf}_{uA}=
      n \zhTB_{uA} + r \partial_r \zhTB_{uA}
     - \frac{1}{r^3} \zspaceD^B h_{AB}
     \,,
\ \
\mbox{with} \ \
    &
     \partial_r \overadd{0}{\Hf}_{uA}
     =
    \frac{1}{r^4}  \zspaceD^B h_{AB}
    \,.
    \label{23IV23.1b}
\end{align}
Assuming \eqref{24V23.1}, for all $i \geq 1$  the equations $\partial_u^i \mcE_{rA}=0$ are equivalent
  to the transport equations
\ptcend{varying end point would introduce a term linear in $\xi^r(r_2)$ multiplied by the integrand at $r_2$; this will introduce innocuous boundary terms multiplying a gauge term, so quadratic, so will not change anything?}
\begin{align}
    \partial_r  \overadd{i}{\Hf}_{uA}
       &=  \zspaceD^B \overadd{i}{\chi}\ofP\, r^{-(i+4)} h_{AB}
     + m^i  \zspaceD^B \overadd{i}{\chi}_{[m]} r^{-(4+i(n-1))} h_{AB}
    \nonumber
    \\
    &\quad
    + \sum_{j,\ell}^{i_*}  m^{j} \alpha^{2\ell} \zspaceD^B \overadd{i}{\chi}_{j,\ell}\ofP\, r^{-(i + 4) - j (n - 2) + 2 \ell} h_{AB} \,,
    \label{6III23.w1a}
\end{align}
 where $\sum_{j,\ell}^{i_*}$ denotes the sum over
 $j,\ell\in \N$ satisfying
 \index{s@$\sum_{j,\ell}^{i_*}$}%
 \begin{equation}
 1 \leq j \leq i-1 \,,\  j+\ell\leq i \,,\   \text{and } 0\le 2\ell\leq  i +j(n-2)
 \,.
 \label{19V23.3-1}
 \end{equation}
 When $\alpha = 0$, this sum reduces to
 \begin{equation}
     \sum_{j,\ell}^{i_*} \Big|_{\alpha= 0}=  \sum_{j=1}^{i-1}\,.
 \end{equation}
We also have%
\index{chi@$\overadd{i}{\chi}\ofP$}%
\index{chi@$\overadd{i}{\chi}_{[m]}$}%
\begin{align}
    \overadd{i}{\chi}\ofP  &= \prod_{j=1}^{i} \ck{-(j+3)}{\ofPnoP}
      \,,
    \quad
    \overadd{0}{\chi}\ofP  := 1 \,,
        \label{6III23.w2}
\\
    \overadd{i}{\chi}_{[m]} &= \prod_{j=1}^{i}  \ckm{-(4+(j-1)(n-1))}
     \,,
    \quad
    \overadd{0}{\chi}_{[m]} := 0\,,
        \label{6III23.w2b}
\end{align}
where
\index{K@$\ckm{k}$}%
\index{K@$\ck{k}{\ofPnoP}$}%
\begin{align}
    \ck{k}{\ofPnoP} &:= -\frac{1}{7 -  n + 2 k} \bigg[\frac{2 (n - 1) P}{(3 + k) (3 -  n + k) }
    +  2 \tric + \TSzlap -   (n - 4 -  k) (2 + k) \myGauss \bigg]\,,
        \label{6III23.w3}
\\
    \ckm{k} &:= \frac{2(4-n+k)^2}{7 -  n + 2 k} \,;
    \label{17IV23.1}
\end{align}
with
\begin{equation}\label{23IV23.11}
 \fbox{$k\in \mathbb{Z}$ satisfying $  k \not \in \{ -3, n-3,  \frac{n-7}{2}\}$}
 \end{equation}
(in fact the numbers $k=-3$ or $k=n-3$ do  not occur  in \eqref{6III23.w2}-\eqref{6III23.w2b}, but these values can occur in \eqref{6III23.w9} below). Thus the operators $\overadd{i}{\chi}\ofP$, respectively  $\overadd{i}{\chi}_{j,\ell}\ofP$, are polynomials in $\TSzlap$ and $P$ of order $i$, respectively  $i-j-\ell\le i-1$.
Next, the fields $\overadd{i}{\Hf}_{uA} $ take the form
$\overadd{i}{\Hf}_{uA} = n \partial^i_u \zhTB_{uA} + r \partial_r\partial_u^i \zhTB_{uA} \, +$ terms which depend on
$(r, \partial_u^{j-1} h_{uA}, \partial_r \partial_u^{j-1} h_{uA}, h_{AB},\partial^j_u h_{AB})_{j=1}^i$,
 and are defined recursively by the equations%
\index{H@$H_{uA}$!$\overadd{i}{\Hf}_{uA}$}%
\begin{align}
\overadd{1}{\Hf}_{uA} &=  \partial_u \overadd{0}{\Hf}_{uA}
    - \zspaceD^B \qh^{(-4)}_{AB} \,,
    \nonumber
    \\
    \overadd{i}{\Hf}_{uA} &=  \partial_u \overadd{i-1}{\Hf}_{uA}
    - \zspaceD^B \overadd{i-1}{\chi}\ofP \, \qh^{(-i-3)}_{AB}
    - m^{i-1}  \zspaceD^B \overadd{i-1}{\chi}_{\!\!\![m]\,} \qh^{-(4+(i-1)(n-1))}_{AB}
    \nonumber
    \\
    &\quad
    - \sum_{j,\ell}^{(i-1)_*}  m^{j} \alpha^{2\ell} \zspaceD^B \overadd{i-1}{\chi}_{j,\ell}\ofP \, \qh^{-(i+3)-j(n-2)+2\ell}_{AB}
    \nonumber
    \\
    &\quad
    -  \alpha^2
    \cka{-(i+3)}
    \hck{-(i+2)}{\TSzlap,\zdivtwo C}
    \overadd{i-2}{\Hf}_{uA} \,, \quad i\geq 2 \,,
    \label{6III23.w4}
\end{align}
where:
\begin{enumerate}
  \item
\index{U@$\tilde{U}\ofDC$}%
 the notation $U\ofP \mapsto \tilde{U}\ofDC$ (a tilde over an operator $U$) denotes the replacement in $U$ of all appearances of the operator $P := C\circ\zdivtwo$, respectively   $\tric$, by the operator $\zdivtwo\circ\, C$,  respectively   $1/2(n-2)\myGauss$;
  \item
   and where%
\index{K@$\cka{k}$}%
\begin{align}
    \cka{k} &:= \frac{(k +4) (n-k -4)}{n - 7 -2 k}\,;
    \label{6III23.w5}
\end{align}
\item
  the field $\qh_{AB}^{(k)}$ is defined as%
\index{qt@$\tilde q_{AB}$}%
\index{qh@$\qh_{AB}^{(k)}$}%
\begin{align}
    \label{5III23.1ah}
    \qh_{AB}^{(k)}& := \ E_{k} r^{k + 1} \partial_u h_{AB} + G_{k} r^{5-n+k +\frac{n-1}{k +4 - n}} \partial_r (r^{\frac{n-1}{n-k -4}+n-3} C(h_{uA}) )
    \nonumber
    \\
    &\quad \quad \quad
    + B_{k} r^{k -  \frac{n-7}{2}} \tilde q_{AB} - H_{k} r^{k} P(h)_{AB} + \frac{r^{k}}{2} (\twoscsign - \alpha^2 r^2- \frac{2m}{r^{n-2}}) h_{AB}
   \,,
\\
 \tilde q_{AB}
  & :=
  \frac{ r^{\frac{n-3}{2}}}{2}  V  \partial_r \zhTB_{AB}
     -  \frac{n-1 }{4}  r^{\frac{n-5}2}   V   \zhTB_{AB}
     +\frac{2r^{\frac{n-3}{2}}}{n+1}P\zhTB_{AB}
     \,,
\end{align}
for some non-zero numbers $E_k$, $G_k$, $B_k$ and $H_k$ to be found in \eqref{3III23.6}, the precise values of which being irrelevant at this stage.
\end{enumerate}
\paragraph{Regularity.}
We see that%
\index{regularity|(}%
\begin{eqnarray}
\lefteqn{%
\text{if } \
  h_{AB} \in \CrHk
  \,,
  \
   h_{uA} \in \CrHkm
  \,,
  \
   V \in \CrHkmm
   \,,
   }
   &&
   \nonumber
\\
   \lefteqn{\text{then we have the equivalences}
   }
 &&
    \nonumber
\\
 &&
   \Big(  \qh^{(\red{k})}_{AB} \in \CrHkmm
  \
  \Longleftrightarrow
  \
  q_{AB} \in \CrHkmm
  \
  \Longleftrightarrow
  \
    \partial_u h_{AB} \in \CrHkmm
    \Big)
    \,.
 \label{3IX23.1}
\end{eqnarray}
Furthermore
\begin{eqnarray}
    \overadd{i}{\Hf}_{uA}|_{r=r_1} \in \Hkmi
  \,,
  \
  h_{AB} \in \CrHk
  \qquad
  \Longrightarrow
  \qquad
    \overadd{i}{\Hf}_{uA} \in \CrHkmi
   \,.
   \label{2X23.1a}
\end{eqnarray}
\index{regularity|)}%

\seccheck{10XII23}
\subsection{Transport equation involving $\partial_u^i h_{AB}$}

For $ 1 \leq i $ the equations $\TS[\partial_u^{i-1} \mcE_{AB}] = 0$ are equivalent
to the transport equations
\begin{align}
    \partial_r  \overadd{i}{q}_{AB}  &=  \overadd{i}{\psi}\ofP\, r^{(n-7-2i)/2} h_{AB}
    + \alpha^{2i} \overadd{i}{\psi}_{[\alpha]} r^{(n-7+2i)/2} h_{AB}
    + m^i \overadd{i}{\psi}_{[m]} r^{\frac{n-7-2i(n-1)}{2}} h_{AB}
    \nonumber
    \\
    &\quad
     + \sum_{j,\ell}^{i_{**}}  m^{j} \alpha^{2\ell} \overadd{i}{\psi}_{j,\ell}\ofP\, r^{\frac{n-7}{2} - i -  j (n-2) + 2 \ell} h_{AB} \, ,
    \label{6III23.w6}
\end{align}
where   $\sum_{j,\ell}^{i_{**}}$ denotes the sum over $j,\ell$ with
\index{s@$\sum_{j,\ell}^{i_{**}}$}%
\begin{equation}
\begin{cases}
 1\leq j \leq i-1 \,,\  j+\ell \leq i \, , & \mbox{if $n$ is even}   \\
 1\leq j \leq i-1 \,,\  j+\ell \leq i  \, ,
 \
  \frac{n-7}{2} - i -  j (n-2) + 2 \ell \leq -4 , & \mbox{if $n$ is odd.}
\end{cases}
 \label{19V23.2b}
\end{equation}
 %
 When $\alpha = 0$, this sum reduces to
 \begin{equation}
     \sum_{j,\ell}^{i_{**}} \quad \overset{\alpha= 0}{=} \quad  \sum_{j=1}^{i-1}\,.
 \end{equation}
In \eqref{6III23.w6}, the field $\overadd{i}{q}_{AB} $ is of
the form  $\overadd{i}{q}_{AB} = r^{\frac{n-1}2}\partial_u^i\zhTB_{AB} \,+ $ terms which depend on fields of lower $u$-derivatives, specifically,
$(r ,\partial^{j-1}_u h_{AB} , \partial_u^{j-1} h_{uA}, \partial_r \partial_u^{j-1} h_{uA})_{j=1}^i$; the operators
$\overadd{i}{\psi}\ofP$ and $\overadd{i}{\psi}_{j,\ell}\ofP$ are polynomials in $\TSzlap$ and $P$ of orders $i$ and $i-j-\ell$ respectively; $\overadd{i}{\psi}_{[\alpha]}$ and $\overadd{i}{\psi}_{[m]}$ are constants. For $i \geq 2$ these are given by the recursion relations%
\index{qi@$\overadd{i}{q}_{AB}$}%
 \ptcheck{17VI; crosschecked with the appendix}
\begin{align}
    \overadd{i}{q}_{AB}&= \partial_u \overadd{i-1}{q}_{AB}
    - \overadd{i-1}{\psi}\ofP \, \qh_{AB}^{(n-5-2i)/2}
    - \alpha^{2(i-1)} \overadd{i-1}{\psi}_{[\alpha]} \qh_{AB}^{(n-9+2i)/2}
    \nonumber
    \\
    &\quad
    - m^{i-1} \overadd{i-1}{\psi}_{[m]} \qh^{(\frac{n-7-2(i-1)(n-1)}{2})}_{AB}
    - \alpha^2 \mathcal{\widehat K}(2(i-1),P)  \overadd{i-2}{q}_{AB}
    \nonumber
    \\
    &\quad
    - \sum_{j,\ell}^{(i-1)_{**}}  m^{j} \alpha^{2\ell} \overadd{i-1}{\psi}_{j,\ell}\ofP \, \qh^{(\frac{1}{2} (n-7- 2 (i-1) - 2 j (n-2) + 4 \ell))}_{AB}
    \,,
    \label{6III23.w8}
\end{align}
\index{psi@$\overadd{i}{\psi}\ofP$}%
\index{psi@$\overadd{i}{\psi}_{[\alpha]}$}%
\index{psi@$\overadd{i}{\psi}_{[m]}$}%
\begin{align}
    \overadd{i}{\psi}\ofP  & = \prod_{j=2}^{i} \ck{\tfrac{n-5-2j}{2}}{\ofPnoP}\overadd{1}{\psi}\ofP \, \,,
    \quad
    \overadd{i}{\psi}_{[\alpha]} = \prod_{j=2}^{i}\cka{\tfrac{n-9+2j}{2}} \overadd{1}{\psi}_{[\alpha]} \,,
    \label{6III23.w9}
\\
    \overadd{i}{\psi}_{[m]} &=\prod_{j=2}^{i} \ckm{\tfrac{n-7-2(j-1)(n-1)}{2}} \overadd{1}{\psi}_{[m]}
    \,,
    \label{6III23.w10}
\end{align}
with the initiating functions
\begin{align}
    \overadd{0}{q}_{AB}&= 0 \,, \quad
    \overadd{1}{q}_{AB}= q_{AB} \,,
\\
    \overadd{1}{\psi}\ofP  &= - \bigg[\frac{4}{n+1} P -  \tric - \frac{1}{2}\TSzlap
       +\frac{(n-3)(n-1)\myGauss}{8}\bigg]\,,
    \label{6III23.w7a}
\\
    \overadd{1}{\psi}_{[\alpha]} &= \frac{1}{8}(n^2-1)\,,
    \quad
    \overadd{1}{\psi}_{[m]} = -\frac{(n-1)^2}{4} \,,
    \label{6III23.w7b}
\end{align}
and where
\index{K@$\mathcal{\widehat K}(2i,P)$}%
\begin{align}
    \mathcal{\widehat K}(2i,P) &:= \ck{\tfrac{n-5-2i}{2}}{\ofPnoP} \cka{\tfrac{n-7-2i}{2}}\,.
\end{align}
Equations \eqref{6III23.w8}-\eqref{6III23.w10} holds for all $i\in \mathbb{Z}^+$ when $n$ \textbf{is even} and until $i=\frac{n-1}2$ when $n>3$ \textbf{is odd}, with $\hat{q}{}^{(-3)}_{AB}$ being given by~\eqref{5III23.1} where necessary. The case $n=3$ is special and has been analysed in~\cite{ChCong1}.

When $n\neq 3$ \textbf{is odd}, at order $i = (n+1)/2$, we have
\begin{align}
    \overadd{\frac{n+1}{2}}{\psi}\ofP = \zck{-3}{\ofPnoP} \overadd{\frac{n-1}{2}}{\psi}\ofP\,,
    \quad
     \overadd{\frac{n+1}{2}}{\psi}_{[\alpha]} = 0\,,
     \label{17IV23.4}
\end{align}
with the second equation agreeing with the second equation of~\eqref{6III23.w9} after noting that $\cka{n-4}=0$. For $i \geq (n+1)/2$, the recursion formula for $\overadd{i}{\psi}\ofP$ and $\overadd{i}{\psi}_{[\alpha]}$ is as given in~\eqref{6III23.w9} (thus $\overadd{i}{\psi}_{[\alpha]}=0$), but with $\ck{-3}{\ofPnoP}$ replaced by $\zck{-3}{\ofPnoP}$ where necessary.
Meanwhile,
$\overadd{i}{q}_{AB}$ is given by
\begin{align}
    \overadd{i}{q}_{AB}&= \partial_u \overadd{i-1}{q}_{AB}
    - \overadd{i-1}{\psi}\ofP \, \qh_{AB}^{(n-5-2i)/2}
    - m^{i-1} \overadd{i-1}{\psi}_{[m]} \qh^{(\frac{n-7-2(i-1)(n-1)}{2})}_{AB}
    \nonumber
    \\
    &\quad
    - \sum_{j,\ell}^{(i-1)_{**}}  m^{j} \alpha^{2\ell} \overadd{i-1}{\psi}_{j,\ell}\ofP \, \qh^{(\frac{1}{2} (n-7- 2 (i-1) - 2 j (n-2) + 4 \ell))}_{AB}
    \,.
    \label{17IV23.3}
\end{align}

The following proposition, which we prove in Appendix~\pref{ss24IX23.1}, will play a key role in what follows:

\begin{Proposition}
\label{P9X23.1m}
We have, for $n\ge 5$ odd and $j\ge0$, and for any vector field $W$ and symmetric two-covariant tensor field $h$,
\begin{eqnarray}
\label{16X.f1m}
     & \overadd{\frac{n-3}{2}+j}{\psi}\ofP  \, h^{\red[S]} \equiv 0 \,, \quad  \overadd{\frac{n-1}{2}+j}{\psi}\ofP \, h^{\red[V]} \equiv 0 \,,
     &
     \\
\label{12VI.1m}
&\overadd{\frac{n-1}{2}+j}{\psi}\ofP  \circ\, C(W) \equiv 0 \,,
 &
\\
 &
\label{5X23f.1m}
\overadd{\frac{n-1}{2}+j}{\psi}\ofP \circ P  \, (h) \equiv 0 \,.
 &
\end{eqnarray}
\end{Proposition}

\paragraph{Regularity.}%
\index{regularity|(}%
We have
\begin{eqnarray}
 h_{AB} \in \CrHk
 \,,
 \
 \forall \ 0\le i\le \ell
 \
  \overadd{i}{q} _{AB}|_{r=r_1}
 \in \Hkmmi
 \quad
  \Longrightarrow
  \quad
 \overadd{\ell}{q}_{AB}
 \in \CrHkmmell
    \,,
 \label{2X23.2}
\end{eqnarray}
as well as the equivalence
\begin{eqnarray}
  \partial_u^i h_{AB} \in \CrHk
   \qquad
  \Longleftrightarrow
   \qquad
 \overadd{i}{q}_{AB} \in   \CrHkmmi
 \label{2X23.5}
\end{eqnarray}
whenever all the previous equations hold and 
\begin{eqnarray}
 h_{AB} \in \CrHk
 \,,
 \
h_{uA} \in \CrHkm
 \,,
 \
\delta V  \in \CrHkmm
 \,,\
 \text{and }
  \nonumber
\\
 \partial_u^j h_{AB}|_{r=r_1} \in
  \Hkmmj \ \text{ for $0\le j \le i$}
    \,.
 \label{2X23.4}
\end{eqnarray}
\index{regularity|)}%

\paragraph{Radial charges.} We end this section by introducing two further radial charges in the \underline{case $m=0$}.
 The first arises from taking a combination of the fields $\overadd{\frac{n-3}2}{q}_{AB}$ and $\chi$. From their respective transport equations \eqref{6III23.w6}, with $i=\frac{n-3}2$, and \eqref{1XI22.w1}, it can be readily seen that we have,%
 \FGp{23X}
\index{Q@$\kQ{3}{}$}%
\begin{align}
    \partial_r
    \underbrace{\bigg(
    \mrL \big( \overadd{\frac{n-3}2}{q}{} ^{[\ker \overadd{\frac{n-3}2}{\psi}]} )
    +   \alpha^{n-3} \frac{n-1}{n-3} \overadd{\frac{n-3}2}{\psi}_{[\alpha]} \chi
    \bigg)}_{
    =:\kQ{3}{}
    } = 0 \,.
    \label{17X23.1}
\end{align}
While relatively straightforward, we note that in arriving at \eqref{17X23.1}, we made use of the 
facts\footnote{Thus we could have replaced $\mrL \big( \overadd{\frac{n-3}2}{q}{}_{AB}^{[\ker \overadd{\frac{n-3}2}{\psi}]} )$ in \eqref{17X23.1} by $\mrL \big( \overadd{\frac{n-3}2}{q}{}_{AB})$, but we chose to keep the projection to emphasise that the first term on the right-hand side of \eqref{6III23.w6}, with $i=\frac{n-3}2$, will not contribute to the right-hand side of \eqref{17X23.1}.}
\begin{align}
   \Big(
   \underbrace{ h\in  (\ker \overadd{\frac{n-3}2}{\psi})^\perp \Rightarrow h = h^{[V]}+  h ^{[\TTt]} }_{\text{cf.\ \eqref{16X.f1m}}}\,, \quad \text{and}
   \quad
   & \mrL(h^{[V]} )
   = 0= \mrL(h^{[\TTt]} )
   \Big)
   \nonumber
   \\
   &
    \implies  \mrL(h^{[\ker \overadd{\frac{n-3}2}{\psi} ]} ) = \mrL(h ) \,.
    \label{18X23.1}
\end{align}
 \ptcheck{23X}

The next radial charge arises from taking a combination of the fields $\overadd{\frac{n-1}2}{q}_{AB}$ and $\overadd{*}{\Hf}_{uA}$. From their respective transport equations \eqref{6III23.w6}, with $i=\frac{n-1}2$, and \eqref{24VII22.1b},
again \underline{in the case $m=0$}, it can be readily seen that we have,%
\index{Q@$\kQ{4}{}$}%
\ptcheck{23X}
\begin{align}
    \partial_r
    \underbrace{\bigg(
    \zspaceD^B \overadd{\frac{n-1}2}{q}{}_{AB}
    +    \frac{\alpha^{n-1}}{n-1} \overadd{\frac{n-1}2}{\psi}_{[\alpha]} \overadd{*}{\Hf}_{uA}
    \bigg)}_{
    =:\kQ{4}{A}
    } = 0 \,,
    \label{17X23.4}
\end{align}
where we note that \eqref{12VI.1m} with $j=0$ was used for the vanishing contribution of the first term on the right-hand side of \eqref{6III23.w6} to the right-hand side of \eqref{17X23.4}. 

\subsection{Gauge transformations}
\label{ss9IX23.1}

The recursion formula \eqref{6III23.w6} gives another family of radially conserved charges for $i\in \Z^+$ \underline{when $m=0=\alpha$}: 
\begin{align}
    \partial_r \int_{\secN} \overadd{i}{\kerpsi}{}^{AB} \overadd{i}{q}_{AB} \sm = 0\,,
    \label{23IV23.2}
\end{align}
where $\overadd{i}{\kerpsi}$ satisfies%
\index{m@$\overadd{i}{\kerpsi}$}%
\begin{align}
    0=\overadd{i}{\psi}\ofP ^\dagger  (\overadd{i}{\kerpsi})
    \equiv    \overadd{i}{\psi}\ofP  (\overadd{i}{\kerpsi})
    \,.
    \label{23IV23.3}
\end{align}
\seccheck{10XII23}

We now move on to determine the gauge dependence of these radial charges. In the rest of this section, unless otherwise indicated, we set \underline{$m=0$}. Dimension considerations show that the $r$-dependence of gauge fields in the gauge transformation of $\overadd{i}{q}_{AB}$ \underline{when $\alpha = 0$} is determined by terms of the form
\begin{align}
r^{\frac{n-3}{2}-i} \xi^u\, \quad \text{and}\quad r^{\frac{n-3}{2}-j} \partial_u^{i-1-j} \xi^A
\label{1V23.1}
\end{align}
for $i\geq 1$ and $0\leq j \leq i-1$,  up to differential operators acting in the $x^A$-variables. However, it follows from the radial conservation of these charges that their gauge transformations must  be $r$-independent. Clearly, from \eqref{1V23.1}, when $n$ is even,
no  gauge term in the gauge transformation of $\overadd{i}{q}_{AB}$ will contain $r$-independent terms,
whatever $i$. This implies in particular that we must have: when $n$ is even, under gauge transformations,%
\index{gauge transformation law!q@$\overadd{i}{q}_{AB}$}%
\index{qi@$\overadd{i}{q}_{AB}$!gauge transformation law}%
\begin{align}
    \int_{\secN} (\overadd{i}{\kerpsi})^{AB} \overadd{i}{q}_{AB}  \sm \rightarrow \int_{\secN} (\overadd{i}{\kerpsi})^{AB} \overadd{i}{q}_{AB} \sm \,,
    \label{23IV23.5}
\end{align}
and \eqref{23IV23.2} gives a family of gauge-invariant radially conserved charges.
The transformation \eqref{23IV23.5} remains true for \underline{$\alpha\neq 0$} by a similar argument when $n$ is even.

When $n$ \textit{is odd}, the smallest $i$ for which an $r$-independent gauge transformation is possible is $i=\frac{n-3}{2}$. In what follows, we will look at the transformations for $\overadd{i}{q}{}^{[\ker\overadd{i}{\psi}]}_{AB}$ in the cases $i\in\{\frac{n-3}{2}\,,\frac{n-1}{2}\,,\frac{n+1}{2}\}$ and $i>\frac{n+1}{2}$ for odd $n>3$.

\paragraph{Transformation when $i=\frac{n-3}{2}$.}
For this value of $i$ it follows from \eqref{1V23.1} that the transformation depends only on the gauge field $\xi^u$ \underline{when $\alpha = 0$}. From the recursion formula \eqref{6III23.w8} of $\overadd{i}{q}_{AB}$, $i\geq 2$, and the gauge transformation \peqref{23IV23.1} below, of
$\qh_{AB}^{(k)}$, the dependence on $\xi^u$ only comes from the term $\overadd{i-1}{\psi}\ofP \, \qh^{((n-5-2i)/2)}_{AB}$. This gives for $i=\frac{n-3}{2}$, $i\geq 2$ (and hence $n>5$),
\index{L@$\hLop$}%
\begin{align}
    \int_{\secN}  \overadd{\frac{n-3}{2})}{\kerpsi}{}_{AB}
     &
      \overadd{\frac{n-3}{2}}{q}{}^{AB}  \sm
    \rightarrow
    \int_{\secN} \overadd{\frac{n-3}{2}}{\kerpsi}{}^{AB} \overadd{\frac{n-3}{2})}{q}_{AB} \sm
    +\int_{\secN} \overadd{\frac{n-3}{2}}{\kerpsi}{}^{AB} {\hLop_{n}}(\xi^u)_{AB} \sm
    \,,
    \label{24IV23.1}
\end{align}
where
 \ptcheck{4X23}
\index{Ln@$\hLop_{n}$}%
\begin{align}
	\hLop_{n}(\xi^u)_{AB}:=
 - \frac{n-4}{(n-5)(n-2)^2} \overadd{\frac{n-5}2}{\psi}\ofP \, \bigg( (n-1) P - 2 (n-2)^2 \myGauss
		\bigg) \TS[\zspaceD_A\zspaceD_B \xi^u]
    \,. 
    \label{24IV23.1xs}
\end{align}
%
In the case $n=5$ we have $i=\frac{n-3}{2}=1$, so we must use the gauge transformation formula \eqref{7III23.w1} of $ \overadd{1}{q}_{AB}= q_{A B}$. This leads to
 \ptcheck{4X23}
\begin{equation}\label{6VI23.f2}
	\hLop_{5}(\xi^u)_{AB}:= -\left(\frac{2}{9}P-\myGauss\right) \TS[\zspaceD_A\zspaceD_B \xi^u]\,.
\end{equation}

It follows from \eqref{24IV23.1} and the definition \eqref{17X23.1} of $\kQ{3}{}$  that \underline{when $m=0$ and for any $\alpha$}, under a gauge transformation, we have%
\index{Q@$\kQ{3}{}$}%
\index{gauge transformation law!Q@$\kQ{3}{}$}%
\index{Q@$\kQ{3}{}$!gauge transformation law}%
\begin{align}
\label{17X23.2}
    \kQ{3}{} \rightarrow \kQ{3}{} + \mrL   \circ \hLop_{n}(\xi^u)\,.
\end{align}

Recall that $\alpha$ has the same   dimension as $r^{-1}$, and that the fields $\overadd{\frac{n-3}{2}}{q}_{AB}$, and $\xi^u$  both have the same   dimension  as $r$, while field $\xi_A$ is dimensionless. Since $\alpha$ appears in the gauge transformations (cf.\ Section~\ref{ss11III23.1}) 
and in the recursion formulae (cf.\ Section~\ref{sec6III23.1}) with positive powers only,
we see that there will be no additional $r$-independent contributions to \eqref{6VI23.f2} and \eqref{17X23.2} \underline{when $\alpha\neq 0$}. 

For later convenience, we define%
\index{Ln@$\underline{\hLop}_{n}$}%
\begin{equation}
    \underline{\hLop}_{n}(\xi_A)
    := \begin{cases}
        -\left(\frac{2}{9}P-\myGauss\right) \circ\, C(\xi) \,, & n=5 
         \\
        - \frac{n-4}{(n-5)(n-2)^2} \overadd{\frac{n-5}2}{\psi}\ofP \, \bigg( (n-1) P - 2 (n-2)^2 \myGauss
		\bigg) \circ\, C(\xi)
    \,,
    & n>5\,,
    \end{cases}
     \label{11V23.21}
\end{equation}
so that
\begin{equation}\label{5XI23.31}
    \hLop_{n}(\xi^u) \equiv  \underline{\hLop}_n(\zspaceD_A\xi^u)
 \,.
\end{equation}
Clearly, $\CKV\in \ker \underline{\hLop}_n$.

\paragraph{Transformation when $i=\frac{n-1}{2}$.} A similar analysis as the above for $i=\frac{n-1}{2}$ with {$n\geq5$} gives
\begin{align}
    \int_{\secN} (\overadd{\frac{n-1}2}{\kerpsi})^{AB} \overadd{\frac{n-1}2}{q}_{AB}  \sm
    \rightarrow &
    \int_{\secN} (\overadd{\frac{n-1}2}{\kerpsi})^{AB} \overadd{\frac{n-1}2}{q}_{AB} \sm
    + \int_{\secN} (\overadd{\frac{n-1}2}{\kerpsi})^{AB}
    \Lop(\xi)_{AB}
    \sm
    \,,
    \label{24IV23.2}
\end{align}
where
\index{Ln@$\Lop$}%
\begin{align}
    \Lop(\xi)_{AB} &:= -\overadd{\frac{n-5}2}{\psi}\ofP \, \bigg[ \tfrac{\mathcal{K}_n\ofP}{(n-3)(n-1)}
        \big(
             2(2-n)(2P+(1-n)\myGauss) C(\xi)_{AB}
            + 4 \TS[\zspaceD_A\zspaceD_B\zspaceD_C\xi^C]
        \big)\bigg]
    \nonumber
\\
   & \qquad\qquad\qquad
   +{\frac{1}{n-1} 	\hLop_{n}(\zspaceD_C\xi^C)_{AB}}
    \,,
    \label{24IV23.3}
\end{align}
with%
\index{K@$\mathcal{K}_n\ofP$}%
\begin{align}
    \mathcal{K}_n\ofP := \begin{cases}
        \overadd{1}{\psi}\ofP \,, & n = 5 \\
        \ck{-1}{\ofPnoP} \,, & \text{otherwise\,.}
    \end{cases}
\end{align}

\begin{remark}
 \label{R16X23.1}
 {\rm
 \ptcheck{16X23}
 For further use we note that \eqref{24IV23.3} can be simplified to
 \begin{equation}\label{5XI23.41}
    \Lop = \overadd{\frac{n-5}{2}}{\psi}\ofP \, \underline{\Lop}
 \,,
 \end{equation}
 with $\overadd{0}{\psi}:=1$ and
\index{Ln@$\underline{\Lop}$}%
\begin{align}
     \underline{\Lop}(\xi)_{AB} &=
     \begin{cases}
        \frac{1}{8} (\TSzlap +2 \tric-4 \myGauss ) (\TSzlap +2 \tric-6 \myGauss )  C(\xi)_{AB}
        &
        \\
        \qquad\qquad\qquad
        - \frac{1}{6}  (\TSzlap + 2 \tric - 5 \myGauss ) \TS[\zspaceD_A\zspaceD_B\zspaceD_C\xi^C] \,, & n=5
        \\
        \frac{1}{(n-1)(n-5)}
     \bigg( (\TSzlap  +2 \tric -2 (n-2) \myGauss) (\TSzlap +2 \tric +(1-n) \myGauss)
     C(\xi)_{AB} &
     \\
     \qquad\qquad\qquad
     -\frac{2 (n-3)}{n-2} (\TSzlap +2 \tric +(5-2 n) \myGauss )
     \TS[\zspaceD_A\zspaceD_B\zspaceD_C\xi^C]
     \bigg) \,, & n\neq 5\,.
     \end{cases}
     \label{30VI.1}
\end{align}
}
Clearly (cf.\ Lemma \ref{L21X23.1}), we have $\CKV\in \ker\underline{\Lop}$.
\qed
\end{remark}

It follows from \eqref{24IV23.2} and the definition of $\kQ{4}{}$ in \eqref{17X23.4} that \underline{when $m=0$ and for any $\alpha$}, under a gauge transformation  we have%
\index{Q@$\kQ{4}{}$}%
\index{gauge transformation law!Q@$\kQ{4}{}$}%
\index{Q@$\kQ{4}{}$!gauge transformation law}%
\begin{align}
\label{17X23.6}
    \kQ{4}{} \rightarrow \kQ{4}{} + \zspaceD^A\Lop(\xi)_{AB}\,.
\end{align}
This is justified by a similar argument to that below \eqref{17X23.2} when $\alpha\neq 0$.

\paragraph{Transformation when $i=\frac{n+1}{2}$.} Returning to our main line of thought, we continue with the case $i=\frac{n+1}{2}$, $ n\geq5$, which gives
\begin{align}
    \int_{\secN}
     &
     (\overadd{(n+1)/2}{\kerpsi})^{AB}
     \overadd{(n+1)/2}{q}_{AB}  \sm
    \rightarrow
    \int_{\secN} (\overadd{(n+1)/2}{\kerpsi})^{AB} \overadd{(n+1)/2}{q}_{AB} \sm
    \nonumber
    \\
    &
    + \int_{\secN} (\overadd{(n+1)/2}{\kerpsi})^{AB}
  \Big(
    \Lop(\partial_u\xi)_{AB}
    +
    \big(n-\frac{2}{n}-\frac{2}{n-1}\big)
   \underbrace{
     \overadd{\frac{n-1}2}{\psi}\ofP C(\partial_u\xi)_{AB}
     }_{=0}
    \Big)
   \,   \sm
    \,,
    \label{24IV23.4}
\end{align}
where the underbraced term arises from \peqref{15V23.1} and
vanishes by Proposition~\pref{P9X23.1} below.

 \FGp{5X23}

\paragraph{Transformation when $i>\frac{n+1}{2}$.} Finally for $i=\frac{n+3}{2}+j$ with $j\geq 0$, $n\geq5$, we have
\begin{align}
    \int_{\secN} (\overadd{(n+3)/2+j}{\kerpsi})^{AB} \overadd{(n+3)/2+j}{q}_{AB}  \sm
    \rightarrow &
    \int_{\secN} (\overadd{(n+3)/2+j}{\kerpsi})^{AB} \overadd{(n+3)/2+j}{q}_{AB} \sm
    \nonumber
    \\
    &
    + \int_{\secN} (\overadd{(n+3)/2+j}{\kerpsi})^{AB}
   \Lop(\partial^{j+2}_u\xi)_{AB}
    \sm
    \,.
    \label{25IV23.1}
\end{align}

\paragraph{The fields $\protect\overadd{i}{\Hf}_{uA}$.} We continue with an analysis of the fields $\protect\overadd{i}{\Hf}_{uA}$. We shall work out  the $r$-independent parts of the associated gauge transformations, which is relevant for gluing. Recall that, \underline{for any $m$ and $\alpha$}, the field $\protect\overadd{1}{\Hf}_{uA}$ is given by (see \eqref{6III23.w4})
\begin{equation}\label{2VI23.31}
  \overadd{1} H_{uA} = \partial_u \overadd{0} H_{uA} - \zspaceD^B \qh^{(-4)}{}_{AB}
   \,.
\end{equation}
After making use of the gauge transformation law \eqref{28II23.1} of $\overadd{0}{\Hf}_{uA}$, and that of $\qh_{AB}^{(-4)}$ in \eqref{23IV23.1}, we  obtain, \underline{for any $m$ and $\alpha$}, 
\FGp{ 14VI23}
\begin{align}
    \overadd{1}{\Hf}_{uA} \rightarrow 
    \overadd{1}{\Hf}_{uA}
    &
     + n \partial_u^2 \xi_A + \alpha^2 n \underbrace{\big(\frac{1}{n-1} \zspaceD_A \zspaceD_B \xi^B + \frac{2}{n+1} \zspaceD^B C(\xi)_{AB}\big)}_{=: \overadd{1,1}{D}(\xi)_A}
 \nonumber
\\
 &    +
    \text{ $r$-dependent terms}
    \,.
    \label{25V23.2}
\end{align}
Making use of \eqref{6III23.w4}, it can be shown inductively that for $i\geq 1$ and \underline{any $m$ and $\alpha$},%
\index{D@$\overadd{i}{D}(\partial^{i-1}_u\xi)$}
\FGp{ 14VI23:}
\begin{align}
    \overadd{i}{\Hf}_{uA} &\rightarrow \overadd{i}{\Hf}_{uA}
    + n \partial^{i+1}_u \xi_A
    + \alpha^2 n \underbrace{
    \Big( \overadd{1}{D} + \sum_{j=1}^i \cka{-(j+3)} \hck{-(j+2)}{\TSzlap,\zdivtwo\circ\, C}
    \Big)(\partial^{i-1}_u\xi)_A
    }_{=:\overadd{i,1}{D}(\partial^{i-1}_u\xi)_A}
    \nonumber
\\
    &\quad
    + O(\alpha^4)
    +
    \text{ $r$-dependent terms}
    \,.
    \label{26V23.1}
\end{align}
\seccheck{11XII23} 
The $O(\alpha^4)$ terms depend on the fields $\xi^u$ and $\partial_u^j\xi^A$ with $1\leq j \leq i-3$,
and their explicit form is not needed for the arguments that follow. 
 
\subsection{Summary on regularity}
\label{ss26IV24.1}
From what has been said so far about the regularity of the fields, assuming
\begin{align}
    h_{AB} \in \CrHk\,,
    \label{23VI24.1}
\end{align}
the following regularity at $r=r_1$ will be preserved by the various transport equations along $\mcN$: for $0\le i \le \kgamma /2-1$,
\begin{align}
  \partial_u^i \partial_r h_{uA}, \partial_u^i h_{uA} &\in \Hkmi
  \\
  \partial_u^i\delta\beta\,, \partial_u^i h_{AB} &\in
  \Hkmmi
  \,, \quad
  \partial_u^i \delta V \in H^{k_{\gamma} - 2i - 2}(\secN)\,.
\end{align}
In addition these will be preserved under gauge transformations by the following regularity of the gauge fields: for $0\le i \le \kgamma /2-1$,
\begin{equation}
  \partial_u^i\xi^u\in  H^{k_{\gamma} - 2i +2}(\secN)
  \,,
  \
  \partial_u^i \xi^A\in   H^{k_{\gamma} - 2i +1}(\secN)
  \,.
  \label{23IV24.2}
\end{equation}
Recall the definition of  linearised Bondi sphere data from \eqref{23III22.992}:
\begin{align}
  x&= (\partial_u^{j}h_{AB}|_{\secN}
  ,\,  \partial_r^jh_{AB}|_{\secN}
  ,\,  \partial_u^{j}\delta\beta|_{\secN},
  \, \partial_u^{j}\delta U^A|_{\secN}
  ,\, \partial_r \delta U^A|_{\secN}
  ,\,   \delta V|_{\secN}
  )_{0\le j\leq k} \in \dt{\secN}
  \,.
\end{align}
In view of the above, we define the following spaces for linearised Bondi sphere data and gauge fields:
\index{H@$\Hkdt$}%
\index{H@$\Hkzeta$}%
\begin{align}
\Hkdt &:= \prod_{j\in[0,k]}
    \bigg( H^{k_{\gamma}-2j}(\secN)
     \times
     \left\{
       \begin{array}{ll}
           H^{k_{\gamma}-j}(\secN), & \hbox{$j=0$} \\
           H^{k_{\gamma}+1-j}(\secN), & \hbox{$j>0$}
       \end{array}
     \right.
    \times
    H^{k_{\gamma}-2j}(\secN) \times
    H^{k_{\gamma}-2j-1}(\secN)\bigg)
    \nn
    \\
    &\qquad \qquad \qquad
    \times
    H^{k_{\gamma}-1}(\secN) \times
    H^{k_{\gamma}-2}(\secN) \,,
    \\
    \Hkzeta :&= \prod_{j\in[0,k+1]}H^{k_{\gamma}-2j+1}(\secN)\times H^{k_{\gamma}-2j+2}(\secN) \,.
     \label{10V24.11}
\end{align}
 We note that the gauge transformation map $\Gmap$ of \eqref{12V24.1}, 
\index{G@$\Gmap$}%
\begin{align}
\tilde x = \Gmap(x)
   \,,
\end{align}
is a linear map preserving the regularity of the given data:
\begin{equation}
   \mbox{if} \ (x,z)\in \Hkdt\times\Hkzeta
   \ \mbox{then}
   \
   \Gmap(x)
   \in 
    \Hkdt\,.
     \label{14V24.1}
\end{equation}

\section{Gluing up to gauge}
 \label{s12I22.1}

We now present a scheme for gluing, up to residual gauge,   the linearised fields
\begin{equation}\label{19XII22.1}
\{\hBo_{\mu\nu},
 \partial_u\hBo_{\mu\nu}\,, \ldots\,, \partial_u^\bluek \hBo_{\mu\nu}\}
\end{equation}
in Bondi gauge,   with $2\le \bluek  < \infty$.
 We will assume for simplicity that each of the fields $ \partial_u^i \hBo_{\mu\nu}\big|_{\{u=0\}}$, $0\le i \le \bluek $, is smooth. The collection of fields of this differentiability class will be denoted by $\Ck$. The case of finite Sobolev regularity will be discussed in the next section.

The formulation of the problem follows that of~\cite[Section 4]{ChCong1}, we repeat it here for
the convenience of the reader: Let $0\le r_0<r_1<r_2<r_3\in \R $.
Consider two sets of vacuum linearised gravitational fields  in Bondi gauge
 of, say for simplicity, $\Ck$-differentiability class,
  defined in spacetime neighborhoods of $\mcN_{(r_0,r_1]} $ and $\mcN_{[r_2,r_3)} $. Denote by $\secN_{1}$ the section of $\mcN_{(r_0,r_1]}$ at $r=r_1$. The linearised gravitational field near $\mcN_{(r_0,r_1]}$
 induces an element, say  $x_{1}$, of the  set $\dt{\secN_1}$ of Bondi cross-section data (cf.\ \eqref{23III22.992}).
 Similarly, we denote by $\secN_{2}$ the section of $\mcN_{[r_2,r_3)}$ at $r=r_2$ and the induced  data by $x_2\in \dt{\secN_2}$.
  To take into account gauge transformation, $\tilde{\secN}_1$ (resp. $\tilde{\secN}_2$) will denote the cross-section obtained by gauge-transforming $\secN_{1}$ (resp. $\secN_2$).
  The associated Bondi data is denoted by $\tilde x_1 \in \dt{\tilde{\secN}_1}$
  (resp. $\tilde x_2 \in \dt{\tilde{\secN}_2}$), and the outgoing null hypersurface on which it lies by $\tmcN _{(r_0,r_1]}$ (resp.
 $\tmcN_{[r_2,r_3)}$).

 The  goal
is to glue $\tilde x_1$ and $\tilde x_2$
 along $\tmcN _{[r_1,r_2] }$ so that the resulting linearised field on $\tmcN _{(r_0,r_3)} $ provide smooth characteristic data for Einstein equations
 together.  Indeed, we will show that
a  $\Ck$-linearised vacuum data set on $\mcN_{(r_0,r_1]}$ can be smoothly glued to another  such set on $\mcN_{[r_2,r_3)}$ up to gauge if and only if the obstructions listed in Tables~\ref{T11III23.2}-\ref{T11XII23.1} in the Introduction  are satisfied.

Let $\interph _{AB}$ be any symmetric traceless tensor field  defined on a neighbourhood of $\mcN_{[r_1,r_2]}$ which interpolates between the original fields $ \hBo_{AB}|_{\mcN_{(r_0,r_1]}}$  and $  \hBo_{AB}|_{\mcN_{[r_2,r_3)}}$,
so that the resulting field on $\mcN_{(r_0,r_3)}$ is as differentiable as the original fields.
When attempting a $\Ck$-gluing, we can add to $\interph _{AB} $  a field $\wh_{AB}|_{[r_1,r_2]}$
which vanishes  smoothly
(i.e.\ together with  $r$-derivatives of all orders)
at the end cross-sections $\{r_1\}\times \secN$ and  $\{r_2\}\times \secN$ without affecting the smoothness of $h_{AB}$.
To take into account the gauge freedom, let $\phi(\tdr)\ge 0$ be a function which equals $1$ near $\tdr=r_1$ and equals $0$ near $\tdr=r_2$. Let $\kxi{1}{}^u$ and $\kxi{1}^A$ be gauge fields which will be used to gauge the metric around $\mcN_{(r_0,r_1]}$, and let $\kxi{2}{}^u$ and $\kxi{2}{}^A$ be  gauge fields which will be used to gauge the metric around $\mcN_{[r_2,r_3)}$. For $r_1\le \tdr \le r_2$ we set%
\index{v@$\interph_{AB}$}%
\index{w@$\wh_{AB}$}
\index{h@$\tilde\hBo_{AB}$}%
\begin{equation}\label{16III22.2old}
  \tilde \hBo_{AB} = \interph _{AB}+ \wh_{AB}
   +  \phi   \tdr^2 \TS[
	\TSoLie_{\kzeta{1}} \ringh_{AB} ] + (1-\phi )  \tdr^2 \TS [
	\TSoLie_{\kzeta{2}} \ringh_{AB} ]
  \,.
\end{equation}
(Recall that $\TSxip^A=\xi^A -\zspaceD^A \xi^u/r$, cf.\ \eqref{1VIII22.1HD}.)

In the gluing problem, the gauge fields evaluated on $\tilde{\secN}_{a}\,, a=1,2$ and the field $\wh_{AB}$ on $\tmcN _{(r_1,r_2)}$ are \textit{free fields} which can be chosen arbitrarily. Our aim in what follows is to show how to choose these fields to solve the transport equations of Section \ref{sec:28VII22.1} and Section \ref{sec6III23.1} to achieve gluing-up-to-gauge.

For the $\Ck$-gluing we will need smooth functions%
\index{kappa@$\kappa_i$}%
\index{iota@$\iota_{\alpha,m}$}%
\begin{align}
\label{26IV24.1}
 \kappa_i:(r_1,r_2)\rightarrow\R
 \,,
 \quad
    i\in \iota_{\alpha,m} :=
    \{k_{[\alpha]},k_{[\alpha]}+ \frac 12 ,k_{[\alpha]}+1 ,\ldots,k_{\red{[m]}}+4\} \subset \frac 12 \Z 
    \,,
\end{align}
where
\begin{align}
    k_{[\alpha]} := \begin{cases}
        4-n \,, & \alpha = 0
        \\
        \min(4-n,\tfrac{7-n-2k}{2}) \,, &\alpha \neq 0
    \end{cases}
    \,,\qquad
    k_{[m]} := \begin{cases}
        k \,, & m = 0
        \\
        k(n-1) \,, & m \neq 0
    \end{cases}
    \,,
\end{align}
and with $\kappa_i$'s satisfying%
\index{k@$k_{[m]}$}%
\index{k@$k_{[\alpha]}$}%
\index{kappa@$\hat{\kappa}_i$}%
\begin{eqnarray}
  &
  \displaystyle
 \ip{\kappa_i}{\hat \kappa_j} \equiv \int_{r_1}^{r_2} \kappa_i(s) \hat \kappa_j(s) \, ds = \delta_{ij}
  \,,
  \ 
 \mbox{where} \ \hat{\kappa}_i(s) :=s^{-i}\,,
  &
  \label{13VIII22.2a}
\end{eqnarray}
and vanishing   near the end points $r\in\{r_1,r_2\}$, which is possible since the $\hat{\kappa}_i$'s are linearly independent.
The existence of such functions is standard, as we will work in a space where only a finite number of the $\kappa_i$'s is needed: Indeed, let $\phi\ge0 $, $\phi \not\equiv 0$, be any smooth  function compactly supported in $(r_1,r_2)$. For $i,j\in \iota_{\alpha,m}$  define
\begin{equation}\label{24VI24.1}
  A_{ij} = \int_{r_1}^{r_2} \phi(s)\hat \kappa_i(s)\hat\kappa_j(s) \, ds
  \,.
\end{equation}
The matrix $A_{ij}$ is symmetric, and positive definite since the $\hat \kappa_i$'s are linearly independent:
\begin{equation}\label{24VI24.2}
  A_{ij}x^i x^j  = \int_{r_1}^{r_2} \phi(s)(\hat\kappa_i(s)x^i)^2 \, ds >0
  \,.
\end{equation}
Hence its inverse, say $A^{ij}$, exists. Using the summation convention, set
\begin{equation}\label{24VI24.4}
     \kappa_i (s) = A^{ik} \phi(s) \hat\kappa_k(s)
     \,.
\end{equation}
Then 
\begin{equation}\label{24VI24.3}
   \int_{r_1}^{r_2}  \kappa_i(s)\hat\kappa_j(s) \, ds
   = A^{ik} \underbrace{
   \int_{r_1}^{r_2} \phi(s)\hat\kappa_k(s)\hat \kappa_j(s) \, ds
   }_{A_{kj}}
   = \delta_{ij}
  \,,
\end{equation}
as desired.

The fields $ \wh_{AB}$ of \eqref{16III22.2old} will be taken of the following form: for  $s\in[r_1,r_2]$,%
\index{phii@$\vphi{i}_{AB}$}
\index{w@$\wh_{AB}$}%
\begin{align}
   \wh_{AB}(r) = \sum_{i\in \iota_{\alpha,m}}\kappa_i(r)
     \vphi{i}_{AB}
   \,.
   \label{27VII22.1a}
\end{align}
We also define
\index{phij@$\kphi{i}_{AB}$}%
\begin{equation}
    \kphi{i}_{AB}\equiv \ip{\hkappa_i}{\wh_{AB}}\,.
   \label{27VII22.1a+}
\end{equation}
In the following, we will show how to determine the fields $\kphi{i}_{AB}$, which we will sometimes refer to as the ``interpolating fields''. The un-hatted fields $\vphi{i}_{AB}$ can be obtained from $\kphi{i}_{AB}$ by solving a linear system of equations, see~\cite[Equation (4.7)]{ChCong1}.
\ptclater{crosscheck the crossrefence when updating or preparing for publication}

We will use the York decomposition of the fields  $ \vphi{i}_{AB}$:%
\index{w@$\TTtpvec{p}$}%
\index{phii@$\vphi{i}_{AB}$!York decomposition}%
\begin{equation}\label{12III23.1}
   \vphi{i}_{AB} = \vphi{i}{}{}^{[\TTtp]}_{AB} +  \vphi{i}{}^{[\TTt]}_{AB}
   \equiv C(\TTtpvec{i})_{AB}+  \vphi{i}{}^{[\TTt]}_{AB}
   \,,
\end{equation}
where $\vphi{i}{}^{[\TTt]}_{AB}$ is transverse and traceless,
\begin{equation}\label{12III23.2}
   \zspaceD^A \vphi{i}{}^{[\TTt]}_{AB} = 0 = \zzhTBW^{AB} \vphi{i}{}^{[\TTt]}_{AB}
   \,,
\end{equation}
and $\TTtpvec{i}$ a uniquely defined vector field which is $L^2$-orthogonal to the space of conformal Killing vector fields of $\zmetric$;
similarly for $\kphi{i}{}_{AB}$, etc.


To achieve linearised characteristic $\Ck$-gluing in the $\delta \beta = 0$ gauge,
 it is sufficient to find a 
smooth interpolation on $\tmcN$ of the field $\tilde h_{AB}$ and a  continuous interpolation  of
the fields
\begin{align}
    \{
   \tilde\chi,\overadd{*}{\tilde\Hf}_{uA}, \overadd{0}{\tilde\Hf}_{uA}, \overadd{\ell}{\tilde\Hf}_{uA}, \overadd{\ell}{\tilde q}_{AB}
    \}_{\ell=1}^{k}
    \,,
    \label{4V23.1}
\end{align}
where we remind the reader that $\tilde{h}$ denotes the gauge-transformed counterpart of a field $h$. We will justify this shortly.

We consider first the case where the pair $(n,k)$ is convenient with \underline{$m=0=\alpha$}, because it contains most of the main ideas without
the technical complications which come with the remaining cases.
Then the transport equation for a field $\chi$,  $H_{uA}$ or $q$
 appearing in \eqref{4V23.1} generally takes the form
  \begin{align}
      \partial_r \Hf = r^{q} \hat{D}(h_{AB})
      \,,
      \label{5V23.2}
  \end{align}
for some $q\in \Z/2$ and for a  linear differential operator $\hat{D}$ in the $x^A$-variables acting on $h_{AB}$;
we list the exponents $q$ and the operators $\hat{D}$ for the various fields in Table~\ref{T6V23.1}.\ptcheck{2VI23}%
\begin{table}[t]
{\small
\hspace{-.9cm}
 \begin{tabular}{||c|c|c|c|c|c||}
  \hline
  \hline
  &&&&&\vspace{-0.3cm}
  \\
     Field
        & Defined in
            & Exponent $q$
                & Gluing operator $\hat D$
                    & Gauge $\myhatopL(\hat\xi)$
                        & Equation
     \\
  \hline
    $\red{\overadd{*}{\Hf}_{uA}}$
    \checkmark
      & \peqref{24VII22.1b}
        &   $n-4$
                & $ \zdivtwo$
                     & 
                     no gauge
                        &   \peqref{7III23.4}
  \\
  \hline
    $  \overadd{0}{H}_{uA}$
    \checkmark
     & \peqref{21II23.1}
        & $-4$
                & $\zdivtwo$
                     & $ \partial_u \kxi{2}_A$
                        &\peqref{7III23.7}
  \\
  \hline
    $\overadd{p}{\Hf}_{uA}$, $ p\ge 1$
     \checkmark
      & \peqref{6III23.w4}
        &  $-p-4$
                &  $\zspaceD^B \overadd{p}{\chi}\ofP  $ (\peqref{6III23.w2})
                     &  $
       \partial_u^{p+1} \kxi{2}_A$
                        &    \peqref{11III23.1}
  \\
  \hline
   $ \chi$
    \checkmark
     &\peqref{20II23.2}
        & $n-5$
            & $\zdivone\circ\zdivtwo$
                & 
                no gauge
                        &  \peqref{7III23.1}
\\
  \hline
   $ \overadd{i}{q}_{AB}$, $n$ even,
    \checkmark
      & \peqref{6III23.w8}
         &  $\frac{n-7-2i}{2}$
             &  $\overadd{i}{\psi}\ofP$ (\peqref{6III23.w9})
                 & no gauge
                        & \peqref{11III23.2}
\\
     or $n$ odd, $i<\frac{n-3}{2}$
        &
            &
                &
                    &
                        &
\\
  \hline
   $ \overadd{i}{q}_{AB}$, $n$ odd,
        &
            &
                &
                    &
                        &
\\
     $i=\frac{n-3}{2}$
        &
            & $-2$
                & $\overadd{\frac{n-3}{2}}{\psi}\ofP$
                    & $\hLop_n(\xi^u)$
                        & \peqref{24IV23.1xs}
\\
     $i=\frac{n-1}{2}$
        &
            & $-3$
                & $\overadd{\frac{n-1}{2}}{\psi}\ofP$
                    & $\Lop(\xi)$
                        & \peqref{24IV23.2}
\\
     $i=\frac{n+1}{2}+j$
        & \peqref{17IV23.3}
            & $-j-4$
                &  $\overadd{\frac{n+1}{2}+j}{\psi}\ofP$
                    & $\Lop(\partial^{j+1}_u\xi) $
                        & \peqref{25IV23.1}
  \\
  \hline
  \hline
\end{tabular}
}
\caption{The operators appearing in \eqref{5V23.2} and \eqref{6V23.p1}.
The exponent $q$ of the third column refers to the exponent of $r$ in \eqref{5V23.2}.  The operator $P$ is defined in \peqref{16V22.1b}.
Radially conserved charges arise from the kernel of the gluing operator $\hat D$ in the case $m=\alpha=0$, whenever non-trivial.
}\label{T6V23.1}
\end{table}
 Integrating \eqref{5V23.2} from $r_1$ to $r_2$ gives
  \begin{align}
       \Hf|_{\secN_2} - \Hf|_{\secN_1} = \int_{r_1}^{r_2} s^{q} \hat{D}( h_{AB}(s)) \, ds\,;
      \label{5V23.2b}
  \end{align}
this integrated transport equation tells us how to choose $h_{AB}$ along $\mcN_{[r_1,r_2]}$ to achieve continuous gluing of $\Hf[x_1]$ with $\Hf[x_2]$.

However, it follows from \eqref{5V23.2} that $\Hf{}^{[\ker \hat{D}^\dagger]}$ is a radially conserved charge, since for any
  \begin{equation*}
      \mu \in \ker \hat{D}^\dagger \,,
  \end{equation*}
we have
\begin{align}
    \partial_r \int_{\secN_r} \mu \Hf \sm
     =
      r^{q} \int_{\secN_r} \mu \hat{D}( h_{AB}) \sm
      =
       r^{q} \int_{\secN_r}\big( \hat{D}^\dagger \mu \big){}^{AB}h_{AB} \sm
       = 0\,.
\end{align}
 This shows that no choice  of $h_{AB}$ along $\mcN_{[r_1,r_2]}$ can  achieve the  gluing of $\Hf{}^{[\ker \hat{D}^\dagger]}$. Therefore, a necessary condition  to achieve characteristic gluing-up-to-gauge is the existence of a gauge transformation so that
\begin{equation}
    \widetilde\Hf{}^{[\ker \hat{D}^\dagger]}|_{\tilde\secN_2} -  \widetilde\Hf{}^{[\ker \hat{D}^\dagger]}|_{\tilde\secN_1} = 0\,,
    \label{5V23.1}
\end{equation}
where $\widetilde{\Hf{}}$ denote the gauge-transformed counterpart of $\Hf{}$.

Let us denote  by $\myhatopL$ the linear differential operator occurring in the formula for the gauge-transformation of the charge associated with $\Hf$:
\begin{equation}
     \int_{\secN_2} \mu \Hf \sm \rightarrow \int_{\tilde\secN_2} \mu (\Hf
      + \myhatopL(\hat \xi) )\sm
     \,,
      \label{6V23.p1}
\end{equation}
where $\hat \xi$ stands for a gauge field from  the collection $\{\partial^i_u\kxi{2}{}^A, \kxi{2}{}^u\}_{i=0}^{k+1}$.

We will check that for all the  operators $ \myhatopL$ and $\hat D$ occurring  in Table~\ref{T6V23.1} we have (see Appendix~\ref{s25X23.1-})
\begin{equation}\label{6V23.1}
 \fbox{$(\ker \myhatopL^\dagger)^\perp = \im \myhatopL$, similarly $(\ker \hat{D}^\dagger)^\perp = \im \hat D$.
 }
\end{equation}
Then for all $\mu\in(\ker \hat{D}^\dagger)\cap(\ker \myhatopL^\dagger)$ it holds that
\begin{align}
    \int_{\secN_2} \mu \tilde\Hf \sm =\int_{\secN_2} \mu (\Hf
     + \myhatopL(\hat \xi) )\sm = \int_{\tilde\secN_2} \mu \Hf \sm\,.
\end{align}
In other words, the $L^2$-orthogonal projection  $\Hf{}^{[(\ker \hat{D}^\dagger)\cap(\ker \myhatopL^\dagger)]}$ of  $\Hf$
is gauge-invariant and the fields $\Hf|_{\secN_1}$ and $\Hf|_{\secN_2}$ cannot be continuously glued-up-to-gauge unless
\begin{equation}\label{15XII23.1}
\fbox{$\Hf{}^{[(\ker \hat{D}^\dagger)\cap(\ker \myhatopL^\dagger)]}|_{\secN_1}=\Hf{}^{[(\ker \hat{D}^\dagger)\cap(\ker \myhatopL^\dagger)]}|_{\secN_2} \,.
$
}
\end{equation}
We shall refer to such projections as \textit{gauge-invariant obstructions} to gluing.

Thus:
\begin{itemize}
    \item[(i)] The $L^2$-orthogonal projection  to
    $$
     \ker (\hat{D}^\dagger)\cap\im (\myhatopL )
     =
     \ker (\hat{D}^\dagger)\cap\ker (\myhatopL^\dagger)^\perp
     $$
      of the gauge-transformed version of the condition \eqref{5V23.1} can be by solving the equation
    \begin{align}
        \myhatopL(\hat\xi^{[(\ker \myhatopL)^\perp]}) = \Hf{}^{[(\ker \hat{D}^\dagger)\cap(\ker \myhatopL^\dagger)^\perp]}|_{\secN_1} - \Hf{}^{[(\ker \hat{D}^\dagger)\cap(\ker \myhatopL^\dagger)^\perp]}|_{\secN_2} \,,
    \end{align}
    for a unique field $\hat\xi^{[(\ker \myhatopL)^\perp]}$.
    \item[(ii)] The projection onto $\im \hat D =(\ker \hat{D}^\dagger)^\perp$  of the gauge-transformed version of the condition \eqref{5V23.2b} results in the equation
    \begin{align}
    \int_{r_1}^{r_2} s^q \hat D (\tilde h_{AB}(s)) ds = \widetilde \Hf{}^{[(\ker \hat{D}^\dagger)^\perp]}|_{\secN_2} - \widetilde \Hf{}^{[(\ker \hat{D}^\dagger)^\perp]}|_{\secN_1}  \,,
     \label{6V23.2}
    \end{align}
    which can be solved for a unique field $\kphi{-q}{}^{[(\ker \hat{D})^\perp]}_{AB}$.
\end{itemize}
To summarise: we solve the integrated transport equation \eqref{5V23.2b} up to gauge by projecting it onto three parts. Then the part

\begin{enumerate}
  \item   $\ker \hat D^\dagger \cap \ker \myhatopL^\dagger$ gives a gauge-invariant obstruction;
  \item   $\ker \hat D^\dagger \cap (\ker \myhatopL^\dagger)^\perp$ is solved by the gauge field $\hat\xi^{[(\ker \myhatopL)^\perp]}$;  and
  \item $(\ker \hat D^\dagger)^\perp$ is solved by the interpolating field $\kphi{-q}{}^{[(\ker \hat{D})^\perp]}_{AB}$.
\end{enumerate}

To make this work one needs to isolate the operators involved, establish their mapping properties, and make sure that the powers $r^q$ arising in the integrals \eqref{6V23.2} are distinct for distinct fields $H$ from the collection \eqref{4V23.1}.
As an example, we provide an analysis of all operators involved in the case when $\secN \approx S^{n-1}$ in Appendix~\ref{ss20X22.1}, with the results summarised in Table~\ref{T8VI23.1}.

The strategy of the case with non-vanishing $\alpha$ and/or $m$, and of the inconvenient case, is more involved due to the coupling between various equations, but comes with less obstructions to gluing. These cases will be treated in detail in Section \ref{ss30VI.1} below.

We sketch now the linearised characteristic $\Ck$-gluing. This can be achieved in three steps:
\begin{enumerate}
  \item
  Rather similarly to the four-dimensional case addressed in~\cite{ChCong1}, we solve the integrated transport equations for each field appearing in the set \eqref{4V23.1}. These will determine the gauge fields $\kxi{2}{}_A$ and $\kxi{2}{}^u$ and the interpolating fields $\kphi{p}_{AB}$, as well as the conditions on $x_1\in \dt{\secN_1}$ and $x_2 \in \dt{\secN_2}$ associated with the gauge-invariant-obstructions.
  We list the  fields involved in Tables \ref{T11III23.2}-\ref{T11XII23.1}, and present more details  in the next section. All remaining free fields, such as $\partial_u^p\kxi{1}$, not appearing in these tables are   set to zero. The last columns (denoted as ``obstructions'') in these tables list the gauge-invariant radial charges associated to the indicated fields appearing in the respective rows.  See Table~\ref{T8VI23.1} for an explicit list of the spaces involved in the case $\secN \approx S^{n-1}$, $n\ge 4$, and $\alpha=0=m$.
\begin{table}[t]
\hspace{-2cm}
{\small
 \begin{tabular}{||c|c|c|c|c|c||}
  \hline
  \hline
     Gluing field
        & Gluing operator $\hat D$
            & $\im(\hat D)^{\perp}$
                & Gauge operator $\myhatopL$
                    & $\im(\myhatopL)^{\perp}$
                        & Obstructions
\\
  \hline
   $ \overadd{i}{q}_{AB}$, $n$ even,
      & $\overadd{i}{\psi}\ofP $
         &  trivial
             &  no gauge
                 & no gauge
                        & -
\\
     or $n$ odd, $i<\frac{n-3}{2}$
        &
            &
                &
                    &
                        &
\\
  \hline
   $ \overadd{i}{q}_{AB}$, $n$ odd,
        &
            &
                &
                    &
                        &
\\
     $i=\frac{n-3}{2}$
        & $\overadd{\frac{n-3}{2}}{\psi}\ofP $
            & $\mathbb{S}_{AB}^I$
                & $\hLop_n(\xi^u)$
                    & $\mathbb{V}_{AB}^I$, $\mathbb{T}_{AB}^I$
                        & -
\\
     $i=\frac{n-1}{2}$
        &  $\overadd{\frac{n-1}{2}}{\psi}\ofP $
            & $\mathbb{S}_{AB}^I$, $\mathbb{V}_{AB}^I$
                & $\Lop(\xi^{[\CKVp]})$
                    & $\mathbb{T}_{AB}^I$
                        & -
\\
          $i=\frac{n+1}{2}$
        & $\overadd{\frac{n+1}{2}}{\psi}\ofP $
            & $\mathbb{S}_{AB}^I$, $\mathbb{V}_{AB}^I$
                &  $\Lop(\partial_u\xi^{[\CKVp]})$
                    & $\mathbb{T}_{AB}^I$
                        & -
\\
     $i=\frac{n+1}{2}+p$
        & $\overadd{\frac{n+1}{2}+p}{\psi}\ofP $
            & $\mathbb{S}_{AB}^I$, $\mathbb{V}_{AB}^I$  ,
                &  $\Lop(\partial^{p+1}_u\xi^{[\CKVp]})$
                    & $\mathbb{T}_{AB}^I$
                        & $\overadd{i}{q}{}_{AB}^{[\mathbb{T}^I]}$,
\\
        &
            & and  $\mathbb{T}_{AB}^I$, $\ell\leq p+1$
                &
                    &
                        &  $\ell\leq p+1$

\\
  \hline
    $\red{\overadd{*}{\Hf}_{uA}}$

      & $ \zdivtwo$
        &  $\CKV$
                & no gauge
                     &
                     no gauge
                        &   $\overadd{*}{\Hf}{}_{uA}^{[\CKV]}$
  \\
  \hline
    $  \overadd{0}{H}_{uA}$

     & $\zdivtwo$
        & $\CKV$
                & $ \partial_u \xi^{[\CKV]}_A $
                     & trivial
                        & -
  \\
  \hline
    $\overadd{p}{\Hf}_{uA}$, $ p\ge 1$

      & $\zspaceD^B \overadd{p}{\chi}\ofP  $
        &  $\zspaceD_A\mathbb{S}^I$, $\mathbb{V}^I_{A}$, $1  \leq \ell \leq i+1$
                & $
       \partial_u^{p+1} \xi^{[\CKV]}_A $
                     &  trivial
                        &    $\overadd{p}{\Hf}{}_{uA}^{[2 \leq\ell\leq p+1]}$
  \\
  \hline
   $ \chi$

     & $\zdivone\circ\zdivtwo$
        & $\mathbb{S}^{[\ell = 0,1]}$
            & no gauge
                & no gauge
                    &  $ \chi^{[\ell=0,1]}$
  \\
  \hline
  \hline
\end{tabular}
}
\caption{Summary of gluing operators on $S^{n-1}$ in the case $\alpha=0=m$. The harmonic tensors $\mathbb{S}_{AB}^I$, $\mathbb{V}^I_{AB}$ and $\mathbb{T}^I_{AB}$ appearing in the table refer to  modes as given in \eqref{K1k}-\eqref{K1kT}. The column labeled ``obstructions'' refers to the gauge-invariant radial charges associated to the gluing fields.
 The notation
``$\protect\overadd{i}{q}{}_{AB}^{[\mathbb{T}^I]}$, $\ell\le p+1$''
refers to the projection of $\protect\overadd{i}{q}{}_{AB}$ onto the space spanned by eigentensors $ \mathbb{T}_{AB}^I $ of the Laplacian with eigenvalues smaller than or equal to $\ell(\ell+d-1)-2$, with  $\ell\le p+1$.
}
\label{T8VI23.1}
\end{table}

    \item Once the gauge fields and the fields $\kphi{p}_{AB}$ with $k_{[\alpha]} \leq p \leq k_{\red{[m]}}+4$
    have been determined, we let $\tilde h_{AB}$ be as in \eqref{16III22.2old} and use this to construct the fields 
\begin{align}
    \{
   \chi,\overadd{*}{\Hf}_{uA}, \overadd{0}{\Hf}_{uA}, \overadd{\ell}{\Hf}_{uA}, \overadd{\ell}{q}_{AB}
    \}_{\ell=1}^{k}
    \label{4V23.12}
\end{align}
    on $\tmcN_{[r_1,r_2)}$ using the explicit formulae of Sections
    \ref{sec:28VII22.1} and \ref{sec6III23.1}:
\seccheck{12XII23}

     \begin{enumerate}

     \item[i.] $\partial_u^p \tilde h_{ur}$ for $0\leq p\leq k$:
     We set $\partial_u^p \tilde h_{ur}|_\tmcN \equiv 0$, which guarantees both smoothness of $\tilde h_{ur}$ and the validity of  the equations, for all $i$,
    \begin{eqnarray}
    0
   &  = &
  \partial_u^i \delta \mcE _{rr}  |_\tmcN
  \equiv
  - \partial_u^i \delta \mcE ^u{}_{r}  |_\tmcN
  \equiv
  \partial_u^i \delta \mcE ^{uu}  |_\tmcN
   \,.
   \label{30IX22.11}
    \end{eqnarray}
         \item[ii.]
         $\overadd{*}{\tilde\Hf}_{uA}$ and $\overadd{p}{\tilde\Hf}_{uA}$ for $0\leq p \leq k$: We determine these fields algebraically using the integrated transport equation~\eqref{25VII22.3} and the integrated version of~\eqref{6III23.w1a}. Here and in what follows, all $h_{\mu\nu}$ fields in the representative formulae are understood to be replaced by $\tilde h_{\mu\nu}$'s.
    Now, satisfying the transport equation \eqref{6III23.w1a} for $i=p$
          guarantees that on $\tmcN_{[r_0,r_2)}$ we have
      \begin{equation}\label{15XI22.w1}
          \partial_u^p \delta  \mcE _{rA}|_{\tmcN_{[r_0,r_2)}}  \equiv -  \partial_u^p \delta \mcE ^u{} _{A}|_{\tmcN_{[r_0,r_2)}} = 0
         \,.
      \end{equation}
    It follows that
    \begin{equation}
        \partial^p_u \delta\mcE^A{}_B|_{\blue{\tmcN_{[r_0,r_2)}}} =  g^{AC}\partial^p_u\delta \mcE_{CB}|_{\blue{\tmcN_{[r_0,r_2)}}}\,.
        \label{15XI22.w2}
    \end{equation}
    The divergence identity for the Einstein tensor with a lower index $A$,
\begin{eqnarray}
  0 &\equiv & \nabla_\mu \delta \mcE ^{\mu}{}_A
    \nonumber
\\
 &= &
 r^{-(n-1)}\partial_r(r^{n-1}  \delta \mcE ^r{}_A)
   + \partial_u \delta \mcE ^u{}_A
    + \zspaceD_B \delta \mcE ^B{}_A
 \,,
 \label{2X22.6}
\end{eqnarray}
together with its $u$-derivatives, shows that we also have
\begin{equation}
 \forall \
  0 \le i \le k-1  \quad
   \Big(
  r^{-(n-1)}\partial_r(r^{n-1}  \partial_u^i\delta \mcE ^r{}_A)
    + \zspaceD_B \partial_u^i\delta \mcE ^B{}_A
    \Big)\Big|_{\tmcN_{[r_0,r_2)}}
     =0
 \,.
  \label{1X22.3}
\end{equation}
         \item[iii.]
         $\tilde \chi$ for $0\leq p \leq k$: We use the integrated~\eqref{1XI22.w1} to set the field $\tilde\chi$ on $\tmcN$.
         \item[iv.] $\overadd{p}{\tilde q}_{AB}$ for $1\leq p \leq k$:
         We use the representation formulae obtained by integrating \eqref{6III23.w6} in $r$ to
         determine the field $\overadd{p}{\tilde q}_{AB}$ on $\tmcN$.
         This ensures that
    \begin{equation}\label{15XI22.w4}
  \TS\big(
   \delta \partial^{p-1}_u \mcE _{AB}
    \big)
  \big|_{\tmcN_{[r_0,r_2)}}
   =0
    \,.
    \end{equation}
    We have thus obtained all fields in the set~\eqref{4V23.12}. We emphasis that the continuity of these fields at $r_2$
    is guaranteed by the construction from Step 1. We continue now to determine the other metric components from them.
    \item[v.] $\partial_u^p \tilde\chi$ for $1\leq p \leq k$:
    For $p=1$, the fields $\overadd{*}{\tilde\Hf}_{uA}$, $\overadd{0}{\tilde\Hf}_{uA}$, $\tilde\chi$ and $\overadd{1}{\tilde q}_{AB}$, as determined from items ii., iii., and iv. above, together determine the fields $\tilde h_{uA}$, $\tilde h_{uu}$ and $\partial_u\tilde h_{AB}$ algebraically.
    Using the field $\partial_u\tilde h_{AB}$ thus obtained, we set $\partial_u\tilde\chi$ on $\tmcN$ according to
         \begin{equation}
         \label{20XII23.12}
             \partial_r \partial_u^p\tilde\chi = \frac{(n-3)r^{n-5}}{n-1}\zspaceD^A\zspaceD^B \partial_u^p \tilde h_{AB} \,,
         \end{equation}
         for $p=1$, with the initial conditions (compare \eqref{1III23.1})
    \begin{align}
    \label{1III23.1c}
    \partial_u^p\tilde\chi|_{r_1} &=
    \partial_u^p \chi|_{r_1}
    - \frac{2r_1^{n-3}(n-2)}{(n-1)^2} \TSzlap ( \TSzlap + (n-1) \twoscsign) \partial_u^p\kxi{1}^u\
    \,.
    \end{align}
    Note that \textit{a priori}, the value of $\partial_u\tilde\chi|_{r_2}$ thus obtained may not agree with the gauged data $\dt{\tilde\secN_2}$. However, we will justify this in item 3. below.

    Next, the field $\overadd{1}{\tilde \Hf}_{uA}$ and the equation $\delta \mcE_{uA} = 0$ (cf.~\eqref{9XI20.t1}, with all $h_{\mu\nu}$'s replaced by $\tilde h_{\mu\nu}$'s) can be used to determine the field $\partial_u\tilde h_{uA}$. Using this, the field $\partial_u^2\tilde h_{AB}$ can then be obtained from the field $\overadd{2}{\tilde q}_{AB}$.  The procedure then iterates: one determines $\partial^2_u\tilde\chi$ from~\eqref{20XII23.12}-\eqref{1III23.1c} with $p=2$ and then the field $\partial_u^2\tilde h_{uA}$ from $\overadd{2}{\tilde\Hf}_{uA}$ and $\partial_u\delta\mcE_{uA}=0$ so on, until $p=k$.

    The following analysis will be useful for step 3. below: Equations~\eqref{20XII23.12}-\eqref{1III23.1c} ensure
      \begin{equation}\label{15XI22.w3}
           \partial^p_u\delta\mcE _{ru}|_{\blue{\tmcN_{[r_0,r_2)}}}
           - \frac 1{(n-2)r}
  \zspaceD ^A  \partial^p_u\delta \mcE _{rA}|_{\blue{\tmcN_{[r_0,r_2)}}}  = 0\,.
      \end{equation}
      Together with \eqref{15XI22.w1}, Equation~\eqref{15XI22.w3} guarantees that
      \begin{equation}\label{1X22.2}
    \partial_u^p \delta  {\mcE}_{ru}|_{\tmcN_{[r_0,r_2)}}
     \equiv
    -
    \partial_u^p \delta  \mcE^u{}_{ u}|_{\tmcN_{[r_0,r_2)}}
      = 0\,.
      \end{equation}
    The $u$-differentiated divergence identity with lower index $r$ reads
    \begin{align}
     0 &\equiv \nabla_\mu \partial_u^p \delta \mcE ^{\mu}{}_r
     \nonumber
     \\
     &=
     \partial^p_u  \delta  \mcE ^u{}_r
   +
   \frac{1}{r^{n-1}} \partial_r  (r^{n-1}  \delta  \partial^{p-1}_u\mcE ^r{}_r )
   \nonumber
    \\
    &\qquad
   +   \frac{1}{\sqrt{|\det \zgamma |}} \partial_A (\sqrt{|\det \zgamma |} \delta  \partial^{p-1}_u\mcE ^A{}_r )
   -
   \frac{1}{r}
        g^{AB}
    \delta \partial^{p-1}_u\mcE _{AB}
   \,,
   \label{C21X22.5ii}
    \end{align}
    so that, in view of \eqref{30IX22.11}, \eqref{15XI22.w1} and \eqref{1X22.2}, we have now
    \begin{equation}\label{1X22.6}
    \forall \
     0 \le i \le k \qquad
      0 =
      \frac{1}{r}
        g^{AB}
      \partial_u^i  \delta \mcE _{AB}
        \big|_{\tmcN_{[r_0,r_2)}}
   \,.
    \end{equation}
    Together with \eqref{15XI22.w4}, it follows that
    \begin{equation}\label{2X22.4}
 \forall \
  0 \le i \le k-1 \qquad
   \partial_u^i \delta \mcE _{AB}
  \big|_{\tmcN_{[r_0,r_2)}}
   =0
    \,.
    \end{equation}
    Equation \eqref{1X22.3} then gives
    \begin{align}
     \forall \
      0 \le i \le k-1 \qquad
        0 & =r^{-(n-1)}\partial_r(r^{n-1}
       \partial_u^i \delta \mcE ^r{}_A) |_{\tmcN_{[r_0,r_2)}}
       \nonumber
       \\
        &=
         - r^{-(n-1)}\partial_r(r^{n-1}
       \partial_u^i
       \delta \mcE _{uA}) |_{\tmcN_{[r_0,r_2)}}
       \,,
        \label{15XI22.w5}
    \end{align}
    where we have used
    $$ \partial_u^i\delta \mcE^{r}{}_A |_{\tmcN_{[r_0,r_2)}}
     = - g_{uu}  \partial_u^i\delta \mcE_{r A}|_{\tmcN_{[r_0,r_2)}} -  \partial_u^i\delta \mcE _{u A}|_{\tmcN_{[r_0,r_2)}}
      =   -  \partial_u^i\delta \mcE _{u A}|_{\tmcN_{[r_0,r_2)}}
      \,;
    $$
    note that the last equality is justified by \eqref{15XI22.w1}. Continuity at $r_1$, where all the fields $
   \partial_u^i  \mcE _{\mu\nu} $, $i\in \N$, vanish when the data there arise from a smooth solution of linearised Einstein equations, together with \eqref{15XI22.w5} implies that
    \begin{equation}
     \forall \
      0 \le i \le k-1 \qquad
       \partial_u^i  \delta \mcE ^r{}_A  |_{\tmcN_{[r_0,r_2)}}=0 = \partial_u^i  \delta \mcE _{uA} |_{\tmcN_{[r_0,r_2)}}
        \,.
         \label{2X22.8}
    \end{equation}

    Meanwhile, the divergence identity for the Einstein tensor with a lower index $u$ now reduces to
    \begin{equation}
     \forall \
      0 \le i \le k-1 \qquad
      0  \equiv  \partial_u^i \nabla_\mu  \delta \mcE ^{\mu}{}_u
        \big |_{\tmcN_{[r_0,r_2)}}
       = r^{-(n-1)} \partial_r (r^{n-1} \partial_u^i \delta \mcE^r {} _u)
        \big |_{\tmcN_{[r_0,r_2)}}
       \,.
    \end{equation}
    Continuity and vanishing at $r_1$
     together with \eqref{30IX22.11}
     and  \eqref{1X22.2}
     implies that
    \begin{equation}\label{2X22.7}
     \forall \
      0 \le i \le k-1 \qquad
       0 = \partial_u^i  \delta \mcE _{uu}
        \big |_{\tmcN_{[r_0,r_2)}}
         =  -  \partial_u^i  \delta \mcE ^r{}_{u}
        \big |_{\tmcN_{[r_0,r_2)}}
         =   \partial_u^i  \delta \mcE ^{rr}
        \big |_{\tmcN_{[r_0,r_2)}}\,.
    \end{equation}
         \end{enumerate}
     \item
     The construction above guarantees
     the continuous gluing at $r_2$ of $\tilde\hBo_{uu}$, $\partial_u\tilde\hBo_{AB}$, $\overadd{p}{\tilde\Hf}_{uA}$ with $0\leq p \leq k$, and $\overadd{p}{\tilde q}_{AB}$ with $2\leq p \leq k$. Continuity of the fields $\partial^p_u \tilde\hBo_{uA}$
     and $\partial_u^p\tilde\hBo_{uu}$ for $1 \leq p \leq k$  and
     $\partial^{i}_u \tilde\hBo_{AB}$ for $2 \leq i \leq k$
     at $r_2$ follows now
          by induction: The integrated transport equation for $\overadd{1}{\widetilde\Hf}_{uA}$ and the explicit form \eqref{9XI20.t1} of the equation $ \delta \mcE _{uA} =0$ can be solved simultaneously for $\partial_u \tilde h_{uA}|_{[r_1,r_2)}$ and $\partial_r \partial_u \tilde h_{uA}|_{[r_1,r_2)}$. The solutions will be given in terms of the previously obtained fields $\{\tilde\hBo_{uu}, \tilde\hBo_{uA}, \tilde\hBo_{AB}, \partial_u\tilde\hBo_{AB}\} $, which are continuous at $r_2$, hence ensuring continuity of $\partial_u \tilde h_{uA}$ and $\partial_r \partial_u \tilde h_{uA}$ at $r_2$. This in turn guarantees the continuity of $\partial_u^2 \tilde h_{AB}$ at $r_2$.

    Meanwhile the explicit form \eqref{13VIII20.t3} of $ \delta \mcE_{u u}=0$ together with continuity  at $r_2$ of $\tilde\hBo_{uu}$, $\tilde\hBo_{uA}$  and $\partial_u\tilde\hBo_{AB}$, ensures the continuity of $\partial_u\tilde\hBo_{uu}$ at $r_2$.

     Now, suppose that the continuity of the fields $\partial_u^p\tilde\hBo_{uu}$, $\partial_r\partial_u^p\tilde\hBo_{uA}$ and $\partial_u^{p+1}\tilde\hBo_{AB}$ has been achieved up to $p=k-1$.
    By differentiation of \eqref{20II23.3} we obtain the explicit form of \eqref{2X22.8} with $i=k-1$:
    \begin{align}
      - r^{n+1}\partial_{r}\partial^k_{u}\left(\frac{h_{uA}}{r^{2}}\right)
& =
r^{n-3}\zspaceD^{B}\zspaceD_{A}\partial^{k-1}_{u}{h}_{uB}
- r^{n-3}\zspaceD^{B}\zspaceD_{B}\partial^{k-1}_{u}{h}_{uA}
+  r^{n-3}\partial^k_{u} \zspaceD^{B}{h}_{A B}
\nonumber
\\
&\quad
- 2 (n-2) r^{n-1}(\alpha^2 + {2m r^{-n}}) \partial^{k-1}_{u}h_{uA}
- r^{2}\partial_{r}(r^{n-3}\zspaceD_{A}\partial^{k-1}_{u}{h}_{uu})
\nonumber\\
&\quad
-r^{(n-2)}(\twoscsign - r^2(\alpha^2+ {2mr^{-n}}) )((n-3)\partial_r \partial^{k-1}_{u} h_{uA} + r\partial^2_r \partial^{k-1}_{u} h_{uA})
\,.
    \label{9XI20.t1c}
    \end{align}
    Equation \eqref{9XI20.t1c} and the integrated transport equation \eqref{6III23.w1a} with $i=k$, together with the continuity of
    $\partial^{k-1}_u\tilde\hBo_{uu}$,
     $\partial_u^{k-1}\tilde \hBo_{uA}$ and $\partial^k_u\tilde\hBo_{AB}$,
       ensures the continuity of $\partial_u^{k} \tilde\hBo_{uA}$ and $\partial_r\partial_u^{k} \tilde\hBo_{uA}$ at $r_2$. (This in turn guarantees the continuity of $\partial_u^{k+1} \tilde\hBo_{AB}$ at $r_2$, but this is irrelevant for $\Ck$-gluing.)

       Finally, the explicit form of \eqref{2X22.7} with $i=k-1$,
    i.e.
    \begin{eqnarray}
      0 & = &  \partial_u^{k-1}  \delta \mcE _{uu}
        \big |_{\mcN_{[r_1,r_2)}}
      \nonumber
    \\
     & = &     \frac{1}{ r^2}\bigg[
       2\partial^k_{u} \zspaceD^{A} {h}_{u A}
        + \partial_r \left(\frac{V}{r}\right)  \zspaceD^{A} \partial_u^{k-1}{h}_{u A}
        - \frac{V}{r^{n-1}}  \partial_{r} \bigg(r^{n-2} \zspaceD^{A}\partial_u^{k-1} {h}_{u A}\bigg)
 \nonumber
\\
 &&
 +\frac{V}{r^3} \bigg(\zspaceD^{A} \zspaceD^{B} -\zR^{AB} \bigg)\partial_u^{k-1} h_{A B}
 -\left(r (n-1) (\partial_{u}
 -\partial_r\left(\frac{V}{r}\right)) +  R[\gamma] + \zspaceD^{A} \zspaceD_{A}\right)\partial_u^{k-1}{h}_{u u}
 \nonumber
\\
& &
   + \frac{(n-1)V}{r^{2n-4}} \partial_{r}(r^{2n-4} \partial_u^{k-1} {h}_{u u})\bigg]
    + 2 \Lambda \partial_u^{k-1} h_{uu}
    \, , \phantom{xxxxx}
    \label{13VIII20.t3rs}
    \end{eqnarray}
    together with smoothness  at $r_2$ of $\partial_u^{k-1}\tilde\hBo_{uu}$, $\partial_u^{k-1}\tilde\hBo_{uA}$, $\partial_u^{k}\tilde\hBo_{uA}$  and $\partial^{k-1}_u\tilde\hBo_{AB}$, ensures the continuity of $\partial^k_u\tilde\hBo_{uu}$ at $r_2$.

 \end{enumerate}

 \seccheck{19XII} 

\section{Solving the integrated transport equations on $\mcN$}
\label{ss30VI.1}
We now provide more details to step 1.\ of the above.
In particular, the results in Section~\ref{s12I22.1} and the current section will establish our main theorem of this paper, Theorem \ref{t12V24.11} below:

 First we recall the definition of the set of Bondi cross-section data
\begin{align}
    \dt{\secN} := \{ x| x=(\partial_u^{j}h_{AB}|_{\secN}
  ,\,  \partial_r^jh_{AB}|_{\secN}
  ,\,  \partial_u^{j}\delta\beta|_{\secN},
  \, \partial_u^{j}\delta U^A|_{\secN}
  ,\, \partial_r \delta U^A|_{\secN}
  ,\,   \delta V|_{\secN}
  )_{0\le j\leq k}\}
   \,.
\end{align}
%
We will use the following function spaces:
\begin{align}
\Hkdt &:= \prod_{j\in[0,k]}
    \bigg( H^{k_{\gamma}-2j}(\secN)
     \times
     \left\{
       \begin{array}{ll}
           H^{k_{\gamma}-j}(\secN), & \hbox{$j=0$} \\
           H^{k_{\gamma}+1-j}(\secN), & \hbox{$j>0$}
       \end{array}
     \right.
    \times
    H^{k_{\gamma}-2j}(\secN) \times
    H^{k_{\gamma}-2j-1}(\secN)\bigg)
    \nn
    \\
    &\qquad \qquad \qquad
    \times
    H^{k_{\gamma}-1}(\secN) \times
    H^{k_{\gamma}-2}(\secN) \,,
    \\
    \Hkzeta :&= \prod_{j\in[0,k+1]}H^{k_{\gamma}-2j+1}(\secN)\times H^{k_{\gamma}-2j+2}(\secN) \,.
\end{align}
We have:

\begin{theorem}
\label{t12V24.11}
Let $\hak,k\in \N$ with $\hak \geq 2k + 2$. To any  two sets of linearised codimension-two Bondi data 
\begin{align}
\label{12V24.2}
    x_{r_1} \in \dt{\secN_{r_1}}\,,
    \
    x_{r_2}
    \in \dt{\secN_{r_2}}
    \,,
    \
    \mbox{with}
    \ x_{r_a}
    \in
    \Hkdt \,,
\end{align}
one can smoothly assign a collection of gauge fields
\begin{align}
\label{12V24.4}
    (\partial^i_u\kxi{2}_A\,,
    \partial^i_u\kxi{2}^u)_{i\in[0,k+1]\cap \Z} \in
    \Hkzeta
\end{align}
on $\secN_{r_2}$,
as well as a symmetric tensor field
\begin{align}
\label{12V24.4c}
    h_{AB}\in \CrHk
    \,,
\end{align}
of the form \eqref{16III22.2old}-\eqref{27VII22.1a} on $\mcN_{[r_1,r_2]}$,
 to achieve $C^k$-gluing-up-to-gauge,  up to a finite-dimensional space  of obstructions.
\end{theorem}
\begin{remarks}
    \begin{enumerate}
        \item[1.] The field $h_{AB}$ is given by~\eqref{16III22.2old}, with (cf.\ also \eqref{26IV24.1} and \eqref{27VII22.1a})
        \begin{align}
        \label{12V24.3}
            \kphi{j}_{AB}\in H^{k_{\gamma}}(\secN)\,, \quad j\in\iota_{\alpha,m}\,.
        \end{align}
        Note to be compatible with \eqref{12V24.4c}, the field $\interph(r,\cdot)$
        in \eqref{16III22.2old}, as well as the remaining fields in \eqref{16III22.2old}, 
        belong to $H^{k_{\gamma}}(\secN)$. It will be part of the work that follows to show the regularity as indicated in \eqref{12V24.4}-\eqref{12V24.3} using Einstein equations.
        
        \item[2.] By achieving ``$C^k$-gluing-up-to-gauge'', we mean the construction of the following linearised fields  on $\mcN_{[r_1,r_2]}$:
        \begin{align}
            &y:= (\partial^{\ell}_u \delta V,\partial^{\ell}_u \delta \beta, \partial^{\ell}_u \delta U^{A},\partial^{\ell}_u h_{AB})_{0\leq\ell\leq k}
            \,,
            \ \mbox{with}
             \
            \\
           &
    y
            \in
            H^{k_{\gamma}-2\ell-2;+}_{
            \mathcal{N}}\times H^{k_{\gamma}-2\ell;+}_{
            \mathcal{N}} \times
    H^{k_{\gamma}-2\ell-1;+}_{
            \mathcal{N}} \times
    H^{k_{\gamma}-2\ell;+}_{
            \mathcal{N}}
    \,,
    \label{12V24.5b}
        \end{align}
        for $0\leq \ell \leq k$, such that:
        \\
         (i) $y$ agrees with   the given data $x_{r_1}$  at $\secN_1$,
        \\
         (ii) $y$ agrees with the gauge-transformed data $x_{r _2}$ at $ \secN_2 $; and
          \\
          (iii) $y$ satisfies  the linearised null constraint equations.

\item[3.] By ``up to a finite-dimensional space of obstructions'' we mean that
         there exists a finite-dimensional subspace $\mathcal{Q}$ of $\Hkdt$, such that point 2.(ii)
         is satisfied  only on the $L^2$-orthogonal complement of this space, i.e.,
         $$
         (y|_{r_2} - x_{r_2})^{[\mathcal{Q}^\perp]} = 0 \,,
         $$
         where we write $y|_{r_2}$ to denote the Bondi cross-section data induced by $y$ at $r=r_2$. 
         
         \item[4.]
When $m\neq 0$, the space of obstructions (for linearised data in any gauge) is provided by $\kQ{1}{}(\pi^A)$, where $\pi^A$ is a Killing vector of the metric $\zgamma$, and $\kQ{2}{}(\lambda)$ with $\lambda=1$; recall that these are defined as:
\begin{align}
           \kQ{1}{}(\pi^A)[x]  &:=
            \int_{\secN}\pi^A \left[r^{n+1}  \partial_r(r^{-2} \red{h}_{uA})  -
   2 {r^{n-1}}\zspaceD_A\delta\beta \right] \,\sm
   \\
    \kQ{2}{}(\lambda)[x]&:= \int_{\secN}\lambda
     \Big[r^{n-3}\delta V
         - \frac{r^{n-2}}{n-1}\partial_r\bigg(r^2 \zspaceD^A \delta U_A\bigg) - \tfrac{2 r^{n-2}}{n-1} \TSzlap \delta\beta
         \nn
         \\
         &\qquad\qquad
         - 2 r^{n-3} V \delta\beta
         \Big]\sm \,.
           \end{align}
Under the assumption $m\ne 0$ the space $\mathcal{Q}$ of  point 3.\ is a $(c_{\ringh}+1)$-dimensional subspace of the space
    $$
     \{h_{uA}^{[\KV]},\partial_r h_{uA}^{[\KV]},\delta V^{[0]}\}\big|_{\secN}
     \,,
     $$
     where $c_\ringh$ is the dimension of the space of Killing vectors of $(\secN,\ringh)$.

\item[5.] The finite-dimensional spaces of obstructions
            in the gauge $\partial_u^i\delta\beta = 0$ are listed in Tables~\ref{T11III23.2}-\ref{T11XII23.1}. The formulae for obstructions in a  gauge where the  $\partial_u^i\delta\beta  $'s do  not vanish can be obtained from the formulae derived as follows:
Let us denote by  
$\Qnobeta $ the functional for one of the radial charges, where   the gauge $\partial_u^i\delta\beta = 0$ has been used. 
By construction, $\Qnobeta $ is gauge invariant under  gauge transformations \eqref{24IX20.1HD}-\eqref{24IX20.23HD} satisfying
$$
\partial_{\tdu}^{i+1}  \xi^{u} -\frac{ \zspaceD_{\tdB} \partial_u^i\TSxi^{B}}{n-1} =0\,.
$$
Thus, under a gauge transformation we have
\begin{equation}\label{5VI24.4}
  \Qnobeta  \mapsto \Qnobeta  
  + \sum_{i=0}^{q}
 F_i  \Big[
	 \partial_{\tdu}^{i+1}  \xi^{u} -\frac{ \zspaceD_{\tdB} \partial_u^i\TSxi^{B}}{n-1} 
\Big]
\,,
\end{equation}
for some $q\in \Z$, $q\geq 0$, and some linear differential operators  $F_i$
(possibly depending upon $r$). 
We set
\begin{equation}\label{5VI24.5}
  Q =   \Qnobeta 
  -
  \sum_{i=0}^{q}
  2 F_i  \Big[
	 \partial_u^i\delta\beta
\Big]
\,.
\end{equation}
It is straightforward to check that $Q$ is gauge-invariant,  keeping in mind the gauge-transformation law of $\delta \beta$:
\begin{equation}
 2 \partial_u^i\delta\beta \mapsto 	 2 \partial_u^i\delta\beta +  \partial_{\tdu}^{i+1}  \xi^{u} -
 \frac{ \zspaceD_{\tdB} \partial_u^i\TSxi^{B}}{n-1} 
 \,.
  \label{5VI24.7}
\end{equation}
It remains to show that the new expression $Q$ will also be radially conserved.
For this, we first take the $r$-derivative of \eqref{5VI24.4}, which gives
\begin{equation}\label{5VI24.4a}
  \partial_r \Qnobeta  \mapsto \partial_r \Qnobeta  
  + \sum_{i=0}^{q}
 \partial_rF_i  \Big[
	 \partial_{\tdu}^{i+1}  \xi^{u} -\frac{ \zspaceD_{\tdB} \partial_u^i\TSxi^{B}}{n-1} 
\Big]
\,.
\end{equation}

Next, 
the $r$-independence of $\Qnobeta $ in the gauge where all the fields  $\partial_u^i\delta\beta$ vanish implies that \emph{in a general gauge where $\partial_u^i\delta \beta \ne  0$} we must have
\begin{align}
    &\partial_r \Qnobeta= \sum_{i=0}^p \hat H_i \partial_u^i \delta\beta
     \,,
    \label{6VI24.1}
\end{align}
for some $p\in \Z$, $p\geq 0$, and some linear differential operators  $\hat H_i$. 
Taking a gauge transformation of \eqref{6VI24.1} 
gives
\begin{align}
    \partial_r \Qnobeta + 
  \sum_{i=0}^q \partial_r F_i  \Big[
	 \partial_{\tdu}^{i+1}  \xi^{u} -\frac{ \zspaceD_{\tdB} \partial_u^i\TSxi^{B}}{n-1} 
\Big] = \sum_{i=0}^p \bigg(\hat H_i \partial_u^i \delta\beta +\frac{1}{2} \hat H_i \Big[
	 \partial_{\tdu}^{i+1}  \xi^{u} -\frac{ \zspaceD_{\tdB} \partial_u^i\TSxi^{B}}{n-1} 
\Big] \bigg)\,.
\end{align}
Since the above holds for any values of $\partial_u^i\xi^u$ and of $ \partial_u^i\TSxi^{B}$ we must have
\begin{align}
   p=q\,,\quad 
    \partial_r F_i = \frac{1}{2} \hat H_i\,.
\end{align}
Thus,
\begin{align}
    \partial_r \Big( \Qnobeta  - \sum_{i=0}^{q}
 F_i  \Big[
	 2 \partial_u^i\delta\beta\Big]
  \Big) 
  =
  \sum_{i=0}^q \hat H_i \partial_u^i \delta\beta
  - \sum_{i=0}^q \hat H_i \partial_u^i \delta\beta = 0
   \,,
\end{align}
as desired.

\item[6.]
As  already pointed out, the differentiability classes in the theorem are more sophisticated than necessary for the linearised problem addressed here, and several  consistent simpler 
 choices of functional spaces are possible. The current setup is dictated by consistency with the nonlinear setting of~\cite{ChCongGray2}. 
\qed
    \end{enumerate}
\end{remarks}

To prove Theorem~\ref{t12V24.11}, the following will be useful:
we will say that two sets of linearised Bondi data $\red{x_1}$ and $\red{x_2}$ are \textit{gauge-equivalent}, and we will write
\index{s@$\gsim$}%
$$\red{x_1} \gsim \red{x_2} \,,$$
if there exist  gauge fields
$z:=(\partial^i_u\xi_A\,,
    \partial^i_u\xi^u)_{i\in[1,k+1]} \in
    \Hkzeta$
    such that
    $$\red{x_2} = z^*(\red{x_1}) \,,$$
    where the map $z^*$ was defined in \eqref{12V24.1}; see also \eqref{14V24.1}.

\begin{lemma}
\label{l12V24.1}
    Let
    \begin{equation}
        \red{x_{r_1;a}}\,,
    \
    \red{x_{r_2;a}}
    \in
    \Hkdt
    \end{equation}
    and
    \begin{equation}
        \red{x_{r_1;b}}\,,
    \
    \red{x_{r_2;b}}
    \in
    \Hkdt
    \end{equation}
    be two sets of gauge-equivalent gluing data, i.e.,
\begin{equation}
    \label{12V24.7}
    \red{x_{r_j;a}}\gsim\red{x_{r_j;b}} \quad j=1,2 \,.
\end{equation}
Then $\{\red{x_{r_1;a}}\,,\red{x_{r_2;a}}\}$ can be $C^k$-glued-up-to-gauge, up to a finite-dimensional space of obstructions, iff $\{\red{x_{r_1;b}}\,,\red{x_{r_2;b}}\}$ can.
\end{lemma}

\begin{proof}
    Since data ``$a$'' and ``$b$'' are gauge-equivalent, there exist gauge fields
    $$(\partial^i_u(\kxi{j}_{a\rightarrow b})_{A}\,,
    \partial^i_u\kxi{j}^u_{a\rightarrow b})_{i\in[1,k+1]} \in
    \Hkzeta \,,\quad j=1,2$$
    such that, if we write
    $ \red{x_{r_j;a}}\equiv \big((\slashed{h}_{a,j})_{\mu\nu}\big)$,
    we have
    \begin{equation}
    \label{12V24.8}
        (\slashed{h}_{a,j})_{\mu\nu} =  (\slashed{h}_{b,j})_{\mu\nu} + \Gmap_{\mu\nu}[(\partial^i_u(\kxi{j}_{a\rightarrow b})_{A}\,,
    \partial^i_u\kxi{j}^u_{a\rightarrow b})]\,,
    \end{equation}
   where the linear map $\Gmap_{\mu\nu}$ denotes the gauge transformation map for the relevant field as given by~\eqref{24IX20.1HD}-\eqref{24IX20.23HD} together with the
$u$- and $\tdr$-derivatives of these equations.

    Suppose now that data ``$a$'' can be $C^k$-glued up-to-gauge. Thus, there exist gauge fields $(\partial^i_u(\kxi{j}_{a})_{A}\,,
    \partial^i_u\kxi{j}^u_{a})_{i\in[1,k+1]} \in
    \Hkzeta$,  and a fields $y\equiv (\slashed{h}_{\mu\nu})
      $
      as in~\eqref{12V24.5b}  on $\mcN_{[r_1,r_2]}$ such that
    \begin{align}
    \label{12V24.21}
        \slashed h_{\mu\nu}|_{r_j} = (\slashed h_{a,j})_{\mu\nu} + \Gmap_{\mu\nu}[(\partial^i_u\kxi{j}_{a,A}\,,
    \partial^i_u\kxi{j}^u_a)]\,,
    \quad j=1,2\,.
    \end{align}
   Set
    \begin{equation}
     (\partial^i_u(\kxi{j}_{b})_{A}\,,
    \partial^i_u\kxi{j}^u_{b})
    = (\partial^i_u(\kxi{j}_{a\rightarrow b})_{A}\,,
    \partial^i_u\kxi{j}^u_{a\rightarrow b})
    + (\partial^i_u(\kxi{j}_{ a})_{A}\,,
    \partial^i_u\kxi{j}^u_{ a})\,.
    \end{equation}
We then have
    \begin{align}
        &(\slashed{h}_{b,j})_{\mu\nu} +  \Gmap_{\mu\nu}[(\partial^i_u(\kxi{j}_{b})_{A}\,,
    \partial^i_u\kxi{j}^u_{b})]
    =  (\slashed{h}_{a,j})_{\mu\nu} + \Gmap_{\mu\nu}[(\partial^i_u(\kxi{j}_{ a})_{A}\,,
    \partial^i_u\kxi{j}^u_{ a})]
    \,,
    \label{12V24.22}
    \end{align}
    where we have used \eqref{12V24.8}-\eqref{12V24.22} and the linearity of $\Gmap_{\mu\nu}$.
    Equations \eqref{12V24.22} and \eqref{12V24.21} then give
    \begin{align}
        \slashed h_{\mu\nu}|_{r_j} = (\slashed h_{b,j})_{\mu\nu} + \Gmap_{\mu\nu}[(\partial^i_u\kxi{j}_{b,A}\,,
    \partial^i_u\kxi{j}^u_b)]\,,
    \quad j=1,2\,.
    \end{align}
    In other words, data ``$b$'' is $C^k$-glued-up-to-gauge using the same interpolating fields $\slashed h_{\mu\nu}$ as ``$a$'', but using the gauge fields $(\partial^i_u(\kxi{j}_{b})_{A}\,,
    \partial^i_u\kxi{j}^u_{b})_{i\in[1,k+1]} \in
    \Hkzeta$, $j=1,2$.
    \qedskip
\end{proof} 

We now proceed with the explicit construction of the solution which achieves $C^k$-gluing-up-to-gauge of some given data
\begin{align}
    x_{r_1} \,,
    \
    x_{r_2}
    \in
    \Hkdt \,.
\end{align}
%

\subsection{$\delta\beta = 0$ gauge}
 \label{ss1VIII22.1}
The gauge functions
 $\partial_{\tdu}^{i+1}\kxi{j}^u|_{{\secN}_j}, \partial_u^i
   \kxi{j}_A|_{{\secN}_j} \in H^{k_{\gamma}-2i}(\secN)$ for $0 \leq i\leq k$ and $j=1,2$ allow
 us to transform $\partial_u^j\delta{\beta}\in H^{k_{\gamma}-2j}(\secN)$
 for $j\leq k $ to zero on $\tilde{\secN}_1$ and $\tilde{\secN}_2$,
 by using e.g.\ the gauge transformation
 \begin{align}
    2\partial_u^i\delta\beta|_{\secN_{r_j}} =
  - \partial^{i+1}_{u} \kxi{j}^u \in H^{\kgamma-2i}
  \,,
  \quad
  \partial_u^i
   \kxi{j}_A = 0.
   \label{17V24.3}
 \end{align}
 These gauge fields transform the given data $x_{r_1}\,, x_{r_2}$ to new data
 $\red{\tilde x_{r_1}}\,,
    \
    \red{\tilde x_{r_2}}
    $ (according to the transformation rules~\eqref{24IX20.1HD}-\eqref{24IX20.23HD}) which continue  to lie in $\Hkdt$.

    By Lemma~\ref{l12V24.1}, it suffices to solve the gluing problem for these new data. In what follows, we proceed to do this. To ease notation, we will continue using the untilded $\red{x_{r_1}}\,,
    \
    \red{x_{r_2}}
    $ to denote the data. Furthermore, any gauge fields appearing will be understood to be associated to gauge transformations of the new data.

 The  $\partial_u^i \delta G_{rr}$ constraint equation \eqref{CHG28XI19.4aa} requires
 \begin{align}
     \forall\ 0\leq i\leq k  \quad \partial_u^i\delta\beta|_{r_1} = \partial_u^i\delta\beta|_{r_2}\,,
     \label{17V24.1}
 \end{align}
 which is automatically satisfied by the (new) data.
 It will be convenient to work only with gauge fields at $r_2$ from now on, that is, we will not perform further gauge transformations at $r_1$. Thus, to preserve~\eqref{17V24.1}, we are left with the residual gauge freedom given by \eqref{5XII19.1aHD}-\eqref{17III22HD}, together with their $u$- and $r$- derivatives.

To simplify notation we omit the ``$|_{\tilde{\secN}_j}$'' on  gauge fields, with the understanding that all
all $\kxi{2}$ fields, and their $u$-derivatives, are evaluated on $\tilde{\secN}_2$.

\subsection{Gauge invariant obstructions and gauge freezing}
 \label{s10IX22.1}

In this section we analyse the  gauge-invariant obstructions arising from the transport equations of the Bondi data, and we determine some of the gauge fields needed to match   gauge-dependent radial charges (cf.\ the last two columns of Tables \ref{T11III23.2}-\ref{T11XII23.1}). More specifically, we provide a prescription how to determine all the gauge fields for the following two cases: $(1)$  convenient pairs $(n,k)$ for any $m$ and $\alpha$, and $(2)$  inconvenient pairs $(n,k)$  but with $m=0$  and any $\alpha$. For the remaining case, that is  $(n,k)$ inconvenient and $m\neq 0$, we only determine here the $\CKV$ part of the gauge fields; the $\CKVp$ part will be taken care of
in Section \ref{ss5IX23.1}.

We start with the field $h_{uu}$. \underline{When  $m=0$} it follows from the conservation of the gauge-independent charge $\kQ{2}{}(\lambda)$ that the gluing of $\hBo_{uu}$ requires the data $x_{r_1}\in\dt{\secN_1}$ and $x_{r_2}\in\dt{\secN_2}$ to satisfy
\begin{equation}
    \kQ{2}{}[x_{r_1}]=\kQ{2}{}[x_{r_2}] \,.
    \label{7III23.3a}
\end{equation}
\underline{When  $m\neq 0$} and on $S^{n-1}$, part of the charge, $\kQ{2}{}{(\lambda^{[=1]})}{}$, can be matched using the gauge field $(\zspaceD_B\kxi{2}^B)^{[=1]}$ by setting:
\begin{align}
\label{16XI23.3}
    (\chi[x_{r_2}])^{[=1]} - (\chi[x_{r_1}])^{[=1]} = \frac{2 n m}{n-1} (\zspaceD_B\kxi{2}^B)^{[=1]}\,.
\end{align}

Let us pass now to  the field $\partial_r\hBo_{uA}$. It follows from the conservation of the radial charge
$\kQ{1}{}(\pi^A)$
that the gluing of $\partial_r\hBo_{uA}$ requires, for $\pi^A\in \CKV$,
\begin{align}
    \kQ{1}{}[x_{r_1}](\pi^A)-\kQ{1}{}[x_{r_2}](\pi^A) =  2 m n  \int_{\secN_2} \pi^A \zspaceD_A\kxi{2}^u \sm= -  2 m n  \int_{\secN_2} (\zspaceD_A\ \pi^A )\kxi{2}^u \sm
    \,.
     \label{5V23.1as}
\end{align}
\underline{When $m\neq 0$}
and when $(\secN, \ringh)$ admits proper conformal
Killing vectors (recall that this occurs only on the round $S^{n-1}$ in our setting), we can use the  freedom
in the choice of   $(\kxi{2}^u)^{[=1]}$ to guarantee that \eqref{5V23.1as} holds for $\pi^A$'s which are proper $\CKV$'s. We are thus left with the following gauge-invariant condition:
\begin{equation}
    \kQ{1}{}[ x_{r_1}]=\kQ{1}{}[ x_{r_2}]
     \ \mbox{for}
\
\pi^A \in
\left\{
  \begin{array}{ll}
     \KV, & \hbox{$m\ne 0$;} \\
    \CKV, & \hbox{$m=0$.}
  \end{array}
\right.
    \label{7III23.5b}
\end{equation}

 We note that since projections onto $\CKV$ is smooth, so are the solutions to \eqref{16XI23.3}-\eqref{5V23.1as}.
At this point, the $(n,k)$ convenient and inconvenient cases need to be considered separately:

\paragraph{The case of convenient pairs $(n,k)$.}
We start with the field $\overadd{i}{\Hf}_{uA}$. \underline{When  $m=0$}, we match the gauge-dependent radial charges using the gauge fields (see \eqref{26V23.1}) $(\partial_u^{i+1}\kxi{2}^A)^{[\ker( \overadd{i}\chi\circ\, C)]}$ according to
\begin{align}
    \overadd{i}{\Hf}{}_{uA}^{[\ker( \overadd{i}\chi\circ\, C)]}\big|_{r_1}^{r_2}
    =\begin{cases}
    - n (\partial_u\kxi{2}^A)^{[\CKV]}
    - \alpha^2 n (\zspaceD_A \kxi{2}^u)^{[\CKV]}
    \,, & i=0\,;
    \\
    - n (\partial_u^{i+1}\kxi{2}^A)^{[\ker(\overadd{i}\chi\circ\, C)]}
    + O(\alpha^2)
    \,,
    & 1 \leq i \leq k      \,.
    \end{cases}
    \label{8V23.4}
\end{align}
(See Appendix \ref{ss20X22.1} for an explicit
description
 of the space $\ker( \overadd{i}\chi\circ\, C)$ when $\secN \approx S^{n-1}$.)
The $O(\alpha^2)$ terms in the second line of~\eqref{8V23.4} contain gauge fields of lower $u$-derivatives, of the form $(\partial_u^j\kxi{2}^A)^{[\ker( \overadd{i}\chi\circ\, C)]}$ with $0 \leq j \leq i-1$,  and the gauge field $\kxi{2}^u$  (cf.\  comments below~\eqref{26V23.1}).
 Equation \eqref{8V23.4} can thus be solved recursively starting from $i=0$ using the value of $(\kxi{2}^u)^{[=1]}$ determined from \eqref{5V23.1as}.
 \underline{When  $m\neq 0$,} we use the same scheme, but with $\ker( \overadd{i}\chi\circ\, C)$ replaced with $\ker C = \CKV$ wherever it appears in this paragraph.
 Since all projections involved in \eqref{8V23.4} are onto smooth spaces, the solutions obtained are also smooth.
 There are no gauge-invariant obstructions associated to the field $\overadd{i}{\Hf}_{uA}$ in the $(n,k)$ convenient case.

Finally, the gauge-invariant radial charges associated to the field $\partial_u^p h_{AB}$ when \underline{$m=0=\alpha$} impose
the following conditions on $x_{r_1}$ and $x_{r_2}$:
        \begin{align}
        \label{21VIII23.w1}
            \int_{\secN_2} (\overadd{p}{\kerpsi})^{AB} \overadd{p}{q}_{AB} \sm
            - \int_{\secN_1} (\overadd{p}{\kerpsi})^{AB} \overadd{p}{q}_{AB} \sm = 0\,,
        \end{align}
for $1 \leq p \leq k$ and where $\overadd{p}{\kerpsi}\in\ker \overadd{p}{\psi}\ofP $, as defined in \eqref{23IV23.3}. (In the $(n,k)$ convenient case, when $m\neq 0$ or $\alpha\neq 0$, there are no radially conserved charges associated to $\overadd{p}{q}_{AB}$.) 

\paragraph{The case of  inconvenient pairs $(n,k)$.} In this case, we start with the field $\overadd{p}{q}_{AB}$. \underline{When $m=0 = \alpha$}, the gauge-invariant radial charges associated to $\overadd{p}{q}_{AB}$ for $1 \leq p \leq \frac{n-5}{2}$ impose the following conditions on $x_{r_1}$ and $x_{r_2}$:
        \begin{align}
        \label{19VIII23.1}
            \int_{\secN_2} \overadd{p}{\kerpsi}{}^{AB} \overadd{p}{q}_{AB} \sm
            - \int_{\secN_1} \overadd{p}{\kerpsi}{}^{AB} \overadd{p}{q}_{AB} \sm = 0\,, \qquad 1 \leq p \leq \frac{n-5}{2}
            \,,
        \end{align}
where $\overadd{p}{\kerpsi}\in\ker \overadd{p}{\psi}\ofP$, as defined in \eqref{23IV23.3}.
\underline{When $m=0$ but $\alpha \neq 0$}, there are no obstructions to gluing the fields $\overadd{p}{q}_{AB} \,, 1 \leq p \leq \frac{n-5}{2}$ {\bluec (cf.\ \eqref{17X23.10})}.

Next, for the field
 $\overadd{\frac{n-3}{2}}{q}_{AB}\in H^{k_{\gamma}-(n-3)}(\secN)$
in \underline{the case $m=0=\alpha$}, we match the associated gauge-dependent radial charge using the field $(\kxi{2}^u)^{[(\ker \hLop_n)^{\perp}]}$
according to
\begin{align}
    & \forall
    \
     \overadd{\frac{n-3}{2}}{\kerpsi}{}^{AB} \in \ker \overadd{\frac{n-3}{2}}{\psi}\ofP \, \cap (\ker \hLop_{n}^\dagger)^{\perp}
 \nonumber
    \\
   & \int_{\secN_1} \overadd{\frac{n-3}{2}}{\kerpsi}{}^{AB} \overadd{\frac{n-3}{2}}{q}_{AB} \sm
            - \int_{\secN_2} \overadd{\frac{n-3}{2}}{\kerpsi}{}^{AB} \overadd{\frac{n-3}{2}}{q}_{AB} \sm
            = 
            \int_{\secN_2} \overadd{\frac{n-3}{2}}{\kerpsi}{}^{AB} \hLop_n(\kxi{2}^u)_{AB}  \sm \,,
            \label{14VI.3}
\end{align}
where the operator $\hLop_n$, which is of order $n-1$, has been defined in \eqref{24IV23.1xs}-\eqref{6VI23.f2}. Thus, the solution $(\kxi{2}^u)^{[(\ker \hLop_n)^{\perp}]}$ lies in $ H^{k_{\gamma}+2}(\secN)$, consistently with \eqref{12V24.4},
 by ellipticity of $\zdivone\circ\zdivtwo\circ\hLop_n$.

The gauge-invariant radial charge imposes further the condition,
\begin{align}
   \forall
    \
    \overadd{\frac{n-3}{2}}{\kerpsi}{}^{AB} \in
     & \ker \overadd{\frac{n-3}{2}}{\psi}\ofP \, \cap (\ker \hLop_{n}^\dagger)
     \nonumber
    \\
   & \int_{\secN_1} \overadd{\frac{n-3}{2}}{\kerpsi}{}^{AB} \overadd{\frac{n-3}{2}}{q}_{AB} \sm
            - \int_{\secN_2} \overadd{\frac{n-3}{2}}{\kerpsi}{}^{AB} \overadd{\frac{n-3}{2}}{q}_{AB} \sm
            = 0
            \,,
            \label{7X23.1b}
\end{align}
on $x_{r_1}$ and $x_{r_2}$.

In \underline{the case $m=0$ but $\alpha \neq 0$}, Equation \eqref{14VI.3} for $\kxi{2}^u$ is replaced by (cf.\ \eqref{17X23.2})
\begin{align}
    (\kQ{3}{}[x_{r_2}]- \kQ{3}{}[x_{r_1}])^{[\im \mrL\circ \hLop_{n}]} = \mrL\circ \hLop_{n}(\kxi{2}^u) \,,
    \label{17X23.7}
\end{align}
 with $\kQ{3}{}\in H^{k_{\gamma}-(n-3)-2}$ (cf.\ \eqref{17X23.1}), which determines $(\kxi{2}^u)^{[(\ker(\mrL\circ \hLop_n))^\perp]}\in H^{k_{\gamma}+2}$ uniquely.
In addition, the radial obstruction 
\eqref{7X23.1b}, is replaced by
\begin{align}
     (\kQ{3}{}[x_{r_2}]- \kQ{3}{}[x_{r_1}])^{[(\im (\mrL\circ \hLop_{n})^\perp]} = 0\,.
     \label{17X23.9}
\end{align}
%

A similar analysis applies to the fields $\overadd{p}{q}_{AB}$ with index $p \geq \frac{n-1}{2}$. For this we first note that
\begin{equation}
\label{15VI.3}
\TTt^\perp \subseteq \ker \overadd{p}{\psi}\ofP
\,,
\end{equation}
as follows from $\TTtp = \im C$ and \peqref{12VI.1m}.
Meanwhile, the expression \eqref{24IV23.3} for $\Lop$ indicates that
\begin{equation}
 \im\,\Lop = (\ker \Lndagger)^\perp  \subseteq \TTt^\perp
   \,,
   \label{13XI23.5}
\end{equation}
and thus
\begin{equation}
 \im\,\Lop \cap \ker \overadd{p}{\psi}\ofP =\im\,\Lop\,.
   \label{13XI23.55}
\end{equation}
%

We consider now the fields
 $\overadd{\frac{n-1}{2} + j}{q}_{AB}\in H^{k_{\gamma}-(n-1)-2j}(\secN)$.
In the case \underline{$m=0=\alpha$},  $j\geq 0$, we match the associated gauge-dependent radial charges using the gauge fields
$(\partial^{j}_u \kxi{2}_A)^{[(\ker \Lop)^{\perp}]}$, which are in $ H^{k_{\gamma}+1-2j}(\secN)$  by ellipticity of $\zdivtwo\circ \Lop$ (which follows from, e.g.,~\eqref{10V24.41}), according to
%
\begin{align}
    \overadd{\frac{n-1}{2} + j}{q}{}_{AB}^{[\im \Lop]}|_{\secN_1}
   - \overadd{\frac{n-1}{2} + j}{q}{}_{AB}^{[ \im \Lop]}|_{\secN_2}
            = \Lop[ (\partial^{j}_u \kxi{2}_A)^{[(\ker \Lop)^{\perp}]}]_{AB}  \,.
            \label{14VI.2}
\end{align}
%
The gauge-invariant radial charge imposes the condition
\begin{align}
\overadd{\frac{n-1}{2} + j}{q}{}_{AB}^{
[
  (\im \Lop)^\perp
   \cap \ker \overadd{p}{\psi}
   ]
    }|_{\secN_1}
 -
\overadd{\frac{n-1}{2} + j}{q}{}_{AB}^{
[
  (\im \Lop)^\perp
   \cap \ker \overadd{p}{\psi}
   ]
    }|_{\secN_2}
            = 0
            \label{26VI.1}
\end{align}
on the data. It follows from~\eqref{13XI23.55} that~\eqref{14VI.2}-\eqref{26VI.1} take care of the projection $\overadd{\frac{n-1}{2} + j}{q}{}_{AB}^{[\ker\overadd{p}{\psi}]}$.

When \underline{$m=0$ but $\alpha\neq 0$}, we continue using \eqref{14VI.2} and \eqref{26VI.1} for $j\geq 1$. However, for $j=0$
Equation \eqref{14VI.2} for $\kxi{2}_A$ is replaced by
\begin{align}
    (\kQ{4}{}[x_{r_2}]- \kQ{4}{}[x_{r_1}])^{[\im (\zdivtwo \circ \Lop)]} = \zdivtwo \Lop(\kxi{2})  \,,
    \label{18X23.7}
\end{align}
which determines $\kxi{2}{}_A^{[(\ker(\zdivtwo \circ \Lop))^\perp]}$ uniquely.
 Since $\kQ{4}{}\in H^{k_{\gamma}-(n-1)-1}(\secN)$ and the elliptic operator $\zdivtwo\circ\Lop$ is of order $n+1$, the solution $\kxi{2}{}_A^{[(\ker(\zdivtwo \circ \Lop))^\perp]}$ lies in $H^{k_{\gamma}+1}(\secN)$.
In addition, the radial obstruction \eqref{26VI.1} is replaced by
\begin{align}
     (\kQ{4}{}[x_{r_2}]- \kQ{4}{}[x_{r_1}])^{[(\im (\zdivtwo \circ \Lop))^\perp]} = 0\,.
     \label{18X23.9}
\end{align}

\underline{When $m\neq 0$}, there are no radially conserved charges and hence no conditions on the data are required for the gluing of $\overadd{p}{q}_{AB}$ for $1\leq p\leq k$. The gauge fields in this case will be determined below (see Section \ref{ss5IX23.1}).

Finally, for the gauge-dependent radial charges $\overadd{i}{\Hf}{}^{[\ker( \overadd{i}\chi\circ\, C)]}_{uA}$, in the case \underline{when $m=0$}, we
can match the $\CKV\in \ker( \overadd{i}\chi\circ\, C)$-part of the charges using the gauge fields $(\partial_u^{i+1}\kxi{u}^A)^{[\CKV]}$ according to, for integers $0\leq i \leq k$ (cf.\ \eqref{26V23.1}),
\ptcheck{7VI24}
\begin{align}
    \overadd{i}{\Hf}{}_{uA}^{[\CKV]}\big|_{r_1}^{r_2}
    =\begin{cases}
    - n (\partial_u\kxi{2}^A)^{[\CKV ]}
    - \alpha^2 n (\zspaceD_A \kxi{2}^u)^{[\CKV]}\,, & i=0\,,
    \\
    - n (\partial_u^{i+1}\kxi{2}^A)^{[\CKV]}
    + O(\alpha^2)\,,
    & 1 \leq i \leq k
    \,,
    \end{cases}
    \label{8V23.4b}
\end{align}
where we note that the fields $(\partial_u^{i+1}\kxi{2}^A)^{[\CKV]}$ remain free at this point since $\CKV\in \ker \Lop$ and $\ker \hLop_n$ (cf.\ below \eqref{5XI23.31} and \eqref{30VI.1}).
As in the convenient case, the $O(\alpha^2)$ terms in~\eqref{8V23.4b} depend only on gauge fields with lower
$u$-derivatives, i.e., $\partial_u^j\kxi{2}^A$, with $0 \leq j \leq i-1$, and the gauge field $\kxi{2}^u$.
 Thus Equation \eqref{8V23.4b} can be solved recursively starting from $i=0$ by using the value of $\kxi{2}^u$ determined from \eqref{5V23.1as} and \eqref{14VI.3}.
 Since the fields resulting from the projections involved in \eqref{8V23.4b} are smooth, so are the solutions obtained.

 Next, we move on to the $\CKVp$-part of the radial charges.
 We show in Appendix \ref{App5VI24.1a}
 that there are no obstructions to gluing $\overadd{i}{\Hf}{}_{uA}^{[\CKVp]}$ when $\myGauss \leq 0$.
On the other hand, when $\myGauss>0$, the obstructions associated to $\overadd{i}{\Hf}{}_{uA}$ for $1\leq i \leq  k - \frac{n+1}{2} $ (there are no obstructions if $k<\frac{n+3}{2}$) impose the following conditions on $x_{r_1}$ and $x_{r_2}$:
 \begin{align}
     \kQ{5,i}{}{}_B(x_{r_1}) = \kQ{5,i}{}{}_B(x_{r_2})\,.
     \label{5VI24.23}
 \end{align}
The obstructions $\kQ{5,i}{}{}_B$ are of the form%
\index{Q@$\kQ{5,i}{}{}$}
\begin{align}
\label{7VI24.21}
    \kQ{5,i}{}{} &:= 
    \zdivtwo\circ \Lop
    \Big[\overadd{i}{\Hf}{}^{[\CKVp\cap\ker(\overadd{i}{\chi}\circ C)]}_{uA} \Big]
     + (...) \,,
    \end{align}
    with $(...)$ depending on the radial charges $\overadd{j}{q}{}^{[\ker\overadd{i}{\psi}]}_{AB}$, $\kQ{3}{}$ and $\kQ{4}{}$ (see \eqref{5VI24.21}-\eqref{5VI24.22} for full expressions).
Now, since $\ker(\zdivtwo\circ \Lop) = \{0\}$ 
 when $\myGauss >0$ (cf.\ \eqref{8VI24.7}), imposing \eqref{5VI24.23} ensures the gluing of the projection
 \begin{align}
     \label{14VI.4}
     \overadd{i}{\Hf}{}_{uA}^{[\CKVp\cap\ker( \overadd{i}\chi\circ\, C)]}
 \end{align}
 after the radial charges $\overadd{j}{q}{}^{[\ker\overadd{i}{\psi}]}_{AB}$, $\kQ{3}{}$ and $\kQ{4}{}$
 have been matched.

 For $\max(k-\frac{n-1}{2} ,1)\leq i \leq k$, the gauge-dependent radial charge $\overadd{i}{\Hf}{}_{uA}^{[\CKVp\cap\ker( \overadd{i}\chi\circ\, C)]}$ is matched using the smooth gauge fields $\partial_u^{i+1}\xi_A^{[\CKVp\cap\ker( \overadd{i}\chi\circ\, C)]}$ according to
 \begin{align}
     \overadd{i}{\Hf}{}_{uA}^{[\CKVp\cap\ker( \overadd{i}\chi\circ\, C)]} = \partial_u^{i+1}\xi_A^{[\CKVp\cap\ker( \overadd{i}\chi\circ\, C)]} + O(\alpha^2)\,.
     \label{6VI24.31}
 \end{align}
 As before, the $O(\alpha^2)$ term in~\eqref{6VI24.31} depends only on gauge fields of lower $u$ derivatives, i.e., $\partial_u^j\kxi{2}^A$, with $0 \leq j \leq i-1$, and the field $\kxi{2}^u$.
 Thus Equation \eqref{6VI24.31} can be solved recursively starting from $i=k-\frac{n-1}{2} $ by using the values of $\partial_u^j\kxi{2}^A$, with $0 \leq j \leq k-\frac{n-1}{2}$, and that of the field $\kxi{2}^u$ as determined from \eqref{14VI.3}-\eqref{18X23.7}.

 \underline{When $m \neq 0$ and for all  $\alpha$}, we continue using \eqref{8V23.4b}; no conditions on $x_{r_1}$ and $x_{r_2}$ are needed in this case. The gauge fields needed for the gluing of $\overadd{i}{\Hf}{}^{[\CKVp]}_{uA}$ will be derived in Section~\ref{ss5IX23.1} below.

The above gauge invariant conditions on $x_{r_1}$ and $x_{r_2}$, together with the  equations for the gauge fields, solve the $L^2$-orthogonal projections of the   transport equations for the fields in the collection \eqref{4V23.1} onto their \textit{radially conserved components}. In what follows, we analyse the conditions which arise from the \textit{radially non-conserved components} of the fields. These will become equations for the interpolating fields $\kphi{p}_{AB}$. Their solutions will depend on the values of the gauge fields as determined above. 

\subsection{Undifferentiated equations}

 We begin with an analysis of the remaining characteristic equations which do not involve $u$-differentiated fields.
 In what follows, the given data at $\secN_1$ and $\secN_2$ will be assumed to be in $\Hkdt$, as defined in \eqref{10V24.11}.

\subsubsection{Continuity of $\partial_r\tilde h_{uA}$}
 \label{ss3VIII22.1}

Taking into account the allowed gauge perturbations to Bondi data~\eqref{30IX23.1} and the transport equation \eqref{24VII22.1b}, the gluing of $\partial_{\tdr}\tilde{\hBo}_{uA}$ requires $\tilde h_{AB}$ to satisfy on $\tmcN_{(r_1,r_2)}$,
\begin{align}
    \red{\overadd{*}{\Hf}_{uA}}|_{\secN_2} - \overadd{*}{\tilde\Hf}_{uA}|_{\tilde\secN_1}
     &=2r_2^{n-2} \Done (\kxi{2}^u)_A + 2 r_2^{n-1}\zspaceD^{\tdB} C(\kzeta{2})_{AB}
     - 2m n \zspaceD_A\kxi{2}^u
                \nonumber
        \\
        &\qquad
        - (n-1)  \int_{r_1}^{r_2} \hkappa_{4-n}(s) \zspaceD^B \tilde h_{AB}
                  \,.
                   \label{7III23.4}
\end{align}
In this transport equation, and the ones that follow, the gauge fields made explicit come from the gauge transformation
at $r=r_2$ of the field that is transported;
$\overadd{*}{\Hf}_{uA}$ in this case. Those gauge fields coming from the gauge transformations at $r=r_1$, and
from the gauge-corrections  \eqref{16III22.2old} both at $r=r_1$ and  $r_2$ of  the field $\tilde h_{AB}$,
 are implicit in the notation. 

For all cases except the $(n,k)$ inconvenient, $m\alpha\neq 0$ case, we solve the projection of \eqref{7III23.4}  onto $\big[\ker (\zdivtwo^\dagger)\big]^\perp$ using the field
$\kphi{4-n}_{AB}^{[\TTt^\perp]}$.

We emphasise that a unique solution $\kphi{4-n}_{AB}^{[\TTt^\perp]}$ to the projection of \eqref{7III23.4} onto
$$
 \big[\ker (\zdivtwo^\dagger)\big]^\perp = \CKVp = \im (\zdivtwo)
$$
exists. The solution will depend on the data $\dt{\secN_1}$ and $\dt{\secN_2}$ as well as the value of the gauge fields
which has been determined in Sections \ref{ss1VIII22.1} and \ref{s10IX22.1}. A straightfoward comparison with what has been said so far about the regularity of the other fields appearing in \eqref{7III23.4} shows that the solution $\kphi{4-n}_{AB}^{[\TTt^\perp]}\in H^{k_{\gamma}}(\secN)$.

The case of  inconvenient pairs $(n,k)$ with $m\alpha\neq 0$ is analysed in Appendix \ref{ss18XI23.1}.

\subsubsection{Continuity of $\tilde\hBo_{uA}$}

Taking into account the allowed gauge perturbations to Bondi data, it follows from~\eqref{21II23.2}
and \eqref{28II23.1} that the gluing of $\tilde{\hBo}_{uA}$ requires
\begin{align}
    & \overadd{0}{\Hf}_{uA}|_{\secN_2} - \overadd{0}{\tilde\Hf}_{uA}|_{\tilde\secN_1} - \frac{1}{r_2}\left ( (\TSzlap +(n-2)\twoscsign) + \frac{n-3}{n-1} \zspaceD_A\zspaceD^B \right) \kxi{2}_B
    \nonumber
\\
    & \quad
    + n \partial_u \kxi{2}_A + \frac{n-4}{r^2} \Done (\kxi{2}^u)_A + \alpha^2 n \zspaceD_A \kxi{2}^u
      =
           \int_{r_1}^{r_2}
        \hkappa_{4}(s) \zspaceD^B \tilde{\hBo}_{AB}ds
          \,.
                   \label{7III23.7}
\end{align}
For all cases except the $(n,k)$ inconvenient, $m\alpha\neq 0$ case, we solve the projection of this equation onto $\big[\ker (\zdivtwo^\dagger)\big]^\perp = \CKVp = \im (\zdivtwo)$ using the field $\kphi{4}^{[\TTt^\perp]}_{AB}$ in terms of the given data and the predetermined gauge fields.

For the case of  inconvenient pairs $(n,k)$ with $m\alpha\neq 0$, see Appendix \ref{ss18XI23.1}. 
\subsubsection{Continuity of $\tilde\hBo_{uu}$}
 \label{ss24IX22.1}

The transport equation \eqref{1XI22.w1} together with the gauge transformation \eqref{1III23.1} of $\chi$ results in the following condition for the continuity of $h_{uu}$ at $r_2$:
\begin{align}
   & \chi|_{\secN_2} - \tilde\chi|_{\tilde{\secN}_1} + \frac{2r_2^{n-2}(n-3)}{n-1} \bigg(\frac{1}{n-1}\TSzlap + \twoscsign\bigg) \zspaceD_B \kxi{2}^B
   -  \frac{2 n m}{n-1} \zspaceD_B \kxi{2}^B
    \nonumber
    \\
    &\qquad
    + \frac{2r_2^{n-3}}{(n-1)^2} \TSzlap ( \TSzlap + (n-1) \twoscsign) \kxi{2}^u = \frac{n-3}{n-1}  \int_{r_1}^{r_2} \hkappa_{5-n}(s)   \mrL(\tilde h) ds \,.
    \label{7III23.1}
\end{align}
For all cases except the $(n,k)$ inconvenient, $m\alpha\neq 0$ case, the continuity at $r_2$ of the part of $\chi$ which lies in the image
$\im \, \mrL$ of the   operator
  $\mrL:= \zdivone\circ\zdivtwo$, can be achieved by solving
the projection onto $\im \, \mrL$ of \eqref{7III23.1} for a unique field
\begin{equation}\label{10V24.42}
 \kphi{5-n}_{AB}^{[(\ker \mrL)^\perp]}\in S\cap H^{k_{\gamma}}(\secN)
\end{equation}
(compare~\eqref{10VI23.3})
in terms of the predetermined gauge fields.
{We note for future reference that $\kphi{5-n}{}{}^{[\ker \mrL]}_{AB}$ remains free up to this point.}
%
The reader is referred to Appendix \ref{ss18XI23.1} for the case of  inconvenient pairs $(n,k)$ with $m\alpha\neq 0$. 


\subsection{Higher derivatives}
\label{ss26XI22.3}

\subsubsection{Continuity of $\partial_u^p \tilde h_{uA}$}

Taking into account the allowed gauge perturbations of  Bondi data, it follows from \eqref{6III23.w1a} and \eqref{26V23.1} that the gluing of $\partial_u^p \tilde h_{uA}$, $1 \leq p \leq k$ requires
 \begin{align}
      \overadd{p}{\Hf}_{uA}\big|_{\secN_2} - \overadd{p}{\tilde\Hf}_{uA}\big|_{\tilde\secN_1}
       &=
       \zspaceD^B \overadd{p}{\chi}\ofP \, \int_{r_1}^{r_2} \hkappa_{p+4}(s) \tilde h_{AB} \ ds
    +  m^p  \overadd{p}{\chi}_{[m]} \int_{r_1}^{r_2} \hkappa_{p(n-1)+4}(s) \zspaceD^B \tilde h_{AB} \ ds
     \nonumber
\\
    & \quad
    + \sum_{j,\ell}^{p_*}  m^{j} \alpha^{2\ell} \zspaceD^B \overadd{p}{\chi}_{j,\ell}\ofP \, \int_{r_1}^{r_2} \hkappa_{(p + 4) + j (n - 2) - 2 \ell}(s) \tilde h_{AB} \ ds
    \nonumber
\\
    &\quad
     - n \partial_u^{p+1} \kxi{2}_A
    + \text{ gauge fields with lower $u$-derivatives} \,;
    \label{11III23.1}
 \end{align}
 recall that the sum $\sum^{p_*}_{j,\ell}$ has been defined in \eqref{19V23.3-1}.

\underline{When $m=0$}, the projection of \eqref{11III23.1} onto
$$
 \red{\im \, [ \zdivtwo\circ\overadd{p}{\chi}\ofP  } ]
 =
 (\ker[ \overadd{p}{\chi}\ofP \, \circ\, C])^\perp
 $$
   can be solved uniquely for a field
 \begin{equation}
  \label{10V24.44}
  \kphi{p+4}_{AB}^{[(\red{\ker( \zdivtwo\circ\overadd{p}{\chi}\ofP } )^\perp]}\in (S\oplus V)\cap H^{k_{\gamma}}(\secN)
  \end{equation}
 (compare~\eqref{10VI23.3})
  in terms of the given data and predetermined gauge fields for $1\leq p \leq k$. We leave out the explicit formula of the gauge fields;
  these can be determined order-by-order in the index $p$ using the recursion formula \eqref{6III23.w4} and the gauge transformations \eqref{17III22.1HD}-\eqref{17III22HD}.

We move on now to the case \underline{$m\neq 0$}.
\paragraph{The \red{case of convenient pairs $(n,k)$}.}
First, recall that the projection of \eqref{11III23.1} onto $\CKV$ has been solved in the paragraph below \eqref{8V23.4}. We thus restrict our attention now to the projection of \eqref{11III23.1} onto $\CKVp$. Since $\zdivtwo$ is surjective onto $\CKVp$,
and
$\ker \zdivtwo = \TTt$,
it makes sense to rewrite this projection as%
\index{S@$\overadd{p}{S}_{AB}$}
 \begin{align}
     & \underbrace{\zdivtwo^{-1}(\overadd{p}{\Hf}_{uA}\big|_{\secN_2} - \overadd{p}{\tilde\Hf}_{uA}\big|_{\tilde\secN_1})^{[\CKVp]}
      + \text{known fields}}_{ =: \overadd{p}{S}_{AB}}
      \nonumber
\\
       &=
       \overadd{p}{\chi}\ofP \, \kphi{p+4}_{AB}^{[\TTtp]}
    +  m^p  \overadd{p}{\chi}_{[m]} \kphit{AB}{p(n-1)+4}^{[\TTtp]}
     \nonumber
\\
    & \quad
    + \sum_{\substack{j,\ell,
    \\
    (p + 4) + j (n - 2) - 2 \ell >4 }}^{p_*}
    m^{j} \alpha^{2\ell} \overadd{p}{\chi}_{j,\ell}\ofP \, \kphit{AB}{(p + 4) + j (n - 2) - 2 \ell}^{[\TTtp]}\,,
    \label{19VIII23.2}
 \end{align}
 with $\overadd{p}{S}_{AB}\in H^{k_{\gamma}-2p }(\secN)$.
In \eqref{19VIII23.2}, the term ``known fields''  denotes fields which have already been determined up to this point. These include (the $\CKVp$ projections of) all the gauge fields and the fields $\interph_{AB}^{[\TTtp]}$ and $\kphi{4}_{AB}^{[\TTtp]}$; this last field is the reason for the exclusion of $(p + 4) + j (n - 2) - 2 \ell=4$ from the summation.
\ptcheck{the whole section up to here 27VIII}

 The last term in the  right-hand side of \eqref{19VIII23.2} takes the form (cf.~\eqref{6III23.w1a}-\eqref{19V23.3-1} with $i$ there replaced by $p$)
\begin{eqnarray}
 \lefteqn{
 \sum_{\substack{j,\ell,
    \\
    (p + 4) + j (n - 2) - 2 \ell >4 }}^{p_*}
    m^{j} \alpha^{2\ell} \overadd{p}{\chi}_{j,\ell}\ofP \, \kphit{AB}{(p + 4) + j (n - 2) - 2 \ell}^{[\TTtp]}
    } \nonumber
    \\
     & = &
      \sum_{\substack{  1\le j \le p-1\,,\ 0\le\ell \le p - j 
    \\
   i =  (p + 4) + j (n - 2) - 2 \ell >4 }}
    m^{j} \alpha^{2\ell} \overadd{p}{\chi}_{j,\ell}\ofP \, \kphi{i}{}_{AB}^{[\TTtp]}
    \,.
    \label{27VIII23.1}
\end{eqnarray}
We have
$$
 i=(p + 4) + j (n - 2) - 2 \ell\le(p + 4) + j (n - 2)  < (p + 4) +p (n - 2)= p(n-1)+4
 \,,
$$
and to simplify notation it is convenient to rewrite \eqref{19VIII23.2} as

 \begin{align}
     \overadd{p}{S}_{AB}
       &=
    \sum_{i=5}^{p(n-1)+4}
    \overadd{p}{\chi}_{i} \kphi{i}_{AB}^{[\TTtp]}\,,
    \label{20VIII23.1}
 \end{align}
 with some operators  $ \overadd{p}{\chi}_{i} $  (possibly zero).
 For example,
 one finds%
 \index{chi@$\overadd{p}{\chi}_{i}$}%
 \ptcheck{27VIII}
 \begin{align}
     \overadd{1}{\chi}_{5} &= \overadd{1}{\chi}\ofP \, \,, \quad
     \overadd{1}{\chi}_{(n-1)+4} =
      m
      \overadd{1}{\chi}_{[m]}\, ,
     \\
     \overadd{2n}{\chi}_{2n+4} &= \overadd{2n}{\chi}\ofP
     + \sum_{j=1}^{2n-1} m^j \alpha ^{\frac{j(n-2)}2}\overadd{2n}{\chi}_{j,\frac{j(n-2)}{2}}\ofP \, \,,  \quad \text{if $n$ is even.}
 \end{align}
 The following facts can be easily verified and will be useful for the analysis below:
 \begin{enumerate}
     \item[(i)]
 \label{e3IX23.1} each of the coefficients $\overadd{p}{\chi}_{j}$ are either numbers, or products of commuting operators of the form
     \index{L@$\operatorname{L}_{a.c}$}%
     \begin{equation}
         \operatorname{L}_{a,c} := a P + \TSzlap + 2 \tric + c\,, \quad a,c \in \R \,,
         \label{20VIII23.3}
     \end{equation}
     or sums thereof;
     
     \item[(ii)] 
     for any $1 \leq p \leq k$, the operators $\overadd{p}{\chi}_j$ are partial differential operators on $\secN$ of order less than or equal $2p$, with equality achieved only by $\overadd{p}{\chi}_{4+p}$. In addition, we have%
 \index{chi@$\overadd{p}{\chi}_{4+p}$}%
  \ptcheck{27VIII}
  \begin{equation}
         \overadd{p}{\chi}_{4+p} = \overadd{p}{\chi}\ofP  + \ell.o. \,,
     \end{equation}
     where ``$\ell.o.$'' refers to operators of lower order (i.e., $< 2p$ in this case). It follows from the ellipticity of $\overadd{p}{\chi}\ofP$ (see Proposition \ref{P18VIII23.1} below) that $\overadd{p}{\chi}_{4+p}$ is elliptic;
     
     \item[(iii)] 
     for each $1 \leq p \leq k$, the coefficients $\overadd{p}{\chi}_{p(n-1)+4}$ are non-vanishing numbers  (cf.\ \eqref{6III23.w2b} and \eqref{17IV23.1}),  and are given by $ m^p \overadd{p}{\chi}_{[m]}$, while $\overadd{p}{\chi}_{j} = 0$ for $j >p(n-1)+4 $.
  \ptcheck{27VIII}
 \end{enumerate}

Now, for $k \ge 1$ we let
$$
 \Phi_k := \left(
             \begin{array}{c}
               \kphi{5}{}^{[\TTtp]}_{AB} \\
               \kphi{6}{}^{[\TTtp]}_{AB} \\
               \vdots \\
                \kphi{4+k}{}^{[\TTtp]}_{AB}  \\
             \end{array}
           \right)
           \,,
           \qquad
 \Psi_k := \left(
             \begin{array}{c}
               \kphi{5+k}{}^{[\TTtp]}_{AB} \\
               \kphi{6+k}{}^{[\TTtp]}_{AB} \\
               \vdots \\
                \kphit{AB}{4+k(n-1)}{}^{[\TTtp]}  \\
             \end{array}
           \right)
           \,,
           \qquad
S_k := \left(
             \begin{array}{c}
                \overadd{1}{S}_{AB} \\
                \overadd{2}{S}_{AB} \\
               \vdots \\
                \overadd{k}{S}_{AB}  \\
             \end{array}
           \right)
           \,.
 $$
 The system \eqref{19VIII23.2} for $1 \leq p \leq k$ takes the form
 \begin{equation}\label{26VII23.1}
    S_k = \chi_k \Phi_k +  M_k \Psi_k
   \,,
 \end{equation}
 where
 \begin{equation}\label{26VII23.2}
   \chi_k =\left(
          \begin{array}{ccc}
            \overadd{1}{\chi}_{5} & \ldots &  \overadd{1}{\chi}_{4+k} \\
            \vdots &  \ddots & \vdots \\
            \overadd{k}{\chi}_{5} & \ldots &  \overadd{k}{\chi}_{4+k}  \\
          \end{array}
        \right)
        \,,
        \quad
   M_k =\left(
          \begin{array}{ccc}
            \overadd{1}{\chi}_{5 + k}  & \ldots &  \overadd{1}{\chi}_{4 + k(n-1)} \\
            \vdots &  \ddots & \vdots \\
             \overadd{k}{\chi}_{5}  & \ldots & \overadd{k}{\chi}_{4 + k(n-1)}  \\
          \end{array}
        \right)
        \,.
 \end{equation}
We continue by noting that the operator $\chi_k$  is elliptic in the sense of Agmon, Douglis and Nirenberg (cf.\ Appendix \ref{ss24IX23.2} or e.g.,~\cite{MorreyNirenberg}). 
  Indeed, it follows from item (ii) above that in the
   $p$'th row
   of $\chi_k$ the operators $\overadd{p}{\chi}_{j}$ are of order less than $2p$, except for the  operator $\overadd{p}{\chi}_{4+p}$ lying on the diagonal, which is elliptic and precisely of order $2p$. This shows that the Agmon-Douglis-Nirenberg condition on the order  $s_p+t_j$ of $\overadd{p}{\chi}_{j}$ holds by setting $s_p=2p  $ and $t_j = 0$.
   The Agmon-Douglis-Nirenberg symbol is a  block diagonal matrix with blocks on the diagonal having non-zero determinants,
    and ellipticity readily follows.
   This results in the estimate~\cite[Theorem~C]{MorreyNirenberg}:

   \begin{proposition}
   \label{pp30IV24.1}
    For all $k_\gamma\ge 2k$ we have
   \begin{equation}\label{28VIII23.11}
     \sum_{p=5}^{4+k} \| \kphi{p}{}\|_{k_\gamma} \le C (k,k_\gamma)
     \sum_{p=5}^{4+k}
     \big(
      \| (S_k - M_k \Psi_k)_p\|_{k_\gamma-2p }
       + \| \kphi{p}{}\|_0
     \big)
     \,,
   \end{equation}
   where $\|\cdot\|_k$ is the $H^k(\secN)$-norm, and where $(S_k - M_k \Psi_k)_p$ denotes the $p$-th entry of the vector $S_k - M_k \Psi_k$.
   \qed
   \end{proposition}

  Note that the operators $\operatorname{L}_{a,c}$ of \eqref{20VIII23.3} are elliptic (see Lemma \ref{L18VIII23.1}), self-adjoint and pairwise commuting (see paragraph above \eqref{20VI23.12}). These properties imply the existence of a complete set of smooth,
 pairwise $L^2$-orthogonal, joint eigenfunctions $\phi_\ell$ of all the $\operatorname{L}_{a,c}$'s appearing in $\chi_k$ and $M_k$, with a corresponding discrete set of eigenvalues   $\lambda_{a,c,\ell}$ satisfying
 $$
  |\lambda_{a,c,\ell}|\to_{\ell \to \infty} \infty
   \,.
$$
We can therefore write
 \begin{equation}\label{29VII23}
   \Phi_k= \sum_{\ell} \Phi_{k,\ell} \phi_\ell
   \,, \quad
   \Psi_k= \sum_{\ell} \Psi_{k,\ell} \phi_\ell
   \,, \quad
  S_k = \sum_{\ell} S_{k,\ell} \phi_\ell
   \,.
 \end{equation}
The rest of the argument is essentially a repetition of that of~\cite{ChCong1}, we reproduce it here for the convenience of the reader:
 It  follows from item (i) above that \eqref{26VII23.1} can be solved mode-by-mode:
 \begin{align}
  \chi_k \Phi_{k,\ell}  + M_k \Psi_{k,\ell}&= S_{k,\ell}
  \quad
    \Longleftrightarrow
    \quad
   \chi_k|_{\operatorname{L}_{a,c}\mapsto \lambda_{a,c,\ell}} \Phi_{k,\ell}+ M_k|_{\operatorname{L}_{a,c}\mapsto \lambda_{a,c,\ell}} \Psi_{k,\ell}
   &= S_{k,\ell}
   \,,
   \label{26VII23.1mode}
 \end{align}
 where $\operatorname{L}_{a,c} \mapsto \lambda_{a,c,\ell}$ means that every occurrence of $\operatorname{L}_{a,c}$ should be replaced by $\lambda_{a,c,\ell}$.

 Now, $\det \chi_k|_{\operatorname{L}_{a,c} \mapsto \lambda_{a,c,\ell}}$ is a polynomial in $\lambda_{a,c,\ell}$.
Keeping in mind that in each line of the matrix $\chi_k$ the highest order operator is on the diagonal, we see that $\det \chi_k|_{\operatorname{L}_{a,c}\mapsto \lambda_{a,c,\ell}}$ is non-zero for $\ell$ large enough, and therefore there exists $N(k)$ such that we can find a unique solution of \eqref{26VII23.1mode} with  $\Psi_{k,\ell}=0$ for all $\ell > N(k)$.
Thus the $\Psi_k$'s are finite combinations of eigenfunctions of the $\operatorname{L}_{a,c}$'s, hence smooth, which renders useful the estimate \eqref{28VIII23.11}.

 It remains to show that  \eqref{26VII23.1mode} can be solved in the finite dimensional space of $\Phi_k$'s and $\Psi_k$'s of the form
 \begin{equation}\label{29VII23.11}
   \Phi_k= \sum_{\ell\le N(k) } \Phi_{k,\ell} \phi_\ell
   \,, \quad
   \Psi_k= \sum_{\ell\le N(k) }  \Psi_{k,\ell} \phi_\ell
   \,, \quad
  S_k = \sum_{\ell\le N(k) }  S_{k,\ell} \phi_\ell
   \,.
 \end{equation}
This is equivalent to the requirement that all the linear maps  obtained by juxtaposing $\chi_k|_{\operatorname{L}_{a,c} \mapsto \lambda_{a,c,\ell}}$ and $ M_k|_{\operatorname{L}_{a,c} \mapsto \lambda_{a,c,\ell}}$ with $\ell < N(k)$ are surjective. (Note, by the way, that we have already established surjectivity for $\ell \ge N(k)$.)  This, in turn, is equivalent to the fact that the adjoints of these linear maps have no kernel.

Let us denote by  $(\chi_k\, M_k)$ the relevant matrices. For simplicity in what follows we will write $\overadd{p}{\chi}_{j}$ for $\overadd{p}{\chi}_{j}|_{\operatorname{L}_{a,c} \mapsto \lambda_{a,c,\ell}}$.
 It follows from item (iii) above that $(\chi_k\, M_k)$ is of the form
 \begin{align}
 \label{22VIII23.5}
     \left(
          \begin{array}{ccccc}
                \begin{array}{|cc|}
                \hline
                  \overadd{1}{\chi}_{5}   & \dots \overadd{1}{\chi}_{(n-1)+4} \\
                  \overadd{2}{\chi}_{5}   & \dots \overadd{2}{\chi}_{(n-1)+4} \\
                  \vdots      & \vdots \\
                  \vdots      & \vdots \\
                  \vdots      & \vdots \\
                  \overadd{k}{\chi}_{5}   & \dots \overadd{k}{\chi}_{(n-1)+4} \\
                \hline
                \end{array}
            & \ldots
            &
                \begin{array}{|ccc|}
                \hline
                  0    &\ldots       &  0\\
                  \vdots   & \ddots  & \vdots \\
                  0     &\ldots      &  0   \\
                  \overadd{k-j}{\chi}_{(k-j-1)(n-1)+5}   & \ldots &\overadd{k-j}{\chi}_{(k-j)(n-1)+4} \\
                  \vdots      & \ddots & \vdots \\
                  \overadd{k}{\chi}_{(k-j-1)(n-1)+5}   & \ldots &\overadd{k}{\chi}_{(k-j)(n-1)+4} \\
                \hline
                \end{array}
            & \ldots
            &
                \begin{array}{|ccc|}
                \hline
                  0         &\ldots  & 0 \\
                  \vdots    &\ddots  & \vdots \\
                  \cdot     &  \cdot & \cdot \\
                  \vdots    &\ddots  & \vdots \\
                  0         &\ldots  & 0 \\
                  \overadd{k}{\chi}_{(k-1)(n-1)+5}   & \ldots & \overadd{k}{\chi}_{k(n-1)+4} \\
                \hline
                \end{array}
          \end{array}
        \right) \,,
 \end{align}
 where the $(i,j)$-th entry is $\overadd{i}{\chi}_{4+j}$.  Note that each  of the $k \times (n-1)$ blocks, as grouped above, has a specific vanishing-block structure. This gives the following adjoint matrix:
 \begin{align}
     \left(\begin{array}{l}
            \begin{array}{|cccccc|}
            \hline
               \overadd{1}{\chi}_{5}            &  \overadd{2}{\chi}_{5}  & \dots \quad \dots    & \ldots \quad \dots  & \dots\quad\dots\quad \dots\quad\dots   \quad \ldots    &  \, \overadd{k}{\chi}_{5} \phantom{LL}
               \\
               \vdots           &  \vdots  &     &   &     &  \, \vdots \phantom{LL}
               \\
               \fbox{ $\overadd{1}{\chi}_{(n-1)+4}$ } &  \overadd{2}{\chi}_{(n-1)+4}  & \dots \quad \dots    & \ldots \quad \dots  & \dots\quad\dots\quad \dots\quad\dots      \quad \ldots & \, \overadd{k}{\chi}_{(n-1)+4} \phantom{LL}
               \\
            \hline
            \end{array}
        \\
            \begin{array}{cccccc}
                 &  \phantom{\overadd{k}{\chi}_{,k}\qquad\qquad\qquad}   &   \phantom{\overadd{k}{\chi}_{}\qquad\qquad}   & \quad \vdots &   &
            \end{array}
        \\
            \begin{array}{|cccccc|}
            \hline
              \phantom{\overadd{k}{\chi}_{,k}} \, 0       & \phantom{\overadd{k}{\chi}_{,k}} \quad \qquad 0  \,\dots        & 0         & \overadd{k-j}{\chi}_{(k-j-1)(n-1)+5}        & \ldots\quad\dots\quad\dots\quad \ldots  & \overadd{k}{\chi}_{(k-j-1)(n-1)+5}
              \\
              \phantom{\overadd{k}{\chi}_{,k}} \, \vdots  & \phantom{\overadd{k}{\chi}_{,k}} \quad \qquad \vdots  \,\ddots  & \vdots    & \vdots                                      &                             & \vdots
              \\
              \phantom{\overadd{k}{\chi}_{,k}} \, 0       & \phantom{\overadd{k}{\chi}_{,k}} \quad \qquad 0  \,\dots        & 0         &  \fbox{$\overadd{k-j}{\chi}_{(k-j)(n-1)+4}$} & \ldots\quad\dots \quad\dots \quad \ldots &  \overadd{k}{\chi}_{(k-j)(n-1)+4}
              \\
            \hline
            \end{array}
        \\
            \begin{array}{cccccc}
                 & \phantom{\overadd{k}{\chi}_{,k}\qquad\qquad}    &  \phantom{\overadd{k}{\chi}_{}\qquad\qquad\qquad}    & \quad \vdots &   &
            \end{array}
        \\
            \begin{array}{|cccccc|}
            \hline
              \phantom{\overadd{k}{\chi}_{,k}} \, 0       & \phantom{\overadd{k}{\chi}_{,k}} \quad \qquad  0 \quad \dots       & \dots \quad \dots   & \ldots \quad \dots 0        &\, \overadd{k-1}{\chi}_{(k-2)(n-1)+5} \quad  & \overadd{k}{\chi}_{(k-2)(n-1)+5}
              \\
              \phantom{\overadd{k}{\chi}_{,k}} \, \vdots  & \phantom{\overadd{k}{\chi}_{,k}} \quad \qquad  \vdots \quad \dots  & \dots \quad \dots   & \ldots \quad \quad \vdots   &\, \vdots \quad                          & \vdots
              \\
              \phantom{\overadd{k}{\chi}_{,k}} \, 0        & \phantom{\overadd{k}{\chi}_{,k}} \quad \qquad  0 \quad \dots        & \dots  \quad \dots  & \ldots \quad \dots 0 &\, \fbox{$\overadd{k-1}{\chi}_{(k-1)(n-1)+4} $}\quad  & \overadd{k}{\chi}_{(k-1)(n-1)+4}
              \\
            \hline
            \end{array}
        \\ \vspace{-.4cm}
        \\
            \begin{array}{|cccccc|}
            \hline
              \phantom{\overadd{k}{\chi}_{,k}} \, 0  & \phantom{\overadd{k}{\chi}_{,k}} \quad \qquad  0 \quad \dots  & \dots \quad \dots   & \ldots \quad \dots 0 \phantom{\overadd{k}{\chi}_{k}\dots} & \quad 0  \quad \,  & \phantom{\overadd{k}{\chi}_{,kk.}\quad} \overadd{k}{\chi}_{(k-1)(n-1)+5}
              \\
              \phantom{\overadd{k}{\chi}_{,k}} \, \vdots  & \phantom{\overadd{k}{\chi}_{,k}} \quad \qquad  \vdots \quad \dots  & \dots \quad \dots   & \ldots \quad \quad \vdots \phantom{\overadd{k}{\chi}_{k}\dots} &  \quad  \vdots \quad \,  & \phantom{\overadd{k}{\chi}_{,kk.}\quad} \vdots
              \\
              \phantom{\overadd{k}{\chi}_{,k}} \, 0  & \phantom{\overadd{k}{\chi}_{,k}} \quad \qquad 0 \quad \dots  & \dots  \quad \dots  & \ldots \quad \dots 0 \phantom{\overadd{k}{\chi}_{k}\dots} & \quad 0  \quad \,  & \phantom{\overadd{k}{\chi}_{,kk.}\quad} \fbox{$\overadd{k}{\chi}_{k(n-1)+4}$}
              \\
            \hline
            \end{array}
            \end{array}
     \right)\,.
     \label{27VII23.w1}
 \end{align}
It follows from item (iii)
  that the  entries $\overadd{p}{\chi}_{p(n-1)+4}$,  boxed in the above, are non-zero when $m\ne 0$.
This easily implies that  that the matrix in \eqref{27VII23.w1} has maximal rank, and thus trivial kernel.

We have therefore proved: 

\begin{theorem}
 \label{t29VII23.1}
Suppose, for simplicity, that  $
 k_\gamma \ge 2k-1$.
 The system  \eqref{19VIII23.2} with $p\in\{1,\ldots,k\}$ can be solved by a choice of interpolating fields
$$
 \kphi{j}_{AB} \in H^{k_{\gamma}}(\secN) \,, \quad 5 \leq j \leq k(n-1)+4
 \,,
$$
for any finite $k$. Its solutions are determined by an elliptic system, uniquely up to a finite number of eigenfunctions of the operators $\operatorname{L}_{a,c}$, and satisfy the estimate \eqref{28VIII23.11}.
 \ptcheck{the whole section up to here 27VIII23}
 \FGp{31X23}
\end{theorem}

\seccheck{19XII}
The case  of inconvenient pairs $(n,k)$ will be addressed in Section \ref{ss5IX23.1} below.


 \subsubsection{Continuity of $\partial_u^p \tilde h_{AB}$}

 It follows from the transport equation \eqref{6III23.w6} that the gluing of $\partial_u^p \tilde h_{AB}$ with $p\geq 1$ requires
\begin{align}
    \overadd{p}{q}_{AB} \big|_{\secN_2}
      -\overadd{p}{ \tilde q}_{AB} \big|_{ \tilde\secN_1}
    &=   \int_{r_1}^{r_2} [ \ \overadd{p}{\psi}\ofP \, \hkappa_{(7-n+2p)/2}(s) \tilde h_{AB}
    + \alpha^{2p} \overadd{p}{\psi}_{[\alpha]} \hkappa_{(7-n-2p)/2}(s) \tilde h_{AB}
    \nonumber
    \\
    &\phantom{ \int_{r_1}^{r_2} [ \ \quad }
    + m^p \overadd{p}{\psi}_{[m]} \hkappa_{(7-n+2p(n-1))/2}(s) \tilde h_{AB}
    \ ] \ ds
    \nonumber
    \\
    &\quad
    + m (\ldots)
    +\text{ gauge fields}
    \label{11III23.2}
\end{align}
where, as already stated, the term ``gauge fields'' come from the gauge transformation
at $r=r_2$ of $\overadd{p}{q}_{AB}$, the exact form of which is unimportant for the argument here;
we will address the  $m(...)$-terms in \eqref{22VIII23.2} below.
\underline{When $m=0$}, the projection of \eqref{11III23.2} onto
$${\redc \im [\overadd{p}{\psi}\ofP ]= } (\ker[\overadd{p}{\psi}\ofP ])^\perp
$$
can be solved uniquely for the field
$\kphit{AB}{(7-n+2p)/2}^{[(\ker\overadd{p}{\psi}\ofP )^\perp]}$.
If in addition we have $\alpha = 0$, then the projection of \eqref{11III23.2} onto $\ker[\overadd{p}{\psi}\ofP ]$ is a gauge-invariant obstruction (see \eqref{21VIII23.w1}).

In the $(n,k)$ convenient case, the operators $\overadd{p}{\psi}\ofP $ are elliptic and the solution $\kphit{AB}{(7-n+2p)/2}^{[(\ker\overadd{p}{\psi}\ofP )^\perp]}\in H^{k_{\gamma}}(\secN)$. The remaining projection of \eqref{11III23.2} onto $\ker[\overadd{p}{\psi}\ofP ]$ has been solved using gauge fields in \eqref{21VIII23.w1}.

In the $(n,k)$ inconvenient case, with $p\geq \frac{n-1}{2}$, we have $(\ker[\overadd{p}{\psi}\ofP ])^\perp = \TTt$. Since the operators $\overadd{p}{\psi}\ofP $ are elliptic when restricted to $\TTt$, the regularity of the solution follows again from ellipticity. The projection onto $\ker[\overadd{p}{\psi}\ofP ]$ has been solved in \eqref{14VI.2}-\eqref{26VI.1}.

Finally, when $p=\frac{n-3}{2}$, we have $(\ker[\overadd{p}{\psi}\ofP ])^\perp = \TTt \oplus V$, where the space $V$ is as defined in \eqref{14XII23.1}. The restriction of $\overadd{p}{\psi}\ofP $ to this space is once again elliptic (cf.\ Proposition \ref{P17X23.1}) and regularity follows from ellipticity. The projection onto $\ker[\overadd{p}{\psi}\ofP ]$ has been solved using gauge fields in \eqref{14VI.3}-\eqref{17X23.9}, after taking obstructions into account.

To proceed further, we need to consider the convenient and inconvenient cases separately.

\paragraph{The \red{case of convenient pairs $(n,k)$}.}
First, we note that for such pairs  $(n,k)$  we have
$$
 \frac{n-7-2p}{2},\frac{n-7+2p}{2},\frac{n-7+2p(n-1)}{2} \neq 4+j, n-5, n-4
 $$
 for any $p\geq 1$, $j \in \mathbb{Z}
 \,.
 $
 Hence the fields $\kphit{AB}{(7-n-2p)/2},\kphit{AB}{(7-n+2p)/2},\kphit{AB}{(7-n+2p(n-1))/2}$ are free at this point.

Thus \underline{when $m=0$ but $\alpha \neq 0$}, the projection of \eqref{11III23.2} onto $\ker[\overadd{p}{\psi}\ofP ]$ can be solved uniquely for the interpolating field
 $\kphit{AB}{(7-n-2p)/2}^{[\ker\overadd{p}{\psi}]}$ for all $p \geq 1$.
 Since the operators $\overadd{p}{\psi}\ofP$ are elliptic (cf.\ Proposition \ref{P22VIII23.1}), $\kphit{AB}{(7-n-2p)/2}^{[\ker\overadd{p}{\psi}]}\in C^{\infty}(\secN)$.
Equation \eqref{11III23.2} with $m=0$ is a system of uncoupled equations for the fields $\kphit{AB}{(7-n+2p)/2}$ and $\kphit{AB}{(7-n-2p)/2}$ for $p\leq k$. However when $m\neq 0$, coupling between the equations with different $p$'s calls for a more involved scheme as follows.

\underline{When $m\neq 0$}, we rewrite \eqref{11III23.2} as
\index{rp@$\overadd{p}{r}_{AB}$}%
\begin{align}
    & \underbrace{\overadd{p}{q}_{AB} \big|_{\secN_2} -\overadd{p}{ q}_{AB} \big|_{\secN_1} + \text{known fields}}_{=: \overadd{p}{r}_{AB}\textcolor{blue}{\in H^{k_{\gamma}-2p}(\secN)}}
    \nonumber
    \\
    &=    \overadd{p}{\psi}\ofP \, \kphit{AB}{(7-n+2p)/2}
    + \alpha^{2p} \overadd{p}{\psi}_{[\alpha]} \kphit{AB}{(7-n-2p)/2}
    + m^p \overadd{p}{\psi}_{[m]} \kphit{AB}{(7-n+2p(n-1))/2}
    \nonumber
    \\
    & \quad
    + \sum_{j,\ell}^{p_{**}}  m^{j} \alpha^{2\ell} \overadd{p}{\psi}_{j,\ell}\ofP \, \kphit{AB}{p - \frac{n-7}{2} +  j (n-2) - 2 \ell}
    \,,
    \quad  1\leq p\leq k\,,
    \label{22VIII23.2}
\end{align}
where we used \eqref{6III23.w6}, and where the ``known fields'' refer to the predetermined gauge fields and the field $\interph$
of \eqref{16III22.2old}.

\begin{theorem}
 \label{t22VIII23.1}
The system  \eqref{22VIII23.2} can be solved by a choice of interpolating fields
$$\kphi{j}_{AB} \textcolor{blue}{\in H^{k_{\gamma}}(\secN)} \,, \quad   j \in \left[\frac{7-n-2k}{2},
\frac{7-n+2k(n-1)}{2}\right]$$
for any finite $k$ with $(n,k)$ convenient. Its solutions are determined by an elliptic system, uniquely up to a finite number of joint eigenfunctions of the operators $\operatorname{L}_{a,c}$ of \eqref{20VIII23.3}.
\end{theorem}

\proof In the same spirit as \eqref{20VIII23.1}, we write \eqref{22VIII23.2} as
\begin{align}
    \overadd{p}{r}_{AB} = \sum_{j=\frac{7-n-2p}{2}}^{\frac{7-n+2p(n-1)}{2}} \overadd{p}{\psi}_j \kphi{j}_{AB} \,,
    \quad
    1 \leq p \leq k \,.
    \label{22VIII23.w1}
\end{align}
The following facts, which will be useful for the analysis below, can be easily verified:

\begin{enumerate}
 \label{e3IX23.2}
     \item[(i)] each of the coefficients $\overadd{p}{\psi}_{j}$ are either  products of operators of the form $\operatorname{L}_{a,c}$ in \eqref{20VIII23.3}
     or sums thereof;
     \item[(ii)] for any $1 \leq p \leq k$, the coefficients $\overadd{p}{\psi}_j$ are differential operators on $\secN$ of order less than or equal $2p$, with equality achieved only by $\overadd{p}{\psi}_{\frac{7-n+2p}{2}}$. In addition, we have
     \begin{equation}
         \overadd{p}{\psi}_{\frac{7-n+2p}{2}} = \overadd{p}{\psi}\ofP  + \ell.o. \,,
     \end{equation}
     where ``$\ell.o.$'' refers to operators of lower order (i.e., $< 2p$). It follows from the ellipticity of $\overadd{p}{\psi}\ofP$ (see Proposition \ref{P22VIII23.1}) that $\overadd{p}{\psi}_{\frac{7-n+2p}{2}}$ is elliptic;
     \item[(iii)] for each $1 \leq p \leq k$, the coefficients $\overadd{p}{\psi}_{\frac{7-n+2p(n-1)}{2}}$ are non-vanishing numbers and are given by $ m^p \overadd{p}{\psi}_{[m]}$, while $\overadd{p}{\psi}_{j} = 0$ for $j > \frac{7-n+2p(n-1)}{2} $. In addition, when $\alpha \neq 0$,
     the coefficients $\overadd{p}{\psi}_{\frac{7-n-2p}{2}}= \alpha^{2p} \overadd{p}{\psi}_{[\alpha]}$ are non-vanishing numbers while the coefficients $\overadd{p}{\psi}_{j} = 0$ for $j< \frac{7-n-2p}{2}$.
 \end{enumerate}

 For $k \ge 1$ we let
\begin{align}
 \Theta_k &:= \left(
             \begin{array}{c}
               \kphi{\frac{7-n-2k}{2}}{} \\
               \vdots \\
                \kphi{\frac{7-n-2}{2}}{}  \\
             \end{array}
           \right)
           \,,
           \qquad
 \Xi_k := \left(
             \begin{array}{c}
               \kphi{\frac{9-n}{2}}{} \\
               \vdots \\
                \kphi{\frac{7-n+2k}{2}}{}  \\
             \end{array}
           \right)
           \,,
           \qquad
\Omega_k := \left(
             \begin{array}{c}
                \kphi{\frac{9-n+2k}{2}}{}  \\
                \vdots \\
                \kphi{\frac{7-n+2k(n-1)}{2}}{}
             \end{array}
           \right)
           \,,
           \\
r_k & := \left(
             \begin{array}{c}
                \overadd{1}{r}_{AB} \\
               \vdots \\
                \overadd{k}{r}_{AB}  \\
             \end{array}
           \right)
           \,.
\end{align}
 The system \eqref{22VIII23.w1} then takes the form
 \begin{equation}\label{26VII23.1b}
   A_k \Theta_k + \psi_k \Xi_k +  N_k \Omega_k = r_k
   \,,
 \end{equation}
 where $A_k$ and $\psi_k$ are $k\times k$ matrices of operators, and $N_k$ is  a $k\times k(n-2)$ matrix of operators,
 \begin{align}\label{26VII23.2b}
    A_k &=\left(
          \begin{array}{ccc}
            \overadd{1}{\psi}_{\frac{7-n-2k}{2}}  & \ldots &  \overadd{1}{\psi}_{\frac{7-n-2}{2}}
            \\
            \vdots &  \ddots & \vdots
            \\
             \overadd{k}{\psi}_{\frac{7-n-2k}{2}}  & \ldots &  \overadd{k}{\psi}_{\frac{7-n-2}{2}}
          \end{array}
        \right) \,, \quad
   \psi_k =\left(
          \begin{array}{ccc}
            \overadd{1}{\psi}_{\frac{9-n}{2}} & \ldots &  \overadd{1}{\psi}_{\frac{7-n+2k}{2}} \\
            \vdots &  \ddots & \vdots \\
            \overadd{k}{\psi}_{\frac{9-n}{2}} & \ldots &  \overadd{k}{\psi}_{\frac{7-n+2k}{2}}  \\
          \end{array}
        \right)
        \,,
    \\
   N_k &=\left(
          \begin{array}{ccc}
              \overadd{1}{\psi}_{\frac{9-n+2k}{2}} &\ldots & \overadd{1}{\psi}_{\frac{7-n+2k(n-1)}{2}}\\
            \vdots & \ddots & \vdots \\
              \overadd{k}{\psi}_{\frac{9-n+2k}{2}} &\ldots & \overadd{k}{\psi}_{\frac{7-n+2k(n-1)}{2}}
          \end{array}
        \right)
        \,,
 \end{align}
 with $A_k$ vanishing when $\alpha = 0$.
 Next, it holds by a similar argument to that below \eqref{26VII23.2} that the operator $\psi_k$ is elliptic in the sense of Agmon, Douglis and Nirenberg after setting $s_i = 2i$, $t_j=0$.
 In addition, there exists a complete set of smooth,
 pairwise $L^2$-orthogonal, joint eigenfunctions $\eta_\ell$ of all the $\operatorname{L}_{a,c}$'s appearing in $\psi_k$ and $N_k$, with a corresponding discrete set of eigenvalues $\lambda_{a,c,\ell}$ with  $|\lambda_{a,c,\ell}|\to_{\ell \to \infty} \infty$.
We can therefore write
 \begin{equation}\label{29VII23b}
    \Theta_k= \sum_{\ell} \Theta_{k,\ell} \eta_\ell
   \,, \quad
   \Xi_k= \sum_{\ell} \Xi_{k,\ell} \eta_\ell
   \,, \quad
   \Omega_k= \sum_{\ell} \Omega_{k,\ell} \eta_\ell
   \,, \quad
  r_k = \sum_{\ell} r_{k,\ell} \eta_\ell
   \,.
 \end{equation}
 It then follows from item (i) above that \eqref{26VII23.1b} can be solved mode-by-mode:
 \begin{align}\label{26VII23.1modeb}
   A_k \Theta_{k,\ell} + \psi_k \Xi_{k,\ell}+ N_k \Omega_{k,\ell} &= r_{k,\ell}
    \\
    \Longleftrightarrow
    \quad
   A_k|_{\operatorname{L}_{a,c} \mapsto \lambda_{a,c,\ell}} \Theta_{k,\ell} + \psi_k|_{\operatorname{L}_{a,c}\mapsto \lambda_{a,c,\ell}} \Xi_{k,\ell}
    + N_k|_{\operatorname{L}_{a,c}\mapsto \lambda_{a,c,\ell}} \Omega_{k,\ell} &= r_{k,\ell}
   \,.
   \label{22VIII23.7}
 \end{align}

Now, by an argument similar to that below \eqref{26VII23.1mode} (with $\chi_k$ there replaced by $\psi_k$), there exists $N(k)$ such that we can find a unique solution of \eqref{22VIII23.7} with  $\Xi_{k,\ell}= 0 = \Omega_{k,\ell}$ for all $\ell > N(k)$.
It remains to show that  \eqref{26VII23.1modeb} can be solved in the finite dimensional space of $\Theta_{k,\ell}$'s, $\Xi_{k,\ell}$'s and $\Omega_{k,\ell}$'s of the form
 \begin{equation}\label{29VII23.11b}
   \Theta_k= \sum_{\ell\le N(k) } \Theta_{k,\ell} \eta_\ell
   \,, \quad
   \Xi_k= \sum_{\ell\le N(k) }  \Xi_{k,\ell} \eta_\ell
   \,, \quad
  \Omega_k = \sum_{\ell\le N(k) }  \Omega_{k,\ell} \eta_\ell
   \,.
 \end{equation}
This is equivalent to the requirement that the linear maps  obtained by juxtaposing $A_k|_{\operatorname{L}_{a,c} \mapsto \lambda_{a,c,\ell}}$, $\psi_k|_{\operatorname{L}_{a,c} \mapsto \lambda_{a,c,\ell}}$ and $ N_k|_{\operatorname{L}_{a,c} \mapsto \lambda_{a,c,\ell}}$ with $\ell < N(k)$ are surjective.  This, in turn, is equivalent to the fact that the adjoints of these linear maps have no kernel.

We shall denote by  $(A_k\,\psi_k\, N_k)$ the relevant matrices. For simplicity, in what follows we shall write $\overadd{p}{\psi}_{j}$ for $\overadd{p}{\psi}_{j}|_{\operatorname{L}_{a,c} \mapsto \lambda_{a,c,\ell}}$. It follows from item (iii) above that the matrix $A_k|_{\operatorname{L}_{a,c}\mapsto \lambda_{a,c,\ell}}$ is of the form
\begin{align}
    A_k = \begin{pmatrix}
0 & \ldots & \overadd{1}{\psi}_{\frac{7-n-2}{2}}\\
 \vdots & \iddots & \vdots
 \\
 \overadd{k}{\psi}_{\frac{7-n-2k}{2}} & \ldots & \overadd{k}{\psi}_{\frac{7-n-2}{2}}
\end{pmatrix}
\end{align}
with all entries on the reverse diagonal non-vanishing when $\alpha \neq 0$, and all entries above the reverse diagonal vanishing.
Clearly, $A_k^\dagger$ has trivial kernel in this case. Meanwhile it follows from (iii) above that the matrix $(\psi_k  \, N_k)$ is of the same form as $(\chi_k\, M_k)$ in \eqref{22VIII23.5}, but with all $\overadd{p}{\chi}_j$'s there replaced by $\overadd{p}{\psi}_{j - \frac{1+n}{2}}$. It follows by the argument below \eqref{22VIII23.5} that $(\psi_k  \, N_k)^\dagger$ has trivial kernel.

We thus conclude that the matrix $(A_k\,\psi_k\, N_k)^\dagger$ has trivial kernel regardless of the value of $\alpha$,
since a block matrix of the form $\begin{pmatrix}
    A \\
    B
\end{pmatrix}$
 (here ``$B=\psi N$'')  has trivial kernel when either $A$ or $B$ does.

Note that we can always set $\Theta_k=0$ and solve only for $\Omega_k$; but if $\alpha \neq 0$, we can set $\Omega_k = 0= \Xi_k$ and solve for $\Theta_k$.
\FGp{31X23}

Finally, since from above, the $\Xi_k$'s and $\Omega_k$'s are finite combinations of eigenfunctions of the $\operatorname{L}_{a,c}$'s, they are smooth. Furthermore, we have the following estimate analogous to that in Proposition \ref{pp30IV24.1}:
    For all $k_\gamma\ge 2k$ we have the Agmon-Douglis-Nirenberg  estimate
    (choosing $\ell= \kgamma-2 \bluek$ in Theorem~\ref{T30XI23.1})
\index{ADN estimates}%
   \begin{equation}\label{14IV24.21}
     \sum_{p=\frac{9-n}{2}}^{\frac{7-n+2k}{2}} \| \kphi{p}{}\|_{k_\gamma} \le C (k,k_\gamma)
     \sum_{p=\frac{9-n}{2}}^{\frac{7-n+2k}{2}}
     \big(
      \| (r_k - A_k \Theta_k - N_k \Omega_k)_p\|_{k_\gamma-2p }
       + \| \kphi{p}{}\|_0
     \big)
     \,,
   \end{equation}
   where $\|\cdot\|_k$ is the $H^k(\secN)$-norm, and where $(r_k - A_k \Theta_k - N_k \Omega_k)_p$ denotes the $p$-th entry of the vector $r_k - A_k \Theta_k - N_k \Omega_k$.
   \ptcheck{17V24}
\qedskip

\subsubsection{The case of  inconvenient pairs $(n,k)$ }
\label{ss5IX23.1}

The case $m=0=\alpha$ has been covered below \eqref{11III23.2}.
\paragraph{The case $m=0$, $\alpha\neq 0$.}
For this case
the projection of \eqref{11III23.2} onto $\ker[\overadd{p}{\psi}\ofP ]$, for $1 \leq p \leq \frac{n-5}{2}$, can be solved uniquely for the smooth interpolating field
$\kphit{AB}{(7-n-2p)/2}^{[(\ker\overadd{p}{\psi} )]}$, which gets rid of the obstructions in \eqref{19VIII23.1} in this range of $p$'s. Explicitly, we set, for $1\leq p\leq \frac{n-5}2$,
\begin{align}
   \big(\overadd{p}{q}_{AB} \big|_{\secN_2}
   &
   -\overadd{p}{ q}_{AB} \big|_{\secN_1}
   - \text{(known fields})
   \big)^{[\ker\overadd{p}{\psi}]}
    =    \alpha^{2p} \overadd{p}{\psi}_{[\alpha]} \kphit{AB}{(7-n-2p)/2}^{[\ker\overadd{p}{\psi} ]}
    \,,
    \\
    (\overadd{p}{q}_{AB} \big|_{\secN_2}
    &
    -\overadd{p}{ q}_{AB} \big|_{\secN_1})^{[(\ker \overadd{p}{\psi})^\perp]}
   - \text{(known fields})^{[(\ker\overadd{p}{\psi})^\perp]}
    \nonumber
     \\
     &=   \overadd{p}{\psi}\ofP \kphit{AB}{(7-n+2p)/2}{}^{[(\ker \overadd{p}{\psi})^\perp]}
    +
    \alpha^{2p}
   \overadd{p}{\psi}_{[\alpha]}
    \underbrace{
    \kphit{AB}{(7-n-2p)/2}^{[(\ker \overadd{p}{\psi})^\perp]}}_{\text{set } \equiv 0}
    \,,
    \label{17X23.10}
\end{align}
where the ``known fields'' on the left-hand sides refer to the field $\interph$ and those gauge fields which have been predetermined in Section \ref{s10IX22.1} for the case $m = 0$ and $\alpha \neq 0$.
We note that \eqref{17X23.10} determines the field
$\kphit{AB}{(7-n+2p)/2}^{[(\ker\overadd{p}{\psi})^\perp]}\in H^{k_{\gamma}}(\secN)$ uniquely.

We continue with $p=\frac{n-3}{2}$, in which case we take care of the projection of \eqref{11III23.2} onto $(\ker[\overadd{\frac{n-3}2}{\psi}\ofP ])^\perp$ by solving the equation
\begin{align}
    \big(\overadd{\frac{n-3}2}{q}_{AB} \big|_{\secN_2}
    &
    -\overadd{\frac{n-3}2}{ q}_{AB} \big|_{\secN_1}
     - \text{(known fields})
     \big)^{[ \ker(\overadd{\frac{n-3}2}{\psi})^\perp]}
     \nonumber
\\
    & =   \overadd{\frac{n-3}2}{\psi}\ofP \kphi{2}_{AB}^{[ \ker(\overadd{\frac{n-3}2}{\psi})^\perp]}
    +
    \alpha^{n-3}\overadd{\frac{n-3}2}{\psi}_{[\alpha]}\underbrace{ \kphi{5-n}_{AB}^{[ \ker(\overadd{\frac{n-3}2}{\psi})^\perp]}}_{\text{set } \equiv 0}
\end{align}
for the field
 $\kphi{2}_{AB}^{[ \ker(\overadd{\frac{n-3}2}{\psi})^\perp]}\in V\cap H^{k_{\gamma}}(\secN)$ (see Proposition~\pref{P17X23.1}). We note that setting $\kphi{5-n}_{AB}^{[ \ker(\overadd{\frac{n-3}2}{\psi})^\perp]}$ to zero does not conflict with \eqref{7III23.1} since $\ker(\overadd{\frac{n-3}2}{\psi})^\perp \subset \ker \mrL$ (cf. \eqref{18X23.1}).
 \ptcheck{25X23 in this section everything until here}

 Next, consider the projection of \eqref{11III23.2} onto $\ker(\overadd{\frac{n-3}2}{\psi})$.
 Assume that the obstruction \eqref{17X23.9} is satisfied by the data and that \eqref{17X23.7} holds.
  This,
  together with the solution of the transport equation for the field $\chi$ solves the projection of \eqref{11III23.2} onto $\ker(\overadd{\frac{n-3}2}{\psi})\cap(\ker\mrL)^\perp$. The remaining projection onto $\ker(\overadd{\frac{n-3}2}{\psi})\cap\ker\mrL$ can be achieved using the
 smooth
  fields $\kphi{5-n}{}_{AB}^{[\ker\overadd{\frac{n-3}2}{\psi}\cap \ker\mrL]}$ according to
   \ptcheck{25X}
\begin{align}
    \big(\overadd{\frac{n-3}2}{q}_{AB} \big|_{\secN_2}
    -\overadd{\frac{n-3}2}{ q}_{AB} \big|_{\secN_1}
   - \text{(known fields})
   \big)^{[\ker\mrL\cap \ker\overadd{\frac{n-3}2}{\psi}]}
    =    \alpha^{n-3}
    \overadd{\frac{n-3}2}{\psi}_{[\alpha]}
    \,
    \kphit{AB}{5-n}^{[\ker\mrL\cap \ker\overadd{\frac{n-3}2}{\psi}]}
    \,.
\end{align}

We continue with the case $p=\frac{n-1}{2}$, where we solve the projection of \eqref{11III23.2} onto $ \TTt$ by solving the projections
\begin{align}
    &\big(\overadd{\frac{n-1}2}{q}_{AB} \big|_{\secN_2} -\overadd{\frac{n-1}2}{ q}_{AB} \big|_{\secN_1}
   - \text{(known fields})
   \big)^{[(\ker\overadd{\frac{n-1}{2}}{\psi})^\perp \cap \TTt]}
   \nonumber
   \\
     &\qquad\qquad\qquad
     =
     \overadd{\frac{n-1}2}{\psi}\ofP
     \kphi{4}_{AB}^{[\ker(\overadd{\frac{n-1}{2}}{\psi})^\perp \cap \TTt]}
     +
     \alpha^{n-1}
    \overadd{\frac{n-1}2}{\psi}_{[\alpha]} \underbrace{\kphi{4-n}_{AB}^{[\ker(\overadd{\frac{n-1}{2}}{\psi})^\perp \cap \TTt]}}_{\text{set } = 0}\,,
    \\
    & \big(\overadd{\frac{n-1}2}{q}_{AB} \big|_{\secN_2} -\overadd{\frac{n-1}2}{ q}_{AB} \big|_{\secN_1}
   - \text{(known fields})
   \big)^{[\ker\overadd{\frac{n-1}{2}}{\psi} \cap \TTt]}
     =
     \alpha^{n-1}
    \overadd{\frac{n-1}2}{\psi}_{[\alpha]} \kphi{4-n}_{AB}^{[\ker\overadd{\frac{n-1}{2}}{\psi}\cap \TTt]}
\end{align}
for the fields
 $\kphi{4-n}_{AB}^{[\ker(\overadd{\frac{n-1}{2}}{\psi})^\perp \cap \TTt]}\in H^{k_{\gamma}}(\secN)$ and $\kphi{4-n}_{AB}^{[\TTt]}\in C^{\infty}(\secN)$
respectively. We note that the latter field remains free after solving \eqref{7III23.4}. Next, \eqref{18X23.7}-\eqref{18X23.9} together with the solution of the transport equation for the field $\overadd{*}{\Hf}_{uA}$ ensures that the projection of \eqref{11III23.2} with $p=\frac{n-1}{2}$ onto $\TTtp$ holds.

For each $p$ with $p\geq \frac{n+1}{2}$, we solve the projection of \eqref{11III23.2} onto $\ker(\overadd{p}{\psi})^\perp$ for the interpolating field
 $\kphit{AB}{(7-n+2p)/2}^{[\ker(\overadd{p}{\psi})^\perp]} \in (V\oplus S)\cap H^{k_{\gamma}}(\secN)$. 
Explicitly, we set, for each $p\geq \frac{n+1}{2}$,
\begin{align}
    (\overadd{p}{q}_{AB} \big|_{\secN_2}
    &
    -\overadd{p}{ q}_{AB} \big|_{\secN_1})^{[(\ker \overadd{p}{\psi})^\perp]}
   - \text{(known fields})^{[(\ker\overadd{p}{\psi})^\perp]}
=   \overadd{p}{\psi}\ofP \kphit{AB}{(7-n+2p)/2}{}^{[(\ker \overadd{p}{\psi})^\perp]}\,.
\end{align}
(The projection of \eqref{11III23.2} onto $\ker(\overadd{p}{\psi})$ has already been addressed in \eqref{14VI.2}-\eqref{26VI.1}.)

 \ptcheck{28IX23 and 8XI23, modulo some leftover problems highlighted by margin notes}

\paragraph{The case $m\neq0$, $\alpha=0$.}
In this case Equation  \eqref{11III23.1} for $\overadd{p}{\Hf}{}^{[\CKVp]}_{uA}$ and \eqref{11III23.2}  for $\overadd{p}{q}{}^{[\TTtp]}_{AB}$ may involve $\kphi{i}_{AB}^{[\TTtp]}$'s with the same index $i$. This forces us to consider the $\overadd{p}{\Hf}{}^{[\CKVp]}_{uA}$- and $\overadd{p}{q}{}_{AB}^{[\TTtp]}$-equations as a system of coupled equations.

First,
note that the transport equations for $\overadd{p}{q}_{AB}$ with $\alpha=0$ take the form (cf.\ \eqref{6III23.w6})
\begin{align}
    \partial_r  \overadd{p}{q}_{AB}  &=  \overadd{p}{\psi}\ofP\, r^{(n-7-2p)/2} h_{AB}
    + m^p \overadd{p}{\psi}_{\red{[m]}} r^{\frac{n-7-2p(n-1)}{2}} h_{AB}
    \nonumber
    \\
    &\quad
     +
     \sum_{j=1}^{p-1} m^{j} \overadd{p}{\psi}_{j,0}\ofP\, r^{\frac{n-7}{2} - p -  j (n-2) } h_{AB} \, ,
     \quad 1 \leq p \leq k \,,
    \label{6III23.w6c}
\end{align}
with
\begin{align}
    \overadd{p}{q}_{AB}&= \partial_u \overadd{p-1}{q}_{AB}
    - \overadd{p-1}{\psi}\ofP \, \qh_{AB}^{(n-5-2p)/2}
    - m^{p-1} \overadd{p-1}{\psi}_{\red{[m]}} \qh^{(\frac{n-7-2(p-1)(n-1)}{2})}_{AB}
    \nonumber
    \\
    &\quad
    -
    \sum_{j=1}^{p-2}
    m^{j}
    \overadd{p-1}{\psi}_{j,0}\ofP \, \qh^{(\frac{1}{2} (n-7- 2 (p-1) - 2 j (n-2)))}_{AB}
    \, ,
     \quad 1 \leq p \leq k \,.
    \label{6III23.w8b}
\end{align}
The associated integrated transport equation then reads
\begin{align}
    & \underbrace{\overadd{p}{q}_{AB} \big|_{\secN_2} -\overadd{p}{ q}_{AB} \big|_{\secN_1} + \text{known fields}}_{=: \overadd{p}{r}_{AB}}
    \nonumber
    \\
    &=    \text{gauge fields } + \overadd{p}{\psi}\ofP \, \kphit{AB}{(7-n+2p)/2}
    + m^p \overadd{p}{\psi}_{\red{[m]}} \kphit{AB}{(7-n+2p(n-1))/2}
    \nonumber
    \\
    & \quad
    +
    \sum_{j=1}^{p-1}
    m^{j}
    \overadd{p}{\psi}_{j,0}\ofP \, \kphit{AB}{p - \frac{n-7}{2} +  j (n-2) }
    \,,
    \quad  1\leq p\leq k\,,
    \label{22VIII23.2c}
\end{align}
where the ``known fields'' on the left-hand side refer to
the contribution to this equation from the field $\interph$; an analysis of the ``gauge fields'' 
(coming from the gauge transformation of $\overadd{p}{q}_{AB} \big|_{\secN_2}$ and the gauge corrections \eqref{16III22.2old} at $r_2$ of  the field $\tilde h_{AB}$)
appearing on the right-hand side are provided in Appendix~\ref{ss24IX23.2}.

The projection of \eqref{22VIII23.2c} onto $\TTt$ can be solved using the fields $ \kphit{AB}{(7-n+2p)/2}{}^{[\TTt]}$ and $ \kphit{AB}{(7-n+2p(n-1))/2}{}^{[\TTt]}$ for each $1\leq p\leq k$ according to:
\begin{align}
    \overadd{p}{r}{}^{[\TTt\cap (\ker\overadd{p}{\psi})^\perp]}_{AB}
    &=   \overadd{p}{\psi}\ofP \, \kphit{AB}{(7-n+2p)/2}{}^{[\TTt\cap (\ker\overadd{p}{\psi})^\perp]}
    + m^p \underbrace{\overadd{p}{\psi}_{\red{[m]}} \kphit{AB}{(7-n+2p(n-1))/2}{}^{[\TTt\cap (\ker\overadd{p}{\psi})^\perp]}}_{\text{set }=0}
    \nonumber
    \\
    & \quad
    +\sum_{j=1}^{p-1}
    \overadd{p}{\psi}_{j,0}\ofP \, \underbrace{\kphit{AB}{p - \frac{n-7}{2} +  j (n-2) }{}^{[\TTt\cap (\ker\overadd{p}{\psi})^\perp]}}_{\text{set } = 0}
    \,,
    \\
    \overadd{p}{r}{}^{[\TTt\cap \ker\overadd{p}{\psi}]}_{AB}
    &=   m^p \overadd{p}{\psi}_{\red{[m]}} \kphit{AB}{(7-n+2p(n-1))/2}{}^{[\TTt\cap \ker\overadd{p}{\psi}]}
    +\sum_{j=1}^{p-1}
    \overadd{p}{\psi}_{j,0}\ofP \, \underbrace{\kphit{AB}{p - \frac{n-7}{2} +  j (n-2) }{}^{[\TTt\cap \ker\overadd{p}{\psi}]}}_{\text{set } = 0}
    \,,
    \label{22VIII23.2d}
\end{align}
where we note that the projection of the ``gauge fields'' in the right-hand side of \eqref{22VIII23.2c} onto $\TTt$ vanishes.

For consideration of the coupled system for $\overadd{p}{\Hf}{}^{[\CKVp]}_{uA}$ and $\overadd{p}{q}{}^{[\TTtp]}_{AB}$, it turns out to be convenient to write the interpolating fields $\kphi{p}{}^{[\TTtp]}$ as%
\index{w@$\TTtpvec{p}$}%
\begin{align}
\label{8IX23.w3}
    \kphi{p}{}^{[\TTtp]} = C(\TTtpvec{p})
\end{align}
where $\TTtpvec{p} \in \CKVp$ is the unique vector solution to \eqref{8IX23.w3} (cf.\  Appendix \ref{ss8IX23.1}).

The following commutation relation will be useful:
    \begin{align}
        \zdivtwo  (  \TSzlap + 2 \tric  ) \hat\varphi
        =
         (\TSzlap + (n-2)\myGauss  ) \zdivtwo \hat\varphi\,.
          \label{11X23.12}
    \end{align}
It implies
    \begin{align}
        \zdivtwo \circ (a P + \TSzlap + 2 \tric + c) \hat\varphi
        =
        (a \, \zdivtwo \circ\, C + \TSzlap + (n-2)\myGauss + c) \zdivtwo \hat\varphi\,.
          \label{11X23.11}
    \end{align}
     We note the following implication of \eqref{11X23.11}:

\begin{Proposition}
 \label{P11X23.1}
 $$
 \overadd{i}{\psi} (\TTt) \subset \TTt
 \,,
 \qquad
 \overadd{i}{\psi} (\TTtp) \subset \TTtp
 \,.
$$
 \end{Proposition}
 
\proof
Since the $\overadd{i}{\psi}$'s are of the form 
$a P + \TSzlap + 2 \tric + c$, the inclusion $ \overadd{i}{\psi} (\TTt) \subset \TTt$ follows immediately from \eqref{11X23.11}. The second inclusion follows from the fact that $a P + \TSzlap + 2 \tric + c$ is 
formally self-adjoint.
\qed

Equation \eqref{11X23.11} and Proposition~\ref{P11X23.1} allow us to write the divergence of the $\TTtp$ projection of \eqref{22VIII23.2c} as
\index{psi@$\overadd{i}{\psi}\ofP$!$\overadd{i}{\tilde\psi}\ofDC$}%
\begin{align}
    \underbrace{\zspaceD^A \overadd{p}{r}{}^{[\TTtp]}_{AB}}_{=:  \overadd{p}{r}_B}
    &=    \text{gauge fields } + \overadd{p}{\tilde\psi}\ofDC \circ \zdivtwo\circ\, C (\TTtpvec{\frac{7-n+2p}{2}})
    \nonumber
    \\
    &\quad
    + m^p \overadd{p}{\psi}_{\red{[m]}} \zdivtwo\circ\, C (\TTtpvec{\frac{7-n+2p(n-1)}{2}})
    \nonumber
    \\
    & \quad
    +
    \sum_{j=1}^{p-1}  m^{j}
    \overadd{p}{\tilde\psi}_{j,0}\red{\ofDC \circ \zdivtwo}\circ\, C ( \TTtpvec{p - \frac{n-7}{2} +  j (n-2)})
    \,,
    \quad  1\leq p\leq k\,;
    \label{9IX23.w1}
\end{align}
\index{U@$\tilde{U}\ofDC$}%
recall  that the notation $\tilde{U}\ofDC$ (a tilde over an operator $U$) denotes the replacement of all appearances of the operator $P := C\circ\zdivtwo$, respectively   $\tric$, in $U$ by the operator $\zdivtwo\circ\, C$,  respectively   $1/2(n-2)\myGauss$ (cf.\ \eqref{11X23.11}).
Similarly to the convenient case (cf.\ \eqref{20VIII23.1}),  in order to simplify notation we group terms together and
 rewrite \eqref{9IX23.w1} as
\begin{align}
    \overadd{p}{r}_B
    = \text{gauge fields}
    + \sum_{j=\frac{7 - n + 2p}{2}}^{\frac{7-n+2p(n-1)}{2}} \overadd{p}{\tilde\psi}_j \circ \zdivtwo \circ\, C (\TTtpvec{j} )\,,
    \quad
    1 \leq p \leq k \,,
    \label{9IX23.w2f}
\end{align}
with some operators $\overadd{p}{\tilde\psi}_j $.
Meanwhile, the integrated transport equation for $\overadd{p}{\Hf}{}^{[\CKVp]}_{uA}$ is as given in \eqref{11III23.1}, which we rewrite as:
 \begin{align}
     & \underbrace{
     \big(
      \overadd{p}{\Hf}_{uA}\big|_{\secN_2} - \overadd{p}{\tilde\Hf}_{uA}\big|_{\tilde\secN_1}
      + \text{known fields}
      \big)^{[\CKVp]}}_{ =: \overadd{p}{s}_{A}}
      \nonumber
\\
       &=
       (\text{gauge fields})^{[\CKVp]} +
       \overadd{p}{\tilde\chi}\red{\ofDC \circ \zdivtwo}\circ\, C(\TTtpvec{p+4})
    +  m^p  \overadd{p}{\tilde\chi}_{\red{[m]}} \zdivtwo\circ\, C(\TTtpvec{p(n-1)+4})
     \nonumber
\\
    & \quad
    + \sum_{j=1}^{p-1}
    m^{j}  \overadd{p}{\chi}_{j,0}\red{\ofDC \circ \zdivtwo}\circ\, C(\TTtpvec{(p + 4) + j (n - 2)})\,,
    \label{19VIII23.2b}
 \end{align}
where $1\leq p \leq k$ and the ``known fields'' on the left-hand side refer to the
contribution to this equation from the
field $\zspaceD^A \interph_{AB}^{[\TTtp]}$
.
Similarly to the convenient case, we rewrite \eqref{19VIII23.2b} as:
 \begin{align}
     \overadd{p}{s}_{A}
     &=
       (\text{gauge fields})^{[\CKVp]} +
    \sum_{i=5}^{p(n-1)+4}
    \overadd{p}{\tilde\chi}_{i} \circ  \zdivtwo\circ\, C(\TTtpvec{i})\,.
    \label{20VIII23.1bf}
 \end{align}

 Recall that the spaces $S,V \subset  \Hkng$, $ k\ge 1$ are defined as 
 %
 $$
 S = \{ \xi_A\in
 \Hkng: \xi_A = \zspaceD_A \phi\,, \phi \in H^{k+1}(\secN) \}
 \,,\quad
 V = \{\xi_A\in
 \Hkng: \zspaceD^A\xi_A = 0  \} \,.
  $$
 When $\secN$ is compact and boundaryless, the spaces $S$ and $V$ are $L^2$-orthogonal, and
any vector field $\xi \in \Hkng$, can be decomposed into its ``scalar'' and ``vector'' parts, denoted as
 \begin{align}
    \xi = \xi^{[S]} + \xi^{[V]} \,.
 \end{align}
 We show in Appendix~\ref{app16X23.1} that, under the current conditions, the scalar-vector decomposition of an element of $\CKVp$ is compatible with  this splitting. That is, if  $\xi \in \CKVp$, then both
$\xi^{\red{[S]}}\,, \xi^{\red{[V]}} \in \CKVp$ as well.

We define the vector field $\upsilon_{k,2}$ as
\begin{align}
    \upsilon_{k,2}^{\red{[S]}}
     = (
     \zspaceD_A \kxi{2}^u
     )^{[\CKVp]}
       \,, \quad \upsilon_{k,2}^{\red{[V]}} = \TTtpvec{2}{}^{\red{[V]}} \,;
        \label{10V24.31}
\end{align}
recall that the fields $\TTtpvec{p} $ have been defined in \eqref{8IX23.w3}. 

We can then write the coupled system as follows: for $n>5$ we let
\begin{align}
\label{13XI23.21}
u_k &:=  \left\{
           \begin{array}{ll}
(
                \overadd{\frac{n-3}{2}}{r} \,,
                \overadd{1}{r} \,,
               \dots \,,
               \overadd{\frac{n-5}{2}}{r} \,,
               \overadd{1}{s} \,,
               \dots \,,
               \overadd{k}{s}
           )^T, &  \\
(               \underbrace{ \overadd{k}{r}  \,,
                 \overadd{k-1}{r}  \,,
                \dots \,,
                 \overadd{\frac{n+1}{2}}{r}}_{k_n \text{ terms}}\,,
                \overadd{\frac{n-1}{2}}{r}\,,
                \overadd{\frac{n-3}{2}}{r} \,,
                \overadd{1}{r} \,,
               \dots \,,
               \overadd{\frac{n-5}{2}}{r} \,,
               \overadd{1}{s} \,,
               \dots \,,
               \overadd{k}{s} 
           )^T, & 
           \end{array}
         \right.
           \\
 \upsilon_{k} &:=
 \left\{
   \begin{array}{ll}
     (
     \upsilon_{k,2} \,,
 \TTtpvec{\frac{9-n}{2}} \,,
 \dots \,,
 \TTtpvec{1}  \,,
 \TTtpvec{5}  \,,
 \dots \,
 \TTtpvec{4+k}
 )^{T },\\
     ( \partial^{k_n}_u \kxi{2}_A \,,
 \partial^{k_n-1}_u \kxi{2}_A \,,
 \dots \,,
 \partial_u \kxi{2}_A\,,
      \kxi{2}_A \,,
     \upsilon_{k,2} \,,
 \TTtpvec{\frac{9-n}{2}} \,,
 \dots \,,
 \TTtpvec{1}  \,,
 \TTtpvec{5}  \,,
 \dots \,
 \TTtpvec{4+k} 
 )^{T [\CKVp]},
   \end{array}
 \right.
\label{13XI23.22}
\end{align}
 where the upper case holds for $k=\frac{n-3}{2}$ and the lower for $k>\frac{n-3}{2}$, and 
 with
 \index{k@$k_n$}
\begin{equation}\label{2X23.1}
 0\le k_n :=\left\{
           \begin{array}{ll}
             0, & \hbox{$k=\frac{n-3}{2}$;} \\
             k-\frac{n-1}{2}, & \hbox{$k\ge \frac{n-1}{2}$\,,}
           \end{array}
         \right.
\end{equation}
where $\cdot^T$ stands for transposition.
 We will return to the case $n=5$ shortly, cf.\ \eqref{14XII23.3}-\eqref{14XII23.2} below.
  Note that $u_k$ contains all the $\overadd{j}{r}$'s and  $\overadd{j}{s}$'s, $j\in\{1,\ldots, k\}$, that  $\overadd{2}{w}{}^{[S]}$, $\overadd{3}{w} $,  $\overadd{4}{w} $ do not appear in the list, and that $\overadd{2}{w}{}^{[V]}$ enters through $v_{k,2}$;
   this particular arrangement of terms in the vectors $u_k$ and $v_k$ is convenient for checking the ellipticity of the system.

We
claim, first,
that the system of coupled equations in~\eqref{9IX23.w2f} and~\eqref{20VIII23.1bf} is equivalent to
%
\begin{align}
    u_k^{\red{[X]}} = \Lambda^{\red{[X]}}_k \upsilon_k^{\red{[X]}} + \mathcal{N}^{\red{[X]}}_k  \rho_k^{\red{[X]}}\,,
    \quad
    X \in \{V,S\}
     \,,
    \label{4IX23.1}
\end{align}
where the $\Lambda^{\red{[X]}}_k $'s are some $2k\times 2k$ matrices of operators, $\mathcal{N}^{\red{[X]}}_{k}$'s are $2k \times k(n-2)$ matrices of operators and
 \begin{align}
     \rho_k :=
     \begin{pmatrix}
         \TTtpvec{5+k}\\
         \TTtpvec{6+k} \\
         \vdots \\
         \TTtpvec{4+k(n-1)}
     \end{pmatrix}\,.
     \label{18XI23.1}
 \end{align}
 For this
we need to show that fields $\upsilon_k^{[S]}$ and $\rho_k^{[S]}$ only produce $u_k^{[S]}$, that is,
\begin{align}
    ( \Lambda^{\red{[S]}} _k \upsilon_k ^{\red{[S]}})^{\red{[V]}} =0
     \,,
\qquad  ( \mathcal{N}^{\red{[S]}} _k  \rho_k^{\red{[S]}})^{\red{[V]}}=0
     \,;
    \label{4IX23.1--}
\end{align}
similarly for $\upsilon_k^{[V]}$, $\rho_k^{[V]}$, and $u_k^{[V]}$.
  This can be justified by first noting
that the  operators appearing in~\eqref{9IX23.w2f} and~\eqref{20VIII23.1bf} are sums of products of elliptic, self-adjoint, pairwise commuting operators of the form%
\index{L@$\tilde{\operatorname{L}}_{a,c,b}$}%
 \begin{align}
     \tilde{\operatorname{L}}_{a,c,b} : &  X \rightarrow X \,, \quad X\in\{S,V\}\,,
   \nonumber
     \\
      \tilde{\operatorname{L}}_{a,c,b} (\xi) = &  a
   \,
   \zdivtwo \circ\, C(\xi) + (\delta_{1b} \TSzlap  + c )\xi
   \label{18XII23.11}
     \\
       = &
     \left\{
       \begin{array}{ll}
        \big(
\frac{a (n-2)}{n-1 } (  \TSzlap  + \myGauss )
+  \delta_{1b} \TSzlap  + c \big)\xi\,,\quad & \hbox{when restricted to  $S$;}
    \\
        \big(
\frac{a }{2 } (  \TSzlap  + (n-2)\myGauss )
+  \delta_{1b} \TSzlap  + c \big)\xi\,, & \hbox{when restricted to $V$,}
       \end{array}
     \right.
     \label{26IX23.1}
 \end{align}
with $a,c\in \R$ and $ b\in\{0,1\}$,
where $\delta_{ij}$ is the usual Kronecker delta. The second equality follows from the results in Appendices \ref{ss18VIII23.1}-\ref{ss4IX23.2}.
Equation \eqref{4IX23.1--} and its $[V]\leftrightarrow [S]$ equivalent is then established by noting that the inclusion $\tilde{\operatorname{L}}_{a,c,b}V \subset V$ follows from \eqref{26IX23.1} and the commutation relation \peqref{14VII23.f1}, while $\tilde{\operatorname{L}}_{a,c,b}S \subset S$ follows from \eqref{26IX23.1} and the commutation relation \eqref{19X23.1}.
  \ptcheck{11XI}

 Now, it follows from \eqref{26IX23.1} that
 the system \eqref{4IX23.1} can be decomposed using
  eigenvectors of $\TSzlap$.
 Let  therefore $\{\red{\eta_\ell}\}$ be a complete set of smooth,
 pairwise $L^2$-orthogonal,
 eigenvectors of $\TSzlap$, with a corresponding discrete set of eigenvalues $\red{\lambda_{\ell}}\to_{\ell \to \infty}-\infty$.
 We can then write, for $X\in\{S,V\}$,
 \begin{align}
     u^{\red{[X]}}_k = \sum_\ell u^{\red{[X]}}_{k,\ell} \red{\eta_\ell} \,,\quad
     \upsilon^{\red{[X]}}_k =  \sum_\ell \upsilon^{\red{[X]}}_{k,\ell} \red{\eta_\ell}
     \,,\quad
     \rho^{\red{[X]}}_k =  \sum_\ell \rho^{\red{[X]}}_{k,\ell} \red{\eta_\ell}
      \,,
     \label{4IX23.7}
 \end{align}
 and we will show that equation \eqref{4IX23.1} can be solved mode-by-mode:
 \begin{align}
        u_{k,\ell}^{\red{[X]}} = \Lambda^{\red{[X]}}_k \upsilon_{k,\ell}^{\red{[X]}} + \mathcal{N}^{\red{[X]}}_k  \rho_{k,\ell}^{\red{[X]}}
    \iff
     u_{k,\ell}^{\red{[X]}} = \Lambda^{\red{[X]}}_k|_{\TSzlap \mapsto \red{\lambda_{\ell}}} \upsilon_{k,\ell}^{\red{[X]}} + \mathcal{N}^{\red{[X]}}_k|_{\TSzlap \mapsto \red{\lambda_{\ell}}}  \rho_{k,\ell}^{\red{[X]}}    \,.
     \label{26IX23.2}
 \end{align}
Note that since  $\CKV = \ker\zdivtwo\circ\, C = \ker \tilde{\operatorname{L}}_{1,0,0}$ by \eqref{18XII23.11} and by Proposition~\pref{P30X22.2}, we see that conformal Killing vectors are eigenvectors of $\TSzlap$. Since the source terms $u^{[X]}_k\in\CKVp$, we only need to consider the modes for which $\red{\eta_\ell}\in\CKVp$.

 The remaining argument for the solvability of equation \eqref{4IX23.1} is similar to the convenient case.
 We show in Appendix~\ref{ss24IX23.2} how to assign to $\Lambda^{[X]}_k $  integers $s_i$ and $t_j$ so
 that the operators $\Lambda^{[X]}_k $ in \eqref{4IX23.1} are elliptic in the sense of Agmon, Douglis and Nirenberg (ADN).
As explained in Appendix~\pref{ss30XI23.2},
this implies that there exists $N(k)$ such that we can find a unique solution of \eqref{26IX23.2} with  $\rho^{\red{[X]}}_{k,\ell}=0$ for all $\ell > N(k)$.
It follows from the ADN estimates   that the sum of these terms converges   when the source terms are sufficiently regular, in Sobolev spaces made precise by Theorem~\ref{T30XI23.1}.
It remains to show that  \eqref{26IX23.2} can be solved in the finite dimensional space of smooth fields $u^{\red{[X]}}_k$, $\upsilon^{\red{[X]}}_k$  and $\rho^{\red{[X]}}_k$  of the form
 \begin{equation}
 \label{26IX23.11}
   u^{\red{[X]}}_k= \sum_{\ell\le N(k) } u^{\red{[X]}}_{k,\ell} \red{\eta_\ell}
   \,, \quad
   \upsilon^{\red{[X]}}_k= \sum_{\ell\le N(k) }  \upsilon^{\red{[X]}}_{k,\ell} \red{\eta_\ell}
   \,, \quad
  \rho^{\red{[X]}}_k = \sum_{\ell\le N(k) }  \rho^{\red{[X]}}_{k,\ell} \red{\eta_\ell}
   \,.
 \end{equation}
This is equivalent to the requirement that all the linear maps  obtained by juxtaposing $\Lambda^{\red{[X]}}_k|_{\TSzlap \mapsto \red{\lambda_{\ell}}}$ and $ \mathcal{N}^{\red{[X]}}_k|_{\TSzlap \mapsto \red{\lambda_{\ell}}}$ with $\ell < N(k)$ and  $\red{\eta_\ell} \in \CKVp$  are surjective.
This, in turn, is equivalent to the fact that the adjoints of these linear maps have no kernel. In what follows, we leave out the ``$|_{\TSzlap \mapsto \red{\lambda_{\ell}}}$'' to ease notation.

It turns out that the question of surjectivity is easier to analyse if we carry out a permutation of the rows and columns of the system \eqref{4IX23.1} as follows: let $\sigma_1$ and $\sigma_2$ be permutations of $\{1,2,...,2k\}$ such that the permuted vectors $(\hat{u}_k)_i := (u_k)_{\sigma_1(i)}$ and $(\hat{\upsilon}_k)_i := (\upsilon_k)_{\sigma_2(i)}$ read
\begin{align}
    \hat{u}_k &= (
                \overadd{1}{r}\,,
                \overadd{1}{s} \,,
                \overadd{2}{r} \,,
                \overadd{2}{s} \,,
               \dots \,,
               \overadd{k}{r} \,,
               \overadd{k}{s}
           )^T \,,
           \label{17XI23.5}
           \\
     \hat{\upsilon}_{k} &:=
\left\{
  \begin{array}{ll}
        (
     \upsilon_{k,2} \,,
 \TTtpvec{\frac{9-n}{2}} \,,
 \dots \,,
 \TTtpvec{1}  \,,
 \TTtpvec{5}  \,,
 \dots \,
 \TTtpvec{4+k}
 )^{T }, & \hbox{$k= \frac{n-3}{2}$;}
\\
        (
      \kxi{2}_A \,,
     \upsilon_{k,2} \,,
 \partial^{k_n}_u \kxi{2}_A \,,
 \partial^{k_n-1}_u \kxi{2}_A \,,
 \dots \,,
 \partial_u \kxi{2}_A
     \,,
 \TTtpvec{\frac{9-n}{2}} \,,
 \dots \,,
 \TTtpvec{1}  \,,
 \TTtpvec{5}  \,,
 \dots \,
 \TTtpvec{4+k}
 )^{T [\CKVp]}, & \hbox{$k> \frac{n-3}{2}$.}
  \end{array}
\right.
\label{17XI23.6}
\end{align}
Then $(\Lambda^{\red{[X]}}_k \, \mathcal{N}^{\red{[X]}}_k)$ is surjective iff
$( \hat{\Lambda}^{\red{[X]}}_k \, \hat{\mathcal{N}}^{\red{[X]}}_k)$
is, where $(\hat{\Lambda}_k^{\red{[X]}})^{i}{}_{j}:= ({\Lambda}_k^{\red{[X]}})^{\sigma_1(i)}{}_{\sigma_2(j)}$
and
$(\hat{\mathcal{N}}^{\red{[X]}}_k)^i{}_{j}:= ({\mathcal{N}}^{\red{[X]}}_k)^{\sigma_1(i)}{}_{j}$.
The benefit in doing so is that the permuted matrix $( \hat{\Lambda}^{\red{[X]}}_k \, \hat{\mathcal{N}}^{\red{[X]}}_k)^\dagger$ then has a similar structure to that in \eqref{27VII23.w1}. Specifically, it can be verified that
\begin{equation}
    \label{16XII23.1}
( \hat{\Lambda}^{\red{[X]}}_k \, \hat{\mathcal{N}}^{\red{[X]}}_k)^\dagger = \begin{pmatrix}
     F^{\red{[X]}} \\ G^{\red{[X]}}
\end{pmatrix} \circ \zdivtwo\circ\, C\,,
\end{equation}
with $G^{\red{[X]}}$  given by
 \begin{align}
     \label{5IX23.w1}
     \left(\begin{array}{l}
            \begin{array}{|cccccc|}
            \hline
               \overadd{1}{\tilde\psi}_{\frac{9-n}{2}}            &  \overadd{1}{\tilde\chi}_{\frac{9-n}{2}}  & \dots \quad \dots    & \ldots \quad \dots  & \dots\quad\dots\quad \dots\quad\dots   \quad \ldots    &  \, \overadd{k}{\tilde\chi}_{\frac{9-n}{2}} \phantom{LL}
               \\
               \vdots           &  \vdots  &     &   &     &  \, \vdots
               \\
               \fbox{ $\overadd{1}{\tilde\psi}_{\frac{7-n}{2}+(n-1)}$ } &  \overadd{1}{\tilde\chi}_{\frac{7-n}{2}+(n-1)}  & \dots \quad \dots    & \ldots \quad \dots  & \dots\quad\dots\quad \dots\quad\dots      \quad \ldots & \, \overadd{k}{\tilde\chi}_{{\frac{7-n}{2}+(n-1)}}
               \\
            \hline
            \end{array}
                \\ \vspace{-.4cm}
        \\
            \begin{array}{|cccccc|}
            \hline
              \phantom{\overadd{k}{\tilde\chi}_{,kkk}} \, 0       &   \hspace{1.5cm}  \overadd{1}{\tilde\chi}_{{\frac{9-n}{2}+(n-1)}}    &  \dots   \quad   &  \dots  \quad     & \ldots\quad\dots\quad\dots\quad \ldots \quad\dots\quad \ldots & \overadd{k}{\tilde\chi}_{{\frac{9-n}{2}+(n-1)}}
              \\
              \phantom{\overadd{k}{\tilde\chi}_{,kkk}} \, \vdots        &     \hspace{2cm}\vdots     &        &                                      &                             & \vdots
              \\
              \phantom{\overadd{k}{\tilde\chi}_{,kkk}} \, 0      &   \hspace{1.5cm} \fbox{$\overadd{1}{\tilde\chi}_{(n-1)+4} $}     &    \dots \quad   &   \dots  \quad & \ldots\quad\dots\quad\dots\quad \ldots\quad\dots\quad \ldots  & \overadd{k}{\tilde\chi}_{(n-1)+4}
              \\
            \hline
            \end{array}
        \\
            \begin{array}{cccccc}
                 & \phantom{\overadd{k}{\tilde\chi}_{,k}\qquad\qquad}    &  \phantom{\overadd{k}{\tilde\chi}_{}\qquad\qquad\qquad}    & \quad \vdots &   &
            \end{array}
        \\
            \begin{array}{|ccccccc|}
            \hline
              \phantom{\overadd{k}{\tilde\chi}_{,kkk}} \, 0      &   \hspace{2cm}     & 0        & \ldots \quad\ldots \quad\ldots \quad \dots \quad \dots & 0        &\hspace{0.3cm}  \overadd{k}{\tilde\psi}_{(k-1)(n-1)+5} \quad  & \overadd{k}{\tilde\chi}_{(k-1)(n-1)+5}
              \\
             \phantom{\overadd{k}{\tilde\chi}_{,kkk}} \, \vdots      &       & \vdots        & \ddots \quad\ddots \quad\ddots \quad \ddots \quad \ddots & \vdots        &\, \vdots \quad  & \vdots
              \\
              \phantom{\overadd{k}{\tilde\chi}_{,kkk}} \, 0      &       & 0        & \ldots \quad\ldots \quad\ldots \quad \dots \quad \dots & 0 &\hspace{0.3cm} \fbox{$\overadd{k}{\tilde\psi}_{\frac{7-n}{2}+k(n-1)} $}\quad  & \overadd{k}{\tilde\chi}_{\frac{7-n}{2}+k(n-1)}
              \\
            \hline
            \end{array}
        \\ \vspace{-.4cm}
        \\
            \begin{array}{|ccccccc|}
            \hline
              \phantom{\overadd{k}{\tilde\chi}_{,kkk}} \, 0      &   \hspace{2cm}     & 0        & \ldots \quad\ldots \quad\ldots \quad \dots \quad \dots & 0        &\hspace{1.4cm}  0 \quad  & \hspace{1.2cm} \overadd{k}{\tilde\chi}_{\frac{9-n}{2}+k(n-1)}
              \\
             \phantom{\overadd{k}{\tilde\chi}_{,kkk}} \, \vdots      &       & \vdots        & \ddots \quad\ddots \quad\ddots \quad \ddots \quad \ddots & \vdots        &\hspace{1.4cm} \vdots \quad  & \hspace{1.2cm} \vdots
              \\
              \phantom{\overadd{k}{\tilde\chi}_{,kkk}} \, 0      &       & 0        & \ldots \quad\ldots \quad\ldots \quad \dots \quad \dots & 0 &\hspace{1.5cm}  0 \quad  & \hspace{1.2cm} \fbox{$\overadd{k}{\tilde\chi}_{k(n-1)+4}$}
              \\
            \hline
            \end{array}
            \end{array}
     \right)\,,
 \end{align}
 where all the operators appearing in \eqref{16XII23.1} and \eqref{5IX23.w1} are understood to be the associated restrictions onto the space $X\in\{S,V\}$.
 In addition, the boxed entries in \eqref{5IX23.w1} are all non-vanishing numbers (cf.\  items (iii) of Sections \peqref{ss26XI22.3} and \peqref{e3IX23.2}). It follows readily that $G^{\red{[X]}}\circ\zdivtwo\circ\, C$,
 and hence
$( \hat{\Lambda}^{\red{[X]}}_k \, \hat{\mathcal{N}}^{\red{[X]}}_k)^\dagger$, has trivial kernel
 on   modes   $\red{\eta_\ell}\in\CKVp$.
\FGp{20XI}

 We show in Appendix \ref{ss30XI23.1} that the resulting solution lies in the expected spaces, specifically
\begin{align}
&
    j \in [1, k_n]\cap \N: \ \partial_u^j \kxi{2}_A 
    \in H^{k_{\gamma} + 1 - 2j}(\secN)
     \,,
    \\
&
    \xi^A \in H^{k_{\gamma}+1}(\secN) \,,\quad
    \xi^u \in H^{k_{\gamma}+2}(\secN)
     \,, \quad
     \kphi{p}{}^{[\TTtp]} \in H^{k_{\gamma}}(\secN) \,.
\end{align}

\paragraph{The case $n=5$.} In this case, we have $\frac{n-3}{2} = 1$. The above analysis can be easily adapted to this case, by removing the rows and columns in the matrix $\Lambda{}^{[X]}_k$ (cf. Figure \ref{f10IX23.1}) associated to the equations for $\overadd{p}{r}$ with
``$1\leq p \leq \frac{n-5}{2}=0$''. Specifically, instead of \eqref{13XI23.21}-\eqref{13XI23.22} we let
\begin{align}
u_k &:=
\begin{cases}
 (
                \overadd{2}{r}\,,
                \overadd{1}{r} \,,
               \overadd{1}{s} \,,
               \dots \,,
               \overadd{k}{s}
           )^T, & k=2\,,
           \\
    (           
               \underbrace{ \overadd{k}{r}  \,,
                 \overadd{k-1}{r}  \,,
                \dots \,,
                 \overadd{3}{r}}_{k_5 \text{ terms}}\,,
                \overadd{2}{r}\,,
                \overadd{1}{r} \,,
               \overadd{1}{s} \,,
               \dots \,,
               \overadd{k}{s} 
           )^T, & k>2\,,
\end{cases}
\label{14XII23.3}
           \\
 \upsilon_{k} &:=
     \begin{cases}
         (
      \kxi{2}_A \,,
     \upsilon_{k,2} \,,
 \TTtpvec{5}  \,,
 \dots \,
  \TTtpvec{4+k}
 )^{T [\CKVp]}, & k=2\,,
 \\
 (      \partial^{k_5}_u \kxi{2}_A \,,
 \partial^{k_5-1}_u \kxi{2}_A \,,
 \dots \,,
 \partial_u \kxi{2}_A\,,
      \kxi{2}_A \,,
     \upsilon_{k,2} \,,
 \TTtpvec{5}  \,,
 \dots \,
 \TTtpvec{4+k} 
 )^{T [\CKVp]}, & k>2\,,
     \end{cases}
\label{14XII23.2}
\end{align}
with $k_5$   as given in \eqref{2X23.1}. The analysis then proceeds analogously to the above.
\FGp{20XI}

\newpage

\newpage


\section{Unobstructed gluing with $m\neq0$}

In this section we conclude by briefly discussing how to remove the obstructions to $\Ck$-gluing imposed by the radially conserved charges (compare with the results summarised in Tables \ref{T11III23.2} and \ref{T11III23.1}). As in previous work~\cite{ChCong1}, the simplest way to fill these gaps is to exploit the linear nature of the field equations and add fields satisfying the linearized equations and  whose charges are specifically chosen to compensate for the obstructions.

For brevity, we shall focus only on the case $m\neq0$, as in Tables \ref{T11III23.2}-\ref{T11XII23.1},
ignoring those radial charges which occur in Table \ref{T11III23.1} when $m=0$.
Then we only have two sets of obstructions to gluing (defined in \eqref{24VII22.4} and \eqref{20VII22.1}) which we reproduce here: 
$$
\kQ{1}{}(\pi^A)  := \int_{\secN}\pi^A \red{\overadd{*}{\Hf}_{uA}}
= \int_{\secN}\pi^A \left[r^{n+1}  \partial_r(r^{-2} \red{h}_{uA})  -
2 {r^{n-1}}\zspaceD_A\delta\beta \right]\,\sm\,,
$$
and
$$ \kQ{2}{}(\lambda):= \int_{\secN}\lambda\Big[r^{n-3}\delta V
- \frac{r^{n-2}}{n-1}\partial_r\big(r^2 \zspaceD^A \delta U_A\big) - \tfrac{2 r^{n-2}}{n-1} \TSzlap \delta\beta
- 2 r^{n-3} V \delta\beta
\Big]\sm  \,.$$
Here $\pi^A\in\CKV$, while
$$
 \TS[\zspaceD_A \zspaceD_B \lambda] = 0\,.
$$
Note again that, on a compact Einstein manifold, the fields $\lambda$ solving this last equation
are only non-constant on the $(n-1)$-sphere, where they are linear combinations of $\ell=0, 1$ spherical harmonics.

Now, as in \cite{ChCong1}, we can differentiate the Birmingham--Kottler metrics with respect to the mass to obtain one family of such data:
\begin{equation}
\frac{d}{dm}
\Big[\big(\twoscsign
-{\alpha^2} r^2  {-\frac{2m}{r^{n-2}}}
\big) du^2-2du \, dr
+ r^2 \ringh_{AB}dx^A dx^B
\Big]
=   -\frac{2}{r^{n-2}}
du^2
\,.
\label{8I23.f1}
\end{equation}
This family can be then used to match all of the $\kQ{2}{}(\lambda=\text{const.})$ charges.

Likewise, a stationary class of data is provided by differentiating the higher dimensional Kerr-(anti) de Sitter metrics with respect to each angular-momentum parameter associated with its respective axis of symmetry. We provide this construction in Bondi coordinates and generalise the linear solution to include Ricci-flat topologies in Appendix~\ref{Kerr-dS appendix}.
This leads to the metric perturbation
\begin{equation}\label{8I23.f2}
\delta V=0=\delta \beta\,,\quad \delta U^A=-\frac{\zlambda^A(x^C) }{r^{n}}\,,
\end{equation}
{where $\zlambda^A$ is any $u$-independent Killing vector on $\secN$. This represents angular momentum about a particular axis of rotation.
One can check that the linearised Einstein equations, discussed in Section~\ref{s3X22.1}, remain satisfied.
When $m\neq0$, it turns out, we need only the linear combination of these two families:}
\begin{equation}\label{18XI24f.3}
	\mathring h =
{r^{-(n-2)}}\left(\zmu \,du^2
	-
	2\mathring \Ipsi_A(x^C)\,
	dx^A  du\right)
	\,.
\end{equation}

To see this for the first charge $\kQ{1}{}(\pi)$ note that, when $m\neq0$, we have used  in \eqref{5V23.1as} the gauge field  $(\kxi{2}^u)^{[=1]}$  to match all of charges associated with the proper conformal Killing vector fields (i.e. $\zspaceD_{A}\pi^{A}\neq0$ which only exist on $\secN$ when it is the round sphere $S^{n-1}$). This leaves only the obstructions arising from the Killing vectors as in \eqref{8I23.f2}. This charge evaluated on the compensating perturbation \eqref{18XI24f.3} is simply
\begin{equation}
	\kQ{1}{}(\pi)= {n} \int_{\secN} \pi^A\lambda_A\, \sm\,.
\end{equation}
%


Finally to see this for the second charge $\kQ{2}{}(\lambda)$ note that, when $m\neq0$, on $S^{n-1}$ in \eqref{16XI23.3} we have matched part of the charge, $\kQ{2}{}{(\lambda^{[=1]})}{}$, using the gauge field $(\zspaceD_B\kxi{2}^B)^{[=1]}$. Thus it remains to match the $\lambda=const.$ charges which is achieved with $\zmu$. In this case the charge associated with the compensating perturbation \eqref{18XI24f.3} is simply
\begin{equation}
	\kQ{2}{}{(1)}{}=
 - \int_{\secN} \, \zmu\, \sm = \zmu\,| S^{n-1}|\,,
\end{equation}
where $|S^{n-1}|$ is the $(n-1)$-dimensional volume of the $S^{n-1}$ sphere. 

This completes the removal of obstructions associated with linearised $\Ck$-gluing when $m\neq0$.

\newpage

\appendix


\section{Transport equations  involving $\partial_u^i \zhTB_{uA}$ and $\partial^{i+1}_u h_{AB}$}
\label{App6III23.1}

In this Appendix we derive the transport equations given in Section \ref{sec6III23.1} of the main text.
The guiding principle behind the calculations that follow is, that an induction argument to determine the higher order $u$-derivatives of the fields will require calculating  $
    r^{k}  \partial_u h_{AB}$ as a sum of terms of  the form $
    r^{\red{k'}} L_{kk'}  h_{AB}$, where $L_{kk'} $ is a partial differential operator in the $x^A$-variables,
together with a term which integrates-out in $r$ to a boundary term determined by the boundary data.

Suppose that
$$
 \mbox{$\mcE_{rA}=0$ and $\TS[\mcE_{AB}]=0$.}
$$
Assume, first, that
\begin{equation}\label{17V23.1}
 \frac{\Z}{2} \ni k\not\in\{- 3,n-3,(n-7)/2\}
 \,.
\end{equation}
We claim that we can then  write
\begin{align}
    r^{k}  \partial_u h_{AB} = \partial_r \qh_{AB}^{(k)} + \ck{k}{\ofPnoP} r^{k-1} h_{AB}  +  \alpha^2 \cka{k} r^{k+1} h_{AB} + m \ckm{k} r^{k + 1 - n} h_{AB}\,,
    \label{3III23.2}
\end{align}
for a field $\qh_{AB}^{(k)}$, some functions $\cka{k}$, $\ckm{k}$ and operators $\ck{k}{\ofPnoP}$.
 (As suggested by the notation, $\alpha^2\cka{k}$ is the term involving $\alpha$ and $m\ckm{k}$ is that involving $m$; in the linear approximation considered here no terms mixing $\alpha$ and $m$ occur).

For this recall  \eqref{3IX22.1HD}, which we reproduce here with a convenient rearrangement of terms:
\FGp{30III23; }
\index{qt@$\tilde q_{AB}$}
\begin{align}
     0= \ & \partial_r \Big[r^{\frac{n-1}2}  \blue{\partial_u \zhTB_{AB}}
      -\frac{2}{(n+1)r^{(n+1)/2}}\partial_r(r^{n-1}TS\big[\zspaceD_A   h_{uB}])
     \nonumber
\\
    &\qquad
    \underbrace{- \frac{ r^{\frac{n-3}{2}}}{2}  V  \partial_r \zhTB_{AB}
     -  \frac{n-1 }{4}  r^{\frac{n-5}2}   V   \zhTB_{AB}
     +\frac{2r^{\frac{n-3}{2}}}{n+1}P\zhTB_{AB}}_{=: \tilde q_{AB}}
     \Big]
   \nonumber
   \\
   &
    - \frac{1}{8}  \left(n^2-1\right)\alpha ^2
   r^{\frac{n-5}2} h_{AB}
     +
     \frac{(n -1)^2 }{4r^{(n+5)/2}} m h_{AB}
    \nonumber
    \\
    &
    + \bigg[\frac{4}{n+1} P -  \tric - \frac{1}{2}\TSzlap
       +\frac{(n-3)(n-1)\myGauss}{8}\bigg]r^{(n-9)/2}h_{AB}
        \,,
        \label{3III23.4}
\end{align}
Recall also \eqref{21II23.1}:
\begin{align}
 0 = \partial_r\bigg( n \zhTB_{uA} + r \partial_r \zhTB_{uA}
     - \frac{1}{r^3} \zspaceD^B h_{AB}
     \bigg) - \frac{1}{r^4}  \zspaceD^B h_{AB}\,.
     \label{3III23.3}
\end{align}
It can be verified by a lengthy calculation, similar to that in Appendix~E
of~\cite{ChCong1},
 that
\wc{see Checking3III23\_5v2.nb}
\FGp{30III23; }
\begin{align}
    & \quad r^{k} \partial_u h_{AB} + B_{k} r^{k - \frac{n-7}{2}} \times \text{RHS of }\eqref{3III23.4} + H_{k} r^{k + 3 } \times C \big[ \text{RHS of }\eqref{3III23.3} \big]
    \nonumber
\\
    & =
    \partial_r \bigg[\ E_{k} r^{k + 1} \partial_u h_{AB} + G_{k} r^{5-n+k +\frac{n-1}{k +4 - n}} \partial_r (r^{\frac{n-1}{n-k -4}+n-3} C(h_{uA}) )
    \nonumber
    \\
    &\quad \quad \quad
    + B_{k} r^{k -  \frac{n-7}{2}} \tilde q_{AB} - H_{k} r^{k} P(h)_{AB} + \frac{r^{k}}{2} (\twoscsign - \alpha^2 r^2 - \frac{2m}{r^{n-2}}) h_{AB} \ \bigg]
    \nonumber
\\
    & \quad
    + \ck{k}{\ofPnoP} r^{k-1} h_{AB} + \alpha^2 \cka{k} r^{k+1} h_{AB}
    + m \, \ckm{k} r^{k+1-n} h_{AB}\,,
    \label{4III23.1}
\end{align}
where
\FGp{30III23;}%
\index{K@$\ck{k}{\ofPnoP}$}%
\index{K@$\cka{k}$}%
\index{K@$\ckm{k}$}%
\begin{align}
    \ck{k}{\ofPnoP} &:= -\frac{1}{7 -  n + 2 k} \bigg[\frac{2 (n - 1) P}{(3 + k) (3 -  n + k) }
    +  2 \tric + \TSzlap -   (n - 4 -  k) (2 + k) \myGauss \bigg]\,,
    \label{3III23.5a}
\\
    \cka{k} &:= \frac{(k +4) (n-k -4)}{n - 7 -2 k} \,,
    \label{3III23.5b}
\\
    \ckm{k} &:= \frac{2(4-n+k)^2}{7 -  n + 2 k}
     \,,
    \label{3III23.5}
\end{align}
and
\FGp{30III23;}
\begin{align}
    B_{k} &:= -\frac{2}{n -7 - 2 k } =: E_{k}\,,
    \nonumber
\\
    H_{k} &:= \frac{2 (n -2 -k)}{(n+1) (k +3) (n - 3 -k)} \,, \quad
   G_{k} := \frac{2 (n -4 -k)}{(k +3) (n - 7 -2 k) (n - 3 -k)}
   \,.
   \label{3III23.6}
\end{align}
We thus obtain the desired form \eqref{3III23.2} with%
\index{qh@$\qh_{AB}^{(k)}$}
\begin{align}
    \label{5III23.1a}
    \qh_{AB}^{(k)}& := \ E_{k} r^{k + 1} \partial_u h_{AB} + G_{k} r^{5-n+k +\frac{n-1}{k +4 - n}} \partial_r (r^{\frac{n-1}{n-k -4}+n-3} C(h_{uA}) )
    \nonumber
    \\
    &\quad \quad \quad
    + B_{k} r^{k -  \frac{n-7}{2}} \tilde q_{AB} - H_{k} r^{k} P(h)_{AB} + \frac{r^{k}}{2} (\twoscsign - \alpha^2 r^2- \frac{2m}{r^{n-2}}) h_{AB}
   \,.
\end{align}

In the special case  $k = -3$ we define%
\index{K@$\zck{-3}{\ofPnoP}$}%
\begin{align}
    \zck{-3}{\ofPnoP} &:= \frac{1}{n-1} \bigg[\frac{(n (1-n) -2 ) P}{n} + 2 \tric +\TSzlap + (n-1) \myGauss \bigg] \,,
    \label{18IV23.1}
    \\
    H_{-3} &:= \frac{n^2-5}{1-n^2} \,,
\quad
    G_{-3} :=  (n-1) - \frac{2(n-2)}{(n-1)n} \,,
\end{align}
and%
\index{qh@$\qh_{AB}^{(k)}$!qh@$\qh^{(-3)}_{AB}$}%
%
\begin{align}
    \label{5III23.1}
    \qh^{(-3)}_{AB}& := E_{-3} r^{-2} \partial_u h_{AB}
    + \partial_r(r^{-1} C(h_{uA}))
    + \frac2n \partial_r\bigg( \log{(\tfrac{r}{r_2})} r^{-1}  C(h_{uA}) \bigg)
    + G_{-3} r^{-2}  C(h_{uA})
    \nonumber
    \\
    &\quad
    + \frac{2(n-1)}{n} r^{-2} \log{(\tfrac{r}{r_2})}  C(h_{uA})
    + B_{-3} r^{(1-n)/2} \tilde q_{AB}
    \nonumber
    \\
    &\quad
     -  H_{-3} r^{-3} P(h)_{AB} - \frac{2}{n} \log{(\tfrac{r}{r_2})} r^{-3} P(h)_{AB}+ \frac{r^{-3}}{2} (\twoscsign - \alpha^2 r^2 - \frac{2m}{r^{n-2}}) h_{AB} \,.
\end{align}
The analysis in Appendices~\ref{A14XI23.1}-\ref{A4IX23.1}  shows that the remaining cases excluded in \eqref{17V23.1} do not occur in our context.

Anticipating, we note that   the $\log$-terms appearing in\eqref{5III23.1} will not contribute to the final recurrence, because of some (unexpected) cancellations which will be established shortly.

In any case, it can   be verified by another straightforward calculation that
 \ptcheck{17V in the mathematica notebook of wan}
\begin{align}
    & \quad r^{-3} \partial_u h_{AB}
    + B_{-3} r^{-\frac{n-1}2} \times \text{RHS of }\eqref{3III23.4}
    - H_{-3}  \times C \big[ \text{RHS of }\eqref{3III23.3} \big]
    \nonumber
    \\
    &
    +\frac2n   \log{(\tfrac{r}{r_2})} \times  C \big[ \text{RHS of }\eqref{3III23.3} \big]
    =
    \partial_r \qh^{(-3)}_{AB}
    + \zck{-3}{\ofPnoP} r^{-4} h_{AB}
    - \frac{2}{n} \log{(\tfrac{r}{r_2})} r^{-4} P(h)_{AB}
    \nonumber
\\
    & \phantom{+\frac2n   \log{(\tfrac{r}{r_2})} \times  C \big[ \text{RHS of }\eqref{3III23.3} \big]=}
    + \alpha^2 \cka{-3} r^{-2} h_{AB}
    + m \, \ckm{-3} r^{-2-n} h_{AB}
    \,,
    \label{31III23.1}
\end{align}
and thus in particular,
 \ptcheck{17V: if the previous is correct than the next is correct; and watch out in a scheme where we do not know that all Einstein equations are satisfied}
\begin{align}
    r^{-3} \partial_u h_{AB} &= \partial_r \qh^{(-3)}_{AB}
    + \bigg[\zck{-3}{\ofPnoP}  - \frac{2}{n}\log{(\tfrac{r}{r_2})}  P \bigg]r^{-4} h_{AB}
    \nonumber
\\
    &\quad
    + \alpha^2 \cka{-3} r^{-2} h_{AB}
    + m \, \ckm{-3} r^{-2-n} h_{AB} \,.
    \label{31III23.2}
\end{align}

We note that the operators $\ck{k}{\ofPnoP}$ and $\zck{-3}{\ofPnoP}$ are self-adjoint and commute with each other.
This is shown explicitly in Appendix~\ref{App19V23.1} below, and ultimately arises from the fact that these operators depend on $P$ and the special combination $ \tric + \frac{1}{2} \TSzlap $ such that $[P,\tric + \frac{1}{2} \TSzlap](h)=0$ see \eqref{20VI23.1}.

\subsection{Integral formulae involving $\partial_u^i h_{uA}$}
 \label{A14XI23.1}

First, taking the $u$-derivative of \eqref{3III23.3} and making use of \eqref{3III23.2}, we have,
 \ptcheck{17V23}
\begin{align}
    \partial_r \bigg( \partial_u \overadd{0}{\Hf}_{uA} - \zspaceD^B\qh^{(-4)}_{AB}
     \bigg) &= \zspaceD^B \ck{-4}{\ofPnoP} r^{-5} h_{AB}
     + m \ckm{-4} r^{-3 - n} \zspaceD^B h_{AB}
     \,,
     \label{5IV23.1}
\end{align}
where we  made use of the fact that $\cka{-4} =0$.

Assuming $\mcE_{rA} = 0$ and $\TS[\mcE_{AB}] = 0$ hold, then it follows by induction that $\partial_u^i\mcE_{rA} = 0$ hold if and only if
\begin{align}
 \forall \    i \geq 1  \quad
    \partial_r \overadd{i}{\Hf}_{uA}
      &=  r^{-(i+4)}\zspaceD^B
      \big(
       \overadd{i}{\chi}\ofP  \, h_{AB}
        \big)
     + m^i \overadd{i}{\chi}_{\red{[m]}} r^{-(4+i(n-1))} \zspaceD^B h_{AB}
    \nonumber
    \\
    &\quad
     + \sum_{j,\ell}^{i_*} m^{j} \alpha^{2\ell}r^{-(i + 4) - j (n - 2) + 2 \ell} \zspaceD^B
     \big(
       \overadd{i}{\chi}_{j,\ell}\ofP  \, h_{AB}
       \big)
     \,,
    \label{4III23.2}
\end{align}
 where $\sum_{j,\ell}^{i_*}$ denotes the sum over $j,\ell$, with
 \index{s@$\sum_{j,\ell}^{i_*}$}%
 \begin{equation}
 1 \leq j \leq i-1 \,,\   j+\ell\leq i \,,\   \text{and } 0\le 2\ell\leq  i +j(n-2)
 \,.
 \label{19V23.3}
 \end{equation}
In the above, the fields $\overadd{i}{\Hf}_{uA}$ depend on $(r, \partial_u^j h_{uA}, \partial_r \partial_u^j h_{uA}, \partial^j_u h_{AB})_{j=0}^i$, with $\overadd{i}{\chi}_{\red{[m]}}\in \R$, and with $\overadd{i}{\chi}\ofP$, respectively  $\overadd{i}{\chi}_{j,\ell}\ofP$, being polynomials in $P$ of order $i$, respectively  $i-j-\ell\le i-1$. We have the recursion formulae:
 \ptcheck{17V partially;  \eqref{5III23.2} not checked either}
\begin{align}
    \overadd{0}{\Hf}_{uA}&:=\Hf_{uA}\,,\quad
    \overadd{1}{\Hf}_{uA}:= \partial_u \Hf_{uA} -\zspaceD^B \qh^{(-4)}_{AB}
\\
    \overadd{1}{\chi}\ofP  &:=  \ck{-4}{\ofPnoP}\,, \quad
    \overadd{1}{\chi}_{\red{[m]}} = \ckm{-4}\,,
\\
    \overadd{i+1}{\chi}\ofP  &= \overadd{i}{\chi}\ofP  \, \ck{-(i+4)}{\ofPnoP}\,,
    \quad
    \overadd{i+1}{\chi}_{\red{[m]}} =  \overadd{i}{\chi}_{\red{[m]}} \ckm{-(4+i(n-1))}
    \label{5III23.2}
\\
    \overadd{i+1}{\Hf}_{uA}
     &=
       \partial_u \overadd{i}{\Hf}_{uA}
    - \zspaceD^B \overadd{i}{\chi}\ofP  \, \qh^{(-i-4)}_{AB}
    - \overadd{i}{\chi}_{\red{[m]}}  m^i  \zspaceD^B \qh^{(-4-i(n-1))}_{AB}
    \nonumber
    \\
    &\quad
    - \blue{ \sum_{j,\ell}^{i_*} } m^{j} \alpha^{2\ell} \zspaceD^B \overadd{i}{\chi}_{j,\ell}\ofP  \, \qh^{(-(i+4)-j(n-2)+2\ell)}_{AB}
    \nonumber
    \\
    &\quad
   -  \alpha^2 \cka{-(i+4)} \hck{-(i+3)}{\TSzlap,\zdivtwo \circ\, C}\overadd{i-1}{\Hf}_{uA}
    \,,
    \label{4III23.3}
\end{align}
with the operator $\hck{k}{\TSzlap,\zdivtwo\circ\, C}$ defined in \eqref{19V.1} below.

 A comment on the powers of $r$ appearing in~\eqref{4III23.2}, and in several similar formulae below, is in order: all the terms appearing in equations such as \eqref{4III23.2} should have the same dimension, keeping in mind that the mass coefficient $m$ has the same dimension as $r^{n-2}$, while $\alpha $ has the same dimension as $r^{-1}$. Hence $m^{j} \alpha^{2\ell}r^{ - j (n - 2) + 2 \ell} $ is dimensionless, as required.

 The last condition  in \eqref{19V23.3}, which is justified below, is the same as $-(i+4)-j(n-2)+2\ell
      \le-4$, which  expresses the fact that the powers of $r$ in the sum do not exceed $-4$. It prevents the appearance of   $\log$ terms in the current recursion.
 Also note that, since $j\le i-1$,
 \begin{equation}\label{21V23.1}
   -(i+4)-j(n-2)+2\ell
    > - (i+4) - i(n-2) = -
    \big(4 + i(n-1)
    \big)
    \,,
 \end{equation}
 so that the most negative power of $r$ comes with the term involving $m^i$ in the first line.

We now prove the recursion by induction. The $i=1$ version of \eqref{4III23.2} is provided by \eqref{5IV23.1}.

To prove \eqref{4III23.2} with $i=2$ we make use of the commutation relation \eqref{11X23.12}, which allows us to write:
\index{K@$\ck{k}{\ofPnoP}$!$\hck{k}{\TSzlap,\zdivtwo\circ\, C}$}%
%
\begin{align}
    \zspaceD^A \ck{k}{\ofPnoP} h_{AB}
 & = \hck{k}{\TSzlap,\zdivtwo\circ\, C} \zspaceD^A h_{AB}
 \,,
 \label{19V.1}
\end{align}
where we recall from the main text that the notation $\tilde{U}\ofDC$ (a tilde over an operator $U$) denotes the replacement of all appearances of the operator $P := C\circ\zdivtwo$, respectively   $\tric$, in $U$ by the operator $\zdivtwo\circ\, C$,  respectively   $(n-2)\myGauss$. 

Then, taking a $u$-derivative of \eqref{5IV23.1} and making use of \eqref{3III23.2} and \eqref{19V.1} we find
 \ptcheck{17V23}
\begin{align}
    &\partial_r\bigg[\partial_u \overadd{1}{\Hf}_{uA} - \zspaceD^B \overadd{1}{\chi}\ofP  \, \qh^{(-5)}_{AB}
    - m \zspaceD^B \overadd{1}{\chi}_{\red{[m]}} \qh^{(-3-n)}_{AB}
    - \alpha^2 \hck{-4}{\TSzlap,\zdivtwo C}\cka{-5} \Hf_{uA}
    \bigg]
    \nonumber
\\
    & =
    \zspaceD^B \overadd{1}{\chi}\ofP  \, \bigg[ \ck{-5}{\ofPnoP} r^{-6} h_{AB}  + m \ckm{-5} r^{-4 - n} h_{AB}\bigg]
    \nonumber
\\
    &\quad
    + m \zspaceD^B \overadd{1}{\chi}_{\red{[m]}}
    \bigg[\ck{-3-n}{\ofPnoP} r^{-4-n} h_{AB}  +  \alpha^2 \cka{-3-n} r^{-2-n} h_{AB} + m \ckm{-3-n} r^{-2-2n} h_{AB} \bigg]
    \nonumber
\\
    & =
    \zspaceD^B \underbrace{\overadd{1}{\chi}\ofP  \, \ck{-5}{\ofPnoP}}_{\overadd{2}{\chi}\ofP } r^{-6} h_{AB}
    + m \zspaceD^B \underbrace{\bigg[ \overadd{1}{\chi}\ofP  \, \ckm{-5}
            +  \overadd{1}{\chi}_{\red{[m]}} \ck{-3-n}{\ofPnoP}\bigg] }_{\overadd{2}{\chi}_{1,0}\ofP }
            r^{-4 - n} h_{AB}
    \nonumber
\\
    &\quad
    +  m \alpha^2  \zspaceD^B \underbrace{\overadd{1}{\chi}_{\red{[m]}} \cka{-3-n}}_{\overadd{2}{\chi}_{1,1}\ofP } r^{-2-n} h_{AB}
    + m^2\zspaceD^B \underbrace{\overadd{1}{\chi}_{\red{[m]}} \ckm{-3-n}}_{\overadd{2}{\chi}_{\red{[m]}}} r^{-2-2n} h_{AB}
     \,,
    \label{5IV23.2}
\end{align}
where we  also made use of \eqref{3III23.3}.
This ends the proof of  \eqref{4III23.2} with $i=2$.

Suppose, next, that \eqref{4III23.2} holds for some $i\ge 2$. Taking  the $u$-derivative of \eqref{4III23.2} and making use of \eqref{3III23.2} we have
 \ptcheck{17V23}
\begin{align}
    &\partial_r \bigg( \partial_u \overadd{i}{\Hf}_{uA}
    - \zspaceD^B \overadd{i}{\chi}\ofP  \, \qh^{(-i-4)}_{AB}
    - \blue{ \sum_{j,\ell}^{i_*} }  m^{j} \alpha^{2\ell} \zspaceD^B \overadd{i}{\chi}_{j,\ell}\ofP  \, \qh^{(-(i+4)-j(n-2)+2\ell)}_{AB}
    \nonumber
    \\
    &\qquad
    - m^i \zspaceD^B \overadd{i}{\chi}_{\red{[m]}} \qh^{(-(4+i(n-1)))}_{AB}
    \bigg)
    \nonumber
\\
    &=
    \zspaceD^B \overadd{i}{\chi}\ofP  \, \bigg[ \ck{-(i+4)}{\ofPnoP} r^{-(i+5)} h_{AB}  +  \alpha^2 \cka{-(i+4)} r^{-(i+3)} h_{AB}
    \nonumber
     \\
     &\qquad\qquad\qquad
     + m \ckm{-(i+4)} r^{-(i+3) - n} h_{AB}
    \bigg]
    \nonumber
\\
    &\quad
    + m^i \zspaceD^B \overadd{i}{\chi}_{\red{[m]}} \bigg[ \ck{-(4+i(n-1))}{\ofPnoP} r^{-(5+i(n-1))} h_{AB}  +  \alpha^2 \cka{-(4+i(n-1))} r^{-(3+i(n-1))} h_{AB}
    \nonumber
     \\
     &\qquad\qquad\qquad
     + m \ckm{-(4+i(n-1))} r^{-(4+(i+1)(n-1))} h_{AB}
    \bigg]
    \nonumber
\\
    &\quad
     + \blue{ \sum_{j,\ell}^{i_*} }  m^{j} \alpha^{2\ell} \zspaceD^B \overadd{i}{\chi}_{j,\ell}\ofP  \, \nonumber
     \\
     &\qquad \quad \times
     \bigg[
     \ck{-(i+4)-j(n-2)+2\ell}{\ofPnoP} r^{-(i+5)-j(n-2)+2\ell} h_{AB}
     \nonumber
     \\
     &\qquad\qquad
     +  \alpha^2 \cka{-(i+4)-j(n-2)+2\ell} r^{-(i+3)-j(n-2)+2\ell} h_{AB}
     \nonumber
     \\
     &\qquad\qquad
     + m \ckm{-(i+4)-j(n-2)+2\ell} r^{-(i+4)-j(n-2)+2\ell - (n-1)} h_{AB}
     \bigg] \,.
    \label{3III23.7}
\end{align}
Let us first justify our  claim above, that the powers of $r$ appearing in the final sum do not exceed $-4$.
 By our induction hypothesis it holds that  $-(i+4)-j(n-2)+2\ell \leq -4$ in each term of the sum appearing in  \eqref{4III23.2}. Now,  out of the three terms in the square brackets, only the  one multiplied by $\alpha^2$ is associated with an increase (compared with $-(i+4)-j(n-2)+2\ell$) in the power of $r$, and this increase is by $1$. Thus the only terms with a power  of $r$ which could possibly exceed  $-4$  are those for which $-(i+3)-j(n-2)+2\ell=-3$. However, all such  terms   are multiplied by $\cka{-4}=0$, which establishes the claim.

Next, all the terms in the last equation are as in \eqref{4III23.2}  except perhaps for the one multiplied by a single power of $\alpha^2$. In order to handle this term we use \eqref{19V.1} and
\eqref{4III23.2} (with $i$ replaced by $i-1$) to write
%
\begin{align}
    \zspaceD^B \overadd{i}{\chi}\ofP  &  r^{-(i+3)} h_{AB}
= \zspaceD^B   \ck{-(i+3)}{\ofPnoP} \overadd{i-1}{\chi}\ofP\, r^{-(i+3)} h_{AB}
        \nonumber
\\
    &=  \hck{-(i+3)}{\TSzlap,\zdivtwo C}
    \zspaceD^B
    \big(
   \overadd{i-1}{\chi}\ofP\, r^{-(i+3)} h_{AB}
   \big)
    \nonumber
\\
    &=    \hck{-(i+3)}{\TSzlap,\zdivtwo C}
    \bigg[ \partial_r\overadd{i-1}{\Hf}_{uA}
    - m^{i-1}  \overadd{i-1}{\chi}_{\red{[m]}} r^{-(4+(i-1)(n-1))}\zspaceD^B h_{AB}
    \nonumber
    \\
    &\quad
     -  \sum_{j,\ell}^{(i-1)_*}  m^{j} \alpha^{2\ell} r^{-(i + 3) - j (n - 2) + 2 \ell}
     \zspaceD^B
     \big( \overadd{i-1}{\chi}_{j,\ell}\ofP  \, h_{AB}\big)
    \bigg]
    \nonumber
\\
    &=    \partial_r\big(\hck{-(i+3)}{\TSzlap,\zdivtwo C}
     \overadd{i-1}{\Hf}_{uA}\big)
     \nonumber
     \\
     &\quad
    - \zspaceD^B \bigg[ m^{i-1} \ck{-(i+3)}{\ofPnoP} \overadd{i-1}{\chi}_{\red{[m]}} r^{-(4+(i-1)(n-1))} h_{AB}
    \nonumber
    \\
    &\quad
     + \sum_{j,\ell}^{(i-1)_*}  m^{j} \alpha^{2\ell} r^{-(i + 3) - j (n - 2) + 2 \ell}
    \ck{-(i+3)}{\ofPnoP} \overadd{i-1}{\chi}_{j,\ell}\ofP  \, h_{AB}
    \bigg]\,,
    \label{4III23.8}
\end{align}
and the induction is completed.

\subsection{Integral formulae involving $\partial_u^i h_{AB}$}
\label{A4IX23.1}
Assume that
$$
\text{$\mcE_{rA} = 0$ and that the equations $\TS[\mcE_{AB}] = 0$ hold.}
$$
We wish to show, by induction, that the equations
$$
 \TS[\partial_u^{i-1} \mcE_{AB}] = 0
$$
hold if and only if
\begin{align}
    \partial_r  \overadd{i}{q}_{AB}
     &=  \overadd{i}{\psi}\ofP\, r^{(n-7-2i)/2} h_{AB}
    + \alpha^{2i} \overadd{i}{\psi}_{[\alpha]} r^{(n-7+2i)/2} h_{AB}
    + m^i \overadd{i}{\psi}_{\red{[m]}} r^{\frac{n-7-2i(n-1)}{2}} h_{AB}
    \nonumber
    \\
    &\quad
     + \sum_{j,\ell}^{i_{**}}  m^{j} \alpha^{2\ell} \overadd{i}{\psi}_{j,\ell}\ofP\, r^{\frac{n-7}{2} - i -  j (n-2) + 2 \ell} h_{AB}\, ,
    \label{5III23.3}
\end{align}
with $i \in \mathbb{Z}^+$ 
(recall that $n >  3$),
and $\sum_{j,\ell}^{i_{**}}$ denotes sum over $j,\ell$ with
\index{s@$\sum_{j,\ell}^{i_{**}}$}%
\begin{equation}
\begin{cases}
 1\leq j \leq i-1 \,,\  j+\ell \leq i \, , & \mbox{if $n$ is even}   \\
 1\leq j \leq i-1 \,,\  j+\ell \leq i    \, ,
 \
  \frac{n-7}{2} - i -  j (n-2) + 2 \ell \leq -4 , & \mbox{if $n$ is odd.}
\end{cases}
 \label{19V23.2}
\end{equation}
The fields $\overadd{i}{q}_{AB}$ appearing  in \eqref{5III23.3} depend on $(r, h_{AB} ,\partial^j_u h_{AB} , \partial_u^{j-1} h_{uA}, \partial_r \partial_u^{j-1} h_{uA})_{j=1}^i$, the operators $\overadd{i}{\psi}\ofP$ and $\overadd{i}{\psi}_{j,\ell}\ofP$ are polynomials in $P$ of orders $i$ and $i-j-\ell$ respectively,
and $\overadd{i}{\psi}_{[\alpha]}$ and $\overadd{i}{\psi}_{\red{[m]}}$ are constants. These are all defined recursively, with 
%
%
\index{qi@$\overadd{i}{q}_{AB}$}%
\begin{align}
    \overadd{0}{q}_{AB}&
     = 0
     \,,
      \quad \overadd{0}{\psi}\ofP  = \overadd{0}{\psi}_{[\alpha]}
     = \overadd{0}{\psi}_{\red{[m]}}=1
     \,,\quad
    \overadd{1}{q}_{AB}= q_{AB}\,, \quad
    \overadd{1}{\psi}_{[\alpha]} = \frac{1}{8}(n^2-1)\,,
\\
    \overadd{1}{\psi}_{\red{[m]}} &= -\frac{(n-1)^2}{4}
    \,,\qquad
    \overadd{1}{\psi}\ofP  = - \bigg[\frac{4}{n+1} P -  \tric - \frac{1}{2}\TSzlap
       +\frac{(n-3)(n-1)\myGauss}{8}\bigg]\,,
    \label{5III23.4a}
\end{align}
and
\begin{align}
    & \overadd{i}{\psi}_{[\alpha]} = \cka{\tfrac{n-9+2i}{2}} \overadd{i-1}{\psi}_{[\alpha]}
    \,,\quad
    \overadd{i}{\psi}_{\red{[m]}} = \ckm{\tfrac{n-7-2(i-1)(n-1)}{2}} \overadd{i-1}{\psi}_{\red{[m]}}
     \,,
     \quad i\geq 2
    \label{5III23.6} \,,
\\
 & \overadd{i}{\psi}\ofP  = \ck{\tfrac{n-5-2i}{2}}{\ofPnoP}\overadd{i-1}{\psi}\ofP \,,\quad 2\leq i \,,\ \mbox{with}\ i\neq \frac{n+1}{2}\ \text{if $n$ is odd}\,,
 \label{28VIII23f.6}
 \\
& \overadd{\frac{n+1}{2}}{\psi}\ofP = \zck{-3}{\ofPnoP}\overadd{\frac{n-1}{2}}{\psi}\ofP\,,
\quad \text{for odd $n$ only}\,,
\label{4IV23.2}
 \\
 & \overadd{i}{q}_{AB}= \partial_u \overadd{i-1}{q}_{AB}
    - \overadd{i-1}{\psi}\ofP  \, \qh_{AB}^{(\frac{n-5-2i}2)}
    - \alpha^{2(i-1)} \overadd{i-1}{\psi}_{[\alpha]} \qh_{AB}^{(\frac{n-9+2i}2)}
     \nonumber
    \\
    &\quad\quad\quad
    - m^{i-1} \overadd{i-1}{\psi}_{\red{[m]}} \qh^{(\frac{n-7-2(i-1)(n-1)}{2})}_{AB}
    - \alpha^2 \mathcal{\widehat K}(2(i-1),P)  \overadd{i-2}{q}_{AB}
    \nonumber
    \\
    &\quad\quad\quad
    - \sum_{j,\ell}^{(i-1)_{**}}  m^{j} \alpha^{2\ell} \overadd{i-1}{\psi}_{j,\ell}\ofP  \, \qh^{(\frac{1}{2} (n-7- 2 (i-1) - 2 j (n-2) + 4 \ell))}_{AB}
     \,,
    \label{5III23.5}
\end{align}
\ptcheck{10III23; checked that denominators don't vanish}
where $\mathcal{\widehat K}(2i,P)$ is defined in \eqref{5III23.11} below. We note that $\cka{n-4}=0$ and hence in fact 
\begin{equation}
    \label{15VI24.4}
    \overadd{i}{\psi}_{[\alpha]} = 0 \quad \text{for}\quad  i \geq \frac{n+1}{2}
\end{equation}
when $n$ is odd. Equation~\eqref{5III23.5} holds for all $i \geq 2$ when $n$ is even and for $2 \leq i \leq \frac{n+1}{2}$ when $n>3$ is odd. When $n>3$ is odd, for $i > \frac{n+1}{2}$, we have instead
\begin{align}
    \overadd{i}{q}_{AB}&= \partial_u \overadd{i-1}{q}_{AB}
    - \overadd{i-1}{\psi}\ofP  \, \qh_{AB}^{(\frac{n-5-2i}2)}
    - m^{i-1} \overadd{i-1}{\psi}_{\red{[m]}} \qh^{(\frac{n-7-2(i-1)(n-1)}{2})}_{AB}
    \nonumber
    \\
    &\quad
    - \sum_{j,\ell}^{(i-1)_{**}}  m^{j} \alpha^{2\ell} \overadd{i-1}{\psi}_{j,\ell}\ofP  \, \qh^{(\frac{1}{2} (n-7- 2 (i-1) - 2 j (n-2) + 4 \ell))}_{AB}
    \,.
    \label{5III23.5b}
\end{align}
We will now prove the recursion by induction. First, \eqref{5III23.3}
 with  $i=1$ is simply \eqref{3IX22.1HD}.
\ptcheck{21V23}

Next, taking the $u$-derivative of \eqref{3IX22.1HD} gives
\begin{align}
    \partial_r \partial_u \overadd{1}{q}_{AB} = \overadd{1}{\psi}\ofP\, r^{(n-9)/2} \partial_u h_{AB} + \alpha^{2} \overadd{1}{\psi}_{[\alpha]} r^{\frac{n-5}2} \partial_u h_{AB}
    + m \overadd{1}{\psi}_{\red{[m]}} r^{-(n+5)/2} \partial_u h_{AB}\,.
\end{align}
We use \eqref{4III23.1} to rewrite this as,
\begin{align}
    &\partial_r \bigg[\partial_u \overadd{1}{q}_{AB}
    - \overadd{1}{\psi}\ofP  \, \qh_{AB}^{(n-9)/2}
    - \alpha^2 \overadd{1}{\psi}_{[\alpha]} \qh_{AB}^{\frac{n-5}2}
    - m \overadd{1}{\psi}_{\red{[m]}} \qh_{AB}^{-(n+5)/2}\bigg]
 \nonumber
 \\
    &=
    \overadd{1}{\psi}\ofP  \, \bigg[\ck{\tfrac{n-9}{2}}{\ofPnoP} r^{(n-11)/2}
     + \alpha^2 \cka{\tfrac{n-9}{2}} r^{(n-7)/2}
     + m \ckm{\tfrac{n-9}{2}} r^{-(n+7)/2}
     \bigg] h_{AB}
     \nonumber
\\
    &\quad
    + \alpha^{2} \overadd{1}{\psi}_{[\alpha]} \big[ \ck {\tfrac{n-5}{2}}{\ofPnoP} r^{(n-7)/2}
     + \alpha^2 \cka{\tfrac{n-5}{2}} r^{\frac{n-3}{2}}
     + m \ckm{\tfrac{n-5}{2}} r^{-(n+3)/2}\big] h_{AB}
     \nonumber
\\
    &\quad
    + m \overadd{1}{\psi}_{\red{[m]}} \big[ \ck{-\tfrac{n+5}{2}}{\ofPnoP} r^{-(n+7)/2}
     + \alpha^2 \cka{-\tfrac{n+5}{2}} r^{-(n+3)/2}
     + m\ckm{-\tfrac{n+5}{2}} r^{-3(n+1)/2}\big] h_{AB}
      \nonumber
 \\
    &=
   \underbrace{ \overadd{1}{\psi}\ofP  \, \ck{\tfrac{n-9}{2}}{\ofPnoP}}_{\overadd{2}{\psi}\ofP } r^{(n-11)/2} h_{AB}
    + \alpha^{4} \underbrace{\overadd{1}{\psi}_{[\alpha]}\cka{\tfrac{n-5}{2}}}_{\overadd{2}{\psi}_{[\alpha]}} r^{\frac{n-3}{2}} h_{AB}
    \nonumber
\\
    &\quad
    + m^2 \underbrace{\overadd{1}{\psi}_{\red{[m]}}\ckm{-\tfrac{n+5}{2}}}_{\overadd{2}{\psi}_{\red{[m]}}} r^{-3(n+1)/2} h_{AB}
    + m \underbrace{\bigg[\overadd{1}{\psi}\ofP  \, \ckm{\tfrac{n-9}{2}}
    + \overadd{1}{\psi}_{\red{[m]}}  \ck{-\tfrac{n+5}{2}}{\ofPnoP}
     \big]}_{\overadd{2}{\psi}_{1,0}\ofP } r^{-(n+7)/2} h_{AB}
     \nonumber
\\
    &\quad
    + m \alpha^2 \underbrace{\bigg[\overadd{1}{\psi}_{[\alpha]} \ckm{\tfrac{n-5}{2}}
    + \overadd{1}{\psi}_{\red{[m]}}  \cka{-\tfrac{n+5}{2}}
     \big]}_{\overadd{2}{\psi}_{1,1}\ofP }  r^{-(n+3)/2} h_{AB}
     \,.
     \label{5III23.7}
\end{align}
where we made use of the fact that
\wc{see Checking3III23\_5v2.nb}
\begin{align}
    \overadd{1}{\psi}\ofP  \, \cka{\tfrac{n-9}{2}} + \overadd{1}{\psi}_{[\alpha]} \ck{\tfrac{n-5}{2}}{\ofPnoP}  = 0 \,,
\end{align}
as verified readily using \eqref{3III23.5a}-\eqref{3III23.5b}, which implies that the two terms involving $r^{(n-7)/2}$ in (the first equality of) \eqref{5III23.7} cancel, after which we are left with the $i=2$ statement.
 \ptcheck{21V23}

Continuing in the same way, we take $\partial_u$ of \eqref{5III23.3} and make use again of \eqref{4III23.1} to obtain:
\begin{align}
    &\partial_r \bigg[
     \partial_u \overadd{i}{q}_{AB}
    - \overadd{i}{\psi}\ofP  \, \qh_{AB}^{(n-7-2i)/2}
    - \alpha^{2i} \overadd{i}{\psi}_{[\alpha]} \qh_{AB}^{(n-7+2i)/2}
    - m^i \overadd{i}{\psi}_{\red{[m]}} \qh^{(\frac{n-7-2i(n-1)}{2})}_{AB}
    \nonumber
    \\
    &\qquad
    - \sum_{j,\ell}^{i_{**}}  m^{j} \alpha^{2\ell} \overadd{i}{\psi}_{j,\ell}\ofP  \, \qh^{(\frac{1}{2} (n-7- 2 i - 2 j (n-2) + 4 \ell))}_{AB}
    \bigg]
    \nonumber
\\
&=
    \underbrace{\overadd{i}{\psi}\ofP  \, \ck{\tfrac{n-7-2i}{2}}{\ofPnoP}}_{\overadd{i+1}{\psi}\ofP  } r^{(n-9-2i)/2} h_{AB}
    + \alpha^{2(i+1)} \underbrace{\overadd{i}{\psi}_{[\alpha]} \cka{\tfrac{n-7+2i}{2}}}_{\overadd{i+1}{\psi}_{[\alpha]}} r^{(n-5+2i)/2} h_{AB}
    \nonumber
\\
    &\quad
    + m^{i+1} \underbrace{\overadd{i}{\psi}_{\red{[m]}} \ckm{\tfrac{n-7-2i(n-1)}{2}}}_{\overadd{i+1}{\psi}_{\red{[m]}} } r^{\frac{n-7-2(i+1)(n-1)}{2}} h_{AB}
    \nonumber
\\
    &\quad
    + m \overadd{i}{\psi}\ofP  \, \ckm{\tfrac{n-7-2i}{2}} r^{\frac{n-7-2i}{2}-(n-1)} h_{AB}
    + m^i \overadd{i}{\psi}_{\red{[m]}}\ck{\tfrac{n-7-2i(n-1)}{2}}{\ofPnoP} r^{\frac{n-7-2i(n-1)}{2}-1} h_{AB}
    \nonumber
\\
    &\quad
    + m \alpha^{2i} \overadd{i}{\psi}_{[\alpha]}\ckm{\tfrac{n-7+2i}{2}} r^{\frac{n-7+2i}{2}-(n-1)} h_{AB}
    + m^i \alpha^{2 } \overadd{i}{\psi}_{\red{[m]}}\cka{\tfrac{n-7-2i(n-1)}{2}} r^{\frac{n-7-2i(n-1)}{2}+1} h_{AB}
     \nonumber
\\
    &\quad
    + \overadd{i}{\psi}\ofP  \, \alpha^2 \cka{\tfrac{n-7-2i}{2}} r^{(n-5-2i)/2} h_{AB}
    + \alpha^{2i} \overadd{i}{\psi}_{[\alpha]} \ck{\tfrac{n-7+2i}{2}}{\ofPnoP} r^{(n-9+2i)/2} h_{AB}
    \nonumber
\\
    &\quad
    + \sum_{j,\ell}^{i_{**}}  m^{j} \alpha^{2\ell} \overadd{i}{\psi}_{j,\ell}\ofP
    [\ldots ]
    \,,
     \label{5III23.9}
\end{align}
where the terms $[\ldots]$, which  have the right structure, have been omitted as their detailed form is irrelevant for further purposes.
\ptcheck{21V23}
 A similar argument to that below \eqref{3III23.7}, together with the fact that for   $i <\frac{n-1}{2}$ we have
 $$\frac{n-7+ 2i}{2}-(n-1) < -3
  \,,
 $$
 justifies the last condition in the $n$-odd case of \eqref{19V23.2}.  It follows from \eqref{5III23.6} that the coefficients in the next-to-last line of \eqref{5III23.9} are given respectively by
\begin{align}
    \overadd{i}{\psi}\ofP  \, \cka{\tfrac{n-7-2i}{2}}
    & = \overadd{i-1}{\psi}\ofP  \, \ck{\tfrac{n-5-2i}{2}}{\ofPnoP}  \cka{\tfrac{n-7-2i}{2}}
     \,,
\\
    \overadd{i}{\psi}_{[\alpha]} \ck{\tfrac{n-7+2i}{2}}{\ofPnoP}
    & = \overadd{i-1}{\psi}_{[\alpha]} \cka{\tfrac{n-9+2i}{2}} \ck{\tfrac{n-7+2i}{2}}{\ofPnoP}  \,.
\end{align}
Next, it can be verified using \eqref{3III23.5} that we have%
\index{K@$\mathcal{\widehat K}(2i,P)$}%
\begin{align}
    \ck{\tfrac{n-5-2i}{2}}{\ofPnoP} \cka{\tfrac{n-7-2i}{2}}  &=  \ck{\tfrac{n-7+2i}{2}}{\ofPnoP} \cka{\tfrac{n-9+2i}{2}}
     \nonumber
\\
    &=: \mathcal{\widehat K}(2i,P)
     \,,
    \label{5III23.11}
\end{align}
\index{K@$\mathcal{\widehat K}(2i,P)$}%
and hence
 \ptcheck{21V23}
\begin{align}
    & \overadd{i}{\psi}\ofP  \, \alpha^2 \cka{\tfrac{n-7-2i}{2}} r^{(n-5-2i)/2} h_{AB}
    + \alpha^{2i} \overadd{i}{\psi}_{[\alpha]} \ck{\tfrac{n-7+2i}{2}}{\ofPnoP} r^{(n-9+2i)/2} h_{AB}
    \nonumber
\\
    & = \alpha^2 \mathcal{\widehat K}(2i,P) \bigg[ \overadd{i-1}{\psi}\ofP\, r^{(n-5-2i)/2} h_{AB}
    + \alpha^{2(i-1)}\overadd{i-1}{\psi}_{[\alpha]} r^{(n-9+2i)/2} h_{AB}
    \bigg]
    \nonumber
\\
    & = \alpha^2 \mathcal{\widehat K}(2i,P) \bigg[\partial_r \overadd{i-1}{q}_{AB}
    - m^{i-1} \overadd{i-1}{\psi}_{\red{[m]}} r^{\frac{n-7-2(i-1)(n-1)}{2}} h_{AB}
    \nonumber
    \\
    &\qquad\qquad\qquad
    - \sum_{j,\ell}^{(i-1)_{**}}  m^{j} \alpha^{2\ell} \overadd{i-1}{\psi}_{j,\ell}\ofP\, r^{\frac{n-5}{2} - i -  j (n-2) + 2 \ell} h_{AB}
    \bigg]\,.
    \label{5III23.10}
\end{align}
Substituting \eqref{5III23.10} back into \eqref{5III23.9} gives an equation of the form
 \ptcheck{21V23}
\begin{align}
     &\partial_r \big[\partial_u \overadd{i}{q}_{AB}
    - \overadd{i}{\psi}\ofP  \, \qh_{AB}^{(n-7-2i)/2}
    - \alpha^{2i} \overadd{i}{\psi}_{[\alpha]} \qh_{AB}^{(n-7+2i)/2}
    - m^i \overadd{i}{\psi}_{\red{[m]}} \qh^{(\frac{n-7-2i(n-1)}{2})}_{AB}
    \nonumber
    \\
    &\qquad
    - \sum_{j,\ell}^{(i-1)_{**}}  m^{j} \alpha^{2\ell} \overadd{i-1}{\psi}_{j,\ell}\ofP  \, \qh^{(\frac{1}{2} (n-7- 2 (i-1) - 2 j (n-2) + 4 \ell))}_{AB}
    - \alpha^2 \mathcal{\widehat K}(2i,P)  \overadd{i-1}{q}_{AB}
    \big]
    \nonumber
\\
    &=
    \overadd{i+1}{\psi}\ofP\, r^{(n-9-2i)/2} h_{AB}
    + \alpha^{2(i+1)} \overadd{i+1}{\psi}_{[\alpha]}  r^{(n-5+2i)/2} h_{AB}
    + m^{i+1} \overadd{i+1}{\psi}_{\red{[m]}}  r^{(n-7-2(i+1)(n-1))/2} h_{AB}
     \nonumber
\\
    &\quad
    + \sum_{j,\ell}^{(i+1)_{**}} m^{j} \alpha^{2\ell} \overadd{i+1}{\psi}_{j,\ell}\ofP\, r^{\frac{n-7}{2} - (i+1) -  j (n-2) + 2 \ell} h_{AB}\,,
\end{align}
which establishes \eqref{5III23.3}  with $i$ replaced by $i+1$. This completes the  induction  for $n$ even, or for $n$ odd and $i\le \frac{n-1}{2}$.

\paragraph{$n$ odd, $i > \frac{n-1}{2}$.}
When $i = \frac{n-1}{2}$ \eqref{5III23.3} reads
\begin{align}
    \partial_r \overadd{\frac{n-1}2}{q_{AB}} &= \overadd{\frac{n-1}2}{\psi}\ofP\, r^{-3} h_{AB}
    + \alpha^{n-1} \overadd{\frac{n-1}2}{\psi_{\alpha}} r^{n-4} h_{AB}
    + m^{\frac{n-1}2} \overadd{\frac{n-1}2}{\psi}_{\red{[m]}} r^{\frac{n-7-(n-1)^2}{2}} h_{AB}
    \nonumber
    \\
    &\quad
     + \sum_{j,\ell}^{(\frac{n-1}2)_{**}}  m^{j} \alpha^{2\ell} \overadd{\frac{n-1}2}{\psi}_{j,\ell}\ofP\, r^{2\ell-3-j(n-2)} h_{AB} \,.
    \label{6III23.3n}
\end{align}
For odd $n$ one can now continue as follows: Taking the $u$-derivative \eqref{6III23.3n} and making use of \eqref{31III23.2} and \eqref{5III23.11}-\eqref{5III23.10} with $i=\frac{n-1}{2}$ gives,
\begin{align}
    &\partial_r \big[\partial_u \overadd{\frac{n-1}2}{q_{AB}}
    - \overadd{\frac{n-1}2}{\psi}\ofP  \, \qh^{(-3)}_{AB}
    - \alpha^{n-1} \overadd{\frac{n-1}2}{\psi_{\alpha}} \qh_{AB}^{(n-4)}
    - \alpha^2 \mathcal{\widehat{K}}(n-1,P) \qh_{AB}^{\frac{n-3}{2}}
    \nonumber
\\
    &\qquad
    - m^{\frac{n-1}2} \overadd{\frac{n-1}2}{\psi}_{\red{[m]}} \qh^{(\frac{n-7-(n-1)^2}{2})}
     - \sum_{j,\ell}^{(\frac{n-1}2)_{**}} m^{j} \alpha^{2\ell} \overadd{\frac{n-1}2}{\psi}_{j,\ell}\ofP  \, \qh_{AB}^{(2\ell-3-j(n-2))}
    \big]
    \nonumber
\\
    &=
     \overadd{\frac{n-1}2}{\psi}\ofP  \, \big[ \big(\zck{-3}{\ofPnoP}  - \frac{2}{n}\log{(\tfrac{r}{r_2})}  P \big) r^{-4}
    + m \, \ckm{-3} r^{-2-n}
     \big] h_{AB}
    \nonumber
\\
    &\quad
    + \alpha^{n-1} \overadd{\frac{n-1}2}{\psi_{\alpha}}  \big[
     \alpha^2 \underbrace{
      \cka{n-4}
      }_{=0}
       r^{n-3}
     + m \underbrace{\ckm{n-4}}_{=0} r^{-3}\big] h_{AB}
     \nonumber
\\
    & \quad
    - \alpha^2 \mathcal{\widehat K}(n-1,P) \sum_{j,\ell}^{(\frac{n-3}{2})_{**}}  m^{j} \alpha^{2\ell} \overadd{\frac{n-3}{2}}{\psi}_{j,\ell}\ofP\, r^{2 (\ell-1) - j (n - 2)} h_{AB}
     \nonumber
\\
     &\quad
     + m^{\frac{n-1}2} \overadd{\frac{n-1}2}{\psi}_{\red{[m]}}\bigg[ \ck{\tfrac{n-7-(n-1)^2}{2}}{\ofPnoP} r^{\frac{n-9-(n-1)^2}{2}}
     +  \alpha^2 \cka{\tfrac{n-7-(n-1)^2}{2}} r^{\frac{n-5-(n-1)^2}{2}}
     \nonumber
     \\&\qquad\qquad
     + m \ckm{\tfrac{n-7-(n-1)^2}{2}} r^{\frac{n-7-(n+1)(n-1)}{2}}\bigg]  h_{AB}
    \nonumber
\\
    &\quad
     + \sum_{j,\ell}^{(\frac{n-1}2)_{**}}  m^{j} \alpha^{2\ell} \overadd{\frac{n-1}2}{\psi}_{j,\ell}\ofP  \, \bigg[ \ck{2\ell-3-j(n-2)}{\ofPnoP} r^{2\ell-4-j(n-2)}
     \nonumber
     \\&\qquad\qquad
     +  \alpha^2 \cka{2\ell-3-j(n-2)} r^{2\ell-2-j(n-2)}
     \nonumber
     \\&\qquad\qquad
     + m \ckm{2\ell-3-j(n-2)} r^{2\ell-2-j(n-2)-n}\bigg]
    h_{AB}
    \nonumber
\\
    &=  \bigg[ \underbrace{\overadd{\frac{n-1}2}{\psi}\ofP  \, \zck{-3}{\ofPnoP}}_{=: \overadd{(n+1)/2}{\psi}\ofP}  - \frac{2}{n}\log{(\tfrac{r}{r_2})} \underbrace{\overadd{\frac{n-1}2}{\psi}\ofP\circ P}_{=0} \bigg]r^{-4} h_{AB}
    \nonumber
\\
    &\quad
    + m^{(n+1)/2} \overadd{(n+1)/2}{\psi}_{\red{[m]}} r^{\frac{n-n^2-6}{2}} h_{AB}
    + \sum_{j,\ell}^{(i+1)_{**}} m^{j} \alpha^{2\ell} \overadd{(n+1)/2}{\psi}_{j,\ell}\ofP\, r^{2(j+\ell-2)-j n} h_{AB}\,,
    \label{6III23.4}
\end{align}
where we show in Appendix~\pref{ss24IX23.1}
that the ``$\log$-term operator'' in the right-hand side of the last equality is zero when acting on symmetric traceless tensors.  More precisely,  for any vector field $X$ we have
\begin{align}
	\overadd{\frac{n-1}2}{\psi}
 &
  \ofP\circ\, C(X)
 =0\,\,  \implies\,\,	
 \nonumber
\\
 &\overadd{\frac{n-1}2}{\psi}\ofP\circ P (h)
 =\overadd{\frac{n-1}2}{\psi}\ofP\circ
 \underbrace{
   P
  \big(
   h^{[\TTt]}}_{\equiv 0} + C(X)
   \big)=0\,.
\end{align}
This also implies that
all logarithmic terms in $\qh^{(-3)}_{AB}$ (cf. \eqref{5III23.1}) appearing in the first line of \eqref{6III23.4} vanish, since they are always multiplied by   $P(h)_{AB}$ or $C(h_{u A})$.

Putting this together, we have
\begin{align}
    \partial_r \overadd{(n+1)/2}{q_{AB}} &=\overadd{(n+1)/2}{\psi}\ofP\, r^{-4} h_{AB}  
    + m^{(n+1)/2} \overadd{(n+1)/2}{\psi}_{\red{[m]}} r^{\frac{n-n^2-6}{2}} h_{AB}
    \nonumber
    \\
    &\quad
     + \sum_{j,\ell}^{i_{**}}  m^{j} \alpha^{2\ell} \overadd{(n+1)/2}{\psi}_{j,\ell}\ofP\, r^{2(j+\ell-2)-j n} h_{AB}
   \,,
    \label{6III23.5}
\end{align}
where $\overadd{(n+1)/2}{q_{AB}}$ is given by the term in the square brackets of the first line of \eqref{6III23.4}.

For $i > (n+1)/2$, the recursion continues as before, with $\overadd{i}{q}_{AB}$ now given by \eqref{5III23.5b}.
 Note that, the only difference in the structure of this expression compared to \eqref{5III23.5} is that the $\alpha$ terms which are not multiplied by factors of $m$ vanish due to the vanishing of $\cka{n-4}$ and $\cka{-4}$.
The recursion again follows by induction   and the calculation is exactly the same as in \eqref{5III23.9}.

\FGp{2VII23}

We also note that when $k\neq -3$, we have under gauge transformations,%
\index{qh@$\qh_{AB}^{(k)}$}%
\index{gauge transformation law!qh@$\qh_{AB}^{(k)}$}%
\index{qh@$\qh_{AB}^{(k)}$!a@gauge transformation law}%
\FGp{: 14VI23
\\--\\
wan: mass term checked 23VI} 
\begin{align}
    \qh_{AB}^{(k)} &\rightarrow \qh_{AB}^{(k)} - \frac{2}{(7+2k-n)(3+k)}
    \nonumber
    \\
    &\phantom{\rightarrow \qh_{AB}^k}
    \times \bigg[
    r^{k+1}\bigg(\frac{(n-1) (k-n+5)}{(n-2) (k-n+3)} P
    + (k+3)  (k-n+5)  \twoscsign
    \nonumber
    \\
    &\phantom{\rightarrow \qh_{AB}^k\times \bigg[r^{k+1}\bigg(}
   - {\frac{2 (k+3) (k-n+4)^2 }{k-n+3} m r^{2-n} } -(k+4) (k-n+2)\alpha^2r^2 \bigg) \TS[\zspaceD_A\zspaceD_B\xi^u]
   \nonumber
\\
   &\qquad\qquad\quad
   - r^{k+2} \bigg( \left(\frac{2}{k-n+3}+2\right)P
   -(k+3)(k-n+4) \left(\alpha ^2 r^2+{2m r^{2-n} }-\twoscsign \right) \bigg) C(\xi)_{AB}
   \nonumber
\\
   &\qquad\qquad\quad
   + r^{k+3} (n-7 -2 k) C(\partial_u\xi)_{AB}
   \nonumber
\\
   &\qquad\qquad\quad
   + r^{k+2} \bigg(\frac{2 (k+3)}{n-1} \bigg) \TS[\zspaceD_A\zspaceD_B\zspaceD_C\xi^C]
    \bigg]
    \,.
    \label{23IV23.1}
\end{align}
\FGp{ 14VI23}
When $k=-3$, $m=0$, the $r-$independent term from the gauge transformation of $\qh^{(-3)}_{AB}$ reads
\index{qh@$\qh_{AB}^{(k)}$}%
\index{gauge transformation law!qh@$\qh_{AB}^{(k)}$!$\qh^{(-3)}_{AB}$}%
\index{qh@$\qh_{AB}^{(k)}$!qh@$\qh^{(-3)}_{AB}$!gauge transformation law}%
\begin{align}
    \qh_{AB}^{(-3)} &\rightarrow \qh_{AB}^{(-3)}
    	+  \left(n-\frac{2}{n}-\frac{2}{n-1}\right)C(\partial_u\xi)_{AB}
    	+\frac{ (n (n+2)-2)}{n} \alpha ^2\TS[\zspaceD_A\zspaceD_B \xi^u]
    	+ \ldots
    	\,,
    \label{15V23.1}
\end{align}
where the $(\ldots)$ above contains $r$-dependent terms. 

\subsection{The case $n=3$}
\label{A19VI24.1}

As mentioned in the introduction, we  provide here the complete list of gauge-invariant obstructions to gluing in spacetime dimension four, as some of the    obstructions  listed in~\cite{ChCong1} were not gauge-invariant. For this, we reexamine our analysis so far with $n=3$.

When $n=3$, the recursion relation, in  Appendix \ref{A14XI23.1} here, for $\overadd{i}{H}_{uA}$ remains unchanged. However, that for the field $\overadd{i}{q}_{AB}$ simplifies due to the vanishing of the operator
\ptcheck{24VI}
\begin{align}
    P - \frac{1}{2}(\TSzlap + 2 \tric)
\end{align}
acting on symmetric trace-free tensors, which follows from \eqref{15VI24.1} after noting that $\TS[R[\gamma]_{AB}] = 0$ when $n=3$. Thus, when $n=3$ and $i=\frac{n-1}{2}=1$ we have (compare also \eqref{5III23.3} with \cite[Equation (101)]{ChCong1}),
\begin{align}
\label{17VI24.21}
    \overadd{\frac{n-1}{2}}{\psi}\ofP =  \overadd{1}{\psi}\ofP = -\frac{4}{n+1} P + \frac{1}{2}(\TSzlap + 2 \tric) = 0 \,.
\end{align}
%
As a result, the equations $\TS[\partial_{u}^{i-1} \mcE_{AB}] = 0$ coincide with  \eqref{5III23.3} where the first and last term vanish, i.e.,
\begin{align}
\forall i \in \Z^+\,, \quad
    \partial_r  \overadd{i}{q}_{AB}
     &=  \alpha^{2i} \overadd{i}{\psi}_{[\alpha]} r^{(n-7+2i)/2} h_{AB}
    + m^i \overadd{i}{\psi}_{\red{[m]}} r^{\frac{n-7-2i(n-1)}{2}} h_{AB}
   \,.
    \label{15VI24.2}
\end{align}
In the rest of this appendix, for ease of comparison  we typically keep $n$ in the equations but  \textit{we assume $n=3$} unless explicitly indicated otherwise.

One verifies that the recursion formulae \eqref{5III23.4a}-\eqref{5III23.6}  for $\overadd{i}{\psi}_{[m]}$ and $\overadd{i}{\psi}_{[\alpha]}$ remain unchanged.

Moreover, it follows from  arguments identical to those leading to \eqref{15VI24.4} that we have, for positive integers $i$,
\begin{align}
    \overadd{i}{\psi}_{[\alpha]}  =
    \begin{cases}
        1 \,, & i=1 \,,\\
        0\,, & \text{otherwise}\,.
    \end{cases}
\end{align}

With some further work one  finds that the recursion formula \eqref{5III23.5} for $\overadd{i}{q}_{AB}$ is now given by
\begin{align}
    \overadd{i}{q}_{AB}&=
    \begin{cases}
        \partial_u \overadd{1}{q}_{AB}
        - \alpha^2 \qh^{(\frac{n-5}{2})}_{AB}
    - m  \qh^{(\frac{n-7-2(n-1)}{2})}_{AB}
     \,, & i=2\\
    \partial_u \overadd{i-1}{q}_{AB}
    - m^{i-1} \overadd{i-1}{\psi}_{\red{[m]}} \qh^{(\frac{n-7-4(i-1)}{2})}_{AB}
    \,, & i>2\,.
    \end{cases}
    \label{15VI24.6}
\end{align}
%

\medskip

\paragraph{Radial charges.} When $n=3$, the obstructions   in Tables \ref{T11III23.2}-\ref{T11XII23.1} for $m\neq0$ remain as they are. When $m=0$, which we will assume for the rest of this appendix, most of the obstructions listed in Table \ref{T11III23.1} remain as presented, with  terms in the obstructions involving fields $\overadd{p}{q}_{AB}$ for $p\leq \frac{n-3}{2}= 0$   understood to be vacuous. This includes the radial charge $\kQ{3}{}$ of \peqref{17X23.1}, which we recall was defined, when $n>3$, as
\begin{align}
    \kQ{3}{}: = \mrL \big( \overadd{\frac{n-3}2}{q}{} ^{[\ker \overadd{\frac{n-3}2}{\psi}]} )
    +   \alpha^{n-3} \frac{n-1}{n-3} \overadd{\frac{n-3}2}{\psi}_{[\alpha]} \chi \,, \quad n>3\,.
\end{align}
When $n=3$, it is convenient to define $\kQ{3}{}$  instead as the radial (but \emph{not} gauge-invariant)
charge $\chi$:
\begin{align}
    \kQ{3}{}: = \chi \,, \quad n = 3\,.
    \label{19VI24.7}
\end{align}
The gauge transformation of $\zspaceD_A \kQ{3}{}$ (cf.\ \eqref{1III23.1} with $n=3$) will play a role in the definition of $\kQ{5,i}{}$ below (cf.\ \eqref{19VI24.21} and following):
\begin{align}
    \zspaceD_A\kQ{3}{} \rightarrow \zspaceD_A\kQ{3}{} -\frac{1}{2}(\TSzlap - \myGauss)(\TSzlap + \myGauss) \zspaceD_A \xi^u\,.
    \label{19VI24.2}
\end{align}

In addition, the expression \eqref{5XI23.41} for the operator $\Lop$ needs to be modified, since $\overadd{\frac{n-5}{2}}{\psi}$ is clearly undefined when $n=3$.
In analogy with the higher dimensional case, we will define the operator $\opL_3$ to be that associated to the gauge transformation of $\overadd{\frac{n-1}{2}}{q}_{AB}$, equivalently $q_{AB}$ for $n=3$, when $\alpha =m= 0$, which by \eqref{7III23.w1} reads
\ptcheck{27Vi24, up to here}
 \begin{align}
     q_{AB} \rightarrow \  & q_{AB}
     \underbrace{-
   \big(
    \TS[\zspaceD_A\zspaceD_B\zspaceD_C\xi^C]
    - \left(  P -\twoscsign  \right)  C(\xi)_{AB}  \big)}_{=: \opL_3(\xi)_{AB}}
    \,.
    \label{19VI24.1}
 \end{align}
Making use of \eqref{26IX23.1} with $a=1$, $b=c=0$, we have
 \ptcheck{27VI; together}
\begin{align}
    \zdivtwo \circ \opL_3 =
    \pm \frac{1}{4}(\TSzlap - \myGauss) (\TSzlap + \myGauss)\,,
     \label{27VI24.11}
\end{align}
where the plus sign is for $V$ and the minus sign is for $S$.
Note that this is  half of the operator acting on $\zspaceD_A\xi^u$ in the gauge transformation of $\zspaceD_A \kQ{3}{}$ in \eqref{19VI24.2}, consistently with the $n>3$ result in
\eqref{17X23.2}, using \eqref{8VI24.21} below.

The radial charge $\kQ{4}{}$ is defined as in \eqref{17X23.4},  and its gauge transformation remains that of \eqref{17X23.6}:
\begin{align}
    \kQ{4}{} \rightarrow \kQ{4}{}+ \zdivtwo \circ \opL_3(\xi) \,.
    \label{19VI24.3}
\end{align}

Finally, a direct calculation of the gauge transformation of the field $\zspaceD^B\overadd{2}{q}_{AB}$ using \eqref{23IV23.1} gives
\ptcheck{27VI; together; and long mathematica calculation...}
\begin{align}
    \zspaceD^B\overadd{2}{q}_{AB} \rightarrow \zspaceD^B\overadd{2}{q}_{AB}  +\zspaceD^B \opL_3(\partial_u \xi)_{AB} + \alpha^2 \zspaceD^B \opL_3(\zspaceD \xi^u)_{AB}
    \,.
    \label{19VI24.6}
\end{align}

\begin{Remark}
\label{R28VI24.1}  The additional $\alpha^2$-term in this gauge transformation was not present in the higher dimensional case  \eqref{24IV23.4} (the gauge transformation of the projection $\overadd{p}{q}{}_{AB}^{[\ker\overadd{p}{\psi}]}$ is given there; however since $\overadd{p}{\psi}=0$ for all integers $p\geq 1$, the projection $\overadd{p}{q}{}_{AB}^{[\ker\overadd{p}{\psi}]}$ coincides with $
 \overadd{p}{q}_{AB}$ when $n=3$, and therefore the transformation law \eqref{24IV23.4} should be read after removing the integral against the $\mu$-field). This additional term is sourced by the field $ \alpha^{2} \qh_{AB}^{(\frac{n-5}2)}$ appearing in $\overadd{2}{q}_{AB}$ of \eqref{15VI24.6}. 
The analogous term in   dimensions $n>3$ is the term $\alpha^{n-1} \overadd{\frac{n-1}{2}}{\psi}_{[\alpha]} \qh_{AB}^{(\frac{n-5}2)}$ appearing in \eqref{5III23.5} 
with $i= \frac{n+1}{2}$. 
Now,
both
 the field $\overadd{\frac{n+1}{2}}{q}_{AB}$ and $\alpha$
 have the same length dimension  as $r^{-1}$, while the gauge field $\xi^u$ has the same dimension as $r$, and $\xi^A$ is dimensionless.
Thus, the only combination of $\alpha$ and gauge fields that is possible in the gauge transformation of $\overadd{\frac{n+1}{2}}{q}{}^{[\ker\overadd{\frac{n+1}{2}}{\psi}]}_{AB}$ are $\partial_u \xi^A$ and $\alpha^2 \xi^u$. A term of the latter form is only possible in $\alpha^{n-1} \overadd{\frac{n-1}{2}}{\psi}_{[\alpha]} \qh_{AB}^{(\frac{n-5}2)}$ when $n-1=2$, i.e., $n=3$, thus explaining the missing contributions of this term in higher dimensions. 
\qedskip
\end{Remark}

It follows inductively from \eqref{15VI24.6} and \eqref{19VI24.6} that for integers $i>2$, under gauge transformations respecting \eqref{5XII19.1aHD},
\begin{align}
    \zspaceD^B\overadd{i}{q}_{AB} \rightarrow \zspaceD^B\overadd{i}{q}_{AB}  +\zspaceD^B \opL_3(\partial^{i-1}_u \xi)_{AB} + \frac{\alpha^2}{2} \zspaceD^B \opL_3(\zspaceD\circ \zdivone \partial_u^{i-3}\xi)_{AB}
    \,.
    \label{19VI24.5}
\end{align}

\section{Some mapping properties}
 \label{s25X23.1-}

In this section we will verify equation \eqref{6V23.1} for various operators which arise in the gluing construction.

\subsection{Elliptic operators}
 \label{s25X23.1}

Consider an elliptic operator  $\nohatL$ of order $\ell$, with smooth coefficients, on  a compact boundaryless $d$-dimensional manifold $\dmanifold$. We will write $\nohatL|_{X}$ to denote the restriction of $\nohatL$ to a subspace $X\subset \MHell$, and $\nohatL^\dagger$ for the formal adjoint of $\nohatL$ obtained by integration by parts against smooth functions; note that no boundary terms arise in the current case.
We  wish to verify a more precise version of \eqref{6V23.1}, namely:

  \begin{proposition}
  \label{P26X23.1}
  Under the conditions just stated, for $k\ge 0$ it holds that
  {\rm
\begin{equation}\label{6V23.1ed}
 (\ker\nohatL^\dagger)^\perp \cap \MMHk= \im\,  (\nohatL|_{\MMHlk})
  \,.
\end{equation}
}
Equivalently, we have  the $L^2$-orthogonal splitting
  {\rm
\begin{equation}\label{6V23.1edxd}
  \MMHk =  \ker\nohatL^\dagger   \oplus \im\,  (\nohatL|_{\MMHlk})
  \,.
\end{equation}
}
\end{proposition}

Here, and elsewhere, orthogonality is meant in $L^2$, in particular $(\ker\nohatL^\dagger)^\perp$ is  a subspace of $L^2$.

Before passing to the proof, let us recall a few standard properties of elliptic operators. Letting $\|\cdot\|_k$ denote the $\MHk$-norm, it holds that:
\begin{enumerate}
  \item If $\nohatL \phi \in \MHk$ in a distributional sense, then $\phi \in \MHlk$ and there exists a constant $C$ such that we have
      \begin{equation}\label{25X23.1}
        \|\phi\|_{\ell+k} \le C
        \big(
        \|\nohatL \phi\|_{k }
        +
        \|\phi\|_{0}
        \big)
        \,.
      \end{equation}
      Note that this immediately implies that elements of the kernel of $\nohatL$ are smooth, in particular $\ker\nohatL|_{\MHlk}=\ker\nohatL|_{\MHell}$ for any $k\ge 0$.
      Similarly $\ker \nohatL^\dagger|_{\MHlk}$ is independent of $k\ge 0$ under the current hypotheses, and we will simply write $\ker \nohatL^\dagger $ for $\ker \nohatL^\dagger|_{\MHlk}$,   except when the indication of the domain is relevant for clarity.

\item The spaces $\ker \nohatL$ and $\ker \nohatL^\dagger$ are  finite dimensional.

  \item Let  a function $\phi\in L^2$ satisfy $\nohatL \phi \in \MHk$ in a distributional sense.
   There exists a constant $C'$ such that if  $\phi$ is $L^2$-orthogonal to $ \ker\nohatL$ then
      \begin{equation}\label{25X23.2}
        \|\phi\|_{\ell+k} \le C'
        \|\nohatL \phi\|_{k }
        \,.
      \end{equation}
\end{enumerate}

\noindent
We are ready now to pass to the
\ptclater{strictly speaking, Lemma \ref{L18XI23.1bx} can be used to streamline many of the arguments below, I propose to leave as is for now and perhaps update later}

\medskip

{\noindent\sc Proof of Proposition~\ref{P26X23.1}:}
To prove \eqref{6V23.1ed} we consider,
first, a function $\phi \in L^2 $ which  is orthogonal to the image of $\nohatL|_{\MHlk}$, thus
\begin{equation}\label{3X23.1}
  \forall \psi \in \MHlk
  \qquad
  \int_{\dmanifold } \phi \, \nohatL \psi = 0
  \,.
\end{equation}
By density of $\MHlk$ in $\MHell$ it also holds that
\begin{equation}\label{3X23.1-a}
  \forall \psi \in \MHell
  \qquad
  \int_{\dmanifold } \phi \, \nohatL \psi = 0
  \,.
\end{equation}
Then $\phi$ is a weak solution of the elliptic equation $\nohatL^\dagger \phi =0$. The operator $\nohatL^\dagger$ is also elliptic, and so by elliptic regularity $\phi\in C^\infty $ and  $\nohatL^\dagger \phi =0$ in the classical sense, thus $\phi \in \ker\nohatL^\dagger$. So
\begin{align}
 (\im\, \nohatL|_{\MHlk}  )^\perp
 &
  = \ker(\nohatL^\dagger|_{\MHell})
 \quad
\Longrightarrow
\quad
\nonumber
\\
  &
   \big(\ker(\nohatL^\dagger|_{\MHell})\big)^\perp  =  \big((\im\, \nohatL|_{\MHlk}  )^\perp\big)^\perp =\overline{\im\, \nohatL|_{\MHlk} }
  \,,\label{3X23.1a}
\end{align}
where $\overline{\im\, \nohatL|_{\MHlk} }$ is the closure of $\im\, \nohatL|_{\MHlk} $ in $L^2$.
We will verify that $\im\, \nohatL |_{\MHell}$ is closed  in $L^2$, and hence \eqref{6V23.1ed} holds with $k=0$:
\begin{equation}\label{3X23.1ab}
 \im\, \nohatL|_{\MHell}  =  \big(\ker(\nohatL^\dagger|_{\MHell})\big)^\perp
  \,.
\end{equation}
(Note that  $\im\, \nohatL|_{\MHlk} $ is not closed in $L^2$ in general, so that from \eqref{3X23.1a} we can only conclude that for $k>0$ we have
\begin{align}
   \big(\ker(\nohatL^\dagger|_{\MHell})\big)^\perp \supset\im\, \nohatL|_{\MHlk}
  \,;
  \label{3X23.1adc}
\end{align}
however, equality suitably understood will follow from elliptic regularity.)

Now, by ellipticity the kernel $ \ker \nohatL |_{\MHell}$ is finite dimensional and therefore   $ \ker \nohatL |_{\MHell}$ is closed (cf.\ e.g.~\cite[Theorem~3.4]{Weidmann})  in $L^2$. Since $A^\perp$ is  closed in a Hilbert space for any set $A$, the set   $( \ker \nohatL |_{\MHell} )^\perp$ is also closed in $L^2$. Therefore we have the splitting $L^2= \ker \nohatL|_{\MHell} \oplus (\ker \nohatL|_{\MHell})^\perp$~\cite[Section~3.1]{Weidmann}.
Hence the map
$$
\nohatL|_{\MHell}: (\ker\nohatL  |_{\MHell})^\perp \to \im\, \nohatL|_{\MHell}
$$
is  surjective.

Consider a sequence $\psi_n $ in $ \im\, \nohatL|_{\MHell}$ converging in $L^2$ to $\psi$.  We have $\psi_n = \nohatL  \phi_n$ for a sequence $\phi_n\in \MHell\cap  (\ker\nohatL |_{\MHell})^\perp$.
Since $\psi_n$ is Cauchy in $L^2$, linearity together with \eqref{25X23.2} shows that the sequence $\phi_n$ is Cauchy in $\MHell $, and therefore there exists $\phi$ such that  $\phi_n\to\phi$ in $\MHell$. By continuity $\psi=\nohatL \phi$, and $\im\,\nohatL |_{\MHell}$ is closed, as claimed.

To finish the proof, consider $\psi \in (\ker\nohatL^\dagger)^\perp \cap \MHk$ with $k>0$. By \eqref{3X23.1ab} we have  $\psi \in \im\, \nohatL|_{\MHell}  $, thus $\psi=\operatorname{L}\phi$ for some $\phi\in \MHell$. By elliptic regularity  $\phi \in \MHlk$, and so $\psi \in  \im\,   \nohatL|_{\MHlk} $, as desired.
\qedskip

\begin{remark}
\label{R25X23.1}
{\rm
Identical arguments apply mutatis mutandi in
 H\"older spaces $C^{k,\lambda}(\dmanifold )$ for any  $\lambda \in (0,1)$, and in Sobolev spaces $W^{k,p}(\dmanifold )$ for any $p\in (1,\infty)$,  and in $L^2$-type Sobolev spaces $H^{\ell+s}(\dmanifold )$ for any $s\in [0,\infty)$, and in fact in any spaces in which the equivalents of the inequalities \eqref{25X23.1}-\eqref{25X23.2} hold.
}
\qed
\end{remark}

\subsection{\protect\rm$\zdivone$}
 \label{ss12XI22.2b}

\index{div@$\zdivone$}%
\index{div@$\zdivone$!image}%
\index{image!d@$\zdivone$}%
\begin{Proposition}
   \label{P30X22.1erd}
   For $k\ge 0$ we have
   {\rm
   $$
   \im\big(\zdivone|_{\MMHkp}\big)= \big(\ker \zdivonedagger \big)^\perp \cap \MMHk
   \,.
   $$
   }
\end{Proposition}

\proof
For all $\xi\in\MHkp$ we have
$$
 \int \zdivone \xi =0
 \,,
$$
hence
$$  \im\big(\zdivone|_{\MHkp}\big)\subset  \{1\}^\perp= \big(\ker \zdivonedagger )^\perp
  \,.
$$
For the reverse inclusion, note that for every $\psi\in  \{1\}^\perp\cap \MHk$ there exists $\phi \in \MHkpp$ such that $\Delta \phi = \psi$. Setting $\xi = \zspaceD\phi\in\MHkp$ one obtains $\psi = \zdivone \xi$, thus
\begin{equation}
\im\big(\zdivone|_{\MHkp}\big)\supset  \{1\}^\perp\cap \MHk
\,.
\tag*{\qed}
\end{equation}

\subsection{\protect\rm$\zdivtwo$}
 \label{ss12XI22.2app}
 
\index{divtwo@$\zdivtwo$!image}%
\index{image!divtwo@$\zdivtwo$}%
While this is irrelevant for our purposes here, we note that $\zdivtwo$ (acting on symmetric traceless two-tensors) is conformally covariant in all dimensions. In particular, in dimension $\ddim\ge 2$, if $g_{AB} = \phi^{-\ell}  \bar g_{AB}$ then
 \begin{equation}\label{18XI22.41}
   D_A h^{AB} =\phi^{ (\ddim+2)\ell/2}  \bar{D}_A(\phi^{ -(\ddim+2)\ell/2}h^{AB})
   \,,
 \end{equation}
where $D$ is the Levi-Civita connection of $g$ and $\bar D$ that of $\bar g$.
This shows that it suffices to understand the kernel of $\zdivtwo$ for, e.g., metrics of constant scalar curvature.

In any case we have:

\index{div@$\zdivtwo$}%
\index{div@$\zdivtwo$!a@image}%
\index{image!div@$\zdivtwo$}
\begin{Proposition}
   \label{P30X22.1}
   For $k\ge 0$ it holds that
   {\rm
   $$
   \im\big(\zdivtwo|_{\MMHkp}\big)= \CKVp\cap \MMHk
   \,.
   $$
   }
   In particular if $\zR_{BC}<0$, then  the operator {\rm ${}\zdivtwo{}$} is surjective.
\end{Proposition}

\proof
For  all conformal Killing vectors $\xi$ and all symmetric traceless tensors $h \in H^1$ we have
 \begin{equation}\label{30X22.CKV5c}
  \int \xi^A \dnabla^B h_{AB}
  =  - \int \dnabla^B\xi^A  h_{AB} =
   - \int \TS (\dnabla^B\xi^A ) h_{AB} = 0
   \,.
 \end{equation}
 Thus
   $$
   \im\big(\zdivtwo|_{\MHkp}\big)\subset  \CKVp\cap \MHk
   \,.
   $$
   The reverse inclusion follows as in the proof of Proposition~\ref{P30X22.1erd} using  the elliptic operator  $\zdivtwo \circ\, C $ instead of the Laplacian.

 Since the space of conformal Killing vectors is trivial when the Ricci tensor is negative by Proposition~\ref{P30X22.2} below, 
   surjectivity for such metrics  follows.
 \qedskip

\subsection{\protect\rm$\zdivone\circ \zdivtwo$}
 \label{ss12XI22.2axb}

\index{div@$\zdivone\circ\zdivtwo$}%
\index{div@$\zdivone\circ\zdivtwo$!image}%
\index{image!d@$\zdivone\circ\zdivtwo$}%
\begin{Proposition}
   \label{P30X22.12}
   Suppose that $(\dmanif,\ringh)$ is Einstein.
   For $k\ge 2$ it holds that
   {\rm
   \begin{equation}\label{5XI23.9}
   \im\big(\zdivone\circ \zdivtwo|_{\MMHkpp}\big)=
   \Big( \ker\big((\zdivone\circ \zdivtwo)^\dagger\big)\Big)^\perp\cap \MMHk
   \,.
   \end{equation}
   }
\end{Proposition}

\begin{Remark}{\rm
It seems clear that the hypothesis that the metric is Einstein is not necessary, but the result in this form is sufficient for our purposes.
\qed
}
\end{Remark}
\proof
For  all functions $\psi \in H^2$ and for for all symmetric and traceless tensors $h \in H^2$ we have
 \begin{equation}\label{30X22.CKV5cd}
  \int \psi \dnabla^B \dnabla^A  h_{AB}
  =   \int \TS(\dnabla^A \dnabla^B \psi)  h_{AB}
   \,.
 \end{equation}
 Thus
 \begin{equation}\label{5XI23.11}
   (\zdivone\circ \zdivtwo)^\dagger
   =
 C\circ \zspaceD
   \,,
 \end{equation}
   and
   $$
   \im\big(\zdivone\circ \zdivtwo|_{\MHkpp}\big)
   \subset
   \Big( \ker ( C\circ \zspaceD)\Big)^\perp\cap \MHk
   \,.
   $$
 To obtain the reverse inclusion,
   let $h_{AB}= C(\zspaceD\phi)_{AB}$ for a function $\phi\in \MHkpppp$. Then
 \begin{equation}\label{5XI23.12}
   \zspaceD^A \zspaceD^B h_{AB}
   =\frac{d-1}{d}\TSzlap ^2 \phi + \frac{1}{2}(\zspaceD^A \zR) \zspaceD_A\phi +\zR_{AB}\zspaceD^A\zspaceD^B\phi
   = \frac{d-1}{d} \zDelta (\zDelta + \myGauss d ) \phi
   \,.
 \end{equation}
 Since $d>2$, the operator at the right-hand side of this equation is an isomorphism between
 $\ker (C\circ \zspaceD)^\perp \cap \MHkpppp$
 and
 $\ker (C\circ \zspaceD)^\perp \cap \MHk$, 
 which finishes the proof.
 \qed

\subsection{$\Lop$ and $\hLop_{n}$}
 \label{ss10XI23.1}

 Recall that (cf.~Remark~\pref{R16X23.1})
\index{Ln@$\underline{\Lop}$}%
\index{Ln@$\Lop$}%
\begin{align}
    \Lop
    &= \overadd{\frac{n-5}{2}}{\psi}\ofP \, \underline{\Lop}
\\
     \underline{\Lop}(\xi)_{AB} &=
     \begin{cases}
        \frac{1}{8} (\TSzlap +2 \tric-4 \myGauss ) (\TSzlap +2 \tric-6 \myGauss )  C(\xi)_{AB}
        &
        \\
        \qquad\qquad\qquad
        - \frac{1}{6}  (\TSzlap + 2 \tric - 5 \myGauss ) \TS[\zspaceD_A\zspaceD_B\zspaceD_C\xi^C] \,, & n=5\,,
        \\
        \frac{1}{(n-1)(n-5)}
     \bigg( (\TSzlap  +2 \tric -2 (n-2) \myGauss) (\TSzlap +2 \tric +(1-n) \myGauss)
     C(\xi)_{AB} &
     \\
     \qquad\qquad\qquad
     -\frac{2 (n-3)}{n-2} (\TSzlap +2 \tric +(5-2 n) \myGauss )
     \TS[\zspaceD_A\zspaceD_B\zspaceD_C\xi^C]
     \bigg) \,, & n\neq 5\,,
     \end{cases}
     \label{30VI.1Ln}
\end{align}
and (cf.~\eqref{5XI23.31})
\index{Ln@$\hLop_{n}$}%
\begin{align}
    \hLop_{n}(\xi^u)
    = \begin{cases}
        -\left(\frac{2}{9}P-\myGauss\right) \circ\, C(\zspaceD_A\xi^u) \,, & n=5 
         \\
        - \frac{n-4}{(n-5)(n-2)^2} \overadd{\frac{n-5}2}{\psi}\ofP \, \bigg( (n-1) P - 2 (n-2)^2 \myGauss
		\bigg) \circ\, C(\zspaceD_A\xi^u)
    \,,
    & n>5\,.
    \end{cases}
     \label{23XI23.1}
\end{align}

The following observation turns out to be useful for the problem at hand:

\begin{Lemma}
 \label{L10xi23.1}
Let $k\ge 0$ and let $A$ and $B$ be   linear partial differential operators of orders $a$ and $b$ such that $A\circ B$ is elliptic.
Then

\begin{enumerate}
  \item $\ker B$ is finite dimensional with $\ker B\subset C^\infty$.
  \item   
If
\begin{equation}\label{14XI23.1a}
  \ker B =  \ker (A\circ B)
   \,,
\end{equation}
then
{\rm
$\im \big(
  B|_{\MMHkab}
  \big)
  \label{10XI23.4}$
}%
is closed in $\MMHka$.
\end{enumerate}
\end{Lemma}

\proof
1. We have  $\ker B\subset \ker A\circ B$, and the result follows form ellipticity.

2.
Let $\psi\in\overline{\im B|_{\MHkab}}$, thus there exists a sequence $\psi_n\in\im B|_{\MHkab}$  converging to $\psi$ in $\MHka$. By definition there exists  $\phi_n\in \MHkab$ such that $\psi_n = B \phi_n$. Since $\ker B$ is finite dimensional we have the $L^2$-orthogonal decomposition
\begin{equation}\label{14XI23.21}
 \MHkab = \ker B \oplus \big( (\ker B)^\perp \cap \MHkab\big)
 \,.
\end{equation}
Indeed, let $\phi_i$, $i=1,\ldots,\dim \ker B$, be a basis of $\ker B$, then $(\ker B)^\perp$ is the intersection of the zero-sets of the finite number of continuous functionals
$$
 \MHkab \ni \psi \mapsto \int  \phi_i B^\dagger \psi
 \,,
$$
and is therefore closed. 

In view of \eqref{14XI23.1a} we can
write $\phi_n = \phi_n^\parallel + \phi_n^\perp$, with $\phi_n^\parallel \in \ker B$ and  $\phi_n^\perp \in \ker B^\perp$, and it holds that
 \begin{equation}\label{11XI23.1}
   \psi_n = B \phi_n^\perp
   \,.
 \end{equation}
The sequence $\psi_n\in\im B|_{\MHkab}$  is Cauchy in $\MHka$. By continuity and linearity the sequence  $A \psi_n$ is Cauchy in $\MHk$. 
 Since
$(\ker B)^\perp = \big( \ker (A\circ B)\big)^\perp$ by hypothesis we have $\psi_n \in \big( \ker (A\circ B)\big)^\perp$.
This shows that we can use \eqref{25X23.2} with $L=A\circ B$ to conclude that the sequence
$\phi_n^\perp$ is Cauchy in $\MHkab$,
thus the limit  $\phi =\lim_{n\to\infty} \phi_n^\perp$ exists.
Continuity of $B:\MHkab\to\MHka$ shows that
\begin{equation}\label{10XI23.3}
 \psi:= \lim_{n\to\infty} \psi_n =
   \lim_{n\to\infty} B\phi_n^\perp = B\phi
   \,,
\end{equation}
and $\im \,B$ is a closed subspace of $\MHka$, as desired.
\qedskip

Recall that $S$ denotes the space of vector fields which are gradients, and $V$ that of vector fields which have vanishing divergence. We have

\begin{Lemma}
 \label{L12XI23.1}
For $k\ge 1$ the spaces
$$
 S^k:= S\cap \MMHk
 \
 \mbox{and}
 \
 V^k:= V\cap \MMHk
$$
are closed, and we have the $L^2$-orthogonal splitting
\begin{equation}\label{13XI23.1k}
\MMHk =  S^k \oplus V^k
\,.
\end{equation}
\end{Lemma}

\proof
$V$ is the kernel of the continuous operator $\zdivone:\MHk \to\MHkm$, hence closed.

The space $S^k$ is the image of $\zspaceD|_{\MHkp}$. The fact that $S^k$ is closed follows from Lemma~\ref{L10xi23.1} with $A=\zdivone$ and $B=\zspaceD$, since $\zdivone\circ \zspaceD=\TSzlap$ is elliptic  and the kernels of $B$ and $A\circ B$ coincide.

The splitting \eqref{13XI23.1k} follows immediately from the equality
\begin{equation}\label{13XI23.2}
  \int \xi_A \zspaceD^A \varphi
 =  -\int \varphi   \zspaceD^A \xi_A
 \,,
\end{equation}
with the last integral vanishing if $\zdivone \xi = 0$.
\qedskip

\index{Ln@$\Lop$!image}%
\index{image!Ln@$\Lop$}%
\begin{Proposition}
\label{P13XI23.1} For $k\ge 0$ and $n\ge 5$ the image 
{\rm
$$
 \im_k\Lop:=\im\, \Lop|_{\MMHktwon}
$$
}is closed in $\MMHkn$, and we have the $L^2$-orthogonal splitting
{\rm
\begin{equation}\label{29XI23.1}
    \MMHkn =  \im_k\Lop \oplus \big(\ker \Lopdagger \cap \MMHkn
                \big)
\,.
\end{equation}
}
\end{Proposition}

\proof
For closedness it suffices to check that the operator $\Lopdagger\circ \Lop$ is elliptic, 
 with the same kernel as $\Lop$;
the result follows then from Lemma~\ref{L10xi23.1} with $A=\Lopdagger$ and $B=\Lop$.

The equality of kernels follows readily from the identity
$$
 \int |\Lop  \psi|^2
  = 
 \int \psi
  \, 
    \Lopdagger \big( \Lop  (\psi)\big)
  \,.
$$

For ellipticity, consider  those terms in  $\Lopdagger$ which are relevant for the determination
of  its principal part. After several commutations, having discarded lower order terms, from \eqref{30VI.1Ln} one obtains
\begin{eqnarray}
  \Lopdagger (h)_A 
   & \sim &
     c_n \,
    E_n \circ \TSzlap \circ \Big( \TSzlap\circ\, C^\dagger (h)_A
     -\frac{2 (n-3)}{n-2} \zspaceD _A D^B D^C h_{BC}
     \Big)
     \nonumber
\\     
    & = &c_n \,
    E_n \circ
     \TSzlap   \circ\Big( \TSzlap \circ\zdivtwo
     -\frac{2 (n-3)}{n-2} \zspaceD \,\circ \zdivone\circ\zdivtwo
     \Big) (h) _A
     \,,\label{14XI23.1}
\end{eqnarray} 
with a dimension-dependent constant $c_n\ne 0$, and
with $E_n$,  which is an elliptic operator of order $n-5$, arising from the principal part of $\overadd{\frac{n-5}{2}}{\psi}\ofP$. Hence
\begin{eqnarray}
  \Lopdagger \circ \Lop 
    & \sim &c_n \,
    E_n \circ
     \TSzlap   \circ\Big(
     \underbrace{
       \TSzlap  
     -\frac{2 (n-3)}{n-2} \zspaceD \,\circ \zdivone
     }_{=: (\diamond)}
     \Big) \circ\zdivtwo \circ \Lop
     \,.\label{14XI23.2}
\end{eqnarray} 
In the proof of Proposition~\pref{P16X23.2} below
we show 
that $\zdivtwo\circ \Lop$ is elliptic. It thus remains to show that $(\diamond)$ is elliptic. Now, the symbol of $(\diamond)$ is
$$
 (\sigma_{(\diamond)}(k)\xi )^A = |k|^2 \xi^A  -\frac{2 (n-3)}{n-2} k_B\xi^B k^A
 \,,
$$
and ellipticity readily follows for $n\ne 4$. 
\qedskip

Similarly, for the operator $\hLop_{n}$, we have:
\index{Ln@$\hLop_{n}$!image}%
\index{image!Ln@$\hLop_{n}$}%
\begin{Proposition}
\label{P23XI23.1} For $k\ge 0$ and $n\ge 5$ the image 
{\rm
$$
 \im_k\hLop_{n}:=\im\, \hLop_{n}|_{\MMHktwonm}
$$
}is closed in $\MMHknm$, and we have the $L^2$-orthogonal splitting
{\rm
\begin{equation}\label{13XI23.1}
    \MMHknm =  \im_k\hLop \oplus \big(\ker \hLopdagger \cap \MMHknm
                \big)
\,.
\end{equation}
}
\end{Proposition}

\proof
The proof is rather  similar to that of Proposition~\ref{P13XI23.1}, we provide the details for completeness.

For closedness it suffices to check that the operator $\hLopdagger\circ \hLop_{n}$ is elliptic;
the result follows then from Lemma~\ref{L10xi23.1} with $A=\hLopdagger$ and $B=\hLop_{n}$.

Thus, consider  those terms in  $\hLopdagger$ which are relevant for the determination
of  its principal part. After several commutations (making use of \eqref{10XI23.1}), having discarded lower order terms, from \eqref{23XI23.1} one obtains
 \ptcheck{25XI23}
\begin{eqnarray}
  \hLopdagger
   & \sim &
     c_n \,
     \TSzlap \circ E_n \circ \zdivone\circ \zdivtwo
\end{eqnarray} 
with a dimension-dependent constant $c_n\ne 0$, 
and with $E_n$,  which is an elliptic operator of order $n-5$, arising from the principal part of $\overadd{\frac{n-5}{2}}{\psi}\ofP$. Hence
\begin{eqnarray}
  \hLopdagger \circ \hLop_{n} 
    & \sim &
    c_n \,
     \TSzlap   \circ  E_n 
     \circ \zdivone\circ \zdivtwo \circ \hLop_{n}
     \,,
     \label{23XI23.2}
\end{eqnarray} 
is elliptic since  $\zdivone\circ\zdivtwo\circ \hLop_{n}$ is (cf.\ the proof of Proposition~\pref{P16X23.1} below).
\qed

\subsection{\protect\rm $\psi$} 
 \label{ss12XI22.2n}

\ptclater{  LemmaClosed.tex can streamline several arguments above ...}
The following result will be useful:

\begin{Lemma}
 \label{L18XI23.1bx}
 Let $\nohatL$ be a linear partial differential operator of order $\ell$ and let $k\ge 0$. Then
 \begin{enumerate}
   \item 
\begin{equation}\label{18XI23.2}
  (\im \, \nohatL|_{\MHlk})^\perp \cap \MHlk = \ker 
   \big(\nohatL^\dagger|_{\MHlk}
   \big)
  \,.
\end{equation}
   \item If 
   
   \begin{enumerate}
     \item  $ \im \,( \nohatL|_{\MHlk})^\perp \cap \MHk$ is $L^2$-dense in $ \im \, (\nohatL|_{\MHlk})^\perp $,  
     \item
     and if 
   the $L^2$-closure of $   \im \,
  \nohatL|_{\MHlk} $ satisfies
\begin{eqnarray}\label{18XI23.29} 
   \overline{
   \im \,
  \nohatL|_{\MHlk} 
  }\cap \MHk
  =
   \im \,
  \nohatL|_{\MHlk} 
 \,,
\end{eqnarray}
   \end{enumerate} 
then \eqref{6V23.1ed} holds.
\end{enumerate}
\end{Lemma}
 
\proof
1. For $\phi\,,\red{\lambda} \in \MHlk $ with $k\ge0$ it holds that
\begin{equation}\label{3X23.11xc}
  \int_{\dmanifold } \phi \, \nohatL \red{\lambda} =\int_{\dmanifold } \red{\lambda} \, \nohatL^\dagger \phi 
  \,.
\end{equation}
If $\nohatL^\dagger \phi =0$ the right-hand side is zero, which shows that
\begin{equation}\label{18XI23.2-}
    \ker 
   \big(\nohatL^\dagger|_{\MHlk}
   \big)
   \subset
   (\im \, \nohatL|_{\MHlk})^\perp 
  \,.
\end{equation}

To obtain the reverse inclusion, let  $\phi \in \MHlk $ be  orthogonal to the image of $\nohatL|_{\MHlk}$.
Then the left-hand side of \eqref{3X23.11xc} vanishes and we find
\begin{equation}\label{3X23.1rp}
  \forall \ \red{\lambda} \in \MHlk
  \qquad
  0 =  \int_{\dmanifold } \red{\lambda} \, \nohatL^\dagger \phi 
  \,.
\end{equation}
By density of $\MHlk$ in $L^2$ it also holds that
\begin{equation}\label{3X23.1-}
  \forall  \ \red{\lambda} \in L^2
  \qquad
    \int_{\dmanifold } \red{\lambda} \, \nohatL^\dagger \phi  = 0
  \,.
\end{equation}
Thus $\nohatL^\dagger \phi =0$, and together with \eqref{18XI23.2-} we obtain \eqref{18XI23.2}.

2. Condition (a) together with point 1.\ and standard properties of orthogonality in $L^2$ (cf., e.g., \cite[Theorem~3.4]{Weidmann}) yield
\begin{eqnarray}\label{18XI23.26}
  \Big( \ker 
   \big(\nohatL^\dagger|_{\MHlk}
   \big)\Big)^\perp
   & = & 
   \Big(
  (\im \, \nohatL|_{\MHlk})^\perp \cap \MHlk
  \Big)^\perp 
   =
   \Big(
   \overline{
  (\im \, \nohatL|_{\MHlk})^\perp \cap \MHlk
  }
  \Big)^\perp
  \nonumber
\\
   & = &
   \Big(
   \overline{
   (\im \,
  \nohatL|_{\MHlk})^\perp
  }
  \Big)^\perp =
   \Big( 
   (\im \,
  \nohatL|_{\MHlk})^\perp 
  \Big)^\perp
  \nonumber
  \\
  &= & 
   \overline{
   \im \,
  \nohatL|_{\MHlk} 
  }
  \,,
\end{eqnarray}
where `` $\overline{ \cdot}$ ''  denotes closure in the $L^2$- topology.  Hence, by hypothesis (b), 
\begin{equation}
  \Big( \ker 
   \big(\nohatL^\dagger|_{\MHlk}
   \big)\Big)^\perp\cap \MHk  = 
   \overline{
   \im \,
  \nohatL|_{\MHlk} 
  }\cap \MHk
  =
   \im \,
  \nohatL|_{\MHlk} 
  \,.
  \tag*{$\Box$}
\end{equation}
%


\index{psi@$\overset{(i)}{\psi}$}%
\index{image!psi@$\overset{(i)}{\psi}$}%
\begin{Proposition}
 \label{P14XI23.1psi}
For $k\ge 1$ the image
{\rm
$$
 \im_k 
 \overset{(i)}{\psi}  
 :=
 \im
 \big( \overset{(i)}{\psi}\ofP|_{\MMHktwoi} \big)
$$
}is closed in $\MMHk$, and we have
{\rm
\begin{eqnarray}
 \im_k
 \overset{(i)}{\psi}  
   &= &
 \big(
   \ker   
 \overset{(i)}{\psi} \ofP  
   \big)^\perp\cap \MMHk 
   \,.
    \label{9XII23.31}
\end{eqnarray}
}
\end{Proposition}

\proof
We start by noting that the operators $\overset{(i)}{\psi}\ofP$, which we will denote by $\overset{(i)}{\psi} $ for brevity, are formally self-adjoint, of order $2i$.
When $i < \frac{n-3}{2}$, or when the pair $(n,\ell)$ is convenient and $i\in \N$,  the  $\overset{(i)}{\psi}$'s are elliptic by Propositions~\ref{P22VIII23.1} and \ref{P17X23.1} below.
Our claim follows then from  Proposition~\ref{P26X23.1}.

To continue, let the pair $(n,\ell)$  be inconvenient and let $i= \frac{n-3}{2}$.  
 We have the $L^2$-orthogonal splittings 
 \begin{eqnarray}\label{23XII23.91} 
   (V\oplus \TTt)\cap \MHktwoin 
     &= &   \ker 
      \big(
    \overset{(\frac{n-3}{2})}{\psi}  |_{V\oplus\TTt}
     \big)
     \oplus 
            \Big(
            \underbrace{
    \Big( \ker 
      \big(
    \overset{(\frac{n-3}{2})}{\psi} |_{V\oplus\TTt}
     \big)
    \Big)^\perp 
    \cap \MHktwoin  
    }_{=: X^{k+n-3}}
        \Big)
         \,,         
\\
 \MHktwoin
  & = &  
   \underbrace{
  \big(
  S\cap \MHktwoin 
  \big) \oplus \ker 
      \big(
    \overset{(\frac{n-3}{2})}{\psi}|_{V\oplus\TTt}
     \big)
     }_{= \ker  \overset{(\frac{n-3}{2})}{\psi}|_{\MHktwoin}}
      \oplus       
         X^{k+n-3}
  \,.
 \end{eqnarray}
 We claim that the operators 
$ \overset{(\frac{n-3}{2})}{\psi}|_{V}$ and of
$ \overset{(\frac{n-3}{2})}{\psi}|_{\TTt}$ are elliptic, in the sense that they are restrictions to $V$, respectively to $\TTt$ of elliptic operators. Indeed, consider
\begin{align}
	{\cal L}_{\TTt} (j) 
  &:=  -\frac{1}{(7 - n + 2 j)}\left( \zTSlap_T+ [4+j(6-n+j)]\myGauss\right) 
\,,
\\
 \overset{(i)}{\psi}_{\TTt}
 & :=  \prod_{j=2}^{i} {\cal L}_{\TTt}(\tfrac{n-5-2j}{2})\overset{(1)}{\psi}\ofP 
 \,,
\end{align}
where%
\index{Delta@$\zTSlap_T$}%
\begin{equation}
	\zTSlap_T h\equiv\left[\TSzlap+2(\tric-(n-2)\myGauss)\right]h
  \,.
\label{8XII23.g6}
\end{equation}
The operators $\overset{(i)}{\psi}_{\TTt}$ are  elliptic, 
and the restriction $ \overset{(\frac{n-3}{2})}{\psi}|_{\TTt}$  coincides with the restriction of $ \overset{(i)}{\psi}_{\TTt}$ to $\TTt$  (compare \eqref{6III23.w9}, together with
\eqref{8XII2.g2} and \eqref{8XII23.g1}, Appendix~\ref{App8XII23.1} below). 

Similarly let
\begin{align}
	{\cal L}_V (j) 
  &\equiv -\frac{(2+j)(4-n + j)}{(7 - n + 2 j)(3-n+j)(3+j)}\left( \zTSlap_T +(1+j)(5-n+j)\myGauss\right) 
\,,
\\
 \overset{(\frac{n-3}{2})}{\psi}_V
 & :=  \prod_{j=2}^{\frac{n-3}{2}}
   {\cal L}_V(\tfrac{n-5-2j}{2})\overset{(1)}{\psi}\ofP 
 \,;
\end{align}
cf.~\eqref{8XII2.g2} and \eqref{8XII23.g3}. The restriction of the elliptic operator $\overset{(\frac{n-3}{2})}{\psi}_V$ to $V$ 
coincides with $\overset{(\frac{n-3}{2})}{\psi} |_V$.

Ellipticity of $ \overset{(\frac{n-3}{2})}{\psi}|_{V}$ and of
$ \overset{(\frac{n-3}{2})}{\psi}|_{\TTt}$, in the sense just explained, 
 shows that the image of $X^{k+n-3}$ by $ \overset{(\frac{n-3}{2})}{\psi}$ is $L^2$-closed and equals $X^{k}$.
 Lemma~\ref{L18XI23.1bx} gives
 \begin{eqnarray}\label{23XII23.92}
 \MHk
  & = &   
  \big(
  S\cap \MHk
  \big) \oplus \ker 
      \big(
    \overset{(\frac{n-3}{2})}{\psi}|_{V\oplus\TTt}
     \big) 
      \oplus       
       X^{k} 
          \nonumber
\\          
  & = &   
  \big(
  S\cap \MHk
  \big) \oplus \ker 
      \big(
    \overset{(\frac{n-3}{2})}{\psi}|_{V\oplus\TTt}
     \big) 
      \oplus        \im \big(
          \overset{(\frac{n-3}{2})}{\psi}|_{X^{k+n-3}}
          \big)
  \,,
 \end{eqnarray}
as the first two summands are $L^2$-closed.

The proof for inconvenient pairs  $(n,\ell)$  with  $i> \frac{n-3}{2}$ is similar, based on 
the $L^2$-orthogonal splittings 
 \begin{eqnarray}\label{23XII23.93}
 \TTt \cap \MHktwoi 
     &= &   \ker 
      \big(
    \overset{(i)}{\psi}  |_{\TTt}
     \big)
     \oplus 
            \Big(
            \underbrace{
    \Big( \ker 
      \big(
    \overset{(i)}{\psi} |_{\TTt}
     \big)
    \Big)^\perp 
    \cap \MHktwoi  
    }_{=: X^{k+2i}}
        \Big)
         \,,         
\\
 \MHktwoi
  & = &  
   \underbrace{
  \big(
  (S\oplus V)\cap \MHktwoi 
  \big) \oplus \ker 
      \big(
    \overset{(i)}{\psi}|_{\TTt}
     \big)
     }_{= \ker  \overset{(i)}{\psi}|_{\MHktwoi}}
      \oplus       
         X^{k+2i} 
  \,.
 \end{eqnarray}
Indeed, ellipticity of $ \overset{(i)}{\psi}|_{\TTt}$ shows that the image of $X^{k+2i}$ by $ \overset{(i)}{\psi}$ is $L^2$-closed and equals $X^{k}$.
Lemma~\ref{L18XI23.1bx} applies as before.
 %
\qed

\seccheck{23XII}

%
%
\subsection{\protect\rm $\zdivtwo\circ\chi$}
 \label{ss12XI22.2}
\index{div@$\zdivtwo$!div@$\zdivtwo \circ \overset{(i)}{\chi}$}%
\index{image!div@$\zdivtwo$!$\zdivtwo \circ \overset{(i)}{\chi}$}%
\begin{Proposition}
 \label{P14XI23.1}
For $k\ge 0 $ the image
{\rm 
$$
 \im_k  
 \big(\zdivtwo \circ \overset{(i)}{\chi}  \big)
 :=
 \im 
 \big(\zdivtwo \circ \overset{(i)}{\chi}\ofP|_{\MMHktwoip} \big)
$$
}is closed in $\MMHk$, and we have 
{\rm
\begin{eqnarray}
\im_k  
 \big(\zdivtwo \circ \overset{(i)}{\chi}  \big)
   &= & 
 \Big(
   \ker \Big(\big(\zdivtwo \circ \overset{(i)}{\chi}  \big)^\dagger 
   \Big)
   \Big)^\perp\cap \MMHk
   \nonumber
\\    
   &\equiv  & 
 \Big(
   \ker \big(\overset{(i)}{\chi}\circ\, C \big)  
   \Big)^\perp\cap \MMHk
   \,.
\end{eqnarray}
}  
\end{Proposition}

\proof
 Let us for conciseness write $ \overset{(i)}{\chi} $ for $ \overset{(i)}{\chi} \ofP$.
The commutation relation \eqref{11X23.11} shows that
\begin{equation}\label{19XI23.11}
  \zdivtwo \circ \overset{(i)}{\chi} = \overadd{i}{\tilde\chi}
  \circ \zdivtwo 
\end{equation}
with an elliptic operator $\overadd{i}{\tilde\chi}$ (cf. Proposition~\ref{P18VIII23.2} below), 
the precise form of which is irrelevant for the current purposes. This shows that $\TTt\subset \ker ( \zdivtwo \circ \overset{(i)}{\chi} )$. The York splitting
\begin{equation}\label{19XI23.12}
  \MHkp =  \TTt \oplus\im\, C
\end{equation}
shows that 
\begin{equation}\label{19XI23.13}
  \im_k  
 \big(\zdivtwo \circ \overset{(i)}{\chi}  \big) 
 = \im\, 
 \big(\zdivtwo \circ \overset{(i)}{\chi}\circ\, C  \big) 
 = \im\, 
 \big( \overadd{i}{\tilde\chi}
  \circ \zdivtwo \circ\, C  \big) 
 \,.
\end{equation}
As $\zdivtwo \circ\, C  $ is elliptic, the operator at the right-hand side of the last equality is elliptic, and the result follows from 
Proposition~\ref{P26X23.1}.
\qed

\section{Operators on $\secN$}
    \label{App12XI22.1h}

The aim of this appendix is to analyse some further mapping properties of operators, acting on fields defined on   cross-sections of $\mcN$, as relevant for the problem at hand. The examples of main interest arise from null hypersurfaces in Birmingham-Kottler metrics, where the metric on cross-sections is Einstein.  In this appendix we thus restrict ourselves to \emph{compact $\ddim$-dimensional Einstein manifolds}
$(\dmanif ,\zgamma)$, $\ddim\ge 3$. The corresponding results with $\ddim=2$ can be found in~\cite{Czimek:2016ydb,ChCong1}.

Some of the arguments below are repetitive with these in $\ddim=2$  in~\cite[Appendix~C]{ChCong1}, we include them here for  the convenience of the reader.

The results here hold mutatis mutandi in weighted H\"older spaces, where orthogonality is always understood in $L^2$.

\subsection{The conformal Killing operator}
 \label{App30X22}

\index{C@$\CKV$}%
 Consider the  \emph{conformal Killing operator} $C$ on a closed $\ddim$-dimensional Riemannian manifold $(\dmanif,\dmetric)$:
 \begin{equation}\label{30X22.CKV1a}
  \xi^A \mapsto  \dnabla _A \xi_B +  \dnabla _B \xi_A
    -  \frac 2 d   \dnabla ^C \xi_C
    \dmetric _{AB}
    =:  2  C(\xi)_{AB}
    \,,
 \end{equation}
 where $\dnabla$ is the covariant derivative operator of the metric $\dmetric$.
 When invoking    $C(\xi)$ we will always assume implicitly that $\xi\in \HkdM$ with some $k\ge1$.

 Recall that: we use the symbol $\CKV$ to denote the space of conformal Killing vectors; $\TTt$ denotes the space of trace-free divergence-free symmetric two-tensors;  orthogonality is defined in $L^2$;  $\TS$ denotes the trace-free symmetric part of a two-tensor.

 We have:
 \seccheck{21XII}

  \begin{Proposition}
   \label{P30X22.2}
   \begin{enumerate}
     \item
  There are no nontrivial Killing vectors or conformal Killing vectors on manifolds with negative Ricci tensor.
     \item On Ricci-flat manifolds all  conformal Killing vectors  are covariantly constant, hence Killing.
         \item
  \label{pP26II23.3a}
          {\rm $\ker  (\zdivtwo\circ\, C) = \CKV$.}
  \item
  \label{pP26II23.3}
   {\rm$\im (\zdivtwo\circ\, C) = \CKVp$.}
  \item For any vector field $\xi\in H^1(\dmanif)$ we have    $  C(\xi) ^{[ \TTt]} =0$.
   \end{enumerate}
  \end{Proposition}

  \proof
1.\ and 2.:
Taking the divergence of the conformal Killing equation and commuting derivatives
leads to
 \ptcheck{20VI by
 wan}
 \begin{equation}\label{30X22.CKV3}
    \dnabla ^A \dnabla _A \xi_B   + \zR_{BC} \xi^C
     + \frac{ \ddim-2 }{\ddim}\dnabla _B \dnabla^A \xi_A
    = 0
    \,.
 \end{equation}
  Multiplying by $\xi^B$ and integrating over $\dmanif$ one finds
 \begin{equation}\label{30X22.CKV4}
   \int_{\dmanif}  
   \big(
    | \dnabla   \xi|^2   + \frac{ \ddim-2 }{\ddim} |  \zdivone \xi |^2  -  \zR_{BC} \xi^B\xi^C
     \big)
   \, d\mu_{\dmetric}
    =  0
    \,.
 \end{equation}
 If $\zR_{BC} \le 0$ we find that $\xi$ is covariantly constant, vanishing if  $\zR_{BC} < 0$.

 \medskip

 3.
 %
Let $\xi_A$ be in the kernel of $\zdivtwo\circ\, C$, we have
 \begin{eqnarray}
  0
     &=&  \int_{\dmanif} \eta^A  \dnabla _B(\dnabla _A \xi_B +  \dnabla _B \xi_A
        -  \frac 2 d  \dnabla ^C \xi_C
    \dmetric _{AB} ) d\mu_{\dmetric}
    =2 \int_{\dmanif} \eta^A \dnabla ^B
    \big( \TS (\dnabla _A \xi_B  )
    \big)
    d\mu_{\dmetric}
     \nonumber
\\
     &=&
      - 2 \int_{\dmanif} \dnabla^B \eta^A \TS (\dnabla _A \xi_B  ) d\mu_{\dmetric}
      =
      - 2 \int_{\dmanif} \TS (\dnabla^B \eta^A) \TS (\dnabla _A \xi_B  ) d\mu_{\dmetric}
      \,.
      \label{21XII23.4}
 \end{eqnarray}
 Setting  $\xi=\eta$ we conclude that $\eta$ is a conformal Killing vector.

 4.
The operator $L:=\zdivtwo \circ\, C$ is elliptic and formally self-adjoint (since \eqref{21XII23.4} is symmetric with respect to 
the interchange of $\eta$ and $\xi$),  with
 $\ker L = \CKV$, and
 %
 we conclude by Proposition~\ref{P26X23.1}.

 \medskip

 5. The field  $  C(\xi) ^{[ \TTt]} $ is obtained by $L^2$-projecting  $  C(\xi)    $ on $\TTt$.
 So let $h\in L^2$ be a weak solution of the equation $\zdivtwo h=0$, i.e., for all $\xi \in H^1(\dmanif)$ we have
 \begin{equation}\label{14VII23.1}
   \int_{\dmanif} \zspaceD^A\xi^B h_{AB} = 0
   \,,
 \end{equation}
 which is precisely the statement that the $L^2$-projection  of $C(\xi ) $ on $\TTt$ vanishes.
%
 \qed

\index{kernel!$C$}%
\begin{Proposition}
 \label{P12XI22.1}
 The conformal Killing operator on $\ddim$-dimensional compact manifolds, $d\ge 3$, has

 \begin{enumerate}
   \item a $ \ddim (\ddim+1)/2$ dimensional kernel  on $S^d$ with the canonical metric;
   \item a $\ddim$-dimensional kernel  on flat manifolds;
   \item no kernel  on negatively Ricci-curved
    manifolds.
 \end{enumerate}
 \end{Proposition}

 \proof
 1. The conformal group of a $\ddim$-dimensional sphere is the same as the Lorentz group in $\ddim+2$ dimensions.

 The statements about the kernel in points 2.\ and 3.\ follow from   Proposition \ref{P30X22.2}.
 \qedskip

We note the following properties of conformal Killing vectors on Einstein manifolds, which will be needed in what follows:

 \begin{Lemma}\label{L21X23.1}
  Let $\chi^A\in \CKV$. If $(\dmanif,\ringh)$ is Einstein,
   then  $  \zspaceD_B \zspaceD_A\chi^A\in \CKV$.
\end{Lemma}
\proof
 We have
 \ptcheck{9X23}
\begin{align}
    \TS[\zspaceD_A\chi_B] = 0
    &\implies
    \frac{1}{d}\ringh_{AB} \zspaceD^C\chi_C = \frac12(\zspaceD_A\chi_B + \zspaceD_B\chi_A)
    \nonumber
    \\
   &\implies
   -\frac{d-2}{ d} \zspaceD_B\zspaceD^C\chi_C =
    \big(\TSzlap +(d-1)\myGauss
     \big)\chi_B \,,
\end{align}
and hence
 \ptcheck{12X23}
\begin{equation}
  -\frac{d-2}{d}  \TS[\zspaceD_A \zspaceD_B\zspaceD^C\chi_C] =  \TS[\zspaceD_A(\TSzlap +(d-1)\myGauss)\chi_B] = \TS[\zspaceD_A\TSzlap\chi_B]= 0
  \,,
   \label{12XII23.22}
\end{equation}
where the last equality follows from the commutation relation \peqref{28VIII23.f3}, using \eqref{16V22.1bxc} for the vanishing of the $\tric$-term there.
\qedskip

 \begin{Lemma}\label{L21X23.2}
  Let $(\dmanif,\ringh)$ be Einstein with scalar curvature equal to  $d(d-1)\myGauss>0$. Then
$$ (\TSzlap-(d-1)\myGauss) \psi \in \CKV
 \qquad
 \Longleftrightarrow
 \qquad
    \psi \in \CKV
    \,.
$$
\end{Lemma}

\proof 
For any vector field $\psi$ we have 
\begin{align*}
   &
    \zdivtwo \circ\, C \circ(\TSzlap-(d-1)\myGauss) \psi = (\TSzlap-(d-1)\myGauss) \circ \zdivtwo \circ\, C (\psi )
    \,.
\end{align*}
Keeping in mind that $\CKV$ coincides with the kernel of $\zdivtwo\circ\, C$, the result follows from the fact that   operators $\TSzlap-\lambda$ with $\lambda>0$ are  isomorphisms.
\qedskip

\subsection{Decompositions of two-tensors}
\label{ss8IX23.1}

A symmetric traceless $2$-covariant tensor field $h$ on a compact
boundaryless $d$-dimensional Riemannian manifold $(\dmanif,\ringh)$
 can be uniquely split into a ``scalar'', a ``vector'', and a ``tensor'' part according to
(cf.\ e.g.~\cite{KodamaIshibashiMaster})%
\index{S@$S$}\index{V@$V$}%
\begin{equation}
\label{7VI17.101PK}
h = h^{\red[S]} + h^{\red[V]} + \red{h^{\red[TT]}}
\,,
\end{equation}
where $\red{h^{\red[TT]}}$ is a $\TTt$-tensor, $h^{\red[V]}$ is the Killing operator acting on a divergence-free vector, and $h^{\red[S]}$ is the trace-free part of the Hessian of a function.

This can be compared to the following (unique) version of the York decomposition: for every symmetric two-covariant tensor
 $h\in H^k(\dmanif)$, $k\ge 1$, we have
 %
%
\begin{equation}
\label{7VI17.101PKYc}
h_{AB} = \phi \ringh_{AB} + \frac12\big(
 \zspaceD _A W_B + \zspaceD _B W_A - \frac 2 d  \zspaceD ^k W_K \ringh_{AB}
   \big)
      + \red{h_{AB}^{[\TTt]}}
\,,
\end{equation}
where $\phi\in H^k(\dmanif)$ is obtained  from  the algebraic trace-part of $h$ (and is zero when $h$ is traceless, as almost always assumed elsewhere in this paper); then $ \red{h^{[\TTt]}}$ is obtained by solving the conformal-vector-Laplacian equation for a unique vector field $W\in H^{k+1}(\dmanif)\cap \CKVp$.
Letting $\varphi\in H^{k+2}(\dmanif)\cap \{1\}^\perp$ be the unique  solution of
\begin{equation}\label{10VI23.1}
  \TSzlap  \varphi = \zspaceD ^k W_K
  \,,
\end{equation}
we set
\begin{equation}\label{10VI23.2}
  V_A = W_A - \zspaceD_A \varphi
  \,.
\end{equation}
\index{S@$S$}\index{V@$V$}%
This allows us to rewrite \eqref{7VI17.101PKYc} as
\begin{equation}
\label{10VI23.3}
h_{AB} =
 \big( \phi
 \hspace{-0.8cm}
 \underbrace{
  - \frac 1 d   \TSzlap  \varphi
 \big) \ringh_{AB}
+
  \zspaceD _A  \zspaceD _B \varphi
 }_{h^{\red[S]} \ \text{when $h$ is traceless, so that $\phi$ vanishes}}
 \hspace{-0.8cm}
 +
 \underbrace{
 	\frac12
 	\big(
 \zspaceD _A V_B + \zspaceD _B V_A
    \big)
 }_{h^{\red[V]}}
      +
      \underbrace{
       \red{h_{AB}^{[\TTt]}}
 }_{h^{\red[TT]}}
\,,
\end{equation}
which justifies \eqref{7VI17.101PK}.

\subsection{Scalar-vector decomposition of vector fields}
\label{app16X23.1}

Let
$\xi_A\in  \MHk$, $k\ge 1$. We  define the vector subspaces $S,V \subset  \MHk$ as
 $$
 S = \{ \xi_A \in \MHk: \xi_A = \zspaceD_A \phi\,,
 \phi \in H^{k+1}(\dmanif) \}
 \,,\quad
 V = \{\xi_A\in\MHk: \zspaceD^A\xi_A = 0  \} \,. $$
 \index{$S$}\index{$V$}%
When $\dmanif$ is compact and boundaryless the spaces $S$ and $V$ are $L^2$-orthogonal, and
any vector field $\xi$ can be decomposed into its ``scalar'' and ``vector'' parts, denoted as
 \begin{align}
    \xi = \xi^{[S]} + \xi^{[V]} \,,
 \end{align}
with $\xi^{[S]}$ and $\xi^{[V]}$  given by
\begin{equation}\label{4X23.1}
  \TSzlap  \phi = \zspaceD_A \xi^A
  \,,
  \quad
  \xi^{[S]}_A  = \zspaceD_A \phi
  \,,
  \quad
  \xi^{[V]}_A  = \xi_A - \zspaceD_A \phi
  \,.
\end{equation}
It will be useful for the following to note that on a Einstein manifold with scalar curvature equal to $(n-1)(n-2)\myGauss$, we have
\begin{align}
    \zspaceD_A\zspaceD^B\xi_B = \zspaceD_A\zspaceD^B\xi^{[S]}_B = \zspaceD_A\TSzlap\phi = (\TSzlap - (n-2)\myGauss)\zspaceD_A\phi = (\TSzlap - (n-2)\myGauss)\xi^{[S]}_A\,.
    \label{25XII23.1}
\end{align}
In the main text we use the following fact:

\begin{Proposition}
 \label{P21X23.1}
 Let  $(\dmanif,\ringh)$ be Einstein with scalar curvature equal to $(n-1)(n-2)\myGauss$. If $\xi\in \CKVp \cap \MMHk$, $k\ge 1$,
 then
 $$
  \xi_A^{[S]} \,,\  \xi_A^{[V]}\in \CKVp\cap \MMHk
  \,.
 $$
\end{Proposition}

\noindent {\sc Proof:}
%
Let $\xi\in \CKVp$. By definition, $\xi^{[S]}_A $ will   be in $\CKVp$ if and only if for all conformal Killing $\chi^A$ vectors we have
\begin{equation}\label{4X23.2}
  0 =  \int_{\dmanif} \chi^A \xi_A^{[S]}
  \equiv
   \int_{\dmanif} \chi^A \zspaceD_A \phi
    \,\sm =  - \int_{\dmanif}  \phi \zspaceD_A\chi^A
    \,\sm
    \,.
\end{equation}
This will certainly be the case when either (see Proposition~\ref{P30X22.2}): a) $\CKV=\{0\}$ (which is the case on negatively Ricci-curved manifolds), or b) $\CKV =\KV$ (which is the case on  Ricci-flat manifolds). So, to show that $\xi^{[S]}\in \CKVp$ it remains to consider positively curved manifolds.

Now, when $\myGauss>0$ the operator $(\TSzlap-(n-2)\myGauss)$ acting on tensors of any rank is an isomorphism. So, given a conformal Killing vector field $\chi$  there exists  vector field $\psi$ such that
$$
  \chi =  (\TSzlap-(n-2)\myGauss)\psi
  \,.
$$
By Lemma~\ref{L21X23.2} we have  $\psi\in \CKV$.
  Then
\begin{align}
\label{5X23.1}
  \int_{\dmanif} \chi^A \xi_A^{[S]}
   \,\sm
  &=
  \int_{\dmanif}
  \big( (\TSzlap-(n-2)\myGauss) \psi^A
  \big)
  \xi_A^{[S]}
   \,\sm
\\
  &=
  \int_{\dmanif} \psi^A
   \big(
    \TSzlap-(n-2)\myGauss
    \big)\xi_A^{[S]}
   \,\sm
   =  \int_{\dmanif} \psi^A \zspaceD_A \zspaceD_B\xi^B \,\sm
  \nonumber
  \\
    & =  \int_{\dmanif}  \underbrace{\zspaceD_B \zspaceD_A\psi^A}_{\in \CKV} \xi^B
    \,\sm = 0
    \,,
\end{align}
where we used \eqref{25XII23.1} and Lemma~\ref{L21X23.1}.
 Hence $ \xi_A^{[S]}\in \CKVp$ in all cases.

 Since $\CKVp$ is a vector space it holds that
$$
 \xi_A^{[V]} = \xi_A - \xi_A^{[S]} \in \CKVp
 \,,
$$
which finishes the proof.
\qed
%

\subsection{${\chi}\ofP $}
 \label{ss18VIII23.1}
 \ptcheck{29VIII; the whole subsection}
\index{chi@$ {\chi}\ofP  $}%

In this appendix we analyse the nature of the operators $\overadd{i}{\chi}\ofP $ appearing in \eqref{6III23.w1a}.
We work on an $(n-1)$-dimensional closed manifold, $n\ge 4$.
We begin with the following lemma:

\begin{Lemma}
   \label{L18VIII23.1}
   Let $a,c \in \R$. The operator
   $$ \operatorname{L}_{a,c} = a P + \TSzlap +2 \tric + c$$
   acting on symmetric, $\ringh$-trace-free $2$-tensors is
 formally self-adjoint, and elliptic if
   $$ a \neq -2, \, \frac{1-n}{n-2} \,.$$
\end{Lemma}

\proof
We start by noting that the operators $P$, $\TSzlap$ and $\tric$ are all formally self-adjoint,
and thus so is $\operatorname{L}_{a,c}$.

Next,
let $k\neq 0$ and let $\sigma(k)_{AB}$ be the symbol of the principal part of $\operatorname{L}_{a,c}$, with kernel determined by the equation
\begin{equation}
\label{18VIII23.1}
    \sigma(k)_{AB} \equiv k^C k_C h_{AB} + \frac{a}{2} (k_A k^C h_{CB} + k_B k^C h_{CA} - \frac{2}{n-1} k^C k^D h_{CD} \ringh_{AB} ) =0 \,.
\end{equation}
 Contracting with $k^A k^B$ gives
 \begin{align}
   0 = k^A k^B \sigma(k)_{AB} &= (|k|^2 + 2 \frac{a}{2} |k|^2- \frac{a}{n-1} |k|^2)  k^A k^B h_{AB}
    \nonumber
    \\
    & = \left( 1 + a\left(1 - \frac{1}{n-1}\right)\right) |k|^2 k^A k^B h_{AB}
    \,
 \end{align}
so that
 \begin{equation}
     \text{either (a) }   k^A k^B h_{AB} = 0 \,, \ \text{or  (b) } a = \frac{1-n}{n-2}\,.
 \end{equation}
 Substituting case (a) into \eqref{18VIII23.1} gives
\begin{equation}
\label{18VIII23.2}
     |k|^2 h_{AB} + \frac{a}{2} (k_A k^C h_{CB} + k_B k^C h_{CA} ) =0 \,.
\end{equation}
 Contracting with $k^B$ gives
\begin{align}
\label{18VIII23.3}
     0& = |k|^2 k^B h_{AB} + \frac{a}{2} (k_A \underbrace{k^B k^C h_{CB}}_{=0} + |k|^2 k^C h_{CA} )
     \\
     &= \bigg( \frac{a}{2} +1 \bigg)|k|^2 k^C h_{CA}
     \,,
\end{align}
 and hence
 \begin{equation}
     \text{either ($\alpha$) }   a = -2 \,, \ \text{or ($\beta$) } k^C h_{CA} = 0\,.
 \end{equation}
 Finally, substituting $(\beta)$ into \eqref{18VIII23.2} gives
 \begin{equation}
     |k|^2 h_{AB} = 0
 \end{equation}
 which completes the proof.
 \qedskip

Using the last lemma one easily concludes that the operators%
\index{K@$\ck{k}{\ofPnoP}$}%
\begin{align}
    \ck{k}{\ofPnoP} &:= -\frac{1}{7 -  n + 2 k} \bigg[\frac{2 (n - 1) P}{(3 + k) (3 -  n + k) }
    +  2 \tric + \TSzlap -   (n - 4 -  k) (2 + k) \myGauss \bigg]
\end{align}
if $\frac{\Z}{2} \ni k\not\in\{- 3,n-3,(n-7)/2\}$,   and%
\index{K@$\zck{-3}{\ofPnoP}$}%
\begin{equation}
	\zck{-3}{\ofPnoP} = \frac{1}{n-1} \bigg[\frac{(n (1-n) -2 ) P}{n} + 2 \tric +\TSzlap + (n-1) \myGauss \bigg] \,,
\end{equation}
are elliptic. The remaining cases of interest, namely $\ck{-1}{\ofPnoP} $ and $\ck{-2}{\ofPnoP} $, will be dealt with in the next section.

To continue, we have from \eqref{6III23.w2}:
\begin{align}
    \overadd{i}{\chi}\ofP  &= \prod_{j=1}^{i} \ck{-(j+3)}{\ofPnoP}\,.
        \label{18VIII23.w1}
\end{align}
Since the composition of elliptic operators is elliptic, we have shown:

 \begin{Proposition}
 \label{P18VIII23.1}%
\index{ellipticity!psi@$\overadd{i}{\psi}\ofP$}%
\index{psi@$\overadd{i}{\psi}\ofP$!kernel}%
\index{psi@$\overadd{i}{\psi}\ofP$!ellipticity}%
 The operators $\overadd{i}{\chi}\ofP $ are
 formally self-adjoint   and elliptic.
 \end{Proposition}

The proof of the following Lemma is similar to that of Lemma \ref{L18VIII23.2}, we leave the details to the reader.

\begin{Lemma}
   \label{L18VIII23.2}
   Let $a,c \in \R$. The operator
   {\rm
   $$ \tilde{L}_{a,c} = a \, \zdivtwo\circ\, C + \TSzlap + c$$
   } is
   formally self-adjoint, and elliptic if
   \begin{equation}
    a \neq -2, \, \frac{1-n}{n-2} \,.
    \tag*{$\Box$}
    \end{equation}
\end{Lemma}

 Recall (cf.\ below \peqref{9IX23.w1}) that the operators%
 \index{chi@$\overadd{i}{\chi}\ofP$!$\overadd{i}{\tilde\chi}\ofDC$}%
 $$\overadd{p}{\tilde\chi}\ofDC$$
 are defined by replacing all appearances of the operators $P$ and $2\tric$ in $\overadd{p}{\chi}\ofP$ by $\zdivtwo \circ\, C$ and $(n-2)\myGauss$ respectively.
It then readily follows
from  Lemma \ref{L18VIII23.2} that:

 \begin{Proposition}
 \label{P18VIII23.2}
 The operators {\rm $\overadd{p}{\tilde\chi}\ofDC$} are
 formally self-adjoint   and elliptic.
 \end{Proposition}

\subsection{$\psi\ofP$}
\label{ss24IX23.1}

	Let us now consider the operators involved in gluing the fields  $\overadd{i}{q}_{AB}$. Recall that the relevant operators are
	\begin{align}
	\overadd{i}{\psi}\ofP  & = \prod_{j=2}^{i} \ck{\tfrac{n-5-2j}{2}}{\ofPnoP}\overadd{1}{\psi}\ofP \, \,,\\
	\overadd{1}{\psi}\ofP  &= - \bigg[\frac{4}{n+1} P -  \tric - \frac{1}{2}\TSzlap
	+\frac{(n-3)(n-1)\myGauss}{8}\bigg]\,,
	\end{align}
	for all $i\geq1$.
	
By Lemma~\ref{L18VIII23.1} the operator $\overadd{1}{\psi}\ofP$
is elliptic for   $n\ne 3,5$. However, we focus on $n>3$ in this work; the case $n=5$ will be addressed below.

Now, a straightforward calculation shows that the operators $\ck{-1}{\ofPnoP}$ and $\ck{-2}{\ofPnoP}$ fail to be elliptic only for odd $n>3$. All other $\ck{k}{\ofPnoP}$ operators appearing in the $\overadd{i}{\psi}$'s are elliptic. 
 \ptclater{finn, reword using your appendix?}
This implies in particular, that:

\index{psi@$\overadd{i}{\psi}\ofP$}
\index{psi@$\overadd{i}{\psi}\ofP$!ellipticity}
	\begin{Proposition}
		\label{P22VIII23.1}
		For convenient pairs $(n,k)$ the operators $\overadd{i}{\psi}\ofP$ are elliptic  and formally self-adjoint. 
	\end{Proposition}

Let us consider now  an inconvenient pair $(n,k)$, where we must be careful about which fields are involved in the gluing.
For this we make use of the York decomposition for $h_{AB}$ (cf.\ \eqref{7VI17.101PKYc}).

First, the restriction of $\operatorname{L}_{a,c}$ acting $\TTt$ is elliptic, since $Ph^{[\TTt]}=0$ and $\TSzlap$ is elliptic.

 Next, the non-elliptic $\mathcal{K}$ operators can be seen to appear (when $n$ is odd) in all the operators
 \ptcheck{19X}
\begin{equation}\label{17X23f.1}
\overadd{\frac{n-1}{2}+j}{\psi}\ofP = \overadd{j-1}{\chi}\ofP \red{\zck{-3}{\ofPnoP}}	\underbrace{\ck{-2}{\ofPnoP} \ck{-1}{\ofPnoP}\overadd{\frac{n-5}{2}}{\psi}\ofP}_{\equiv \overadd{\frac{n-1}{2}}{\psi}\ofP}
\end{equation}
with $j\geq1$. This expression follows from the recursion relations \eqref{28VIII23f.6}-\eqref{4IV23.2}.

We will use the underbraced term in \eqref{17X23f.1} to show the following:

\begin{Proposition} \label{P9X23.1}
We have, for $n\ge 5$ odd and $j\ge0$,
\begin{eqnarray}
    \label{16X.f1}
     & \overadd{\frac{n-3}{2}+j}{\psi}\ofP  \, h^{\red[S]} \equiv 0 \,, \quad  \overadd{\frac{n-1}{2}+j}{\psi}\ofP \, h^{\red[V]} \equiv 0 \,, 
     &
     \\
      &
\label{12VI.1}
\overadd{\frac{n-1}{2}+j}{\psi}\ofP  \circ\, C(W) \equiv 0 \,,
 &
\\
 &
\label{5X23f.1}
\overadd{\frac{n-1}{2}+j}{\psi}\ofP \circ P  \, (h) \equiv 0 \,.
 &
\end{eqnarray}
\end{Proposition}

\proof
%
The starting point is to consider the operator $P$ acting on  $\TTtp$-tensors. That is,
\begin{align}
P C(W)_{AB} &= \frac{1}{2} \TS[\zspaceD_A\zspaceD^D(\zspaceD_D W_B  + \zspaceD_B W_D - \frac{2}{n-1} \ringh_{BD} \zspaceD^C W_C)]
\nonumber
\\
&=
\frac{1}{2} \TS[\zspaceD_A\TSzlap W_B  +  \zspaceD_A\zspaceD^D \zspaceD_B W_D
- \frac{2}{n-1} \zspaceD_A\zspaceD_B\zspaceD^C W_C]
\nonumber
\\
&=
\frac{1}{2} \TS[\zspaceD_A\TSzlap W_B
+ \frac{n-3}{n-1} \zspaceD_A\zspaceD_B\zspaceD^C W_C
-\zspaceD_A \zR^E{}_D{}^D{}_B W_E]
\nonumber
\\
&=
\frac{1}{2} \TS[\TSzlap \zspaceD_A W_B- (n-2)\myGauss \zspaceD_A W_B - 2\zR^D{}_B{}_A{}^C \zspaceD_C W_D
\nonumber
\\
& \qquad  \quad
+ \frac{n-3}{n-1} \zspaceD_A\zspaceD_B\zspaceD^C W_C
+ (n-2) \myGauss \zspaceD_A W_B]
\nonumber
\\
&=
\frac{1}{2} \TS
\big[\TSzlap \zspaceD_A W_B]
+  \tric(\zspaceD W )_{AB}
+ \frac{n-3}{2(n-1)} \TS[\zspaceD_A\zspaceD_B\zspaceD^C W_C
\big]
 \,;
\label{26III23.2}
\end{align}
recall that  $\tric$ has been defined in \eqref{16V22.1b}.
	Thus
	\begin{equation}
	\label{22VIII23.f1}
	P C(W)_{AB}
	=
	\left(\frac12 \TSzlap + \tric\right) C(W)_{AB}
	+ \frac{n-3}{2(n-1)} \TS[\zspaceD_A\zspaceD_B\zspaceD^CW_C]\,.
	\end{equation}

	Next we make use of the scalar-vector-tensor decomposition \eqref{10VI23.3} to write this as
	\begin{eqnarray}
	P C(W)_{AB}
 & = &
  P(h^{\red[S]} +h^{\red[V]})
 \nonumber
\\
 & 	=
	&
\left(\frac12 \TSzlap + \tric\right)\left(\TS[\zspaceD_{A}\zspaceD_{B}\varphi]+\zspaceD_{(A}V_{B)}\right)
	+ \frac{n-3}{2(n-1)} \TS[\zspaceD_A\zspaceD_B\TSzlap\varphi]\,,
\nonumber
\\
&&
	\end{eqnarray}	
	where we have also used \eqref{10VI23.1}. This means that we can immediately write for the vector part
	\begin{equation}
		P(h^{\red[V]})_{AB}=\frac12(\TSzlap +2\tric)h^{\red[V]}_{AB}\,.
	\end{equation}
	Second, focusing on the scalar part. We commute the derivatives (using \eqref{19X23.1} in the second line, \eqref{28VIII23.f3} in the third line, and \eqref{16V22.1bxc} in the fourth one)
\ptcheck{19X}
	\begin{align}
		P(h^{\red[S]})
		&=\left(\frac12 \TSzlap + \tric\right)
 \TS[\zspaceD_{A}\zspaceD_{B}\varphi]
		+ \frac{n-3}{2(n-1)} \TS[\zspaceD_A\zspaceD_B\TSzlap\varphi]
		\nonumber\\
		&=\left(\frac12 \TSzlap + \tric\right)h^{\red[S]}+ \frac{n-3}{2(n-1)} \TS[\zspaceD_A(\TSzlap-(n-2)\myGauss)\zspaceD_B\varphi]
			\nonumber\\
		&=\left(\frac12 \TSzlap + \tric\right)h^{\red[S]}+ \frac{n-3}{2(n-1)} \TS[(\TSzlap+2[\tric -(n-2)\epsilon])\zspaceD_A\zspaceD_B\varphi]
			\nonumber\\
		&=\left(\frac12 \TSzlap + \tric\right)h^{\red[S]}+\frac{n-3}{2(n-1)} (\TSzlap+2[\tric -(n-2)\epsilon])h^{\red[S]}
		\nonumber\\
		&=\frac{n-2}{n-1}\left(\TSzlap +2\tric- (n-3)\myGauss\right)h^{\red[S]}_{AB}\,.
	\end{align}
	Putting these results together we have
	\begin{align}
	\label{28VIII23.f4}
	P(h^{\red[S]})_{AB}&=\frac{n-2}{n-1}\left(\TSzlap +2\tric- (n-3)\myGauss\right)h^{\red[S]}_{AB}\,,\\
	\label{28VIII23.f5}
	P(h^{\red[V]})_{AB}&=\frac12(\TSzlap +2\tric)h^{\red[V]}_{AB}\,.
	\end{align}
We are ready now to prove the first part of  \eqref{16X.f1}:
for $j\geq0$
\begin{equation}
\overadd{\frac{n-3}{2}+j}{\psi}\ofP  \, h^{\red[S]}=0\,.
\end{equation}
This follows from substituting \eqref{28VIII23.f4} into the definition of $\overadd{\frac{n-3}{2}}{\psi}\ofP  \, h^{\red[S]}$. To see this in the case $n=5$, using \eqref{5III23.4a}  we find
\ptcheck{4IX23}
\begin{align}\label{29VIII23.f3}
\overadd{\frac{n-3}{2}}{\psi}\ofP  \, h^{\red[S]}&=\overadd{1}{\psi}\ofP  \, h^{\red[S]}=\frac12\left(-\frac{4}{3}P +( \TSzlap +2\tric) -2\myGauss\right) h^{\red[S]}\,\nonumber
\\
&=\frac12\left( -\frac{4}{3}
\times \frac{5-2}{5-1}(\TSzlap +2\tric -(5-3)\myGauss) +( \TSzlap +2\tric) -2\myGauss \right) h^{\red[S]}\nonumber\\
&=0\,.
\end{align}
Next for $n> 5$,
from \eqref{28VIII23f.6} we have
\begin{equation}
\overadd{\frac{n-3}{2}}{\psi}\ofP =\overadd{\frac{n-5}{2}}{\psi}\ofP \, \ck{-1}{\ofPnoP}\,,
\end{equation}
and from \eqref{3III23.5a}
\ptcheck{4IX23}
\begin{align}\label{29VIII23.f4}
\ck{-1}{\ofPnoP}h^{\red[S]}
=&
\frac{1}{n-5}\left(-\frac{n-1}{n-2} P +\TSzlap +2\tric -(n-3)\epsilon\right)h^{\red[S]}\nonumber\\
=&
\frac{1}{n-5}\Bigg[-\left(\frac{n-1}{n-2}\right)\left(\frac{n-2}{n-1}
\big(
\TSzlap +2\tric- (n-3)\myGauss
\big)
\right)
\nonumber\\
&
+\TSzlap +2\tric -(n-3)\epsilon\Bigg]h^{\red[S]}=0\,.
\end{align}

A similar calculation using \eqref{28VIII23.f5} shows that this does \emph{not} hold when acting on $h^{\red[V]}$. Instead we have
\ptcheck{4IX23}
\begin{equation}\label{16X23.f1}
\ck{-1}{\ofPnoP} h^{\red[V]}=
\begin{cases}
\frac{1}{6}\left( \TSzlap +2
    (\tric -3\epsilon) \right)h^{\red[V]}\quad \text{if } n=5\,,
     \\
        \frac{(n-3)}{2(n-2)(n-5)}\left( \TSzlap +2
         (\tric -(n-2)\epsilon)
           \right)h^{\red[V]}\quad \text{if } n>5\,.
\end{cases}
\end{equation}
It follows that the principal symbol of the restriction of $\ck{-1}{\ofPnoP}$ on $\{h:\, h=h^{[V]}\}$ comes only from $\TSzlap$ and is thus elliptic.
	
We move on now to the operator
	\begin{align}
	\overadd{\frac{n-1}{2}}{\psi}\ofP  =
	\overadd{\frac{n-3}{2}}{\psi}\ofP \, \ck{-2}{\ofPnoP}\,.
	\end{align}
	Using the definitions of the operators  (see \eqref{3III23.5a}) and substituting for $Ph^{\red[V]}$ with \eqref{28VIII23.f5}, we have
	\ptcheck{4IX23}
	\begin{equation}\label{29VIII23.f5}
	\ck{-2}{\ofPnoP} h^{\red[V]}
	\propto
	\left(P-\frac12(\TSzlap+2\tric)\right)h^{\red[V]}=0\,,
	\end{equation}
	and so we obtain the second   part of  \eqref{16X.f1}:  for $j\geq0$,
	\begin{equation}
	\overadd{\frac{n-1}{2}+j}{\psi}\ofP  \, h^{\red[V]}=0\,.
	\end{equation}
Now, writing $C(W)$ as $h^{[S]} + h^{[V]}$, \eqref{12VI.1} follows from  \eqref{16X.f1}: 
	\ptcheck{4IX23}
	\begin{equation}
	\overadd{\frac{n-1}{2}+j}{\psi}\ofP  \big(C(W) \big)
 =0\,.
 \label{6XI23.1}
	\end{equation}
Finally, to obtain \eqref{5X23f.1} 
we write
\begin{eqnarray}
\overadd{\frac{n-1}{2}+j}{\psi}\ofP \circ P  \, (h)   &=& \overadd{\frac{n-1}{2}+j}{\psi}\ofP \circ P  
 \, \big(h^{[\TTt]} +C(W)\big)
  \nonumber
 \\
    &=&    \overadd{\frac{n-1}{2}+j}{\psi}\ofP \circ P
 \, \big( C(W)\big)
  \nonumber
=P \circ   \overadd{\frac{n-1}{2}+j}{\psi}\ofP
 \, \big( C(W)\big)
 \\ 
 & = 
 & 0
\end{eqnarray}
by \eqref{6XI23.1}, 
since $\overadd{i}{\psi}\ofP$ commutes with $P$ 
(see \peqref{20VI23.12} below) 
 and $P(h^{[\TTt]})=0$.
 
\qedskip

Summarising, we have the following:
\index{psi@$\overadd{i}{\psi}\ofP$}
\index{psi@$\overadd{i}{\psi}\ofP$!ellipticity}
	\begin{Proposition}
		\label{P17X23.1}
		For  inconvenient pairs $(n,k)$  the operators $\overadd{i}{\psi}\ofP$
		\begin{enumerate}
			\item acting on $h^{\red[S]}$ are elliptic for $i<\frac{n-3}{2}$ and vanish for $i\geq \frac{n-3}{2}$;
			
			\item acting on $h^{\red[V]}$:
				 \begin{enumerate}
						\item  are elliptic for $i\leq \frac{n-3}{2}$,
						\item  vanish when $i\geq \frac{n-1}{2}$;
				\end{enumerate}
			
			\item acting on $h^{[\TTt]}$ are elliptic.
			
		\end{enumerate}
	\end{Proposition} 

 We end this section with a remark  on the kernel of the operator $\ck{-1}{\ofPnoP}$. First,  $\ck{-1}{\ofPnoP}$ vanishes identically  on $S$: $\ck{-1}{\ofPnoP}h^{[S]}=0$.
 Next, from  \eqref{16X23.f1} and the commutation relation \eqref{28VIII23.f3} we have,
for vector fields with vanishing divergence,%
\index{kernel!$\ck{-1}{\ofPnoP}$}
\begin{align}
		\ck{-1}{\ofPnoP} h^{\red[V]}\propto
   & \left( \TSzlap +2
   (\tric -(n-2)\epsilon)
     \right)\zspaceD_{(A}V_{B)}
      =\zspaceD_{(A}\left( \TSzlap -(n-2)\epsilon \right)V_{B)}\,.
\end{align}
Now, the operator inside of $\zspaceD_{A}$ on the right-hand side of the last equality is just the (negative of) the Hodge Laplacian.
Therefore, any $h^{\red[V]}_{A B}=\zspaceD_{(A} V_{B)}$ constructed from a harmonic vector lies in the kernel of  $\ck{-1}{\ofPnoP}$. Moreover, expanding any $V_A$, which generates an $h^{\red[V]}_{A B}$ lying in the kernel of $\ck{-1}{\ofPnoP}$, in an orthonormal eigenbasis of the Hodge Laplacian implies either $V$ is a Killing vector (in which case $h^{\red[V]}_{A B}=0$) or $V$ is harmonic. Together, this implies
\begin{equation}
\ker \left( \ck{-1}{\ofPnoP}\big|_{h^{\red[V]}} \right)=\left\{h^{\red[V]}=\zspaceD_{(A} V_{B)} \,\,: \,\, \zspaceD_{[A}V_{B]}=0=\zspaceD^A V_A\right\}
\,.
\end{equation}

\subsection{The gauge operators $\hLop_n$ and $\Lop$}
\label{ss4IX23.2}
In this section we analyse the gauge operators involved in the gauge transformations of the radial charges associated to $\overadd{\frac{n-3}{2}+j}{q}_{AB}$, $j\geq 0$, with $n$ odd. In particular, we wish to show that the gauge-invariant charges in the $m=0=\alpha$ case (cf.\ Table \ref{T11III23.1}) associated to these fields are smooth and live in a finite dimensional space.

We  will make use of  the following commutation relation:
\begin{align}
     \zspaceD^A \big(a
      \, \zdivtwo \circ\, C + \TSzlap\big) \xi_A
    =
    \Big(
    (\frac{a(n-2)}{n-1}+1)\TSzlap + (a+1)(n-2)\myGauss
    \Big) \zspaceD^A\xi_A\,.
    \label{10XI23.1}
\end{align}

We begin the analysis with $j=0$, with the relevant operator being $\hLop_n$ of
\eqref{24IV23.1xs}.
We wish to show:

\begin{proposition}
 \label{P16X23.1}$\ker (\hLop_n^\dagger) \cap \ker(\overadd{\frac{n-3}{2}}{\psi}\ofP)$ is smooth and finite dimensional.
 \end{proposition}

\proof
From the expression \eqref{24IV23.1xs} of $\hLop_n$, the principal symbol of $\hLop_n$ arises from the operator
 $$  c_n \overadd{(n-5)/2}{\psi}\ofP \, P\circ\, C \circ \zspaceD
 $$
acting on \red{functions},
 with a constant $c_n$ depending upon the dimension $n$.

 We first analyse the nature of the operator
 \begin{equation}
     \zdivone\circ\zdivtwo\circ\overadd{(n-5)/2}{\psi}\ofP \circ P\circ\, C \circ \zspaceD
     \,.
     \label{10XI23.2}
 \end{equation}
For this let us consider the kernel of the symbol of the self-adjoint operator $\zdivone\circ\zdivtwo\circ P\circ\, C \circ \zspaceD $: for $k\neq 0$, we have
 \ptcheck{16X23}
\begin{align}
    0 = \sigma(k)(\xi^u) &:=  k^A k^B \TS
     \big[k_A k^C\TS[k_C k_B \xi^u]
     \big]
    \nonumber
    \\
    & =
    \frac{ (n-2)^2}{(n-1)^2}|k|^6 \xi^u  \iff \xi^u = 0\,.
\end{align}
Thus, for $n\neq 1,2$, the operator
$\zdivone\circ\zdivtwo\circ P\circ\, C \circ \zspaceD$
is elliptic.

Next, we write
\index{psi@$\overadd{(n-5)/2}{\hat\psi}$}
\begin{align}
    &\zdivone\circ\zdivtwo\circ\overadd{(n-5)/2}{\psi}\ofP \,
    \circ
    P\circ\, C \circ \zspaceD
    \nonumber
    \\
    &\qquad\qquad\qquad =
    \underbrace{
    \zdivone\circ\overadd{(n-5)/2}{\tilde\psi}\ofDC
    }_{=: \overadd{(n-5)/2}{\hat\psi}(\TSzlap)
    \, \circ
    \,
    \zdivone}
    \circ
    \,
    \zdivtwo
     \circ P\circ\, C \circ \zspaceD
    \nonumber
    \\
    & \qquad\qquad\qquad =
    \overadd{(n-5)/2}{\hat\psi}(\TSzlap) \circ \zdivone\circ\zdivtwo  \circ P\circ\, C \circ \zspaceD
    \,,
\end{align}
where the first equality follows from the commutation relation \eqref{11X23.11}. The second equality makes use of fact that the operator $\overadd{p}{\tilde\psi}$ is a product of operators of the form $\tilde{\operatorname{L}}_{a,c}$ of Lemma \ref{L18VIII23.2}; the operator $\overadd{(n-5)/2}{\hat\psi}(\TSzlap)$ is obtained by using the commutation relation \eqref{10XI23.1} to commute $\zdivone$ and $\overadd{(n-5)/2}{\tilde\psi}\ofDC$. Clearly, $\overadd{(n-5)/2}{\hat\psi}(\TSzlap)$ is a product of elliptic operators and is hence itself elliptic. We thus conclude that the operator \eqref{10XI23.2}, and hence $\zdivone\circ\zdivtwo\circ \hLop_n$, are elliptic.

Now,  we have
$h^{[S]}=\im (C\circ \zspaceD)$. Next,
 by Proposition~\ref{P17X23.1}, $\ker(\overadd{\frac{n-3}{2}}{\psi}\ofP)= h^{[S]}$
plus possibly a finite dimensional
space of smooth tensor fields in $h^{[V]}$.
Finally, noting that
\begin{align}
    \zdivone \circ \zdivtwo\circ \hLop_n
    &
    =\hLop_n^\dagger \circ\, C \circ \zspaceD
    \,,
\end{align}
we conclude that $\ker \hLop_n^\dagger  \cap \ker(\overadd{\frac{n-3}{2}}{\psi}\ofP)$ is smooth and finite dimensional.
\qedskip

\seccheck{21XII23; stopping here}

Next, we move on to the operator $\Lop$ \eqref{24IV23.3}, relevant for the gauge transformation of the fields $\overadd{\frac{n-1}{2}+j}{q}_{AB}$ for $j\geq 0$.

\index{Ln@$ \Lopdagger $}%
\index{Ln@$ \Lopdagger $!kernel}%
\index{kernel!Ln@$ \Lopdagger$}%
\begin{proposition}
 \label{P16X23.2}
 The kernel $\ker \Lopdagger \cap \TTtp$ is smooth and finite dimensional.
 \end{proposition}

\proof
Recall from \eqref{30VI.1} that
\begin{align}
     \underline{\Lop}(\xi)_{AB} &=
     \begin{cases}
        \frac{1}{8} (\TSzlap +2 \tric-4 \myGauss ) (\TSzlap +2 \tric-6 \myGauss )  C(\xi)_{AB}
        &
        \\
        \qquad\qquad\qquad
        - \frac{1}{6}  (\TSzlap + 2 \tric - 5 \myGauss ) \TS[\zspaceD_A\zspaceD_B\zspaceD_C\xi^C] \,, & n=5
        \\
        \frac{1}{(n-1)(n-5)}
     \bigg( (\TSzlap  +2 \tric -2 (n-2) \myGauss) (\TSzlap +2 \tric +(1-n) \myGauss)
     C(\xi)_{AB} &
     \\
     \qquad\qquad\qquad
     -\frac{2 (n-3)}{n-2} (\TSzlap +2 \tric +(5-2 n) \myGauss )
     \TS[\zspaceD_A\zspaceD_B\zspaceD_C\xi^C]
     \bigg) \,, & n\neq 5\,,
     \end{cases}
     \label{26VIII23.1}
\end{align}
and
\index{Ln@$\Lop$}
\begin{equation}
    \Lop = \overadd{\frac{n-5}{2}}{\psi}\ofP \underline{\Lop}\,, \quad
    \Lopdagger = (\underline{\Lop})^\dagger\overadd{\frac{n-5}{2}}{\psi}\ofP\,.
    \label{10V24.41}
\end{equation}

For the purpose of calculating the space $\ker \big((\underline{\Lop})^\dagger\big) \cap \TTtp$, we start by analysing the symbol of $(\underline{\Lop})^\dagger \circ\, C $ given by, when $n\neq 5$,
\begin{equation}
    \sigma(k)_{B} =
    -
    \frac{1}{(n-5)(n-1)}\bigg(|k|^4 k^A 
    \TS[k_A \xi_B]
    - \frac{2(n-3)}{(n-2)} |k|^2 k_B k^C k^D
    \TS[k_C \xi_D]
    \bigg)
    \,.
    \label{26VIII23.2}
\end{equation}
To find its kernel for $k\neq 0$, we contract \eqref{26VIII23.2} with $k^B$ and set it to zero giving
\begin{align}
    0 &=
     \frac{n-4}{(n-2)(n-5)(n-1)} |k|^4 k^C k^D
    \TS[k_C \xi_D]
    \implies
   \frac{n-4}{(n-5)(n-1)^2} |k|^6 k^A\xi_A
    = 0
    \,,
\end{align}
and hence $k^A\xi_A=0$ since $|k|^2,(n-5),(n-1),(n-4) \neq 0$. This can be substituted back into \eqref{26VIII23.2} to give
\ptcheck{11XI; all the above up to here}
\begin{equation}
    0= |k|^4 k^A
    \TS[k_A \xi_B] = \frac{1}{2}|k|^4 |k|^2 \xi_B
    \implies
    \xi_B
    = 0
    \,.
    \label{26VIII23.3}
\end{equation}
Thus $(\underline{\Lop})^\dagger \circ\, C $ is elliptic. This, together with the ellipticity of  $\overadd{\frac{n-5}{2}}{\psi}\ofP$, and the fact that $\overadd{\frac{n-5}{2}}{\psi}(\TTtp) \subseteq \TTtp$ (cf.\ Proposition~\pref{P11X23.1}), allows us to conclude that
$\ker \Lopdagger \cap \TTtp$ is finite dimensional and its elements are smooth.

The same conclusion holds for $n=5$, with \eqref{26VIII23.2}  replaced by
 \ptcheck{11XI}
\begin{equation}
    \sigma(k)_{B} = -\frac{1}{8}|k|^4 k^A 
    \TS[k_A \xi_B]
    + \frac{1}{12} |k|^2 k_B k^C k^D
    \TS[k_C \xi_D]
    \,.
    \label{26VIII23.25}
\end{equation}
\qed

\section{$\partial_u^i h_{uA}^{[\CKVp]}$ obstructions for $m=0$, $\alpha\neq 0$}
\label{App5VI24.1}
In this appendix, we derive the expressions for the obstructions associated to the fields $\overadd{i}{\Hf}{}_{uA}^{[\CKVp]}$ when $m=0$ and $\alpha\neq 0$. \emph{Throughout this section we assume that $n$ is odd and $m=0$.}
This is most conveniently done using the machinery developed in Section \ref{ss5IX23.1} of the main text, making use of the $S$ and $V$ decomposition of $\CKVp$.

Indeed, it follows from the arguments there that (cf.\ around \eqref{4IX23.1--}) the transport equation for $\overadd{i}{\Hf}{}_{uA}^{[\CKVp]}$ can be solved by considering the $X= S,V$ projections separately. In the current case, this reads (cf.\ \eqref{6III23.w1a})
\begin{align}
    \partial_r  \overadd{i}{\Hf}{}_{uA}^{[\CKVp\cap X]}
       &=   r^{-(i+4)} \overadd{i}{\tilde\chi} \circ \zspaceD^B C(Y^{[X]})_{AB}\,,
\end{align}
where $Y\in\CKVp$ is the unique field such that $h_{AB}^{[\TTtp]} = C(Y^{[X]})_{AB}$. Thus, the radial charges are the projections of $\overadd{i}{\Hf}{}_{uA}$ onto the spaces $\CKVp\cap X\cap\ker \overadd{i}{\tilde\chi}$.

\subsection{The case $\myGauss\leq 0$}
\label{App5VI24.1a}

\begin{lemma}\label{l5VI24.1}
    For $\myGauss \leq 0$, the space $\CKVp\cap X\cap\ker \overadd{i}{\tilde\chi} = \emptyset$.
\end{lemma}
\begin{proof}
Recall the definition of the operators $\overadd{i}{\tilde\chi}\ofDC$:
\begin{align}
        \overadd{i}{\tilde\chi}\ofDC  &= \prod_{j=1}^{i} \hck{-(j+3)}{\ofDCnoDC}
      \,,
    \quad
    \overadd{0}{\tilde\chi}\ofDC  := 1 \,,
\end{align}
where
\begin{align}
    \hck{j}{\ofDCnoDC} &:= c_j \bigg[\frac{2 (n - 1) \zDivtwo\circ C}{(3 + j) (3 -  n + j) }
    +  (n-2)\myGauss + \TSzlap -   (n - 4 -  j) (2 + j) \myGauss \bigg]\,,
\end{align}
with $c_j:= -\frac{1}{7 -  n + 2 j}$.

Let $\xi\in X \cap\CKVp$. From \eqref{26IX23.1}, we have for $n>3$, $j\geq 1$,
\ptcheck{7VI24}
\begin{align}
    \hck{-(j+3)}{\ofDCnoDC}\xi = \begin{cases}
        \frac{(j+2) (j+n-2) }{j (j+n) (2 j+n-1)} \big(\TSzlap +\underbrace{(j  (j+n)+ 1) }_{\geq 0}\myGauss \big) \,, & X = S
        \\
        \frac{(j+1) (j+n-1) }{j (j+n) (2 j+n-1)} \big(\TSzlap + \underbrace{(j (j+n)+n-2) }_{\geq 0} \myGauss \big)\,.& X=V
    \end{cases}
    \label{8VI24.2}
\end{align}
Thus $\hck{-(j+3)}{\ofDCnoDC}\xi^{[\CKVp]} \neq 0$ by negativity of $\TSzlap$ when $\myGauss\leq 0$.
\qedskip
\end{proof}

It follows from Lemma \ref{l5VI24.1} that there are no obstructions associated to $\overadd{i}{\Hf}{}^{[\CKVp]}_{uA}$ when $\myGauss\leq 0$.

\subsection{The case $\myGauss > 0$}

Next, we move on to the case $\myGauss>0$. For this, we first derive a relation between the operators $\hLop_n$ and $\Lop$. We begin by looking at the operator $\hLop_n$ (cf.\ \eqref{11V23.21}-\eqref{5XI23.31}):
setting
$$a_n=\begin{cases}
    -\frac{1}{18} \,, & n=5\,,
    \\
    - \frac{n-4}{(n-5)(n-2)^2}\,, & n>5 \,,
\end{cases}$$
we have
\begin{align}
   \zspaceD_C \mrL   \circ \hLop_{n}(\xi^u) &= \zspaceD_C \zspaceD^A\zspaceD^B \underline{\hLop}_n(\zspaceD \xi^u)_{AB}
    \nn
    \\
    & =
     a_n \zspaceD_C \zspaceD^A\zspaceD^B
   \overadd{\frac{n-5}2}{\psi}\ofP \, \big( (n-1) P - 2 (n-2)^2 \myGauss
		\big) \, C(\zspaceD \xi^u)_{AB}
\nn
  \\
  &\overset{(\ref{28VIII23.f4})}{=}
  a_n \zspaceD_C \zspaceD^A\zspaceD^B
   \overadd{\frac{n-5}2}{\psi}\ofP \,
   \big( (n-2) (\TSzlap + 2 \tric ) - (n-2)(3n-7) \myGauss
		\big)
   C(\zspaceD \xi^u)_{AB}
  \nn
  \\
  &\overset{(\ref{11X23.11})}{=}
  a_n \zspaceD_C \zspaceD^A
   \overadd{\frac{n-5}2}{\tilde\psi}\ofDC \,
   (n-2) (\TSzlap - (2n - 5)\myGauss )
  \zspaceD^B C(\zspaceD \xi^u)_{AB}
  \nn
  \\
  & \overset{(\ref{26IX23.1})}{=}
  a_n \zspaceD_C \zspaceD^A
   \overadd{\frac{n-5}2}{\tilde\psi}\big(\TSzlap,\tfrac{n-2}{n-1}(\TSzlap+\myGauss)\big) \,
   \tfrac{(n-2)^2}{n-1} (\TSzlap - (2n - 5)\myGauss )
  (\TSzlap+\myGauss) \zspaceD_A \xi^u
  \nn
  \\
  &
  \overset{(\ref{19X23.1})}{=}
  a_n \tfrac{(n-2)^2}{n-1}
   \overadd{\frac{n-5}2}{\tilde\psi}\big(\TSzlap,\tfrac{n-2}{n-1}(\TSzlap+\myGauss)\big) \,
    (\TSzlap - (2n - 5)\myGauss )
  (\TSzlap+\myGauss) (\TSzlap-(n-2)\myGauss)\zspaceD_C \xi^u
\end{align}
where we made use of the fact that $[(\zspaceD\circ\zDivone),\TSzlap] = 0$ in the last equality.

Meanwhile, for the operator $\Lop$, we had from \eqref{5XI23.41},
 \begin{equation}\label{17VI24.1}
    \Lop = \overadd{\frac{n-5}{2}}{\psi}\ofP \, \underline{\Lop}
 \,,
 \end{equation}
where the operator $\underline{\Lop}$ was given in \eqref{30VI.1}. In particular, for
$\xi^A\in S$,
\begin{align}
     \zspaceD^A\underline{\Lop}(\xi)_{AB} &=
        b_n
     \zspaceD^A\big[ (\TSzlap  +2 \tric -2 (n-2) \myGauss) (\TSzlap +2 \tric +(1-n) \myGauss)
     C(\xi)_{AB}
     \nn
     \\
     &\quad\quad\quad
     -\frac{2 (n-3)}{n-2} (\TSzlap +2 \tric +(5-2 n) \myGauss )
     \TS[\zspaceD_A\zspaceD_B\zspaceD_C\xi^C]
     \big]
     \nn
     \\
     & =
      b_n
     \zspaceD^A\big[ (\TSzlap  +2 \tric -2 (n-2) \myGauss) (\TSzlap +2 \tric +(1-n) \myGauss)
     C(\xi)_{AB}
     \nn
     \\
     &\quad\quad\quad
     -\frac{2 (n-3)}{n-2} (\TSzlap +2 \tric +(5-2 n) \myGauss )
     (\TSzlap +2 \tric -2 (n-2) \myGauss)C(\xi)_{AB}
     \big]
     \nn
     \\
     &=
    -\frac{(n-4)b_n}{(n-2)}(\TSzlap - (2n - 5)\myGauss )
   (\TSzlap-(n-2)\myGauss)\zspaceD^AC(\xi)_{AB}
    \nn
     \\
     &=
    -\frac{(n-4)b_n}{(n-1)} (\TSzlap - (2n - 5)\myGauss )
   (\TSzlap-(n-2)\myGauss)(\TSzlap+\myGauss)\xi_B \,,
   \label{8VI24.1}
\end{align}
where
\begin{align}
    b_n = \begin{cases}
        \frac{1}{8}\,, & n=5 \,,\\
        \frac{1}{(n-1)(n-5)}\,, & n>5\,.
    \end{cases}
\end{align}
Thus,  for $\xi^A\in S$,
\begin{align}
    \zspaceD^A\Lop(\xi)_{AB} & = -\frac{(n-4)b_n}{(n-2)}  \overadd{\frac{n-5}2}{\tilde\psi}\big(\TSzlap,\tfrac{n-2}{n-1}(\TSzlap+\myGauss)\big)(\TSzlap - (2n - 5)\myGauss )
   (\TSzlap-(n-2)\myGauss)(\TSzlap+\myGauss)\xi_B
    \nn
    \\
    & =
    \frac{1}{n-1}
    \zspaceD_B \mrL   \circ \underline{\hLop_{n}}(\xi)
    \label{8VI24.21}
\end{align}
A similar calculation shows that \eqref{8VI24.21} continues to hold for $n=5$.

Next, for $i\in \Z, i \geq 0$, it can be shown inductively that under residual gauge transformations, the $r$-independent part of the gauge transformation of the fields $\overadd{i}{\Hf}_{uA}$ reads,%
\index{gauge transformation law!H@$H_{uA}$!$\overadd{i}{\Hf}_{uA}$}%
\index{H@$H_{uA}$!$\overadd{i}{\Hf}_{uA}$!gauge transformation law}%
\begin{align}
    \overadd{i}{\Hf}_{uA} &\rightarrow
    \begin{cases}
        \overadd{i}{\Hf}_{uA} +
    n\partial_u^{i+1}\xi_A
    + \sum_{j=1}^{i/2} \alpha^{2j} n \overadd{i,j}{D}(\partial_u^{i+1-2j}\xi_A)
    + \alpha^{i+2} n \overadd{i,\frac{i+2}{2}}{D}\circ \zspaceD (\xi^u )\,, & i \, \text{even}
    \\
     \overadd{i}{\Hf}_{uA} +
    n\partial_u^{i+1}\xi_A
    + \sum_{j=1}^{(i+1)/2} \alpha^{2j} n \overadd{i,j}{D}(\partial_u^{i+1-2j}\xi_A)\,, & i\, \text{odd} \,.
    \label{4VI24.2}
    \end{cases}
\end{align}
The operators
\index{D@$\overadd{p,q}{D}$}%
$$
 \overadd{p,q}{D}
$$
appearing in \eqref{4VI24.2} are sums of products of operators of the form $\tilde{\operatorname{L}}_{a,c,b}$ of \eqref{18XII23.11} and preserve the spaces $S$ and $V$. They commute with $\TSzlap$ when restricted to $S$ or to $V$.

Thus, for $\myGauss>0$, we have the gauge-invariant radial charge: for even $2\leq i\leq k-\frac{n+1}{2}$,%
\index{Q@$\kQ{5,i}{}{}$}
\begin{align}
    \kQ{5,i}{}{}^{[X]}_B&:= \zspaceD^C \Lop(\overadd{i}{\Hf}{}^{[X\cap \CKVp\cap\ker(\overadd{i}{\chi}\circ C)]}_{uA} )_{CB}
    - n  (\zspaceD^C  \overadd{\frac{n+1}{2}+i}{q}_{CB}){}^{[X\cap \CKVp\cap\ker(\overadd{i}{\chi}\circ C)]}
    \nn
    \\
    &\quad
    - \sum_{j=1}^{i/2} \alpha^{2j} n \overadd{i,j}{D}(\zspaceD^A  \overadd{\frac{n+1}{2}+i-2j}{q}_{AB}){}^{[X\cap \CKVp\cap\ker(\overadd{i}{\chi}\circ C)]}
    \nn
    \\
    &\quad
    - \alpha^{i+2} \frac{n}{n-2} \overadd{i,\frac{i+2}{2}}{D} \Big( \zspaceD_B(\overadd{3}{Q}){}^{[X\cap \CKVp\cap\ker(\overadd{i}{\chi}\circ C)]} \Big)\,,
    \label{5VI24.21}
\end{align}
for $X = S,V$, with the last term vanishing for $X=V$. We have made use of \eqref{8VI24.21} for the last term. For odd $1 \leq i\leq k-\frac{n+1}{2}$, the charge is given by a similar expression,
\ptcheck{7VI; both cases}
\begin{align}
    \kQ{5,i}{}{}^{[X]}_B&:=
    \begin{cases}
        \zspaceD^C \Lop(\overadd{1}{\Hf}{}^{[X\cap\CKVp \cap\ker(\overadd{1}{\chi}\circ C)]}_{uA} )_{CB}
    - n  (\zspaceD^C  \overadd{\frac{n+3}{2}}{q}_{CB}){}^{[X\cap \CKVp\cap\ker(\overadd{i}{\chi}\circ C)]} &
    \\
    \quad
    - \alpha^{2} n \overadd{1,1}{D}(\overadd{4}{Q}_B){}^{[X\cap\CKVp \cap\ker(\overadd{1}{\chi}\circ C)]}\,, & i=1
    \\
    \zspaceD^C \Lop(\overadd{i}{\Hf}{}^{[X\cap \CKVp\cap\ker(\overadd{i}{\chi}\circ C)]}_{uA} )_{CB}
    - n  (\zspaceD^C  \overadd{\frac{n+1}{2}+i}{q}_{CB}){}^{[X\cap \CKVp\cap\ker(\overadd{i}{\chi}\circ C)]} &
    \\
    \quad
    - \sum_{j=1}^{(i+1)/2} \alpha^{2j} n \overadd{i,j}{D}(\zspaceD^A  \overadd{\frac{n+1}{2}+i-2j}{q}_{AB}){}^{[X\cap \CKVp\cap\ker(\overadd{i}{\chi}\circ C)]}\,, & i\geq 3 \,.
    \end{cases}
    \label{5VI24.22}
\end{align}
The gauge-invariance of $\kQ{5,i}{}{}_B$ can be verified by making use of the fact that $[\zDivtwo\circ \Lop, \overadd{p,q}D] = 0$.

\paragraph{Interlude: the case $n=3$.} From the higher dimensional (HD) case, we see that to obtain the gauge-invariant radial charges associated to $\overadd{i}{\Hf}_{uA}$, it is useful to first obtain the radial charges whose gauge transformations do not involve fields $\partial^i_u\xi$ with different   $i$'s. These are provided by the fields $\zspaceD^B\overadd{i}{q}_{AB}$ when $n>3$.

However, when $n=3$, from \eqref{19VI24.5} we have for $i\geq 3$, using \eqref{27VI24.11},
 \ptcheck{27VI24}
\begin{align}
    (\zspaceD^B \overadd{i}{q}_{AB})^{[X]} \rightarrow
    \begin{cases}
       (\zspaceD^B \overadd{i}{q}_{AB})^{[V]} +  \zspaceD^B \opL_3(\partial^{i-1}_u \xi^{[V]})_{AB}\,, & \\
       (\zspaceD^B \overadd{i}{q}_{AB})^{[S]} +  \zspaceD^B \opL_3(\partial^{i-1}_u \xi^{[S]})_{AB}
       +\frac{\alpha^2}{2}(\TSzlap-\myGauss) \zspaceD^B\opL_3( \partial_u^{i-3}\xi^{[S]})_{AB} \,, &
    \end{cases}
    \label{19VI24.21}
\end{align}
where the upper case holds for $X=V$ and the lower case for $X=S$. 
See Remark \ref{R28VI24.1} for the origin of the   $\alpha^2$-term in the last line above, not present in the HD case.

The case $X=V$ agrees with the HD case and thus the charges $\kQ{5,i}{}^{[V]}$ as defined in \eqref{5VI24.21}-\eqref{5VI24.22} are also gauge-invariant obstructions when $n=3$. 
 
For $n=3$ and $X=S$,
we define recursively the radial charges $\overadd{i}{q}_A$, for integers $i\geq 3$, as 
\begin{align}
    \overadd{3}{q}_A &:=
        (\zspaceD^B\overadd{3}{q}_{AB})^{[S]} - \frac{\alpha^2}{2} (\TSzlap-\myGauss)\kQ{4}{}{}_A^{[S]}\,,
\quad
    \overadd{i}{q}_A :=
        (\zspaceD^B\overadd{i}{q}_{AB})^{[S]} - \frac{\alpha^2}{2}(\TSzlap-\myGauss) \overadd{i-2}{q}_A\,,
\end{align}
where the $\alpha^2$-terms are included to compensate for the last term in the second line of \eqref{19VI24.21}.
It then follows from \eqref{19VI24.3} and \eqref{19VI24.21} that, under gauge-transformations,
 \ptcheck{27VI24}
\begin{align}
    \overadd{i}{q}_{A} \rightarrow \overadd{i}{q}_{A} + \zspaceD^B\opL_3(\partial_u^{i-1}\xi^{[S]})_{AB} \,.
\end{align}
Thus,%
\index{Q@$\kQ{5,i}{}{}$}
 when $X=S$, we define the charges $\kQ{5,i}{}^{[S]}$ as in \eqref{5VI24.21}-\eqref{5VI24.22}, but with all $\zspaceD^A\overadd{i}{q}_{AB}$'s replaced by $\overadd{i}{q}_{B}$'s.
\qedskip

Returning to the higher dimensional case, we end this with the following proposition:

\begin{proposition}
When $\myGauss>0$, we have
{\rm
\begin{align}
    \CKVp\cap{\ker}({}\zDivtwo\circ \Lop) = \{0\} \,.
    \label{8VI24.7}
\end{align}
}
\end{proposition}

\proof
We start by noting that for $n>5$, $\myGauss>0$ and $\xi_A\in \CKVp$, we have (cf.\ \eqref{17VI24.1})
\begin{align}
    \zDivtwo \circ \Lop(\xi) &= \zDivtwo\circ \overadd{\frac{n-5}{2}}{\psi} \circ \underline\Lop (\xi)
    \nn
    \\
    & =  \overadd{\frac{n-5}{2}}{\tilde\psi} \circ \zDivtwo \circ \underline\Lop (\xi)
    \nn
    \\
    & = \prod_{j=2}^{i} \hck{\tfrac{n-5-2j}{2}}\ofDCnoDC
    \circ
    \overadd{1}{\tilde\psi} \circ \zDivtwo \circ \underline\Lop (\xi)
    \label{17VI24.3}
\end{align}
For $n=5$, the first equality holds with $\overadd{0}{\psi}:= 1$.

We note next  that $\zDivtwo \circ \underline\Lop $, and each of the remaining   operators appearing in \eqref{17VI24.3}, is formally self-adjoint and  preserves  $\CKVp$. We will show that these operators are  isomorphisms on $\CKVp$, for this it suffices to show that they  have  trivial kernels on  $\CKVp$.

Recall that $\myGauss>0$ throughout this proof by assumption.

We begin with the operator $\zDivtwo \circ \underline\Lop (\xi)$.
A similar calculation to that of \eqref{8VI24.1} shows that for $\xi^A\in V\cap\CKVp$,
\begin{align}
\label{8VI24.5a}
    \zspaceD^A\underline{\Lop}(\xi)_{AB} = c_n (\TSzlap - \myGauss) (\TSzlap - (n-2) \myGauss) \zspaceD^A C(\xi)_{AB} \,,
\end{align}
where $c_n$ is a non-vanishing constant depending on $n$. The negativity of the Laplacian, together with \eqref{8VI24.1}, implies that $\CKVp\cap\ker\zDivtwo \circ \underline\Lop = \emptyset$.
This completes the proof for $n=5$.

Next, for the purpose of investigating the kernel of $\overadd{p}{\tilde\psi}$, $p=\frac{n-5}{2}$, $n>5$,
it follows from the commutation relations in Appendix~\ref{App8XII23.1}, in particular  from \eqref{8XII23.g1-} and \eqref{8XII23.g3} with  appropriate indices,
that
\ptcheck{17VI; actually previously, by Finn}
 \FGp{Checked 18VI24}
\begin{align}
    \hck{\frac{n-5-2j}{2}}{\ofDCnoDC}\xi_A = \begin{cases}
        \frac{(2 j-n+3) (2 j+n-5) }{8 (j-1) (2 j-n-1) (2 j+n-1)} \left(4 \TSzlap -   c_S \myGauss \right)
        \xi_A \,,
        \\
        \frac{(2 j-n+1) (2 j+n-3)  }{ 8 (j-1) (2 j-n-1) (2 j+n-1)} (4 \TSzlap - c_V \myGauss) \xi_A \,,&
    \end{cases}
    \label{8VI24.3}
\end{align}
with
\begin{align}
    c_S:= \left(n^2 - 5 - 4 (j-1) j\right)\,,\quad
    c_V: =( (n-4) n + 7 - 4 (j-1) j ) \,,
\end{align}
and where the upper case holds for $\xi_A\in \CKVp\cap S$ and the lower for $\xi_A\in\CKVp\cap V$. Now for $j\leq \frac{n-5}{2}$ and $n > 5$, we have
\begin{align}
    c_S & \geq n^2 - 5 - 4 \Big(\frac{n-5}{2}-1 \Big) \frac{n-5}{2}
    =
    12n - 40
    > 0 \,,
    \nn
    \\
    c_V &\geq  (n-4) n + 7 - 4 \Big(\frac{n-5}{2}-1\Big) \frac{n-5}{2}
    = 8n - 28 > 0\,,
\end{align}
and thus for $0\ne \xi\in\CKVp$,
\begin{align}
\label{8VI24.5b}
    \hck{\frac{n-5-2j}{2}}{\ofDCnoDC}\xi_A \neq 0
\end{align}
for $\myGauss>0$ and $j \leq \frac{n-5}{2}$. Finally, we have (cf.\ \eqref{6III23.w7a}), again for $n>5$, $\myGauss>0$ and $0\ne \xi\in\CKVp$,
\begin{align}
    \overadd{1}{\tilde\psi}(\ofDCnoDC) \xi_A& = \begin{cases}
       \frac{(n-5) (n-3) }{8 \left(n^2-1\right)}
        \left(4 \TSzlap - \left(n^2-5\right) \myGauss \right) \xi_A
        \,, & \xi_A \in S
        \\
        \frac{(n-3)}{8 (n+1)}
         (4 \TSzlap - ((n-4) n+7) \myGauss ) \xi_A \,, & \xi_A \in V\
    \end{cases}
    \nn
    \\
    & \neq 0 \,.
    \label{8VI24.5c}
\end{align}
The proof is complete.
\qedskip

\section{Operators on spheres}
 \label{ss20X22.1}

 In this appendix we describe in detail the operators and obstructions appearing in the main text in the case when $\secN=S^d$. In what follows, unless explicitly indicated otherwise we  assume that $d\equiv  n-1 >2$, and for simplicity we assume that  $n> 5$.

 \subsection{Decomposition of tensors on $S^d$}

First, recall from \eqref{7VI17.101PK} the decomposition of a $2$-covariant, symmetric trace-free tensor $h$ according to
\begin{equation}
    h = h^{\red[V]} + h^{\red[S]} + h^{\red[TT]} \,.
    \label{7IX23.11}
\end{equation}
On $S^d$, the three parts in \eq{7IX23.11} can be expanded into eigenfunctions of the Laplacian as in~\cite[Sections 2.1, 5.1 and 5.2]{KodamaIshibashiMaster}
\begin{align}
	h ^{\red[S]}_{AB}&=\sum_{I} H_{I}^{\red[S]}\mathbb{S}^{I}_{AB}
	\,, \quad
	\label{06IX17.2}
	h ^{\red[V]}_{AB}=\sum_{I}  H_{I}^{\red[V]}
	\mathbb{V}^{I}_{AB}
	\,, \quad
	h ^{[\TTt]}_{AB}=\sum_{I} H_{I}^T
	\mathbb{T}^{I}_{AB}\,,
\end{align}
with expansion coefficients $H_I^{*}$, where
\begin{align}
	\label{7IX17.61PK}
	\mathbb{S}^{I}_{AB}
	&=
	\frac{1}{k^2}
TS[
\zspaceD_A\zspaceD_B\mathbb{S}^{I}
]
=
	\frac{1}{k^2}
\zspaceD_A\zspaceD_B\mathbb{S}^{I}+\frac{1}{d}\ringh_{AB}\mathbb{S}^{I}
	\,,
	\quad
	k \ne 0
	\,,
	\\
	\mathbb{V}^{I}_{AB}&=-\frac{1}{2k_V}(\zspaceD_A\mathbb{V}^{I}_B
	+\zspaceD_B\mathbb{V}^{I}_A
    )
	\,,
	\quad
	k_V \ne 0
        \,,
        \quad
       \zspaceD^A \mathbb{V}^{I}_A = 0
	\,,
 \label{29V.1}
\end{align}
with the coefficients $H_I^*$ vanishing if $k=0$ or $k_V=0$.
The fields $\mathbb{S}^{I}$, $\mathbb{V}_A^{I}$, $\mathbb{T}_{AB}^{I}$ appearing above are scalar, vector, and $\TTt$ tensor-harmonics, i.e.
\begin{equation}
(\TSzlap+k^2)\mathbb{S}^{I}=0\,,
\quad(\TSzlap+k_V^2)\mathbb{V}^{I}_A=0\,,\quad
(\TSzlap+k_T^2)\mathbb{T}^{I}_{AB}=0\,,
\end{equation}
\begin{equation}
\mathbb{T}^{I}_{AB}=\mathbb{T}^{I}_{BA}\,,\qquad
\ringh^{AB}\zspaceD_A \mathbb{T}^{I}_{BC}=0\,,\qquad
\ringh^{AB} \mathbb{T}^{I}_{AB}=0\,,
\label{8VI23.7}
\end{equation}
with eigenvalues $k^2 = k^2(I)$, $k_V^2 = k_V^2(I)$, $k_T^2= k_T^2(I)$.
On $S^{d}$ the eigenvalues are%
~\cite[Sections 2.1, 5.1 and 5.2]{KodamaIshibashiMaster} 
\begin{align}
	k^2&=\ell (\ell+d-1)\,, & \ell&=0,1,2,\dots\,,\label{K1k}\\
	k_V^2&=\ell (\ell+d-1)-1\,, & \ell&=1,2,3\dots\,,\label{K1kV}\\
	k_T^2&=\ell(\ell+d-1)-2\,, & \ell&=2,3,4,\dots\,, \quad d>2\label{K1kT}\,.
\end{align}

For the purposes of the main body of this paper we  need to calculate $Ph^{\red[S]}$, $Ph^{\red[V]}$, and $P\red{h^{\red[TT]}}$. We start by noting that  for any symmetric traceless 2-covariant tensor $h$, we have
$$
 P\red{h^{\red[TT]}} \equiv C \zdivtwo \red{h^{\red[TT]}} = 0
  \,.
$$
Next, in order to determine $Ph^{\red[S]}$ and $Ph^{\red[V]}$, the following identities will be useful (the reader might find useful the commutation relations of Appendix~\ref{App19V23.1} in their derivation):
\ptcheck{10VI, ptc+Wan}
\begin{align}
    \TSzlap \mathbb{S}_{AB}^I
     \checkmark
      &= (-k^2 + 2 d)\mathbb{S}_{AB}^I \,,
    \label{5VI23.2a}
\\
    \TSzlap \mathbb{V}_{AB}^I
     \checkmark
     &= (-k_V^2 + d + 1) \mathbb{V}_{AB}^I \,,
    \label{5VI23.2}
\\
    \zspaceD^A \mathbb{S}_{AB}^I
     \checkmark
     &= \frac{(d - k^2)(d-1)}{k^2d} \zspaceD_B \mathbb{S}^I \,, \quad k^2>0 \,,
    \label{8VI23.1}
\\
    \TSzlap  \zspaceD_A \mathbb{S}^I
     \checkmark
      &= (d-1- k^2) \zspaceD_A \mathbb{S}^I \,, \quad k^2\geq 0 \,,
    \label{8VI23.3}
\\
    \zspaceD^A \mathbb{V}_{AB}^I
     \checkmark
      &= \frac{k_V^2 -d+1}{2 k_V} \mathbb{V}_{B} \,, \quad k_V>0 \,.
    \label{8VI23.2}
\end{align}
Then:
\ptcheck{10VI, PTC + Wan}
\begin{align}
    P (h^{\red[V]})_{EF} &= - \sum_{I}  H_{I}^{\red[V]} \frac{1}{2k_V} C \Big[\zspaceD^A(\zspaceD_A\mathbb{V}^{I}_B
	+\zspaceD_B\mathbb{V}^{I}_A
    )
    \Big]_{EF}
    \nonumber
    \\
    &=
    - \sum_{I}  H_{I}^{\red[V]} \frac{1}{2k_V} C \Big[-k_V^2 \mathbb{V}^{I}_B
	+ \zspaceD^C\zspaceD_B\mathbb{V}^{I}_C
    \Big]_{EF}
    \nonumber
    \\
    &=
     \sum_{I}  H_{I}^{\red[V]} \frac{1}{2k_V} C \Big[(k_V^2 - (d-1) )\mathbb{V}^{I}_B
    \Big]_{EF}
    \nonumber
    \\
    &=
    \sum_{I} - H_{I}^{\red[V]}   (k_V^2 - (d-1) )\mathbb{V}^{I}_{EF}
 \,,
 \label{1III23.6}
\end{align}
and
\ptcheck{10VI, PTC + Wan}
\begin{align}
    \red{P(h^{\red[S]})}_{EF}  &= \sum_I H_I^{\red[S]} C\Big[\zspaceD^A\Big(\frac{1}{k^2} \zspaceD_A \zspaceD_B \mathbb{S}^I + \frac{1}{d} \ringh_{AB} \mathbb{S}^I \Big) \Big]_{EF}
    \nonumber
\\
            &= \sum_I H_I^{\red[S]} C\Big[\Big(\frac{1}{k^2} \TSzlap \zspaceD_B \mathbb{S}^I + \frac{1}{d}\zspaceD_B \mathbb{S}^I \Big) \Big]_{EF}
            \nonumber
\\
            &= \sum_I H_I^{\red[S]} C\Big[\Big(\frac{1}{k^2}\zspaceD_B  \TSzlap \mathbb{S}^I + \frac{d-1}{k^2}\zspaceD_B \mathbb{S}^I + \frac{1}{d}\zspaceD_B \mathbb{S}^I \Big) \Big]_{EF}
            \nonumber
\\
            &= \sum_I H_I^{\red[S]} C\Big[\Big(- \zspaceD_B   \mathbb{S}^I + \frac{d-1}{k^2}\zspaceD_B \mathbb{S}^I + \frac{1}{d}\zspaceD_B \mathbb{S}^I \Big) \Big]_{EF}
            \nonumber
\\
            &= \sum_I H_I^{\red[S]} \Big(- 1 + \frac{d-1}{k^2} + \frac{1}{d}\Big) C[\zspaceD_B \mathbb{S}^I]_{EF}
            \nonumber
\\
            &= \sum_I H_I^{\red[S]} \Big(- k^2 + d-1 + \frac{k^2}{d}\Big) \mathbb{S}_{EF}^I
            \nonumber
\\
            &=-\sum_I H_I^{\red[S]} \left({\frac{(d-k^2)(d-1) }{d}}\right) \mathbb{S}_{EF}^I \,.
            \label{1III23.5}
\end{align}
Note that if we write $k^2 = \ell(\ell+d -1)$, then   \ptcheck{10VI, PTC + Wan}
\begin{align}
   {\frac{(d-k^2)(d-1) }{d}}=  {\frac{(\ell-1)(d-1)(\ell+d)}{d}} \,.
            \label{1III23.5b}
\end{align}

\subsection{Conformal Killing vectors on the sphere}

 We have the following properties of conformal Killing vectors on $S^d$: 

\begin{Proposition}
 \label{P11VI23.1}
 Consider $S^d$, $d\ge 2$, with the canonical metric.
\begin{enumerate}
  \item The space of proper conformal Killing vectors is spanned by gradients of $\ell=1$ spherical harmonics.
  \item Let $C_A$ be a conformal Killing vector, then $\zspaceD^A C_A$ is an $\ell=1$ spherical harmonic.
\end{enumerate}
\end{Proposition}

\proof
1.
Equation \eqref{8VI23.1} implies   that
\begin{align}
    \int_{\secN} C(\zspaceD_C \mathbb{S}^I)^{AB} C(\zspaceD_C \mathbb{S}^I)_{AB} \sm
    &
    =
     -\int_{\secN} \zspaceD^A \mathbb{S}^I \zspaceD^B
      \underbrace{
       C(\zspaceD_C \mathbb{S}^I)_{AB}
       }_{k^2 \mathbb{S}^I _{AB}}\sm
    \nonumber
    \\
    & =
    -\int_{\secN} \zspaceD^A \mathbb{S}^I \frac{(d-k^2)(d-1)}{d}\zspaceD_A \mathbb{S}^I\sm
     \,,
\end{align}
which vanishes if and only if $k^2 = d$;  equivalently  $\ell=1$. The result follows by noting that the dimensions of both spaces coincide.

2.{Taking the divergence of  \eqref{30X22.CKV3} on a $d$-dimensional sphere  with the canonical metric and using the commutation relation \eqref{14VII23.f1}, one finds
 \begin{equation}
 	\frac{2(d-1)}{d}\left(\TSzlap +d \right)\zspaceD^AC_A=0\,,
\end{equation}
which is the scalar spherical-harmonic equation \eqref{K1k} with $\ell=1$.

%
\qed

\subsection{$\hat D$}
\label{ss9VI23.1}

\index{kernel!$\overadd{i}{\chi}\ofP \, \circ\, C$}%
\index{kernel!$\overadd{i}{\psi}\ofP$}%
We now determine the kernels of the operators $\overadd{i}{\chi}\ofP \, \circ\, C$ and $\overadd{i}{\psi}\ofP$ appearing in \eqref{6III23.w2} and \eqref{6III23.w9} on $S^{n-1}$ (and hence $\myGauss=1$, $\tric(h)_{AB} = - h_{AB}$).

We begin with $\overadd{i}{\chi}\ofP \, \circ\, C$, for which we need the the operator $\ck{-(i+3)}{\ofPnoP}$ acting on symmetric traceless 2-covariant tensors $h$ on $S^{n-1}$:
 \ptcheck{8VI23}
\begin{align}
    \ck{-(i+3)}{\ofPnoP} h_{AB} \equiv
    \frac{1}{2 i + n-1}
    \bigg(
    \frac{2 (n-1)}{i (i+n) } P
    + \TSzlap +i (i+n)+ n - 3 \bigg)
    h_{AB}\,.
    \label{5VI23.1}
\end{align}
Using the decomposition \eqref{7VI17.101PK}, we shall analyse the action of these operators 
on $\mathbb{S}^I_{AB}$ and $\mathbb{V}_{AB}^I$  separately (we leave out the analysis for $\mathbb{T}^I_{AB}$ as these are irrelevant here).

First, using \eqref{5VI23.2a} and \eqref{1III23.5} we have
\ptcheck{10VI; with mathematica}
\begin{align}
    \ck{-(i+3)}{\ofPnoP} \mathbb{S}^I_{AB}
    &=
    \frac{1}{2 i + n-1}
    \bigg(
    \frac{2 (n-1)}{i (i+n) } P
    + \TSzlap +i (i+n)+ n - 3 \bigg)
    \mathbb{S}^I_{AB}
    \nonumber
\\
    &=
    \frac{1}{2 i + n-1}
    \bigg(
    -\frac{2 (\ell-1)(n-2)(\ell+n-1)}{i (i+n) }
    -\ell(\ell+n-2)
    \nonumber
    \\
    & \qquad \qquad \qquad \qquad
    + 2 (n-1)
    +i (i+n)+ n - 3 \bigg)
    \mathbb{S}^I_{AB}
    \nonumber
\\
    &=
    \frac{(i+2) (i-\ell+1) (i+n-2) (i+\ell+n-1)}{i (i+n) (2 i+n-1)}\mathbb{S}^I_{AB} \,,
    \label{5VI23.5}
\end{align}
which gives zero if and only if $\ell = 1+i$ under our restrictions on the parameters involved. Thus the spherical harmonic tensor $\mathbb{S}^I_{AB}$ of mode $\ell = 1+i$ lies in the kernel of $ \ck{-(i+3)}{\ofPnoP}$. Similarly, using \eqref{5VI23.2} and \eqref{1III23.6} one obtains,
\ptcheck{10VI; against mathematica file}
\begin{align}
    \ck{-(i+3)}{\ofPnoP} \mathbb{V}^I_{AB}
    &=
    \frac{(i+1) (i-\ell+1) (i+n-1) (i+\ell+n-1)}{i (i+n) (2 i+n-1)} \mathbb{V}^I_{AB} \,,
    \label{5VI23.4}
\end{align}
which gives zero if and only if $\ell = 1+i$
\ptcheck{10VI}.

We conclude that:
\ptcheck{10VI}
\begin{itemize}
\index{kernel!$\overadd{i}{\chi}\ofP \, \circ\, C$}
    \item[(i)] spherical harmonic vectors $\zspaceD_A\mathbb{S}^I$ and $\mathbb{V}^I_{A}$ of modes $\ell = 1,2,...,i+1$ span the kernel of $\overadd{i}{\chi}\ofP \, \circ\, C$ for $i\geq 0$.
\end{itemize}

We move on now to $\overadd{1}{\psi}\ofP$ acting on symmetric traceless 2-covariant tensors $h$ on $S^{n-1}$:
\begin{align}
     \overadd{1}{\psi}\ofP  \, h_{AB}& \equiv - \bigg[\frac{4}{n+1} P + 1 - \frac{1}{2}\TSzlap
       +\frac{(n-3)(n-1)}{8}\bigg] h_{AB}
       \,.
       \label{24III23.1}
\end{align}
%
The first term in \eqref{24III23.1} gives zero when acting on $\red{h^{\red[TT]}}_{AB}$ in which case we are left with
\begin{align}
     \overadd{1}{\psi}\ofP \, \red{h^{\red[TT]}}_{AB}& =
       \frac{1}{2}\bigg( - 2
       + \TSzlap
       - \frac{(n-3)(n-1)}{4} \bigg) \red{h^{\red[TT]}}_{AB} \,.
       \label{24III23.2}
\end{align}
Since $\TSzlap$ is negative, it follows that $\overadd{1}{\psi}\ofP$ has no kernel when acting on $\TTt$ tensors.
\ptcheck{10VI; up to here}

Next, we consider the action of $\overadd{1}{\psi}\ofP$ on $\mathbb{V}^I_{AB}$. From \eqref{5VI23.2} and \eqref{1III23.6}, we have,
\begin{align}
     \overadd{1}{\psi}\ofP \, \mathbb{V}_{AB}^I& \equiv - \bigg[\frac{4}{n+1} P +  1  - \frac{1}{2}\TSzlap
       +\frac{(n-3)(n-1)}{8}\bigg] \mathbb{V}_{AB}^I
       \nonumber
\\
    &=
    -\frac{(n-3) (2 \ell+n-3) (2 \ell+n-1)}{8 (n+1)}
    \mathbb{V}_{AB}^I
    \,,
       \label{26III23.1}
\end{align}
which gives zero for $n>3$ only when $\ell = \frac{1-n}{2},\frac{3-n}{2}$, both of which are negative. Thus the harmonic tensors $\mathbb{V}_{AB}^I$ do not lie in the kernel of $\overadd{1}{\psi}\ofP$.
Meanwhile, a similar calculation using \eqref{5VI23.2a} and \eqref{1III23.5} gives,
\begin{align}
     \overadd{1}{\psi}\ofP \, \mathbb{S}_{AB}^I& =
    -\frac{(n-5) (n-3) (2 \ell+n-3) (2 \ell+n-1)}{8 \left(n^2-1\right)}
    \mathbb{S}_{AB}^I
    \,.
       \label{26III23.1b}
\end{align}
Once again, the right-hand side evaluates to zero for $n>5$ only when $\ell = \frac{1-n}{2},\frac{3-n}{2}$, both of which are negative. We conclude that, when $n>5$, the operator  $\overadd{1}{\psi}\ofP$ has trivial kernel.
\ptcheck{10VI; up to here against mathematica file operatorsS2casev2.nb}

Finally,  in order to analyse  $\overadd{i}{\psi}\ofP$, $i>1$, we need to understand the operators
%
\begin{equation}
    \ck{\frac{n-5-2i}{2}}{\ofPnoP}
    \equiv
   \frac{1}{i-1} \bigg(
   \frac{4 (n-1) }{\left((1-2 i)^2-n^2\right)} P
   + \frac{1}{2} \TSzlap - 1
   + \frac{(2 i-n+1) (2 i+n-3)}{8}
   \bigg)  \,.
\end{equation}
Proceeding as before, making use of \eqref{5VI23.2a}, \eqref{K1k} and \eqref{1III23.5} gives,
\begin{align}
    \ck{\frac{n-5-2i}{2}}{\ofPnoP}
   &\,   \mathbb{S}^I_{AB} 
     \nonumber
\\
 &
 \hspace{-1.5cm}
 =
    \frac{(2 i-n+3) (2 i+n-5) (2 i-2 \ell-n+1) (2 i+2 \ell+n-3)}{8 (i-1) (2 i-n-1) (2 i+n-1)}
    \mathbb{S}^I_{AB}\,.
\end{align}
Keeping in mind that $i>1$ and $\ell \geq 0$, the right-hand side vanishes if and only if $i=\frac{n-3}{2}$ or $\ell = i-\frac{n-1}{2}$. In addition, making use of \eqref{5VI23.2}, \eqref{K1kV} and \eqref{1III23.6} gives,
\begin{align}
    \ck{\frac{n-5-2i}{2}}{\ofPnoP}
     &\, \mathbb{V}^I_{AB}
     \nonumber
\\
 &
 \hspace{-1.5cm}
 =
    \frac{(2 i-n+1) (2 i+n-3) (2 i-2 \ell-n+1) (2 i+2 \ell+n-3)}{8 (i-1) (2 i-n-1) (2 i+n-1)}
    \mathbb{V}^I_{AB}\,;
\end{align}
the right-hand side vanishes if and only if $i=\frac{n-1}{2}$ or $\ell = i-\frac{n-1}{2}$. Next, making use of \eqref{K1kT}, we have,
\begin{align}
    \ck{\frac{n-5-2i}{2}}{\ofPnoP}\mathbb{T}^I_{AB}
    &=
    \frac{(2 i-2 \ell-n+1) (2 i+2 \ell+n-3)}{8 (i-1)}\mathbb{T}^I_{AB}\,,
\end{align}
the right-hand side of which is zero if and only if $\ell = i-\frac{n-1}{2}$.

Recall that when $n$ is odd and $i=\frac{n+1}{2}$ we have
\begin{equation}
    \overadd{\frac{n+1}{2}}{\psi}\ofP  = \mathring{\mathcal{K}}(-3,P) \overadd{\frac{n-1}{2}}{\psi}\ofP \, \,.
\end{equation}
It can be verified by a similar calculation as above
that $\mathring{\mathcal{K}}(-3,P)$ has trivial kernel.

Summarising (cf. Proposition~\ref{P9X23.1}):

\begin{proposition}
 \label{P28XI23.1}
 On $S^d$, $d\ge 2$, with the canonical metric,
\index{kernel!psi@$\overadd{i}{\psi}\ofP$}%
\index{psi@$\overadd{i}{\psi}\ofP$!kernel}%
%
\begin{itemize}
    \item[(i)] When $n$ is even, the operators $\overadd{i}{\psi}\ofP$ have trivial kernel for all $i \geq 1$.
    \item[(ii)] When $n$ is odd, the operators $\overadd{i}{\psi}\ofP$ have trivial kernel for $i<\frac{n-3}{2}$;
    the kernel of the operator $\overadd{\frac{n-3}{2}}{\psi}\ofP$ is spanned by all spherical harmonic tensors $\mathbb{S}^I_{AB}$;
    the kernels of the operators $\overadd{\frac{n \pm 1}{2}}{\psi}\ofP$ are spanned by all spherical harmonic tensors $\mathbb{V}^I_{AB}$ and $\mathbb{S}^I_{AB}$; the kernels of the operators $\overadd{\frac{n + 1}{2}+j}{\psi}\ofP$ for $j\geq 1 $ are spanned by the spherical harmonic tensors $\mathbb{V}^I_{AB}$ and $\mathbb{S}^I_{AB}$, and spherical harmonic tensors $\mathbb{T}^I_{AB}$ of modes $\ell\leq j+1$.
\end{itemize}
\end{proposition}
\ptcheck{30XI23; the whole section up to here}

\subsection{$\myhatopL$}
\label{ss14VI23.1}
\index{L@$\myhatopL$}%
\index{kernel!$\myhatopL$}%

We now consider the gauge operators $\myhatopL$ appearing in column $5$ of Table \ref{T6V23.1}, obtaining their kernels and the kernels of their adjoints.

On $S^{n-1}$ the kernel of the operator $\hLop_n$ of  \eqref{24IV23.1xs} is spanned by $\ell=0,1$ spherical harmonic functions. This follows from the negativity of the operator $P$ and the fact that $\overadd{\frac{n-5}{2}}{\psi}\ofP$ has trivial kernel by Proposition~\ref{P28XI23.1}.

For the adjoint operator, again if $n\ne 5$ we have
$$\hLopdagger  = c_n
\zdivone \circ \zdivtwo \circ \bigg( (n-1) P - 2 (n-2)^2 \bigg) \overadd{\frac{n-5}2 }{\psi}\ofP   \,,$$
for some constant $c_n\ne 0$.
For the tensor fields $  \mathbb{S}^I_{AB} = \frac{1}{k^2}\TS[\zspaceD_A\zspaceD_B \mathbb{S}^I]$, we have, after making use of \eqref{8VI23.1},
\begin{align}
    \zdivone \circ \zdivtwo \mathbb{S}_{AB}^I
    & =
    \frac{(k^2-n+1)(n-2)}{n-1} \mathbb{S}^I
    \,,
\end{align}
the right-hand side of which vanishes for $k^2 = n -1$, corresponding to $\ell=1$. We also have
\begin{equation}
    \zdivone \circ \zdivtwo \mathbb{V}_{AB}^I =0 =
    \zdivone \circ \zdivtwo \mathbb{T}_{AB}^I \,.
\end{equation}
We conclude that the kernel of $\hLopdagger$ is spanned by the tensors $\mathbb{S}_{AB}^I$ of mode $\ell = 1$, as well as $\mathbb{V}_{AB}^I$ and $\mathbb{T}_{AB}^I$ of all modes.
 \ptcheck{30XI23}

Next, using
 \eqref{5XI23.41}-\eqref{30VI.1} the operator $\Lop$ appearing in \eqref{24IV23.3} can be rewritten as
 \ptcheck{8XII23}
\begin{align}
    \Lop(\xi)_{AB}
   & =  \frac{\overadd{\frac{n-5}{2}}{\psi}\ofP}{(n-1)(n-5)}
     \bigg(
     (\TSzlap - 2 (n - 1) ) (\TSzlap  - (n+1) )
     C(\xi)_{AB}
     \nonumber
     \\
     & \qquad\qquad\qquad
     -\frac{2 (n-3)}{n-2} (\TSzlap + 3 -2 n  )
     \TS[\zspaceD_A\zspaceD_B\zspaceD_C\xi^C]
     \bigg) \,.
    \label{6VI.1}
\end{align}
Letting $\xi_A = \frac{1}{k^2} \zspaceD_A \mathbb{S}^I$ in \eqref{6VI.1} gives
\begin{align}
     \frac{1}{k^2} \Lop( \zspaceD_C \mathbb{S}^I)_{AB}
     =
     - \overadd{\frac{n-5}{2}}{\psi}\ofP
     \frac{\ell (l+1) (n-4) (\ell+n-3) (\ell+n-2)}{(n-5) (n-2) (n-1)}\mathbb{S}^I_{AB}\,,
\end{align}
the right-hand side of which vanishes for $\ell=0$, and is non-vanishing for $\ell>0$ when $n>3$ is odd.
Letting $\xi_A = -\frac{1}{k_V}  \mathbb{V}^I_A$ in \eqref{6VI.1} gives
\begin{align}
     -\frac{1}{k_V}  \Lop(\mathbb{V}^I)_{AB}
     =
     \overadd{\frac{n-5}{2}}{\psi}\ofP \, \frac{\ell (\ell+1) (\ell+n-3) (\ell+n-2)}{(n-5) (n-1)} \mathbb{V}^I_{AB}\,,
\end{align}
the right-hand side of which is non-vanishing for $\ell>0$. We conclude that the kernel of $\Lop$ consists of conformal Killing vectors.

\FGp{18XII}
In order to determine the adjoint operator $\Lndagger = \underline{\Lndagger} \circ \overadd{\frac{n-5}{2}}{\psi}\ofP$,
the notation of Appendix~\ref{App8XII23.1} together with  the commutation relations therein are useful; one finds
\begin{eqnarray}\label{9XII23.1}
    &
      (\TSzlap - 2 (n - 1) ) (\TSzlap  - (n+1) )
     C(\xi)_{AB} = \TS
     \Big[
     \zspaceD_A  (\zTSlap_V + n-3 ) \zTSlap_V \xi_B
     \big]
     \,,
     &
\\
 &
  (\TSzlap + 3 -2 n  )
     \TS[\zspaceD_A\zspaceD_B\zspaceD_C\xi^C]
      = \TS\big[
           \zspaceD_A \zspaceD_B \zspaceD_C (\zTSlap_V + 1) \xi^C
           \big]
           \,,
\end{eqnarray}
}
which immediately leads to
\begin{align}
    \underline{\Lndagger}(h)_C
     &  =  \frac{2(n-3)(\zTSlap_V +1)}{(n-5) (n-2) (n-1)} \zspaceD_C \zspaceD^A\zspaceD^B h_{AB}
    - \frac{ \zTSlap_V(\zTSlap_V + n-3 )}{(n-5) (n-1)} \zspaceD^A h_{AC}
    \nonumber
\\
  &  =  \frac{2(n-3)(\TSzlap+3-n)}{(n-5) (n-2) (n-1)} \zspaceD_C \zspaceD^A\zspaceD^B h_{AB}
    - \frac{(\TSzlap-(n-2)) (\TSzlap -1)}{(n-5) (n-1)} \zspaceD^A h_{AC}     \,.
    \label{8VI23.4}
\end{align}
Clearly, $\red{h^{\red[TT]}}$ belongs to the kernel of $\underline{\Lndagger}$.

Setting $h_{AB} = \mathbb{S}^I_{AB}$ we have, after making use of \eqref{8VI23.1}-\eqref{8VI23.3},
\begin{align}
     \underline{\Lndagger}(\mathbb{S}^I)_C
     &=
     -\frac{(n-4) \left(k^2-n+1\right) \left(k^2+n-3\right)}{(n-5) (n-1)^2}
     \zspaceD_A \mathbb{S}^I
     \,,
\end{align}
the right-hand side of which is only vanishing for $k^2 = n-1$ for $n$'s of current interest, which corresponds to the mode $\ell=1$.

Setting, next, setting $h_{AB} = \mathbb{V}_{AB}^I$ in \eqref{8VI23.4} and making use of \eqref{8VI23.2} gives
\begin{align}
    \underline{\Lndagger}(\mathbb{V}^I)_C
    & = \frac{\left(2- k_V^2\right) \left(k_V^2-n+1\right) \left(k_V^2-n+2\right)}{2 k_V (n-5) (n-1)} \mathbb{V}^I_C\,,
\end{align}
the right-hand side of which vanishes for $k_V^2 = 2,n-1, n-2$, the last of which corresponding to the mode $\ell = 1$. 
We conclude that the kernel of $ \underline{\Lndagger}$, and hence of $\Lndagger$, is spanned by $\TTt$ tensors, 
together with $\mathbb{V}^I_{AB}$'s and $\mathbb{S}^I_{AB}$'s of mode $\ell=1$.
\ptcheck{9XII23;the whole section up to here against mathematica file SimplfyingL.nb}

\section{The matrices $\Lambda^{[X]}_k$}
\subsection{Ellipticity}
\label{ss24IX23.2}
 \ptcheck{28IX23, read the whole section without checking the calculations; and then the calculations as well later}

\begin{figure}
\centering
\begin{tikzpicture}
  \matrix (m)[
    matrix of math nodes,
    nodes in empty cells,
    minimum width=width("999988"),
    minimum height=8mm,
  ] {
        & \partial_u^{k_n}\kxi{2}_A &\dots & \partial_u\kxi{2}_A & \kxi{2}_A  & \upsilon_{k,2} & \TTtpvec{\frac{9-n}{2}}  & \dots     & \TTtpvec{1} & \TTtpvec{5}  & \dots   & \TTtpvec{4+k} \\
    \overadd{k}{r} & \frac{n+1}{2}  & & & k & k & &   &  & & & D_1  \\
    \vdots & &\ddots & &  \vdots   &  \vdots  &    &       &     & & &   \\
    \overadd{\frac{n+1}{2}}{r} & C   & & \frac{n+1}{2}  & \frac{n+1}{2}    &\frac{n+1}{2}  &    &       &     & & & \\
    \overadd{\frac{n-1}{2}} r  & B & & & \hlight{\frac{n+1}{2}}   & \frac{n-1}{2}  &  0 &  &  & & & D_2  \\
    \overadd{\frac{n-3}{2}}{r} & & & & \frac{n-1}{2}   & \hlight{\frac{n-1}{2} } &  &  &  & & &A_2  \\
    \overadd{1}{r}  &  & &  & 2  &  2 &\hlight{2} &    &   & & &A_1 \\
    \vdots & & &  & \vdots & \vdots  &  & \ddots & & & & & \\
    \overadd{\frac{n-5}{2}}{r}  & & & &  \frac{n-3}{2}  & \frac{n-3}{2}  &  &  &  \hlight{\frac{n-3}{2}}   & & &  \\
    \overadd{1}{s} & & &  & 2    &  2  &    &   &     & \hlight{2} & &  \\
    \vdots  & & & &  \vdots   &  \vdots  &    &   &     & & \ddots &  \\
    \overadd{k}{s} & & & \phantom{k+1}  &  k+1   & k+1   &    &       &     & &   & \hlight{k+1} \\
  } ;

    \draw (-7.8,5.5) -- (7,5.5);
  \draw (-6.5,-7) -- (-6.5,6.5);
  \draw[rounded corners] (m-2-12.north east) rectangle (m-5-7.north west); 
  \draw[rounded corners] (m-5-12.south east) rectangle (m-5-8.north west); 
  \draw[rounded corners] (m-6-12.south east) rectangle (m-6-7.north west);
  \draw[rounded corners] (m-7-7.north west) rectangle (m-12-12.south east) ;
  \draw[rounded corners] (m-5-2.north west) rectangle (m-12-4.south east) ; 
  \draw[rounded corners] (m-2-2.north west) rectangle (m-4-4.south east) ; 
\end{tikzpicture}
\caption{The matrix
$\Lambda^{[S]}_k$, $k>\frac{n-1}{2}$, for $n>5$; the remaining matrices $\Lambda^{[X]}_k$ have similar structure. The numbers $x$ appearing in the entries indicate that the operator appearing at the associated position is of   order  $2x$. The alphabets $A_1, A_2,B, C, D_1,D_2$ are used to denote the various \red{block}s of the matrix for easier reference in the text. In each row the highlighted number equals $s_i/2$, with the $t_j$'s  given by \eqref{15IX23.3}. See Figure \ref{f10IX23.2} for the submatrix denoted $C$.}
\label{f10IX23.1}
\end{figure}

\begin{figure}
\centering
\begin{tikzpicture}
  \matrix (m)[
    matrix of math nodes,
    nodes in empty cells,
    minimum width=width("999988"),
    minimum height=8mm,
  ] {
    & t_j/2 & -k_n & -(k_n-1) & \dots & -3 & -2 & -1 \\
    s_i/2 && \partial_u^{k_n}\kxi{2}_A & \partial_u^{k_n-1}\kxi{2}_A  & \dots   & \partial_u^{3}\kxi{2}_A & \partial_u^{2}\kxi{2}_A  & \partial_u\kxi{2}_A  \\
    k+1 & \overadd{k}{r}    & \frac{n+1}{2} & \frac{n+1}{2}   & \dots & k-3 & k-2 & k-1  \\
    k & \overadd{k-1}{r}    & \frac{n-1}{2} & \frac{n+1}{2}   & \dots & k-4 & k-3 & k-2  \\
    \vdots & \vdots   & \vdots & \vdots  & \ddots  & \vdots & \vdots & \vdots  \\
    \frac{n+7}{2} & \overadd{\frac{n+5}{2}}{r}    & 0
    & 0 & \dots & \frac{n+1}{2} & \frac{n+1}{2} & \frac{n+3}{2}\\
    \frac{n+5}{2} & \overadd{\frac{n+3}{2}}{r}    & 0
    & 0 & \dots & \frac{n-1}{2} & \frac{n+1}{2} & \frac{n+1}{2} \\
    \frac{n+3}{2} & \overadd{\frac{n+1}{2}}{r}    & 0 
    & 0 & \dots & \frac{n-3}{2} & \frac{n-1}{2} &\frac{n+1}{2} \\
  } ;
 \draw (-4,2.2) -- (5.2,2.2);
 \draw (m-2-2.north east) -- (m-8-2.south east);
\end{tikzpicture}
\caption{Submatrix of $\Lambda_k$ corresponding to block $C$ of figure \ref{f10IX23.1}. A choice of $\frac{s_i}{2}$ and $\frac{t_j}{2}$'s required to show ellipticity of $\Lambda_k$ is also indicated in the first column and first row respectively. For convenience and without loss of generality, we have assumed a sufficiently large $k$ ($k > n-1$) here.
}
\label{f10IX23.2}
\end{figure}

In this appendix we show that for each $X\in\{S,V\}$, the matrix of operators $\Lambda^{[X]}_k$ of \eqref{4IX23.1} is elliptic in the sense of Agmon, Douglis and Nirenberg.

First, we recall the definition of ellipticity in the sense of Agmon, Douglis and Nirenberg (ADN). Consider a differential system on $\secN$ of the form
\begin{align}
    \Lambda(\zspaceD) v(x) = u(x)
    \label{9IV24.1}
\end{align}
where $v(x)$ and $u(x)$ are $N$-vector valued functions on $\secN$ and $\Lambda$ is a $N\times N$ matrix of linear partial differential operators, with each $(\Lambda(\zspaceD))^i{}_j$ being a polynomial in $\zspaceD_A$. Here and below, the notation $(A\red{)^i{}_j}$ denotes the $ij$-entry of a matrix $A$. The system \eqref{9IV24.1} is said to be \textit{elliptic}  if there exists integers $s_i, t_j$ with $i,j \in [1,N]$, such that the order of the operator $(\Lambda)^i{}_j$ does not exceed $s_i+t_j$ and furthermore
\begin{align}
    \det \Lambda^0(k) \neq 0 \ \text{for all} \ k \neq 0 \ \text{vectors}\,.
    \label{9IV24.2}
\end{align}
In \eqref{9IV24.2}, $\Lambda^0$ denotes the matrix obtained from $\Lambda$ by keeping only those operators of order exactly $s_i+t_j$ in each $(\Lambda)^i{}_j$ entry.

\index{ellipticity!Lambda@$\Lambda^{[X]}_k$|(}%
\index{Lambda@$\Lambda_k$!ellipticity|(}%
To show that the matrix of operators $\Lambda^{[X]}_k$ of \eqref{4IX23.1} is elliptic in the sense just defined we choose $s_i \,, t_j$ as (see also Figure~\ref{f10IX23.1} and \ref{f10IX23.2})
 \begin{align}
     (s_i) &=
     \begin{cases}
          2 \times (\frac{n-1}{2}\,, 2\,, 3\,, \dots \,, \frac{n-3}{2} \,, 2 \,, 3\,, \dots \,, k+1) \,, & k= \frac{n-3}{2} \,,
          \\
           2 \times (k+1\,, k\,, k-1\,, ... \,, \frac{n-1}{2}\,, 2\,, 3\,, \dots \,, \frac{n-3}{2} \,, 2 \,, 3\,, \dots \,, k+1) \,, & k \geq \frac{n-1}{2} \,,
     \end{cases}
     \label{15IX23.2}
     \\
     (t_j) & =
     \begin{cases}
          (0 \,, \dots \,, 0 ) \,, & k=\frac{n-3}{2}
          \\
          2 \times ( \underbrace{ -k_n \,, -(k_n-1) \,, \dots\,, -2, -1 }_{k_n \text{ terms}} \,, \underbrace{0 \,, \dots \,, 0 }_{2k-k_n \text{ terms}}) \,, & k \geq \frac{n-1}{2}
           \,,
     \end{cases}
     \label{15IX23.3}
 \end{align}
 (recall that $k_n$ has been defined in \eqref{2X23.1}).
 Let $ \Lambda_k^{[0,X]}$ denote the matrix obtained by
  keeping in $(\red{\Lambda_k^{[X]}}\red{)^i{}_j}$
 only those operators which are of order $s_i + t_j$.
We wish to show
that  $ \Lambda_k^{[0,X]}$ is a lower triangular  matrix,
 with the diagonal entries being
 \FGp{14XI}
 \begin{align}
 \begin{cases}
      \big(
     d_2^X \,,
     (\overadd{1}{\tilde\psi}
     \,,
     \overadd{2}{\tilde\psi}
     \,,
 \dots \,,
 \overadd{\frac{n-5}{2}}{\tilde\psi}
 \,,
 \overadd{1}{\tilde\chi}
 \,,
  \overadd{2}{\tilde\chi}
  \,,
 \dots \,
  \overadd{k}{\tilde\chi}
  )\circ \zdivtwo \circ\, C
 \big)
   \,,
 \\
  \big(  \underbrace{\zdivtwo \circ \Lop \,,
 \dots \,,
 \zdivtwo \circ \Lop }_{k_n \text{ terms}}\,,
     \zdivtwo\circ  {\Lop} \,,
     d_2^X \,,
     (\overadd{1}{\tilde\psi}
     \,,
     \overadd{2}{\tilde\psi}
     \,,
 \dots \,,
 \overadd{\frac{n-5}{2}}{\tilde\psi}
 \,,
 \overadd{1}{\tilde\chi}
 \,,
  \overadd{2}{\tilde\chi}
  \,,
 \dots \,
  \overadd{k}{\tilde\chi}
  )\circ \zdivtwo \circ \, C \,,
 \big)  \,,
 \end{cases}
 \label{9IX23.11}
 \end{align}
with the first case for $k=\frac{n-3}{2}$ and the second case for $k\geq\frac{n-1}{2}$, and with
\FGp{14XI}
 \begin{equation}
    d_2^X :=
    \begin{cases}
        \zdivtwo\circ  \underline{\hLop}_n \,, & X= S \,,
        \\
        \overadd{\frac{n-3}{2}}{\tilde\psi}\circ \zdivtwo \circ\, C \,, & X=V\,,
    \end{cases}
 \label{9IX23.12}
 \end{equation}
where the operators $\overadd{p}{\tilde\psi} = \overadd{p}{\tilde\psi}\ofDC$,  $\overadd{p}{\tilde\chi} = \overadd{p}{\tilde\chi}\ofDC$ and $\underline{\hLop}_n\,, \Lop$ are the gauge operators defined in Section \ref{ss9IX23.1}. We emphasise that the operators appearing in \eqref{9IX23.11} are all elliptic (see  Appendix \ref{App12XI22.1h}).

By definition of $u_k$ and $\upsilon_k$, it is clear that the terms in \eqref{9IX23.11} will appear in the respective spots on the diagonal of $ \Lambda_k^{[0,X]}$. Indeed, this follows directly from the transport equations and the gauge transformation of $\overadd{p}{q}_{AB}$ for $p \geq\frac{n-3}{2}$. It remains to show that $ \Lambda_k^{[0,X]}$ is lower-triangular.

First, consider the submatrix of $\Lambda^{[X]}_{k}$ denoted as $A_1$ in Figure \ref{f10IX23.1}. For these entries, the analysis is similar to the convenient case -- in each row, the operators involved are either $\overadd{i}{\tilde\psi}_n \circ \zdivtwo \circ C$ or $\overadd{i}{\tilde\chi}_n\circ \zdivtwo \circ C$ for some $n\in\{\frac{9-n}{2},4+k\}$. By the same reasoning as items (ii) of Section \pref{e3IX23.1} and Section \pref{e3IX23.2}, the operator of highest order amongst the $\overadd{i}{\tilde\psi}_n $'s, resp. $\overadd{i}{\tilde\chi}_n $'s, is
\begin{equation}
    \overadd{i}{\tilde\psi}\ofDC\,, \quad \text{resp.} \quad \overadd{i}{\tilde\chi}\ofDC \,,
    \label{10IX23.9}
\end{equation}
which is of order $2i$; all other $\overadd{i}{\tilde\psi}_n $ and $\overadd{i}{\tilde\chi}_n $ operators are of lower order. We note for later use that this is also true for the block $A_2$ in Figure~\ref{f10IX23.1}.
The above justifies that the submatrix $A_1$ of $\Lambda^{[0,X]}_{k}$ is lower triangular.
\FGp{14XI}

Next, we consider the operators acting on the various ``gauge fields'' appearing in \eqref{9IX23.w1}. Most 
gauge transformations of the fields $\zspaceD^A\overadd{i}{q}_{AB}$. From \eqref{6III23.w8b}, we have
\begin{align}
    \zspaceD^A\overadd{i}{q}_{AB}&= \partial_u \zspaceD^A\overadd{i-1}{q}_{AB}
    - \underbrace{\overadd{i-1}{\tilde\psi}\ofDC
     }_{
     { \text{order $2(i-1)$}}}
       \, \zspaceD^A\qh_{AB}^{(n-5-2i)/2}
    - m^{i-1} \underbrace{\overadd{i-1}{\psi}_{\red{[m]}}}_{
     { \text{order $0$}}
     }
     \zspaceD^A\qh^{(\frac{n-7-2(i-1)(n-1)}{2})}_{AB}
    \nonumber
    \\
    &\quad
    - \sum_{j=1}^{p-1}  m^{j}
    \underbrace{\overadd{i-1}{\tilde\psi}_{j,0}\ofDC}_{
     { \text{order $2(i-1-j)$}}
     } \, \zspaceD^A\qh^{(\frac{1}{2} (n-7- 2 (i-1) - 2 j (n-2)))}_{AB}
    \, ,
     \quad 1 \leq i \leq k \,.
    \label{10IX23.1}
\end{align}
Using \peqref{7III23.w1} and \peqref{23IV23.1} it can be shown inductively
 that for each $1\leq i \leq \frac{n-3}{2}$, in the gauge transformation of $\zspaceD^A\overadd{i}{q}_{AB}$,

\begin{enumerate}
    \item[1.]   the operators acting on the gauge fields $\zspaceD_A \xi{}^u$ and $\xi_A$ are both
     {\text{of order $2(i+1)$}},
         \ptcheck{29IX23 }
         \FGp{14XI}
    \item[2.] 
    the order of the operators acting on $\partial^j_u\xi_A$, $1\leq j\leq i-1$, is { $2(i-j-1)$}
(relevant for region $B$ in Figure~\ref{f10IX23.1}).
    \ptcheck{29IXVI}
\end{enumerate}
To see this, note that the gauge transformation law of $\zspaceD^A\qh^{(n)}_{AB}$, given in \eqref{23IV23.1}, involves only the gauge fields $\zspaceD_A\xi^u$, $\xi_A$ and $\partial_u \xi_A$. The associated operators acting on these fields are of orders four, four and two
respectively. Thus, modulo the first term in \eqref{10IX23.1}, the operators of the highest order in the gauge transformation of $\zspaceD^A\overadd{i}{q}_{AB}$ must come from the second term in \eqref{10IX23.1}, and are given by that acting on $\zspaceD_A\xi^u$ and $\xi_A$. Furthermore, they are of order $2(i+1)$.
Next, the fact that   the operators acting on $\partial^j_u\xi_A$, $1\leq j\leq i-1$, are of order { $2(i-(j-1))$}
follows easily by induction; we leave the details to the reader, after using  \eqref{7III23.w1} for initialisation.
\seccheck{21XII}
Concerning the first line of region $B$ of Figure~\ref{f10IX23.1},  we note that the part of items 1.\ and 2.\ above concerning $\xi_A$ and its $u$-derivatives continues to hold for the gauge transformation of $\overadd{\frac{n-1}{2}}{q}_{AB}$.
 By a similar analysis for the gauge transformation of the $\overadd{i}{\Hf}_{uA}$ fields,
 using the recursion formula \eqref{6III23.w4}, and \eqref{30IX23.1} and \eqref{28II23.1} for initialisation, we have:
\ptcheck{ 29IX }\FGp{14XI}%
for each $1\leq i \leq k$,  the operators acting on the gauge fields $\zspaceD_A \xi{}^u$ and $\xi_A$ are
     {\text{of order $2(i+1)$}},
while those acting on $\partial^j_u\xi_A$, $1\leq j\leq i+1$, are of order
     { $2(i-(j-1))$}. 
Thus, in the  region $B$, the order of the operator  appearing at an $(ij)$-entry is lower than or equal to $s_i + t_j$. 

For the entry immediately to the left of $D_2$, we have $(\Lambda^{[X]}_k\red{)^{k_n+1}{}_{k_n+3}}=0$, which follows from \eqref{22VIII23.2c} with $p=\frac{n-1}{2}$.
Next,
we also have $(\Lambda^{[V]}_k\red{)^{k_n + 1}{}_{k_n+2}}=0$, while $(\Lambda^{[S]}_k\red{)^{k_n+1}{}_{k_n+2}}$
     { \text{is of order $n-1$}}.
The first equality follows from the fact that the field $\kphi{2}_{AB}$ does not appear in the integrated transport equation \eqref{9IX23.w1} of $\overadd{\frac{n-1}{2}}{q}_{AB}$. The second statement follows from the following calculation of the gauge term associated to $\zspaceD_A\xi^u$ in the gauge transformation of $\overadd{\frac{n-1}{2}}{q}_{AB}$:
 Namely, using \eqref{23IV23.1},  \eqref{10IX23.1},  and what has been said so far, this term is of the form
\begin{align}
    -\frac{2}{(n-2)r}
    &
    \underbrace{ \overadd{\frac{n-3}{2}}{ \psi}\ofP  \circ P \circ\, C(\zspaceD_A\xi^u)}_{=0}
    \nonumber
\\
 &     +
    { \text{an operator of order $(n-1)$ acting on}}
    \ \zspaceD_A\xi^u
    \,,
    \label{15IX23.1}
\end{align}
 where  the second term comes from the last term on the right-hand side of \eqref{10IX23.1} with $i=\frac{n-3}{2}$ and $j=1$.
The first term vanishes by the fact that
\begin{equation}
    \overadd{\frac{n-3}{2}}{ \psi}\ofP h^{\red[S]} = 0 \,,
\end{equation}
see \eqref{29VIII23.f3}.
 \ptcheck{3X23}
 \FGp{14XI}

That the operators appearing in region $D_2$ in Figure~\ref{f10IX23.1} are all of orders $\leq \frac{n-1}{2}$ follows from the same argument as that around \eqref{10IX23.9} (for regions $A_1$ and $A_2$). For the region $D_1$, the orders of the operators appearing in each row are lower than or equal to $2i$, since the would-be-of-highest-order-$2(i+1)$ operator $\overadd{i}{\tilde\psi}\circ \zdivtwo \circ C $ vanishes for  $i>\frac{n-1}{2}$
by \eqref{5X23f.1m}, Proposition \ref{P9X23.1m}.

Let us turn our attention now to the first $k_n$-lines of Figure~\ref{f10IX23.1} (region $C$).
The vanishing of the first term in \eqref{15IX23.1} modifies the order of the operators appearing in the gauge transformation of $\overadd{i}{q}_{AB}$, $i \geq \frac{n+1}{2}$ ,
as compared to point 2.\ below \eqref{10IX23.1}. Indeed, restarting the induction from $\overadd{\frac{n-1}{2}}{q}_{AB}$, it follows from \eqref{10IX23.1} (where now the second term at the right-hand side vanishes) that
\begin{enumerate}
    \item[3.]  the operators acting on the gauge fields $\zspaceD_A\xi^u$ and $\xi_A$ are both
     {\text{of order $2i$}};
    \item[4.] that acting on $\partial_u^j \xi_A$ is
     {\text{of order $2(i-j)$}}
for $1 \leq j \leq i-\frac{n+1}{2}$, and
     {\text{of order $2\big(i-(j-1)\big)$}}
for $i-\frac{n-1}{2} \leq j \leq i-1$.
 \ptcheck{4X23}
 \FGp{14XI}
\end{enumerate}
One readily verifies that the choice \eqref{15IX23.2} and \eqref{15IX23.3} results in a lower triangular matrix block $C$ of
Figure~\ref{f10IX23.1}. See also Figure~\ref{f10IX23.2}.
 \ptcheck{4X23}
 
Next, we consider the ``gauge fields'' appearing in the integrated transport equations \eqref{9IX23.w1} for $\overadd{i}{r}$ and \eqref{9IX23.w2f} for $\overadd{i}{s}$, arising from the gauge corrections  \eqref{16III22.2old} at $r_2$ of  the field $\tilde h_{AB}$. These involve only the gauge fields $\kxi{2}_A$ and $\zspaceD_A\kxi{2}^u$. The highest order operators acting on these fields are (cf.\ \eqref{6III23.w6c} and \eqref{11III23.1}), 
up to multiplicative constants,
\begin{align}
   \overadd{i}{\tilde\psi}\circ \zdivtwo\circ C
   \quad
   \text{and}
   \quad
   \overadd{i}{\tilde\chi}\circ \zdivtwo\circ C
\end{align}
respectively for $\overadd{i}{r}$, $1\leq i\leq \frac{n-3}{2}$, and $\overadd{i}{s}$, $1\leq i\leq k$, both of  order $2(i+1)$. For $\overadd{i}{r}$, $ i > \frac{n-3}{2}$ the highest order operators are
\begin{align}
   \overadd{i}{\tilde\psi}_{1,0} \circ \zdivtwo\circ C \,,
\end{align}
of order $2i$. 

Finally, we come to the row associated to the field $\overadd{\frac{n-3}{2}}{r}$. Specifically, for the term on the diagonal, it follows directly from the transport equation \eqref{9IX23.w1} with $i= \frac{n-3}{2}$ 
that this term is of order $ {n-1} $ for 
$X=V$. This remains true for $X=S$
(cf.\ previous paragraph item 1 below \eqref{10IX23.1}).
 \ptcheck{4X23}
\FGp{20XI}

We now justify the first $k_n+1$ entries of equation \eqref{9IX23.11}. The analysis in Section \ref{ss9IX23.1} of the main text shows that in the gauge transformation formulae
of $\overadd{\frac{n-1}{2}+j}{r}$ for $j=0,...,k_n$, respectively $\overadd{\frac{n-3}{2}}{r}$, the   operators $\zdivtwo\circ\Lop$ acting on the field $\partial^j_u\kxi{2}^A$,  resp.  $\zdivtwo \circ \underline{\hLop}_n$ acting on $\zspaceD_A\kxi{2}^u$,
appear on the diagonal  of $\Lambda_k$.
It is not difficult to see that these are
 the only operators of order  $ {n+1} $ acting on $\partial^j_u\kxi{2}^A$, respectively of order $ {n-1} $ acting on $\zspaceD_A\kxi{2}^u$.
Indeed, a comparison of length-dimensions of the fields $\partial^j_u\kxi{2}^A$ and $\overadd{\frac{n-1}{2}+j}{r}$, resp. $\zspaceD_A\kxi{2}^u$ and $\overadd{\frac{n-3}{2}}{r}$, shows that
the gauge operators must be dimensionless. Thus each term contributing to these operators must  a)  either be  $r$-independent, in which case we obtain
 $\zdivtwo\circ\underline{\hLop}_n$, of order ${n+1} $, for
 $\overadd{\frac{n-1}{2}+j}{r}$ and $\zdivtwo\circ\hLop_n $, of order ${n-1} $,  for $\overadd{\frac{n-3}{2}}{r}$; or
   b)
   come with the dimensionless combination $m^i\alpha^{2\ell}r^{-i(n-2)+2\ell}$ for some $0\leq i,\ell$, and $i+\ell\geq 1$; the analysis involving $\alpha$ is included for later reference in Appendix \ref{ss18XI23.1}. However, looking at \eqref{5III23.5}-\eqref{5III23.5b} and \eqref{23IV23.1} we see that each factor of $m$ or $\alpha^2$ will decrease the order of the operator by two.
   It ensues that a gauge operator associated with the field $\overadd{\frac{n-1}{2}+j}{r}$, respectively with the field $\overadd{\frac{n-3}{2}}{r}$,  and containing a prefactor $  m^i\alpha^{2\ell} $  is of order ${n+1} -2i-2\ell $, respectively $ {n-1} -2i-2\ell$. This justifies the first $k_n+1$ entries of equation \eqref{9IX23.11}.

It follows that $ \Lambda_k^{[0,X]}$ is lower triangular, with elliptic operators lying on the diagonal.

\subsection{Mode solvability}
 \label{ss30XI23.2}
 In this section, we explain the mode solvability of the system~\eqref{4IX23.1}.
Suppose that all entries of $u_k^{[X]} $, $\upsilon_k^{[X]} $ and $\rho_k^{[X]}$ are eigenvectors of the Laplace operator with the same eigenvalue $-\lambda_\ell\ge 0 $. We then have
$$
  \red{\Lambda_k^{[X]}} \upsilon_k^{[X]} = B_{k,\ell}^{[X]} \upsilon_k^{[X]}
  \,,
  \qquad
  \mathcal{N}^{[X]} _k \, \rho_k ^{[X]} = D_{k,\ell}^{[X]} \rho_k^{[X]}
  \,,
$$
with  real-valued matrices $B_{k,\ell}^{[X]}$ and $D_{k,\ell}^{[X]}$. The $(ij)$ entries $(B_{k,\ell}^{[X]})^i{}_j$ of the matrix $B_{k,\ell}^{[X]}$ are polynomials in $\lambda_\ell$ of  order lower than or
equal to
$(s_i+t_j)/2$ (recall that all $s_i$'s and $t_j$'s are even), or vanish. The system
\begin{align}
    u_k^{[X]} = \Lambda _k^{[X]} \upsilon_k^{[X]} + \mathcal{N} _k^{[X]}  \rho_k^{[X]}
    \label{30XI23.1-}
\end{align}
becomes
\begin{align}
    u_k^{[X]} = B_{k,\ell}^{[X]} \upsilon_k^{[X]} + D_{k,\ell}^{[X]} \rho_k^{[X]}
     \,.
    \label{30XI23.1}
\end{align}
For $\lambda _\ell  \ne 0$ define  the (invertible) diagonal matrices $A_{\ell} $ and $C_{\ell} $ as
$$
   A_{\ell}  = \mbox{\rm diag} (\lambda_\ell ^{-s_i/2}  )
  \,,
  \quad
   C_{\ell}  = \mbox{\rm diag} (\lambda_\ell ^{-t_i/2}  )
   \,.
$$
The equation \eqref{30XI23.1} can be rewritten as
\begin{align}
     \check u^{[X]}:= A_\ell  u_k^{[X]} =
     \underbrace{A_{\ell}  B_{k,\ell}^{[X]} C_\ell
     }_{=:\check B^{[X]}}
     \underbrace{  C_\ell ^{-1}  \upsilon_k^{[X]}
      }_{=: \check \upsilon^{[X]}} +
 A_{\ell}    D_{k,\ell}^{[X]}  \rho_k^{[X]}
     \,.
    \label{30XI23.3}
\end{align}
The $(ij)$ entry of the matrix $ A_{\ell}  B_{k,\ell}^{[X]} C_{\ell} $ equals
$$
  (A_{\ell}  B_{k,\ell}^{[X]} C_{\ell} )^i{}_j = \lambda_{\ell}^{-s_i/2} (B_{k,\ell}^{[X]})^i{}_j \lambda_{\ell}^{-t_j/2}
   \,.
$$
Since  $|\lambda_{\ell}|\to_{\ell\to\infty}\infty$, from what has been said the matrix
$ A_\ell B_{k,\ell}^{[X]} C_\ell$ tends to a lower-triangular matrix with non-zero entries at the diagonal when $\ell$ tends to infinity. It follows that there exists $N$ so that,  for any given $\check u^{[X]} $ and $\rho_k^{[X]} $, the system \eqref{30XI23.3}
\begin{align}
     \check u^{[X]} = \check B^{[X]}  \check \upsilon^{[X]} +  A_\ell ^{[X]}  D_{k,\ell}^{[X]}  \rho_k^{[X]}
    \label{30XI23.4}
\end{align}
has a solution $\check \upsilon^{[X]}$ for  all $\ell\ge N$:
\begin{align}
      \check \upsilon^{[X]}  = (\check B^{[X]})^{-1} (\check u^{[X]} -  A_{\ell}    D_{k,\ell} ^{[X]} \rho_k^{[X]})
      \,.
    \label{30XI23.4b}
\end{align}
 Setting
\begin{align}
     u_k^{[X]} =  (A_{\ell} )^{-1} \check u^{[X]}
     \,,
     \quad  \upsilon_k^{[X]} = C_{\ell}   \check \upsilon^{[X]}
     \,,
    \label{30XI23.3b}
\end{align}
one obtains a solution of the original system \eqref{30IX23.1}.

\subsection{ADN estimates}
 \label{ss30XI23.1}
\index{ADN estimates|(} 
Writing
$  \red{(u_{k})_i}$ for the components of the vector $u_k$, similarly for $\upsilon_k$ and
  $\rho_k$, we set
\begin{equation}\label{28XI23.1}
  \|u_k\|_{{\mycal H}^\ell} = \sum_i \|\red{(u_{k})_i}\|_{H^{s-s_i +\ell}(\secN)}
    \,,
    \quad
  \|\upsilon_k\|_{{\mathcal H}^\ell} = \sum_j \|\red{(\upsilon_{k})_j}\|_{H^{s+t_j +\ell}(\secN)}
  \,,
  \quad
  s = \max\{s_i\}
   \,.
\end{equation}
The spaces ${\mycal H}^{\ell}$ and ${\mathcal H}^\ell$ are defined as completions of $C^\infty$ in the above norms.

From what has been said we infer:

\begin{theorem}
  \label{T30XI23.1} 
  Let  $(s_i,t_j)$ be as defined above,  and let $X\in \{S\cap\CKVp,V\cap\CKVp\}$ (cf.\ Section~\ref{app16X23.1}).
  For
  $\ell\ge  0$
  the maps
\begin{align}
  \big( {\mycal H}^\ell
  \cap X
  \big) \times \big(  C^\infty
  \cap X \big)
   \ni (\upsilon^{[X]}_k,\rho^{[X]}_k)\mapsto  \Lambda^{[X]} _k \upsilon^{[X]}_k  + \mathcal{N}^{[X]} _k  \rho^{[X]}_k  \in {\mathcal H}^\ell\cap X
    \label{4IX23.1app}
\end{align}
are surjective, and for every $k\ge 0$  there exists a constant $C$ so that the  Agmon, Douglis, Nirenberg elliptic estimates hold:%
\index{ADN estimates}
 \ptcheck{norms to check}
  \begin{equation}\label{27XI23.1}
    \|\upsilon^{[X]}_k\|_{{\mathcal H}^\ell} \le C \big(
      \|u^{[X]}_k\|_{{\mycal H}^\ell}   +\| \rho^{[X]}_k  \|_{{\mycal H}^\ell}
      + \|\upsilon^{[X]}_k\|_{L^2(\secN)}
      \big)
      \,.
  \end{equation}
\end{theorem}

\proof
Recall that the entries of the $\Lambda^{[X]}_k$-matrices are operators acting on  vector fields.  The ADN principal symbol of $\Lambda^{[X]}_k$  is  obtained by replacing each entry $(\Lambda_k^{[0,X]})^i{}_j$ of $ \Lambda_k^{[0,X]}$  by a square matrix, namely  the principal symbol of the   operator $(\Lambda_k^{[0,X]})^i{}_j$.

Now, the principal part of each of the operators acting on these blocks is proportional to some power of the Laplacian, in particular the principal symbol of each block lying on the diagonal is diagonal. Hence we obtain   a lower-triangular matrix, where the determinant of each block on the diagonal is non-zero for $k_A\ne 0$ by ellipticity of the Laplace operator. Ellipticity in the ADN sense follows. The estimate \eqref{27XI23.1} follows from~\cite[Theorem~C]{MorreyNirenberg}.
\qedskip
%
%
%
%
%
%
\index{ADN estimates|)}
 
 We note:
 
\begin{corollary}\label{C10V24}
Let $\kgamma\in\N\cup\{\infty\}$, $k\in \N$, suppose that $n$ is odd and let $k_n$ be as in \eqref{2X23.1}.
Under the hypotheses of Theorem~\ref{T30XI23.1}, assume moreover that $ k_{\gamma}-2k-1\ge 0$. Then
\begin{align}
&  \forall\  j =  1,\ \cdots, k_n: \ \partial_u^j \kxi{2}_A  \in H^{k_{\gamma} + 1 - 2j}(\secN)
 \,, \label{11IV24.0}
    \\
&    \xi^A \in H^{k_{\gamma}+1}(\secN) \,,\quad
    \xi^u \in H^{k_{\gamma}+2}(\secN)
     \,, \quad
     \kphi{p}{}^{[\TTtp]} \in H^{k_{\gamma}}(\secN) \,,
    \label{11IV24.2a}
\end{align}
for $p\in ([\frac{9-n}{2},1]\cup[5,4+k])\cap\Z$.
\end{corollary}

\proof
%
Recall that
\begin{equation}
\overadd{i}{r}, \overadd{i}{s} \in H^{k_{\gamma}-2i-1}\,.
\label{11IV24.2}
\end{equation}
Meanwhile, the associated $s_i$ for each field $\overadd{i}{r}$ or $\overadd{i}{s}$ appearing in $\Lambda_k$ equals $2(i+1)$ (cf. \eqref{15IX23.3}). Thus $s = \max {s_i} = 2(k+1)$. To make use of the estimate \eqref{27XI23.1}, we require $u_k \in {\mycal H}^\ell$, which would be true for finite $\kgamma$ if
$$ s-s_i + \ell \leq k_{\gamma}-2i-1 \implies \ell \leq k_{\gamma}-2k-1\,.$$
Thus in particular, we can take $\ell = k_{\gamma}-2k-1$. The estimate \eqref{27XI23.1} then gives
\begin{equation}
 (v_k)_j \in H^{s+t_j+\ell} =H^{k_{\gamma} + 1+t_j}
 \,.
 \label{11IV24.1}
\end{equation}
%
%
Finally, the choice of $t_j$'s (cf. \eqref{15IX23.2}) for each $\partial_u^j \kxi{2}_A$, $j\in[0,k_n]$, is $t_j = -2j$, while that for each $\TTtpvec{p}$ is $t_j = 0$. The regularity \eqref{11IV24.0}-\eqref{11IV24.2a} follows  after recalling that $\kphi{p}{}^{[\TTtp]} = C(\TTtpvec{p})$.
\qed
%

\index{ellipticity!Lambda@$\Lambda^{[X]}_k$|)}%
\index{Lambda@$\Lambda_k$!ellipticity|)}%

\section{$(n,k)$ inconvenient, $m \alpha\neq 0$}
\label{ss18XI23.1} 

The conditions for the interpolating fields $\kphi{i}_{AB}$ and gauge fields in the \red{case of inconvenient pairs $(n,k)$} and  $m \alpha \neq 0$ are very similar to that for $\alpha = 0$. For completeness we briefly explain the procedure required here; notations are as defined in, or analogous to those of  Section \ref{ss5IX23.1}, and will be used  without further comments.

To begin, for inconvenient pairs $(n,k)$ and  $m \alpha \neq 0$,
it follows from \eqref{11III23.2} (cf.\ also \eqref{6III23.w6}) that \eqref{9IX23.w1} gains additional terms arising from $\alpha$:%
\index{rp@$\overadd{p}{r}_B$}%
\begin{align}
    \underbrace{\zspaceD^A \overadd{p}{r}{}^{[\TTtp]}_{AB}
    }_{=:  \overadd{p}{r}_B}
    &=    \text{gauge fields } + \overadd{p}{\tilde\psi}\ofDC \circ \zdivtwo\circ\, C (\TTtpvec{\frac{7-n+2p}{2}})
    \nonumber
    \\
    &\quad
    +\alpha^{2p} \overadd{p}{\psi}_{[\alpha]} \zdivtwo\circ\, C (\TTtpvec{\frac{7-n-2p}{2}})
    + m^p \overadd{p}{\psi}_{[m]} \zdivtwo\circ\, C (\TTtpvec{\frac{7-n+2p(n-1)}{2}})
    \nonumber
    \\
    & \quad
    +
    \sum_{j,\ell}^{p_{**}}
     m^{j}
    \alpha^{2\ell}
    \overadd{p}{\tilde\psi}_{j,\ell}\red{\ofDC \circ \zdivtwo}\circ\, C ( \TTtpvec{ \frac{7-n+2p}{2} +  j (n-2) - 2 \ell})
    \,,
    \quad  1\leq p\leq k\,,
    \label{9IX23.w121}
\end{align}
with $\overadd{p}{\psi}_{[\alpha]} = 0$ for $p>\frac{n-1}{2}$
(which follows from \eqref{6III23.w5} and  \eqref{6III23.w9}),
and with the field $\overadd{p}{r}_{AB}$   defined analogously as \eqref{22VIII23.2c}.

 In what follows it will be relevant to keep track of the $\CKV$- and $\CKVp$-parts of the equation, so we note that  \eqref{9IX23.w121} has been obtained by applying $\zdivtwo$ to \eqref{11III23.2}, hence all the terms there are in $\CKVp$. Further note that both in  \eqref{9IX23.w121} and in the equations that follow, the $\CKV$-part of the $\TTtpvec{\cdot}$ fields drops out because
\begin{equation}\label{18XII23.1}
    \big(
     \zdivtwo\circ\, C
     \big)
      \big|_{\CKV} = 0
 \,.
\end{equation}

Similarly to the $\alpha=0$ case,  in order to simplify notation we group terms together and
 rewrite \eqref{9IX23.w121} as
  \ptcheck{9XII}
\begin{align}
    \overadd{p}{r}_B
    = \text{gauge fields}
    + \sum_{j=4-n}^{\frac{7-n+2p(n-1)}{2}} \overadd{p}{\tilde\psi}_j \circ \zdivtwo \circ\, C (\TTtpvec{j} )\,,
    \quad
    1 \leq p \leq k \,,
    \label{9IX23.w2}
\end{align}
with some operators $\overadd{p}{\tilde\psi}_j $. As easily seen from \eqref{9IX23.w121}, for each $1 \leq p \leq \frac{n-1}{2}$, we have
 \ptcheck{10XII}
\begin{equation}\label{15XII23.21}
\overadd{p}{\tilde\psi}_{\frac{7-n-2p}{2}} = \alpha^{2p} \overadd{p}{\psi}_{[\alpha]} \neq 0\,.
\end{equation}

Similarly, the integrated transport equation for $\overadd{p}{\Hf}{}^{[\CKVp]}_{uA}$, namely the $\CKVp$-projection of  \eqref{11III23.1},  gains additional terms arising from $\alpha$ (compare~\eqref{19VIII23.2b}), and reads:
 \begin{align}
      \overadd{p}{s}_{A}
      &=
       (\text{gauge fields})^{[\CKVp]} +
       \overadd{p}{\tilde\chi}\red{\ofDC \circ \zdivtwo}\circ\, C(\TTtpvec{p+4})
      \nonumber
\\
      &\quad
    +  m^p  \overadd{p}{\tilde\chi}_{[m]} \zdivtwo\circ\, C(\TTtpvec{p(n-1)+4})
     \nonumber
\\
    & \quad
    + \sum_{j,\ell}^{p_*}  m^{j} \alpha^{2\ell} \overadd{p}{\chi}_{j,\ell}\red{\ofDC \circ \zdivtwo}\circ\, C(\TTtpvec{(p + 4) + j (n - 2) - 2 \ell})\,,
    \label{16XI23.1}
 \end{align}
where $0\leq p \leq k$.
As above, we rewrite \eqref{16XI23.1} as:
 \ptcheck{9XII23 and 12XII23}
 \begin{align}
     \overadd{p}{s}_{A}
     &=
       (\text{gauge fields})^{[\CKVp]} +
    \sum_{j=4}^{p(n-1)+4}
    \overadd{p}{\tilde\chi}_{j} \circ  \zdivtwo\circ\, C(\TTtpvec{j})\,.
    \label{20VIII23.1b}
 \end{align}
 with some operators $\overadd{p}{\tilde\chi}_j $.
 \FGp{20XI} 
 Recall that the $\overadd{p}{s}_{A}$'s are in $\CKVp$ by definition.

 Before continuing, we summarise the main differences between the case here and that when $\alpha=0$:
 \begin{enumerate}
     \item the addition of the new $\alpha$ terms in \eqref{9IX23.w121} leads to the coupling of the equations for $\overadd{\frac{n-3}{2}}{r}_{A}$, which now has an $\alpha^{n-3} \overadd{\frac{n-3}2}{\psi}_{[\alpha]} \zdivtwo\circ\, C(\TTtpvec{5-n})$ term, with that for $\chi$ (cf. \eqref{7III23.1}). In addition, the equation for $\overadd{\frac{n-1}{2}}{r}_{A}$, which now has an $\alpha^{n-4} \overadd{\frac{n-4}2}{\psi}_{[\alpha]} \, \zdivtwo\circ\, C(\TTtpvec{4-n})$ term, is now coupled with that for $\overadd{*}{\Hf}{}^{[\CKVp]}_{uA}$ (cf. \eqref{7III23.4}).
          \ptcheck{9XII23}
     \item The occurrence of new $\alpha$ terms in \eqref{9IX23.w121} and \eqref{16XI23.1}, in particular of the non-vanishing terms involving the field $C(\TTtpvec{4})$, requires us to include the case $p=0$ to \eqref{16XI23.1}, whereas $p=0$ was decoupled in the $\alpha=0$ case.
 \end{enumerate}

In view of  point 1.\  above, we now rewrite the equations for the fields $\chi$ and $\overadd{*}{\Hf}_{uA}$ into a form that is compatible with \eqref{9IX23.w121} and \eqref{16XI23.1}. For the field $\chi$ (cf. \eqref{7III23.1}), we continue to use \eqref{7III23.3a}-\eqref{16XI23.3}
to take care of the projection $\chi^{[\im (\mrL)^\perp ]}$ .
  For the projection $\chi^{[\im (\mrL)]}$, we have the obvious inclusion $\TTt \subseteq \ker\mrL$. This, together with \eqref{28XI23.f2} and \eqref{12XII23.21} shows that
\FGp{20XI}
$$
V\oplus \TTt \subseteq \ker\mrL \implies \mrL(h_{AB})
=\mrL(h^{[S]}_{AB})
 \,.
$$
Furthermore, we have
$\zspaceD_A\chi^{[\im (\mrL)]} \in \CKVp$, since for any $\xi\in\CKV$ and any field $\phi \in \im\,\mrL$, thus $ \phi = \mrL h = \zspaceD^A\zspaceD^B h_{AB}$ for some tensor $h$, it holds that
\begin{equation}
    \int_{\secN} \xi^A \zspaceD_A \phi = \int_{\secN} \xi^A \zspaceD_A \zspaceD^C\zspaceD^D h_{CD} =  \int_{\secN} \TS[\zspaceD^C\zspaceD^D\zspaceD_A\xi^A] h_{CD} = 0\,,
\end{equation}
where the last equality follows the calculations in Lemma~\ref{L21X23.1} (cf.~\eqref{12XII23.22}).
Thus, taking the gradient of the projection of \eqref{7III23.1} onto $\im\,\mrL$ gives,
\index{rstar@$\overadd{*}{r}$}%
\begin{align}
   \underbrace{\zspaceD_A( (\chi|_{\secN_2} - \chi|_{{\secN}_1})^{[\im\,\mrL]} )}_{=:\overadd{*}{r}_A}
    = \text{gauge fields}
    + \frac{n-3}{n-1}  
    \hspace{-0.6cm}
    \underbrace{
    	 \zspaceD _A
     \big(\mrL\circ\, C(\TTtpvec{5-n}{}^{[S]})
     \big)
        }_{
        \big(
          (\TSzlap-(n-2)\myGauss) \circ \zdivtwo\circ\, C(\TTtpvec{5-n}{}^{[S]})
          \big)_A}
          \hspace{-0.6cm}\,.
    \label{17XI23.2}
\end{align}
As can be readily verified, we have $\im\,\mrL\subseteq(\ker \zdivone)^\perp$. Thus solving \eqref{17XI23.2} is equivalent to solving the transport equation for the projected field $(\chi^{[\im\,\mrL]})^{[(\ker \zdivone)^\perp]} = \chi^{[\im\,\mrL]}$ and completes the gluing of $\chi$.

Next, we move on to the field $\overadd{*}{\Hf}{}^{[\CKVp]}_{uA}$.
For this, we take the $\CKVp$ projection of \eqref{7III23.4}, which results in
\index{s@$\overadd{*}{s}$}%
\begin{align}
    \underbrace{(\overadd{*}{\Hf}_{uA}|_{\secN_2} - \overadd{*}{\Hf}_{uA}|_{\secN_1})^{[\CKVp]}}_{=: \overadd{*}{s}_A}
     =
      (\text{gauge fields})^{[\CKVp]}
        - (n-1) \  \zdivtwo\circ\, C(\TTtpvec{4-n})
                  \,.
                   \label{17XI23.1}
\end{align}
\FGp{20XI}

We are ready now to write the full coupled system in matrix form. Indeed,
 the system   \eqref{9IX23.w2}, \eqref{20VIII23.1b}, \eqref{17XI23.2} and \eqref{17XI23.1} can be written as (cf.\ \eqref{4IX23.1} in the $\alpha=0$ case)
\begin{align}
    u_k^{\red{[X]}} = A^{[X]}_k \Theta^{[X]}_k + \Lambda^{\red{[X]}}_k \upsilon_k^{\red{[X]}} + \mathcal{N}^{\red{[X]}}_k  \rho_k^{\red{[X]}}\,,
    \quad
    X \in \{V,S\}
     \,,
    \label{17XI23.3}
\end{align}
for some matrix of operators $A_k^{[X]}$, whose exact form is not important, and where $\rho_k$ is as given in \eqref{18XI23.1}; the vectors $u_k$ and $\upsilon_k$ now gain additional terms arising from $\overadd{*}{r}$, $\overadd{*}{s}$ and $\overadd{0}{s}$ (cf.\ \eqref{13XI23.21}-\eqref{13XI23.22}): for $n>5$,
\FGp{20XI}
\begin{align}
\label{17XI23.4}
u^{[X]}_k &:=  \left\{
           \begin{array}{ll}
(
                \overadd{\frac{n-3}{2}}{r} \,,
                \hspace{-0.5cm}
                \underbrace{\overadd{*}{r}}_{\text{for }X = S \text{ only}}
                \hspace{-0.5cm}
                \,,
                \overadd{1}{r} \,,
               \dots \,,
               \overadd{\frac{n-5}{2}}{r} \,,
               \overadd{0}{s}\,,
               \overadd{1}{s} \,,
               \dots \,,
               \overadd{k}{s}
           )^{T [X]},  
           &\hbox{$k=\frac{n-3}{2}$;}  \\
(               \underbrace{ \overadd{k}{r}  \,,
                 \overadd{k-1}{r}  \,,
                \dots \,,
                 \overadd{\frac{n+1}{2}}{r}}_{k_n \text{ terms}}
                \overadd{\frac{n-1}{2}}{r}\,,
                \overadd{\frac{n-3}{2}}{r} \,,
                \overadd{*}{s}\,,
                \hspace{-0.5cm}
                \underbrace{\overadd{*}{r}}_{\text{for }X = S \text{ only}}
                \hspace{-0.5cm}\,,
                \overadd{1}{r} \,,
               \dots \,,
               \overadd{\frac{n-5}{2}}{r} \,,
               \overadd{0}{s}\,,
               \overadd{1}{s} \,,
               \dots \,,
               \overadd{k}{s}
           )^{T [X]}, & \hbox{$k>\frac{n-3}{2}$,}
           \end{array}
         \right.
           \\
 \upsilon_{k}^{[X]} &:=
 \left\{
   \begin{array}{ll}
     (
     \upsilon_{k,2} \,,
     \hspace{-0.5cm}
      \underbrace{\TTtpvec{5-n}}_{\text{for }X = S \text{ only}}
      \hspace{-0.5cm}\,,
 \TTtpvec{\frac{9-n}{2}} \,,
 \dots \,,
 \TTtpvec{1}  \,,
  \TTtpvec{4}  \,,
 \TTtpvec{5}  \,,
 \dots \,
 \TTtpvec{4+k}
 )^{T [X]}, & \hbox{$k=\frac{n-3}{2}$;} \\
     (
 \text{gauge fields}\,,
      \kxi{2}_A \,,
     \upsilon_{k,2} \,,
     \TTtpvec{4-n} \,,
     \hspace{-0.5cm}
     \underbrace{\TTtpvec{5-n}}_{\text{for }X = S \text{ only}}
     \hspace{-0.5cm}\,,
 \TTtpvec{\frac{9-n}{2}} \,,
 \dots \,,
 \TTtpvec{1}  \,,
  \TTtpvec{4}  \,,
 \TTtpvec{5}  \,,
 \dots \,
 \TTtpvec{4+k} 
 )^{T [X]}, & \hbox{$k>\frac{n-3}{2}$.}
   \end{array}
 \right.
\label{17XI23.5f}
 \\
 \Theta_k^{[X]} &: =
 \begin{cases}
      \big( \ \TTtpvec{\frac{7-n-2i}{2}} \ \big)_{1\leq i \leq \frac{n-5}{2}} \,, & X = S\,,
      \\
       \big( \ \TTtpvec{\frac{7-n-2i}{2}} \ \big)_{1\leq i \leq \frac{n-3}{2}} \,, & X=V\,;
 \end{cases}
\end{align}
the ``gauge fields'' in \eqref{17XI23.5f} are the same as  in \eqref{13XI23.22}. The case $n=5$ is handled in the same way after performing a trivial elimination of some of the terms from the system (compare the last paragraph of Section \ref{ss5IX23.1}).

 The analysis for solving \eqref{17XI23.3} then proceeds in a similar way to the $\alpha = 0 $ case.
 First, we write \eqref{17XI23.3} in a mode decomposition, which can be done since, as before, all operators appearing in the matrices are sums of products of operators of the form $\tilde L_{a,b,c}$. Next, using the same arguments as   in Appendix \ref{ss24IX23.2}, one can verify that the matrix $\Lambda_k^{[0,X]}$ remains lower triangular and hence $\Lambda_k^{[X]}$ is elliptic in the sense of Agmon, Douglis and Nirenberg. See
 Figure~\ref{f17XI23.1}
\FGp{20XI: Checked the figure}
for a submatrix of $\Lambda_k^{[S]}$ which includes all important additions. Thus, following the same arguments as in the $\alpha=0$ case, there exists $N(k)\in \N$ such that for all modes $\ell\geq N(k)$, the system can be solved by setting $\Theta_{k,\ell}^{[X]}=0=\rho_{k,\ell}^{[X]}$.

\begin{figure}
\centering
\begin{tikzpicture}
  \matrix (m)[
    matrix of math nodes,
    nodes in empty cells,
    minimum width=width("999988"),
    minimum height=8mm,
  ] {
        & \kxi{2}_A  & \upsilon_{k,2} & \TTtpvec{4-n} & \TTtpvec{5-n}  & \TTtpvec{\frac{9-n}{2}}  & \dots     & \TTtpvec{1} & \TTtpvec{4} & \TTtpvec{5}  & \dots   & \TTtpvec{4+k}
        \\
    \overadd{\frac{n-1}{2}} r  & \hlight{\frac{n+1}{2}} & 0 & 1  & \frac{n-1}{2}  &  0 &  &  & & &  & D
    \\
    \overadd{\frac{n-3}{2}}{r} & \frac{n-1}{2}   & \hlight{\frac{n-1}{2} } & 0 & 0  &  &  & & & &  & A_2
    \\
    \overadd{*}{s}
    & 1 & 1 & \hlight{1} & 0 &  &  &  &  &  & & A_1
    \\
    \overadd{*}{r}
    & 2 & 2 & 0 & \hlight{2} &  &  &  &  &  & &
    \\
    \overadd{1}{r}  & 2  &  2 &&&\hlight{2} &    &   & & & &   \\
    \vdots  & \vdots & \vdots && &  & \ddots & & & & & &\\
    \overadd{\frac{n-5}{2}}{r}  &  \frac{n-3}{2}  & \frac{n-3}{2}  & && &  &  \hlight{\frac{n-3}{2}}   & & & &  \\
    \overadd{0}{s}
    & 1 & 1 &  &  & & &  & \hlight{1} & &
    \\
    \overadd{1}{s}  & 2    &  2  &    &   &   && & & \hlight{2} & &  \\
    \vdots  &  \vdots   &  \vdots  &    &   &   &&  && & \ddots &  \\
    \overadd{k}{s}  &  k+1   & k+1   &    &       &   &  & & &&  & \hlight{k+1}\\
    \\
  } ;

    \draw (-7.8,6) -- (7,6);
  \draw (m-1-2.north west) -- (m-12-2.south west);

   \draw[rounded corners] (m-2-12.south east) rectangle (m-2-7.north west);
  \draw[rounded corners] (m-3-12.south east) rectangle (m-3-6.north west);
  \draw[rounded corners] (m-4-4.north west) rectangle (m-12-12.south east) ;

\end{tikzpicture}
\caption{A submatrix of
$\Lambda^{[S]}_k$, $k>\frac{n-1}{2}$, for $n>5$ and $\alpha \neq 0$; see figure~\ref{f10IX23.1} for the meaning of the entries. The same arguments in Appendix \ref{ss24IX23.2} can be used to show that the block $A_1$ is diagonal and that the entries in the blocks labeled ``D'' and ``$A_2$'' are lower than the highlighted entries of the first and second rows respectively. The remaining rows and columns associated to $\protect\overadd{i}{r}\,, 1 \leq i \leq \frac{n-5}{2}$, which are not displayed here, are identical to that in figure~\ref{f10IX23.1}. The $t_j$'s are zero and the $s_i$'s are twice the figures highlighted on the diagonal.
}
\label{f17XI23.1}
\end{figure} 

 Next, to show that, for each $\ell$ with  $\eta_{\ell}\in \CKVp$
  and  $\ell< N(k)$, the system
 admits a solution, we  perform a permutation which now brings $u_k$ and $\upsilon_k$ to (cf. \eqref{17XI23.5}-\eqref{17XI23.6})
\begin{align}
    \hat{u}^{[X]}_k &:= (
                \overadd{*}{s}\,,
                \hspace{-0.6cm}
                \underbrace{\overadd{*}{r}}_{\text{for }X=S\text{ only}}
                \hspace{-0.6cm}\,,
                \overadd{0}{s} \,,
                \overadd{1}{r}\,,
                \overadd{1}{s} \,,
                \overadd{2}{r} \,,
                \overadd{2}{s} \,,
               \dots \,,
               \overadd{k}{r} \,,
               \overadd{k}{s}
           )^{T [X]} \,,
           \label{17XI23.5b}
           \\
     \hat{\upsilon}^{[X]}_{k} &:=
\left\{
  \begin{array}{ll}
        (
    \upsilon_{k,2} \,,
    \TTtpvec{4-n}  \,,
        \hspace{-0.6cm}
    \underbrace{ \TTtpvec{5-n}  }_{\text{for }X=S\text{ only}}
    \hspace{-0.6cm}
    \,,
 \TTtpvec{4}  \,,
 \TTtpvec{\frac{9-n}{2}} \,,
 \dots \,,
 \TTtpvec{1}  \,,
 \TTtpvec{5}  \,,
 \dots \,
 \TTtpvec{4+k}
 )^{T [X]}, & \hbox{$k= \frac{n-3}{2}$;}
\\
 (
 \text{first terms}
     \,,
 \TTtpvec{4-n}  \,,
        \hspace{-0.6cm}
    \underbrace{ \TTtpvec{5-n}  }_{\text{for }X=S\text{ only}}
    \hspace{-0.6cm}
    \,,
 \TTtpvec{4}  \,,
 \TTtpvec{\frac{9-n}{2}} \,,
 \dots \,,
 \TTtpvec{1}  \,,
 \TTtpvec{5}  \,,
 \dots \,
 \TTtpvec{4+k}
 )^{T [X]}, & \hbox{$k> \frac{n-3}{2}$,}
  \end{array}
\right.
\label{17XI23.6b}
\end{align}
where ``first terms'' denotes the fields $
      \kxi{2}_A \,,
     \upsilon_{k,2} \,,
 \partial^{k_n}_u \kxi{2}_A \,,
 \partial^{k_n-1}_u \kxi{2}_A \,,
 \dots \,,
 \partial_u \kxi{2}_A
     $, as in \eqref{17XI23.6}.

Finally, we need to show that the matrix $(A^{[X]}_k\, \Lambda^{\red{[X]}}_k \, \mathcal{N}^{\red{[X]}}_k)$ is surjective or, equivalently, that its adjoint has trivial kernel. As in the $\alpha=0$ case, the matrix $(\Lambda^{\red{[X]}}_k \, \mathcal{N}^{\red{[X]}}_k)^\dagger$ has trivial kernel iff
$( \hat{\Lambda}^{\red{[X]}}_k \, \hat{\mathcal{N}}^{\red{[X]}}_k)^\dagger$
does.
Additionally, it can again be verified that
\begin{equation}
\label{16XI23.4}
 ( \hat{\Lambda}^{\red{[X]}}_k \, \hat{\mathcal{N}}^{\red{[X]}}_k)^\dagger = \begin{pmatrix}
     F^{\red{[X]}} \\ G^{\red{[X]}}
\end{pmatrix} \circ \zdivtwo\circ\, C
 \,,
\end{equation}
where $G^{\red{[X]}}$ can be found in Figure~\ref{F15XII23.1}  for $X=S$ and $k>\frac{n-3}{2}$.   All operators appearing in \eqref{16XI23.4} and Figure~\ref{F15XII23.1} are understood to be their restrictions onto the relevant $X$ space. 
\FGp{20XI}
\begin{figure}
 \begin{align*}
     \left(\begin{array}{l}
            \begin{array}{|ccccccc|}
            \hline
                \fbox{$\overadd{*}{\tilde\chi}_{4-n}$} & 0 & \hspace{0.2cm} 0 & \hspace{1cm} 0 & \hspace{2cm} 0 & \hspace{1cm} \dots \quad \dots  \quad  \ldots   \quad\dots\quad \dots\quad\ldots & \hspace{0.5cm} 0 \hspace{1.2cm}
               \\
              0 &  \fbox{$\overadd{*}{\tilde\psi}_{5-n}$} & \hspace{0.2cm} 0 & \hspace{1cm} 0 & \hspace{2cm} 0 & \hspace{1cm} \dots \quad \ldots   \quad\dots\quad \dots\quad\dots   \quad \ldots  & \hspace{0.5 cm} 0 \hspace{1.2cm}
               \\
            \hline
            \end{array}
                \\
                \vspace{-.4cm}
                \\
             \begin{array}{|ccccccc|}
            \hline
                \\
                \vspace{-1cm}
                \\
                \hspace{0.5cm} 0 & \hspace{1cm } 0 & \hspace{0.5cm} \fbox{$\overadd{0}{\tilde\chi}_4$} & \hspace{0.6cm} \overadd{1}{\tilde\psi}_4 & \hspace{1.5cm} \overadd{1}{\tilde\chi}_4 & \hspace{0.8cm} \dots \quad  \ldots   \quad\dots\quad \dots\quad\dots   \quad \ldots & \hspace{0.3cm} \overadd{k}{\tilde\chi}_4 \hspace{1cm}
               \\
            \hline
            \end{array}
                \\
                \vspace{-.4cm}
                \\
            \begin{array}{|cccccccc|}
            \hline
                 \hspace{0.5cm} 0 & \hspace{1 cm } 0
                 &  \hspace{0.8cm}  0 &
               \overadd{1}{\tilde\psi}_{\frac{9-n}{2}}            &  \overadd{1}{\tilde\chi}_{\frac{9-n}{2}}  & \dots \quad \dots    & \ldots  \quad\dots\quad \dots  \quad \ldots    &  \, \overadd{k}{\tilde\chi}_{\frac{9-n}{2}} \phantom{LL}
               \\
               \hspace{0.5cm} \vdots & \hspace{1 cm } \vdots &
               \hspace{0.8cm} \vdots     &\vdots      &  \vdots  &       &     &  \, \vdots
               \\
                \hspace{0.5cm} 0 & \hspace{1 cm } 0 & \hspace{0.8cm} 0 &
               \fbox{ $\overadd{1}{\tilde\psi}_{\frac{7-n}{2}+(n-1)}$ }
               &  \overadd{1}{\tilde\chi}_{\frac{7-n}{2}+(n-1)}  & \dots \quad \dots    & \ldots   \quad\dots\quad \dots    \quad \ldots & \, \overadd{k}{\tilde\chi}_{{\frac{7-n}{2}+(n-1)}}
               \\
            \hline
            \end{array}
                \\
                \vspace{-.4cm}
        \\
            \begin{array}{|cccccccc|}
            \hline
            \hspace{0.5cm} 0 & \hspace{1 cm } 0 & \hspace{0.8cm} 0 &
              \phantom{\overadd{k}{\tilde\chi}_{,kkk}} \, 0       &   \hspace{1.4cm}  \overadd{1}{\tilde\chi}_{{\frac{9-n}{2}+(n-1)}}    & \quad \dots     & \ldots\quad\dots\quad \quad\dots\quad \ldots & \overadd{k}{\tilde\chi}_{{\frac{9-n}{2}+(n-1)}}
              \\
              \hspace{0.5cm} \vdots & \hspace{1 cm } \vdots & \hspace{0.8cm} \vdots &
              \phantom{\overadd{k}{\tilde\chi}_{,kkk}} \, \vdots        &     \hspace{1.2cm}\vdots     &        &                                                                & \vdots
              \\
              \hspace{0.5cm} 0 & \hspace{1cm } 0 & \hspace{0.8cm} 0 &
              \phantom{\overadd{k}{\tilde\chi}_{,kkk}} \, 0      &   \hspace{1.5cm} \fbox{$\overadd{1}{\tilde\chi}_{(n-1)+4} $}     &   \quad \dots   & \ldots\quad\dots\quad\dots\quad \ldots  & \overadd{k}{\tilde\chi}_{(n-1)+4}
              \\
            \hline
            \end{array}
        \\
            \begin{array}{cc}
                \hspace{8cm}\vdots &
            \end{array}
        \\
            \begin{array}{|ccccccccc|}
            \hline
            \hspace{0.5cm} 0 & \hspace{1 cm } 0 &
              \hspace{0.8cm } 0   & \hspace{1,1cm} 0   &   \hspace{2.3cm}    0        & \hspace{0.2cm} \ldots  \, \dots & \hspace{0.2cm}  0        &\hspace{0.2cm}  \overadd{k}{\tilde\psi}_{(k-1)(n-1)+5} \quad  & \overadd{k}{\tilde\chi}_{(k-1)(n-1)+5}
              \\
             \hspace{0.5cm} \vdots & \hspace{1 cm } \vdots &
             \hspace{0.8cm } \vdots      &   \hspace{1.1 cm} \vdots    &  \hspace{2.3cm} \vdots        & \hspace{0.2cm}  \ddots \, \ddots &  \hspace{0.2cm}  \vdots        &\, \vdots \quad  & \vdots
              \\
              \hspace{0.5cm} 0 & \hspace{1 cm } 0 &
              \hspace{0.8cm } 0   & \hspace{1,1cm} 0   &   \hspace{2.3cm}    0        & \hspace{0.2cm}  \ldots  \, \dots & \hspace{0.2cm} 0 &\hspace{0.2cm} \fbox{$\overadd{k}{\tilde\psi}_{\frac{7-n}{2}+k(n-1)} $}\quad  & \overadd{k}{\tilde\chi}_{\frac{7-n}{2}+k(n-1)}
              \\
            \hline
            \end{array}
        \\ \vspace{-.4cm}
        \\
            \begin{array}{|ccccccccc|}
            \hline
            \hspace{0.5cm} 0 & \hspace{1 cm } 0 &
              \hspace{0.8cm } 0   & \hspace{1,1cm} 0   &   \hspace{2.3cm}    0  & \hspace{0.2cm} \ldots \,\dots & \hspace{0.2cm}  0        &\hspace{1.4cm}  0 \quad  & \hspace{1.2cm} \overadd{k}{\tilde\chi}_{\frac{9-n}{2}+k(n-1)}
              \\
             \hspace{0.5cm} \vdots & \hspace{1 cm } \vdots &
             \hspace{0.8cm } \vdots      &   \hspace{1.1 cm} \vdots    &  \hspace{2.3cm} \vdots        & \hspace{0.2cm}  \ddots \, \ddots & \hspace{0.2cm} \vdots        &\hspace{1.4cm} \vdots \quad  & \hspace{1.2cm} \vdots
              \\
             \hspace{0.5cm} 0 & \hspace{1 cm } 0 &
              \hspace{0.8cm } 0   & \hspace{1,1cm} 0   &   \hspace{2.3cm}    0  & \hspace{0.2cm}  \ldots \,\dots & \hspace{0.2cm}  0 &\hspace{1.4cm}  0 \quad  & \hspace{1.2cm} \fbox{$\overadd{k}{\tilde\chi}_{k(n-1)+4}$}
              \\
            \hline
            \end{array}
            \end{array}
     \right)
 \end{align*}
  \caption{The submatrix $G^{[X]}$ for $k>(n-3)/2$.}\label{F15XII23.1}
\end{figure}
The matrix for $X=V$ can be read from Figure~\ref{F15XII23.1} after deleting the 
 second row and second column there; the matrix for $k=\frac{n-3}{2}$ is obtained after deleting the first row and first column there.
 The new boxed entries in Figure~\ref{F15XII23.1}, as compared to \eqref{5IX23.w1}, are%
 \index{chi@$\overadd{*}{\tilde\chi}_{4-n}$}%
 \index{psi@$\overadd{*}{\tilde\psi}_{5-n}$}%
 \begin{align}
     \overadd{*}{\tilde\chi}_{4-n} & := 1\,,\quad
     \overadd{*}{\tilde\psi}_{5-n} := (\TSzlap - (n-2)\myGauss)|_{\TSzlap \mapsto \lambda_{\ell}} \,,\quad
     \overadd{0}{\tilde\chi}_4  = 1 \,,
 \end{align}
 where the second term is only relevant for $X=S$.
Now, we have 
 $$
  (\TSzlap - (n-2)\myGauss)|_{\TSzlap \mapsto \lambda_{\ell}}\neq 0
  \,,
 $$
 equivalently, $(\TSzlap - (n-2)\myGauss) \eta_{\ell}^{[S]} \neq 0$, as follows from the commutation relation (cf.~\eqref{19X23.1})
 \begin{align}
     &(\TSzlap - (n-2)\myGauss)\zspaceD_A\phi = \zspaceD_A \TSzlap \phi\,, 
 \end{align}
 therefore, if we write $ \eta_{\ell}^{[S]} = \zspaceD_A \phi_\ell$, then 
 \begin{align}
     &(\TSzlap - (n-2)\myGauss)\eta_{\ell}^{[S]} = 0
      \
      \iff
      \
       \TSzlap \phi_\ell = 0
      \
      \iff
      \
       \phi_\ell = \const
      \
      \iff
      \
       \eta_{\ell}^{[S]}=0\,.
 \end{align}
  It readily follows  that $\overadd{*}{\tilde\psi}_{5-n}$ does not vanish for such modes.
Since all the remaining boxed entries of $G^{\red{[X]}} $ are $\ell$-independent non-zero numbers, we conclude that $G^{\red{[X]}}\circ \zdivtwo\circ\, C$ has trivial kernel  on modes with $\eta_{\ell}\in \CKVp$. Furthermore, the regularity of the solution is as given in \eqref{11IV24.0} and \eqref{11IV24.2a} with $p\in (\{n-4,n-5\}\cup[\frac{9-n}{2},1]\cup[4,4+k])\cap\Z$. The proof of this follows analogously to that of Corollary \ref{C10V24} and is left to the readers.
 
 Finally, let us consider the $\TTt$ projection of the fields $\overadd{p}{r}_{AB}$, $1\leq p \leq k$. The transport equations for these fields are coupled amongst themselves, but are not coupled with any other fields. As a result, the analysis is very similar to that in the proof of Theorem~\ref{t22VIII23.1}. Indeed, the equations we are looking at is the $\TTt$ projection of \eqref{22VIII23.2}:
\begin{align}
    & \underbrace{(\overadd{p}{q}_{AB} \big|_{\secN_2} -\overadd{p}{ q}_{AB} \big|_{\secN_1} + \text{known fields})^{[\TTt]}}_{=: \overadd{p}{r}{}^{[\TTt]}_{AB}}
    \nonumber
    \\
    &=    \overadd{p}{\psi}\ofP \, \kphit{AB}{(7-n+2p)/2}^{[\TTt]}
    + \alpha^{2p} \overadd{p}{\psi}_{[\alpha]} \kphit{AB}{(7-n-2p)/2}{}^{[\TTt]}
    + m^p \overadd{p}{\psi}_{[m]} \kphit{AB}{(7-n+2p(n-1))/2}{}^{[\TTt]}
    \nonumber
    \\
    & \quad
    + \sum_{j,\ell}^{p_{**}}  m^{j} \alpha^{2\ell} \overadd{p}{\psi}_{j,\ell}\ofP \, \kphit{AB}{p - \frac{n-7}{2} +  j (n-2) - 2 \ell}{}^{[\TTt]}
    \,,
    \quad  1\leq p\leq k\,;
    \label{22XI23.2}
\end{align}
recall that $\overadd{p}{\psi}_{[\alpha]} = 0$ for $p>\frac{n-1}{2}$, and that ``known fields'' refer to the contribution from the field $\interph$. We then have the following:

 \begin{theorem}
 \label{t22XI23.1} 
 Let the pair $(n,k)$ be inconvenient.
The system  \eqref{22XI23.2} can be solved by a choice of interpolating fields
$$
 \kphi{j}{}^{[\TTt]}_{AB} \in
  H^{\kgamma}(\secN)
  \,, \quad   j \in \left[\max\{\frac{7-n-2k}{2},4-n\},
\frac{7-n+2k(n-1)}{2}\right]$$
for any finite $k$. Its solutions are determined by an elliptic system, uniquely up to possibly a finite number of  eigenfunctions of the Laplacian  acting on tensors.
\end{theorem}

The proof of this theorem proceeds in a way identical to that of Theorem~\ref{t22VIII23.1} with the following minor modifications:

\begin{enumerate}
    \item[1.] All two-covariant tensor fields appearing there are replaced with their $\TTt$ projections.
    \item[2.]
    The coefficients $\overadd{p}{\tilde\psi}_{\frac{7-n-2p}{2}} $  are still  equal to $ \alpha^{2p}\overadd{p}{\psi}_{[\alpha]}$ (cf.~\eqref{15XII23.21}), and
    are non-zero for $p\le\frac{n+1}{2}$. Now they vanish  for $p > \frac{n+1}{2}$, so that the matrix $A_k$ is no longer surjective. This does not affect the solvability of the system, but only removes  the  option   of setting $\Omega_k = 0 =\Xi_k$ .
    \item[3.]
    Finally, regularity follows from ellipticity of $\overadd{p}{\psi}\ofP$ acting on $\TTt$.
\end{enumerate}

\ptclater{KerrdeSitter.tex and two following commented out, replaced by the version from January}
\section{Linearized Kerr-deSitter in Bondi Gauge}\label{Kerr-dS appendix}
The full Kerr-(A)dS metric in $(n+1)$-spacetime dimensions is given by
\cite{Gibbons:2004js}:%
\footnote{To facilitate the comparison with that reference we  note that $a_i$ is the negative of the one there, and recall that $a_{N+\delta}=0$ for $\delta=1$.}
\begin{align}\label{eq: Kerr AdS}
ds^2=&-W(1-\alpha^2r^2)dt^2+\frac{2m}{U}\Bigl(Wdt+\sum_{i=1}^\red{N} \frac{a_i \mu_i^2 d\phi_i}{\Xi_i}\Bigr)^2\nonumber\\
&+\sum_{i=1}^\red{N}\frac{r^2+a_i^2}{\Xi_i}\bigl(\mu_i^2d \phi_i^2+d\mu_i^2)
+\frac{Udr^2}{V-2m}+\delta r^2 d\mu_{\red{N}+\delta}^2\nonumber\\
&+\frac{\alpha^2}{W(1-\alpha^2r^2)}\Bigl(\sum_{i=1}^\red{N} \frac{r^2+a_i^2}{\Xi_i}\mu_i d\mu_i+\delta r^2 \mu_{\red{N}+\delta} d\mu_{\red{N}+\delta}\Bigr)^2\,,
\end{align}
where
\begin{align}
W&=\sum_{i=1}^{\red{N}} \frac{\mu_i^2}{\Xi_i}+\delta \mu_{\red{N}+\delta}^2\,,\quad V=r^{\delta-2}(1-\alpha^2r^2)\prod_{i=1}^\red{N}(r^2+a_i^2)\,,\nonumber\\
U&=\frac{V}{1-\alpha^2r^2}\Bigl(1-\sum_{i=1}^\red{N}\frac{a_i^2\mu_i^2}{r^2+a_i^2}\Bigr)\,,\quad \Xi_i=1+\alpha^2 a_i^2\,.
\end{align}
Here, $\delta=1, 0$ for even, odd dimensions, $N=\bigl[\frac{n}{2}\bigr]$ (where $[A]$ denotes the integer part of $A$), and the coordinates $\mu_i$ obey a constraint
\begin{equation}
\sum_{i=1}^{N+\delta} \mu_i^2=1\,.
\end{equation}
In these coordinates there is no rotation at spatial infinity and the angles $\phi_i$ are periodic with period $2\pi$.

We linearise the metric by expanding to linear order in the various rotation parameters $a_i$ which gives the higher dimensional Lense--Thirring metric
\begin{align}\label{eq: LTHD}
ds^2_{LT}&=-fdt^2+\frac{dr^2}{f}+r^2\left(\sum_{i=1}^\red{N} \mu_i^2d\phi_i^2+\sum_{i=1}^{\red{N}+\delta}\!d\mu_i^2 \right)
+\frac{4m}{r^{n-2}}\sum_{i=1}^\red{N} a_i\mu_i^2dtd\phi_i
 +\red{O}(a_i^2)
 \,,
\\
  f&=\!1-\alpha^2r^2-\frac{2m}{r^{n-2}}+\red{O}(a_i^2)\,.
\end{align}
We can rewrite this as a perturbation on the Birmingham--Kottler metrics\footnote{At least for spherical topologies.}
\eqref{23VII22.3} by noting that the term in the brackets of \eqref{eq: LTHD} is just the metric on the $(n-1)$-sphere $\ringh_{AB}dx^A dx^B$ and by performing the coordinate transformation
\begin{equation}
	dt=du+\frac{dr}{f}\,.
\end{equation}
Then we have,
\begin{equation}\label{21V24f.1}
	ds^2_{LT}={\nobarzg}_{\a \b} dx^\a dx^\b +\frac{4m}{r^{n-2}}\sum_{i=1}^\red{N} a_i\mu_i^2d\phi_i(du+dr/f)+\red{O}(a_i^2)\,.
\end{equation}

This is manifestly not in Bondi gauge as there are $g_{rA}$ components, however we can address this in two ways which turn out to be equivalent. First, outside of all horizons,  we could accompany the change from the $t$- to the $u$-coordinate with   $r$-dependent rotations of the  $(\mu_i,\phi_i)$ planes,
\begin{equation}\label{eq: Coord change}
	d\varphi_i=d\phi_i-\frac{2ma_i}{r^{n} f}dr
\,,
\end{equation}
which transforms the metric to
\begin{equation}\label{eq: Lin Kerr Bondi}
ds^2_{LT}={\nobarzg}_{\a \b} dx^\a dx^\b +\frac{4m}{r^{n-2}}\sum_{i=1}^\red{N} a_i\mu_i^2d\varphi_idu+\red{O}(a_i^2)\,.
\end{equation}%

Second we could perform an infinitesimal gauge transformation generated by $\zeta^\mu$ to put the perturbation into Bondi form (\`a la Section 3 of~\cite{ChHMS}). Denoting the perturbation in \eqref{21V24f.1} by $h_{\mu\nu}$, the vector $\zeta$ satisfies
\begin{align}
	{\cal L}_\zeta (g_{rr} +h_{rr}) &=0\,,\\
	{\cal L}_\zeta (g_{rA} +h_{rA}) &=0\,,\\
	g^{AB}{\cal L}_\zeta (g_{AB} +h_{AB}) &=0\,.
\end{align}
Given that the nonzero components of $h_{\mu\nu}$ are
\begin{align}
h_{rA}=\frac{2ma_i}{r^{n-2}f} \ringh_{AB}\delta^B_{\phi_i}\,,\quad
h_{uA}=\frac{2m a_i}{r^{n-2}} \ringh_{AB}\delta^B_{\phi_i}
 \,.
\end{align}
The first two conditions are equivalent to
\begin{align}
\partial_r \zeta^u-\frac{1}{2}h_{rr} &=0\,,\\
\partial_r \zeta^A-\frac{\ringh^{AB}}{r^2}\left(\partial_B \zeta^u-h_{rB}\right)&=0\,,
\end{align}
and one can check that~\cite[Section~3]{ChHMS}
\begin{equation}
	\zeta^r=-\frac{1}{n-1}\left(r\zspaceD_B\zeta^B+\frac{1}{2r}\ringh^{AB}h_{AB}\right)\,.
\end{equation}

Given that $h_{rr}=0=h_{AB}$ we can choose $\zeta^u=0$, and then the only non zero components of the vector are
\begin{equation}
	\zeta^A=2ma_i\delta^A_{\phi_i}\int_{r_0}^{r}\frac{1}{s^nf}ds\,,
\end{equation}
which puts the metric exactly into \eqref{eq: Lin Kerr Bondi}. Using that $\zeta^A$ generates the infinitesimal transformation $x^A\to x^A+\xi^A$ one can check that this is equivalent to \eqref{eq: Coord change}.

Finally, identifying the angular momentum as  $J_i=m a_i$,
 clearly the linearised Kerr--(A)dS spacetime adds the following perturbations
\begin{equation}
	\delta V=0=\delta \beta\,,\quad \delta U^A=-\frac{2J_i }{r^{n}}\,\delta^A_{\phi_i} \,.
\end{equation}

{In this form the (linearised) solution can easily be upgraded to  
Ricci-flat $(\secN,\zgamma)$ 
by replacing $f$ with $\zguu:=
-\alpha^2r^2-{2m}{r^{-(n-2)}}$, the metric on the sphere with the Ricci-flat metric on the transverse manifold $\secN$, and setting
\begin{equation}
\delta U^A=\frac{\zlambda^A(x^C) }{r^{n}} \,,
\end{equation}
where $\zlambda_A(x^C)$ is an arbitrary Killing vector of the transverse space~$\secN$. One can then check that the linearised Einstein equations in Section~\ref{s3X22.1} continue to hold for such a perturbation.
}

\section{Some commutation relations}
 \label{App19V23.1}

Let $h_{AB}$ be a 2-covariant, symmetric trace-free tensor.
For the convenience of the reader we collect here several identities, some of which trivial, which are repeatedly used in the main body of the paper. Recall that $n=d+1$, where $d$ is the dimension of the manifold carrying the Einstein metric $\ringh_{AB}$.

We have:
\begin{align}
   & (i)
   & \zspaceD_B \TSzlap \xi^u
   &= \TSzlap \zspaceD_B \xi^u - (n-2)\myGauss \zspaceD_B\xi^u
   \label{19X23.1}
\\
   & (ii) \,
   & \TS[\zspaceD_A \TSzlap \xi_B]
   &=
    \TS[\TSzlap \zspaceD_A\xi_B- (n-2)\myGauss \zspaceD_A \xi_B + 2
     \tric (\zspaceD \xi)_{AB} ]   \label{28VIII23.f3}
\\
   & (iii) \,
   & \zspaceD^B \TSzlap \xi_B
   &=  \TSzlap \zspaceD^A\xi_A + (n-2)\myGauss \zspaceD^A \xi_A \label{14VII23.f1}
\\
   &(iv)\,
   & \zspaceD^B \TSzlap h_{AB}
  &=
   \TSzlap \zspaceD^B h_{AB}
   - 2   \zspaceD^B \tric h_{AB}
   + (n-2)\myGauss \zspaceD^B h_{AB}
    \label{18V23.1}
\\
   &(v)\label{27XI23.f1}
   & \zspaceD^A C(\xi)_{AB}
   &=
   \frac{1}{2} \bigg(
   (\TSzlap + (n-2)\myGauss) \xi_B
   + \frac{n-3}{n-1}\zspaceD_B \zspaceD_C\xi^C
   \bigg)
   \\
    &(vi)\label{27XI23.f2}
    & \zspaceD^B \TS[\zspaceD_A\zspaceD_B\zspaceD^C\xi_C]
    &= \frac{n-2}{n-1}\zspaceD_{\tdA} \, [\, \TSzlap  + (n-1)\twoscsign ] \zspaceD_C\xi^C
\\
 &(vii)\label{20VI23.1}
&\zspaceD_AD^C \TSzlap h_{B C}
&=
\TSzlap \zspaceD_AD^C h_{B C}\nonumber
\\
&&&+2\left( \zR^E{}_A{}^F{}_B P(h)_{EF}
  -D_A\zspaceD^C \tric (h)_{B C}+\frac{2(n-2)}{(n-1)} \myGauss\zspaceD^ED^Fh_{E F}\ringh_{AB}\right)
\end{align}
Note, this last equation implies the commutation of the operators $ \ck{{k}}{\ofPnoP}$.
 Indeed, \eqref{3III23.5} shows 
   that we can write $\ck{{k}}{\ofPnoP}=A_k( B_kP+\tric + \frac{1}{2} \TSzlap+C_k)$ for some  numbers  $A_k$, $B_k$ and $C_k$. 
   Expanding the commutator gives
\begin{align}\label{20VI23.12}
	\left[\ck{{k}}{\ofPnoP},\ck{{k'}}{\ofPnoP}\right](h)=A_kA_{k'}(B_k-B_{k'})[P,\tric + \frac{1}{2} \TSzlap](h)=0\,,
\end{align}
where the second equality follows from applying the $\TS$ operator to \eqref{20VI23.1}.

\section{Further commutation relations}
 \label{App8XII23.1}

In this section we prove some further useful  commutation relations.
We assume throughout that the metric $\ringh$ is Einstein.

In the equations we often (implicitly) have the appearance of the following operator mapping symmetric trace-free tensors to themselves%
\index{Delta@$\zTSlap_T$}%
\begin{equation}
	\zTSlap_T h:=\left[\TSzlap+2(\tric-(n-2)\myGauss)\right]h \,.
\label{8XII23.g5}
\end{equation}
This is the negative of the Lichnerowicz Laplacian.
 Likewise, we have the ``natural'' Laplacians acting on vectors $V$ and functions $
 \varphi$
\begin{align}\label{18VI24.f1}
	\zTSlap_V V&:= \left(\TSzlap -(n-2)\myGauss\right)V\,,
	\\
	\label{18VI24.f2}
	\zTSlap_S\varphi&:=\TSzlap\varphi\,,
\end{align}
which are ``natural'' because these are simply the (negative of) the Hodge Laplacian acting on one-forms and \red{functions}.

Now these are particularly useful in our context because they satisfy the following commutation properties:
\begin{align}
	\zTSlap_T \TS[\zspaceD_{A} V_{B}] &= \TS[\zspaceD_{A} \zTSlap_V V_{B}]\,,
	\\
	\zTSlap_V(\zspaceD_{A}\varphi) &= \zspaceD_{A} (\zTSlap_S \varphi)\,.
\end{align}
The first of these follows from \eqref{28VIII23.f3} and \eqref{16V22.1bxc}, and  the second from \eqref{19X23.1}. Together, in particular, they imply
\begin{align}
	\zTSlap_T \TS[\zspaceD_{A}\zspaceD _{B}\varphi] &=  \TS[\zspaceD_{A}\zspaceD _{B}\zTSlap_S\varphi]\,.
\end{align}
It is a further useful result that
\begin{align}
	\zdivtwo\circ\zTSlap_T &=\zTSlap_V\circ \zdivtwo\,,
	\\
	\zdivone\circ\zTSlap_V &=\zTSlap_S\circ \zdivone\,,	
\end{align}
which follow from \eqref{18V23.1} and \eqref{14VII23.f1}. Similarly this implies
\begin{equation}
	\zdivone\circ\zdivtwo\circ\zTSlap_T =\zTSlap_S\circ \zdivone\circ\zdivtwo \,.
\end{equation}

Now we can do a similar thing with the operator $P$. Defining%
\index{P@$P$!$P_T$}%
\begin{equation}
	P_T\equiv P = C\circ \zdivtwo
\end{equation}
to act specifically on symmetric trace-free two tensors,  we can try to define operators $P_V$ and $P_S$ such that,
\begin{align}
\zdivtwo\circ P_T&\equiv P_V\circ\zdivtwo\,,\label{28XI23.f1}
\\
\zdivone\circ P_V&\equiv P_S\circ\zdivone\,.\label{28XI23.f2}
\end{align}
We have already seen (and it is obvious) that%
\index{P@$P$!$P_V$}%
\begin{equation}
	P_V := \zdivtwo\circ\, C
 \label{12XII23.21}
\end{equation}
works.
Moreover, \eqref{27XI23.f1} and \eqref{14VII23.f1} together with \eqref{28XI23.f2} (see also \eqref{10XI23.1}) imply%
\index{P@$P$!$P_S$}%
\begin{equation}
	P_S\equiv P_S(\zTSlap_S):=\frac{n-2}{n-1}\left(\zTSlap_S +(n-1)\myGauss\right)\,.
\end{equation}

Finally we can factor-in the decomposition $\xi=\xi^{[S]}+\xi^{[V]}$ (cf. \eqref{4X23.1}), and act with $P_V$  on these parts separately. This gives,
\begin{align}
	P_V(\xi^{[S]}+\xi^{[V]}) &=	P_V(\xi^{[S]}) +P_V(\xi^{[V]})
	\\
	&=\underbrace{\frac{n-2}{n-1}\left(\zTSlap_V	
 	  +\myGauss\right)}_{=: P_V^{[S]}(\zTSlap_V)}\xi^{[S]}
	+ \underbrace{\frac{1}{2}\left(\zTSlap_V +(n-2)\myGauss\right)
        }_{=: P_V^{[V]}(\zTSlap_V)}\xi^{[V]}
         \,.
\end{align}
Here we have used  \eqref{27XI23.f1} and  \eqref{4X23.1}  to obtain the second line.

Notice that $P_V$ restricted to acting on $\xi^{[S]}$
 is just $P_S$ promoted to act on vectors by replacing $\zTSlap_S$ with $\zTSlap_V$. We can do the same trick to calculate the action of $P_T$ on the scalar, vector, and tensor (transverse) parts of $h=h^{[S]}+h^{[V]}+h^{[\TTt]}$ (cf. \eqref{10VI23.3}):%
\index{P@$P$!$P_T$}%
\begin{align}
	P_T (h^{[S]}+h^{[V]}+h^{[\TTt]})&=P_T(h^{[S]})+P_T(h^{[V]})+P_T(h^{[\TTt]})
	\\
	& =:
         P^{[S]}_T(h^{[S]}) +P^{[V]}_T(h^{[V]})+P^{[\TTt]}_T(h^{[\TTt]})\,.
\end{align}
Using \eqref{28VIII23.f4} and \eqref{28VIII23.f5} we find
\begin{align}
	P^{[S]}_T&=P_S(\zTSlap_T)\,,\quad P^{[V]}_T=P^{[V]}_V(\zTSlap_T)\,,\quad P^{[\TTt]}_T=0\,.
\end{align}

Now, since each of these operators depends only on the ``natural'' Laplacians we find the following commutation relations:
\begin{align}
P_T \left(\TS[\zspaceD_{A} V_{B}]\right) &= \TS[\zspaceD_{A} P_V V_{B}]\,,
\\
P_V(\zspaceD_{A}\varphi) &= \zspaceD_{A} (P_S \varphi)\,,
\\
P_T \left(\TS[\zspaceD_{A}\zspaceD _{B}\varphi]\right) &=  \TS[\zspaceD_{A}\zspaceD _{B}P_S\varphi]\,.
\end{align}

Consider, next,  the   operator $ \ck{k}{\ofPnoP}$ of \eqref{3III23.5a}. Assuming as usual that the metric $\ringh$ is Einstein, rewritten in terms of the Lichnerowicz Laplacian $\zTSlap_T$  we find%
\index{K@${\cal K}_T(k)$}%
\begin{align}
	\ck{k}{\ofPnoP}
  & =
  -\frac{1}{7 -  n + 2 k} \bigg[\frac{2 (n - 1) P_T}{(3 + k) (3 -  n + k) }
	+  \zTSlap_T +(4+k(6-n+k)) \myGauss \bigg]
 \nonumber
\\
 & =:{\cal K}_T(k)
  \,.
 \label{8XII2.g2}
\end{align}
We have:%
\index{K@${\cal K}_V(k)$}%
\index{K@${\cal K}_S(k)$}%
\begin{align}
\zdivtwo\circ	{\cal K}_T(k) &=	{\cal K}_V(k)\circ \zdivtwo\,,
\\
\zdivone\circ	{\cal K}_V(k) &=	{\cal K}_S(k) \circ \zdivone\,,	
\\
\zdivone\circ\zdivtwo\circ	{\cal K}_T(k) &=	{\cal K}_S(k) \circ \zdivone\circ\zdivtwo \,.
\end{align}
where the subscript ${\cal K}_X(k)$ indicates that operator is a function of the corresponding $\zTSlap_X$ and $P_X$ and maps $X\to X$.
 Moreover these operators will respect the scalar, vector, and tensor decomposition in that
\begin{align}
{\cal K}_T(k) \TS[\zspaceD_{A} \xi_{B}] &= \TS[\zspaceD_{A} {\cal K}_V(k) \xi_{B}]\,,
\\
{\cal K}_V(k)(\zspaceD_{A}\varphi) &= \zspaceD_{A} ({\cal K}_S(k) \varphi)\,,\\
{\cal K}_T(k) \TS[\zspaceD_{A}\zspaceD _{B}\varphi] &=  \TS[\zspaceD_{A}\zspaceD _{B}{\cal K}_S(k)\varphi]\,.
\end{align}
Thus we can write the explicit expressions for ${\cal K}_T(k)$ acting on
$h^{[X]}$:%
\index{K@${\cal K}_T(k)$!${\cal K}^{[\TTt]}_T(k)$}%
\index{K@${\cal K}_T(k)$!${\cal K}^{[V]}_T(k)$}%
\index{K@${\cal K}_T(k)$!${\cal K}^{[S]}_T(k)$}%
\begin{align}
{\cal K}^{[S]}_T(k)&
  := -\frac{(1+k)(5-n + k)}{(7 - n + 2 k)(3-n+k)(3+k)}\left( \zTSlap_T +(2+k)(4-n+k)\myGauss\right)\,,
  \label{8XII23.g1-}
\\
{\cal K}^{[V]}_T(k)&
  := -\frac{(2+k)(4-n + k)}{(7 - n + 2 k)(3-n+k)(3+k)}\left( \zTSlap_T +(1+k)(5-n+k)\myGauss\right)\,,
  \label{8XII23.g3}
\\
 {\cal K}^{[\TTt]}_T(k)&
   :=-\frac{1}{(7 - n + 2 k)}\left( \zTSlap_T+ [4+k(6-n+k)]\myGauss\right)
 \,.
  \label{8XII23.g1}
\end{align}
 The expressions for ${\cal K}^{[X]}_{V,S}$ acting on scalars $\varphi$ and vectors $\xi$ follow readily by replacing the Lichnerowicz Laplacian with the scalar or vector Laplacian $\zTSlap_S$ or $\zTSlap_V$ from \eqref{18VI24.f1} and \eqref{18VI24.f2} respectively.
In this form it is clear that whenever the prefactor is defined,
the operators are  elliptic   since their principal part is the Laplacian.

Note that $k=-3$ is not allowed in ${\cal K}_V(k)$ and ${\cal K}_S(k)$, but $k=-3$ does not appear when trying to set-up a recurrence for the vector and scalar parts of  the transport equations \eqref{6III23.w6}.
Next, one  checks that ${\cal K}^{[\TTt]}_T(-3)$,   coincides with the expression for the special case $k=-3$ of \eqref{18IV23.1} restricted to act on $h^{[\TTt]}$. It follows that one can use  these transport equations for each field $ \overadd{i}{q}{}_{AB}^{[X]}$ independently, and that the operators appearing in the recursions are only those in \eqref{8XII23.g1-}-\eqref{8XII23.g1}.

We finally note that \eqref{8XII23.g1-}-\eqref{8XII23.g1} can be grouped together. If $s=0,1,2$ for  scalars, vectors, and $\TTt$ tensors respectively, then
\begin{multline}\label{28XI23.f3}
	{\cal K}^{[X_s]}_T(k)\equiv -\frac{(1+s+k)(5-s-n + k)}{(7 - n + 2 k)(3-n+k)(3+k)}\\
	\times\left( \zTSlap_T +[s^2(s-1)+(2-s+k)(4+s-n+k)]\myGauss\right)\,.
\end{multline}
Clearly, when $s=0,1$ (scalars and vectors) these operators vanish on the subspaces when $k=-(s+1)$. This implies that the scalar and the vector parts drop out of the recursion relations for the $ \overadd{i}{q}_{AB}$'s.
 \ptcheck{XII23: operators checked with Finn's mathematica file}

%
%
%
\newpage
\printindex

\newpage
\bibliographystyle{JHEP}
\bibliography{0HDNullGluing-minimalCleaned,references}

\end{document}